\documentclass[12pt]{article}

\usepackage{amsfonts,amssymb,amsmath,amsthm,epsfig,euscript,bbm,amsthm}
\usepackage{algorithm}
\usepackage{algpseudocode}
\usepackage{amsthm}
\usepackage[margin=1in]{geometry}
\usepackage{setspace}
  
\usepackage{amsfonts,amssymb,amsmath,amsthm,epsfig,euscript,bbm,amsthm}
\usepackage[usenames,dvipsnames]{xcolor}
\usepackage{algorithm}
\usepackage[dvipsnames]{xcolor}
\usepackage{algpseudocode}
\usepackage{amsthm}
\usepackage{threeparttable,subcaption}
 \usepackage{booktabs}
\newcommand{\ra}[1]{\renewcommand{\arraystretch}{#1}}
 \usepackage{titletoc}
\usepackage{rotating,multirow,makecell,array}
\usepackage{enumitem}
\usepackage{rotating}
\usepackage{lscape}
\usepackage{float}
\usepackage{graphicx}
\usepackage{subcaption}
\usepackage{tikz}
\usetikzlibrary{arrows,calc}
\usepackage{titletoc}
 
\usepackage{rotating,multirow,makecell,array}
\usepackage{bm}
\usepackage{tikz}
\usetikzlibrary{arrows}
\usepackage {graphicx}
\usepackage{enumerate}
\definecolor{darkbrown}{rgb}{0.4, 0.26, 0.13}

\usepackage{subcaption}
\usepackage{caption}
\usepackage[authoryear]{natbib}
\usepackage[titletoc,title]{appendix}
\usepackage[pdftex,colorlinks]{hyperref}
\usepackage{mathtools}
 \usepackage{xcolor}
\usepackage{xr}
\usepackage{algorithm,algpseudocode}

\newcounter{algsubstate}
\renewcommand{\thealgsubstate}{\alph{algsubstate}}
\newenvironment{algsubstates}
  {\setcounter{algsubstate}{0}%
   \renewcommand{\State}{%
     \stepcounter{algsubstate}%
     \Statex {\footnotesize\thealgsubstate:}\space}}
  {}

\usepackage{graphicx}
\usepackage{subcaption}
 \usetikzlibrary{graphs, graphs.standard}
 \definecolor{orange}{RGB}{230,170,120}
  \definecolor{green}{RGB}{120,200,120}

 \hypersetup{
 	colorlinks=true,
 	linkcolor=blue,
 	filecolor=blue,      
 	urlcolor=cyan,
 }
 \hypersetup{colorlinks,linkcolor={blue},citecolor={blue}} 
 \usepackage{geometry}                
 \usepackage{multirow}
\usepackage[section]{placeins}
 \usepackage{bbm}
 \usepackage{graphicx}
 \usepackage{amssymb}
 \usepackage{epstopdf}
 \usepackage{tikz} 
 \usepackage{amsmath} 
\usepackage[utf8]{inputenc}
\usepackage{algorithm}
\usetikzlibrary{calc}
\usepackage{setspace}
 \usepackage{epstopdf}
\usetikzlibrary{arrows, chains, calc, positioning, shapes.multipart}

 \newtheoremstyle{mytheoremstyle}
  {2pt}      
  {3pt}      
  {\itshape} 
  {}         
  {\bfseries}
  {}         
  {3pt}      
  {}         

\usepackage{titlesec}

\titlespacing{\section}
  {0pt}   
  {1pt}   
  {1pt}   

\titlespacing{\subsection}
  {0pt}
  {4pt}
  {2pt}

\theoremstyle{mytheoremstyle}

 \newtheorem{thm}{Theorem}[section]
 \newtheorem{lem}[thm]{Lemma}
 \newtheorem{prop}[thm]{Proposition}
 \newtheorem{cor}{Corollary}
 
 \theoremstyle{definition}
 \newtheorem{defn}{Definition}[section]
 
 \newtheorem{exmp}{Example}[section]
  \newtheorem{ass}{Assumption}[section]
  \usepackage{multibib}
 \theoremstyle{definition}
 \newtheorem{rem}{Remark}
 
 \makeatletter
 \def\BState{\State\hskip-\ALG@thistlm}
 \makeatother
 
 \usepackage{authblk}
\def\spacingset#1{\renewcommand{\baselinestretch}%
{#1}\small\normalsize} \spacingset{1}
 \usepackage{xr}
\DeclarePairedDelimiter{\nint}\lfloor\rceil
\definecolor{coquelicot}{rgb}{1.0, 0.22, 0.0}

\algnewcommand\algorithmicforeach{\textbf{for each}}
\algdef{S}[FOR]{ForEach}[1]{\algorithmicforeach\ #1\ \algorithmicdo}
 
 \newcommand{\blind}{0}

\begin{document}
  
\if0\blind
{  
  \title{\vspace{-13mm} \bf Policy design in experiments with unknown interference\footnote{We are especially grateful to Graham Elliott, James Fowler, Paul Niehaus, Yixiao Sun, and Kaspar W\"uthrich for their continuous advice and support, and Karun Adusumilli, Isaiah Andrews, Tim Armstrong, Peter Aronow, Susan Athey,  Emily Breza, Ivan Canay, Raj Chetty, Arun Chandrasekhar, Tim Christensen, Aureo de Paula, Tjeerd de Vries, Matt Goldman, Peter Hull, Guido Imbens, Brian Karrer, Max Kasy, Larry Katz, Michal Kolesar, Toru Kitagawa, Pat Kline, Michael Kremer, Amanda Kowalski, Michael Leung, Xinwei Ma, Elena Manresa, Craig Mcintosh, Konrad Menzel, Karthik Muralidharan, Ariel Pakes,  Jack Porter, Max Tabord-Meehan, Ulrich Mueller, Gautam Rao, Jonathan Roth, Cyrus Samii, Pedro Sant'Anna, Fredrik S\"avje, Azeem Shaikh, Jesse Shapiro, Pietro Spini, Elie Tamer, Alex Tetenov, Ye Wang, Andrei Zeleneev, the editor and referees, and participants at numerous seminars for helpful comments. We particularly thank the teams at Precision Development and the Center for Economic Research in Pakistan, especially Jagori Chatterjee, Khushbakht Jamal, and Adeel Shafqat.  The experiment is registered at \url{https://www.socialscienceregistry.org/trials/9945}. The experiment has received IRB approval from Stanford University, where DV was affiliated as a postdoctoral fellow during the data collection, and the University of Chicago. We thank the Harvard Griffin Fund in Economic Research, the Harvard Chae Family Economics Research Fund, the International Fund for Agricultural Development (IFAD), the NBER Social Learning Fund, and the Development Innovation Lab for generous funding support. Helena Franco, Benjamin Zeisberg, Angelina Zhang provided excellent research assistance.
  All mistakes are our own.}}
\author{Davide Viviano\footnote{Department of Economics, Harvard University, 1805 Cambridge St, Cambridge, MA 02138.  Email: dviviano@fas.harvard.edu.} $\quad$ Jess Rudder\footnote{Department of Economics, Oregon State University, 1500 SW Jefferson Way, Corvallis, OR 97331. Email: jessica.rudder@oregonstate.edu} }
    \date{First Version: November, 2020; \\ This Version: October, 2026
   }
    
  \maketitle
} \fi

\if1\blind
{
  \bigskip
  \bigskip
  \bigskip
  \ 
 \\
 
 \ 
 \\

  \medskip
} \fi

\vspace{-5mm} 
 \begin{abstract}
\noindent  

This paper studies experimental designs for estimation and inference on policies with spillover effects. Units are organized into a finite number of large clusters and interact in unknown ways within each cluster. First, we introduce a single-wave experiment that estimates the marginal effect of a change in treatment probabilities, taking spillover effects into account. Using the marginal effect, we propose a policy optimality test and improve the population’s average outcome through marginal changes. Second, we design a multiple-wave experiment to estimate optimal treatment rules, with strong regret guarantees. We provide an implementation in a large-scale field experiment studying climate adaptation.



\end{abstract}

\noindent%
{\it Keywords:} Experimental Design, Spillovers, Welfare Maximization, Causal Inference. \\
{\it JEL Codes:} C31, C54, C90.
\vfill

\newpage
 \section{Introduction} 
\spacingset{1.5} 

One of the goals of a government or NGO is to use experimental evidence to choose policies that improve outcomes in the target population. Interference makes this problem challenging: treating one individual may also affect others. As a result, treatment-effect contrasts are not, on their own, sufficient for welfare maximization.\footnote{Examples include the direct effect of treatment and the overall effect, i.e., the effect of treating all individuals compared with treating none. For welfare maximization, neither estimand is sufficient in general. The direct effect ignores spillovers, whereas the optimal rule may treat some but not all individuals because of treatment costs or constraints.} For example, in an information campaign, information may be costly and may generate large direct effects for individuals in remote areas, while generating smaller spillovers.  Measuring spillovers is also difficult: when they operate through an unobserved network, collecting network data is often costly or infeasible \citep[see][]{breza2017using}. In response, a common approach to measure spillovers is to vary treatment probabilities between many clusters, assumed to be independent \citep{baird2018optimal}. These designs are natural when interference is contained within those (small but many) clusters, but are less appropriate when spillovers occur between such units, such as villages or neighborhoods, a common concern in practice \citep[e.g.,][]{egger2019general}.

This paper studies experimental designs with interference in a complementary regime where we have limited information about interference, formalized by assuming a small (finite) number of large clusters, such as regions or districts. Within each cluster, units may interact in unknown ways; across clusters, spillovers are assumed to be negligible.\footnote{A finite number of large clusters allows researchers to remain agnostic about spillovers within each region and only requires approximate independence across regions. Namely, the number of individuals who interact across regions is small relative to the number of individuals in a region \citep[][]{leung2023network}.} The key restriction is that these few but large clusters are comparable experimental environments. In the baseline model, if two clusters were exposed to the same treatment rule, their average outcomes would be the same in expectation, up to additive cluster and time effects. This restriction allows realized networks and individual outcomes to differ across clusters, but rules out arbitrary cluster-specific responses to the policy through unobserved heterogeneity. 

With a few clusters, a single-wave experiment generally cannot recover the full welfare function over the policy space. In addition, practical or political constraints may make large departures from a baseline policy infeasible, especially when they require implementing substantially different treatment rules across regions. We therefore focus on local policy changes around a baseline assignment rule, bringing the local-perturbation logic familiar from the sufficient-statistics literature in public finance and macroeconomics \citep{barnichon2023sufficient, chetty2009sufficient} to experimental design. The idea is simple: assign two comparable clusters nearby values of the policy parameter (possibly matching on covariates), inducing nearby treatment probabilities, and compare their average outcomes after rescaling by the size of the perturbation. This contrast estimates the marginal policy effect (MPE): the derivative of the average outcome with respect to the policy parameter, taking both direct effects and spillovers into account.\footnote{The MPE is distinct from the marginal treatment effect in observational studies \citep{carneiro2010evaluating}, which depends on individual selection into treatment.} The MPE identifies the direction of a welfare-improving local change and provides the basis for testing whether the baseline policy is locally optimal.

Specifically, the paper makes three contributions. First, we introduce a single-wave experiment for estimating and conducting inference on the MPE with a finite number of large clusters. The design pairs comparable clusters, assigns nearby values of the policy parameter within each pair, and uses the resulting difference-in-differences contrasts to estimate marginal policy effects. The same experiment also provides estimates of direct effects and marginal spillover effects, which can be of independent interest, by pooling clusters with different local perturbations. We derive consistency and inference guarantees for these policy-relevant estimands with a finite number of large clusters. 

Second, we implement the single-wave perturbation logic in a large-scale field experiment. In collaboration with Precision Development (PxD), an NGO providing agronomy advice in developing countries, we conducted an experiment with over 250,000 farmers across 25 counties in rural Pakistan. PxD provided geo-localized weather forecasts to improve farmers' information about local weather conditions, which are central for climate adaptation. Spillovers are relevant in this setting: in a survey conducted by PxD, \(80\%\) of surveyed farmers reported sharing weather information with other farmers. The experiment consisted of two pre-specified waves that implemented local perturbations around two treatment probabilities, \(50\%\) in the first wave and \(70\%\) in the second wave. This allows us to estimate marginal policy effects and spillovers at two policy-relevant saturation levels.

Using high-frequency survey data merged with daily weather information, we show that farmers improve their beliefs about one-day ahead weather forecasts, and the program generates spillovers, where beliefs are often central for climate adaptation and behavior, e.g. \cite{burlig2024long}. Specifically,  we observe positive marginal policy effects over the first wave and close to zero marginal effects over the second wave when treating $70\%$ of individuals. 
Treating only $70\%$ of individuals instead of the full population can reduce the cost of the intervention by $\sim$ one million US dollars/year if implemented at scale in Pakistan.

Third, we study how MPE estimates can be used sequentially in a multiple-wave experiment. In many applications including those involving our partner NGO Precision Development \citep[e.g.][]{kasy2019adaptive}, it is common to conduct experiments over multiple waves, although previous literature has mostly focused on settings with no spillovers. The adaptive design pairs clusters, estimates local marginal effects within each pair, and updates policies using information from a different pair. This cross-pair update is important because repeated sampling and serial dependence can otherwise make the policy assigned to a cluster endogenous to that cluster's own potential outcomes.

We derive regret guarantees for both out-of-sample regret, i.e., the welfare loss when the learned policy is deployed on a new population, and in-sample regret, i.e., the welfare loss borne by experimental participants. The main guarantees require the same cluster-comparability condition used for identification, and smoothness and strong concavity of the welfare function, as in settings with globally decreasing marginal returns. Under these conditions, regret decreases at a fast rate \(1/K\), up to logarithmic factors, when the number of clusters is proportional to the number of iterations. When strong concavity fails, Appendix~\ref{sec:quasi_concavity} derives nearly \(1/K\) rates, up to a small polynomial exponent, under a weaker strict quasi-concavity condition and strong concavity only local at the optimum, but at the expense of a larger per-cluster sample size. Conversely, under strong concavity and larger per-cluster sample size, exponential out-of-sample rates are feasible. These results illustrate a trade-off: the more globally regular the welfare function is, the more aggressively the experiment can use local marginal information to learn from a few clusters. In applications, we view $1/K$-rate as the natural benchmark, and exponential rates as a favorable case.

The appendix studies dynamic treatments, observed cluster heterogeneity proposing a matching strategy on observable cluster characteristics (such as the size of the cluster), and alternative fixed-effect structures. When the number of clusters is finite, however, we mantain some form of cluster-level homogeneity, possibly controlling for cluster-level characteristics.

Our paper relates to the literature on inference under interference and draws from \cite{hudgens2008toward} for definitions of potential outcomes.  \cite{leung2022causal, aronow2017estimating, manski2013identification, li2020random} assume an observed network, while  \cite{sinclair2012detecting},  \cite{vazquez2017identification}, \cite{ibragimov2010t} consider clusters among others. \cite{savje2017average} study inference of the direct effect only. Discussion about relevant estimands in these frameworks can also be found in more recent work by \cite{hu2021average}. None study experimental design.

We contribute to the literature on \textit{single-wave} experiments, where existing network experiments include clustered experiments and saturation designs  \citep{imai2021causal, baird2018optimal}. References with observed networks include  \cite{basse2018model}, 
 \cite{viviano2020experimental} among others. For the analysis of the bias of average treatment effect estimators with interference, see also \cite{basse2016analyzing}, \cite{johari2020experimental}, and \cite{imai2009essential}, and \cite{tabord2018stratification} with $i.i.d.$ data. These authors study experimental designs for inference on treatment effects but not inference on optimal policies. Different from these references, we propose a design to identify the marginal policy effect under interference, used for hypothesis testing and improving policies.

With multiple-wave experiments, we connect to the recent literature on adaptive exploration \citep[e.g.,][among others]{bubeck2012regret, kasy2019adaptive}, and the one on derivative-free stochastic optimization, dating back to \cite{kiefer1952stochastic}, and \cite{flaxman2004online, kleinberg2005nearly, shamir2013complexity, agarwal2010optimal}, among others.
These references do not study the problem of interference (and inference, see Section \ref{sec:main_design} for a comprehensive discussion). 
\cite{wager2019experimenting} study price estimation in a single market \citep[and similarly][]{munro2021treatment}.  They assume infinitely many individuals and an explicit model for market prices under which agents are asymptotically independent. As noted by the authors, the assumptions in their paper do not allow for spillovers on a network, different from here. 
Our setting differs from all these references because individuals are organized into finitely many independent clusters, where unobserved (network) spillovers may occur. These differences motivate (i) our design, which exploits two-level randomization at the cluster and individual level instead of individual-level randomization, and (ii) cluster-level perturbations. From a theoretical perspective, dependence and repeated sampling induce novel challenges.

More broadly, we connect to the treatment choice literature on estimation \cite{manski2004, KitagawaTetenov_EMCA2018, athey2017efficient,  viviano2019policy, ananth2020optimal, rai2018statistical,  hadad2019confidence}. This literature considers an existing experiment instead of experimental designs, and has not studied policy design with unobserved interference. We also relate to the literature on targeting on networks \citep[e.g.,][]{bloch2019centrality}, peer-group composition \citep{graham2010measuring}, and pioneering work on vaccination campaigns \citep{manski2010vaccination, manski2017mandating}. None of these study experimental designs.

Finally, our focus on a few large clusters connects to recent work on experimental design with aggregate units, but absent interference. \cite{abadie2021synthetic} use synthetic-control ideas to select treated and comparison units when only one or a few units can be treated. \cite{ni2025enhancing} study switchback experiments, where treatment is varied over time in settings in which individual-level randomization is impractical and carryover effects may arise. Our setting is complementary as here the key source of variation are local perturbations of treatment probabilities to identify the MPE (in settings with interference/dependence).

\section{Setup and overview} \label{sec:2nd}

\subsection{Notation}

We first introduce necessary notation. We observe \(K\) large clusters, where \(K\) is fixed and, for notational convenience, even. Cluster \(k\) contains \(\tilde N^{(k)}\) individuals, with \(\tilde N^{(k)} \asymp N\) for all \(k \in \{1,\ldots,K\}\). Observables and unobservables are independent across clusters, but may be dependent within clusters because of interference.

For each cluster \(k\), period \(t\), and sampled individual \(i\), we observe an outcome, treatment, and baseline covariates,
$ 
Y_{i,t}^{(k)} \in \mathcal Y, 
D_{i,t}^{(k)} \in \{0,1\}, 
X_i^{(k)} \in \mathcal X \subseteq \mathbb R^L .
$ 
In each period, researchers observe a random subsample
$
\Big(Y_{i,t}^{(k)},X_i^{(k)},D_{i,t}^{(k)}\Big)_{i=1}^n, n \asymp N,
$ 
where, for expositional convenience, \(n\) is the same across clusters. The same individuals may or may not be sampled in different periods; with some abuse of notation, we index sampled units in every period by \(i \in \{1,\ldots,n\}\). In our asymptotic analyses, we let $N$ grow through a sequence of
data-generating processes and let $K$ be fixed.  

\paragraph{Potential outcomes} Potential outcomes allow for within-cluster interference. Let
\[
Y_{i,t}^{(k)}\big(\mathbf d_1^{(k)},\ldots,\mathbf d_t^{(k)}\big),
\qquad
\mathbf d_s^{(k)} \in \{0,1\}^{\tilde N^{(k)}},\quad s\le t,
\]
denote the potential outcome of individual \(i\) in cluster \(k\) at time \(t\), as a function of the treatment assignments of all units in the same cluster up to period \(t\). We use capital letters, such as \(D_{i,t}^{(k)}\), for realized treatment assignments. This notation imposes no cross-cluster interference and no anticipation \citep[e.g.][]{athey2018design}. We write \(Y^{(k)}(\cdot)\) for the collection of potential outcome functions in cluster \(k\), and \(X^{(k)}\) for the corresponding covariates. We take a super-population perspective where potential outcomes are random. 

\paragraph{Assignment policies} 
We consider a parametric class of treatment-assignment rules indexed by a policy parameter $\beta$,
$ 
\pi(\cdot;\beta):\mathcal X \to [0,1],
$ in the interior of a compact support $\mathcal{B}$ throughout,
where $\pi(x;\beta)$ is the probability of assigning treatment to an individual with covariates $x$. We assume $\pi(x;\beta)$ to be twice differentiable in $\beta$. For example, with discrete covariates, $\beta$ may directly collect type-specific treatment probabilities, $\pi(x;\beta)=\beta^{(x)}$, while other smooth parametrizations are also possible. In practice, we will think of $\pi(x;\beta)=\beta^{(x)}$ as the benchmark, which generalizes \cite{baird2018optimal} to discrete covariates. 

In the experiment, the researcher may choose cluster- and time-specific policy parameters $\hat\beta_{k,t}$. We use hats to denote parameters chosen by the experimental design, and write $\beta$ for a generic non-random policy parameter. We let $\mathbb E_\beta[\cdot]$ denote expectation over treatment assignments generated by $\pi(\cdot;\beta)$.

We focus on two-stage experiments. First, the researcher chooses the policy parameter $\hat\beta_{k,t}$ for cluster $k$ at time $t$. Second, conditional on this parameter and covariates, treatments are assigned independently across individuals. This class of assignments is easy to implement and nests standard saturation designs \citep[e.g.][]{baird2018optimal, imai2021causal}. 

\begin{ass}[Treatment assignments in the experiment] \label{defn:bernoulli} 
For $\hat{\beta}_{k,t} \perp \big(X^{(k)},Y^{(k)}(\cdot)\big)$, let
\[
D_{i,t}^{(k)} \mid X^{(k)},Y^{(k)}(\cdot),\hat{\beta}_{k,t}
\sim_{i.n.i.d.}
\mathrm{Bern}\!\left(\pi\!\left(X_i^{(k)};\hat{\beta}_{k,t}\right)\right),
\]
where $i.n.i.d.$ indicates independent but not necessarily identically distributed assignments across individuals within cluster $k$.
\end{ass}

Assumption \ref{defn:bernoulli} requires the policy parameter used in a cluster to be exogenous with respect to that cluster's potential outcomes, as in a randomized experiment (the independence with covariates simplifies notation but it is not necessary). 

\subsection{Overview: ``simple world''}
\label{sec:baseline_model}

We begin with an illustrative discussion of the simplest environment underlying the design. Suppose there is a single period of experimentation, no covariates, and all clusters have the same size \(N=\tilde N^{(k)}\). The relevant potential outcome in cluster \(k\) is then
$ 
Y_{i,1}^{(k)}(\mathbf d_1^{(k)}),
\mathbf d_1^{(k)}\in\{0,1\}^N,
$ 
where outcomes may depend on the treatment assignments of all units in the same cluster. Ideally, if a large number of independent clusters were available, one could vary the policy parameter \(\beta\) across clusters and trace out the average reward function
\[
\bar W(\beta)
=
\frac{1}{NK}
\sum_{k=1}^K
\sum_{i=1}^N
\mathbb E_{\beta}
\!\left[
Y_{i,1}^{(k)}(D_1^{(k)})
\right],
\]
where \(D_1^{(k)}\) is drawn according to the assignment rule indexed by \(\beta\). This object is the average response in the population of clusters as a function of the policy. With many independent clusters, researchers can estimate such a response function by assigning different values of \(\beta\) to different clusters, without requiring the response function to be identical across clusters.

This strategy is less useful in the regime studied here, for two reasons. First, and foremost, the researcher has only a small, finite number of independent clusters. This is the relevant regime when natural small units, such as villages or neighborhoods, may be connected by spillovers and therefore cannot be treated as independent experimental units. Second, assigning significantly different treatment probabilities to different regions can sometimes be infeasible due to political or logistical constraints. 

We therefore take a different route. We impose a comparable-clusters restriction: if two large clusters are assigned the same policy parameter \(\beta\), their average outcomes are the same in expectation. In the simple model without fixed effects, this means that, for each cluster,
\[
\frac{1}{N}
\sum_{i=1}^N
\mathbb E_{\beta}
\!\left[
Y_{i,1}^{(k)}(D_1^{(k)})
\right]
=
W(\beta),
\]
for a common reward \(W(\cdot)\). This condition requires that clusters have the same average response to the same assignment rule in expectation. The formal model below relaxes equality in levels by allowing additive cluster and time fixed effects, but the key restriction remains that fixed effects do not change the marginal response to the policy. (In practice, we recommend matching cluster on observable characteristics and test for balance accordingly.) 

Under this comparable-clusters restriction, local perturbations $\eta$ identify a policy-relevant derivative even when the number of clusters is small. Assign one cluster the nearby policy parameter \(\beta+\eta\) and another comparable cluster the nearby policy parameter \(\beta-\eta\). Then
\[
\frac{W(\beta+\eta)-W(\beta-\eta)}{2\eta} \approx \frac{\partial W(\beta)}{\partial \beta}
\]
as we take $\eta$ to be sufficiently small. 
The experimental analogue replaces \(W(\beta+\eta)\) and \(W(\beta-\eta)\) with the average outcomes observed in the two clusters. This local object indicates the direction for improvement, and allows us to test if  \(\beta\) is an interior local optimum.

To further relate to common practice, note that with a few clusters researchers often choose few (e.g., two) treatment probabilities ($\beta_1, \beta_2$), and assign clusters to \textit{each} of these probabilities. Section \ref{sec:designs} shows that the local randomization approach around these probabilities can return the same contrasts these common designs identify by pooling clusters \textit{around} each treatment probability (therefore without loosing sample size) at the expense of an asymptotically negligible bias; at the same time, it also identifies policy gradients at each of these points, for a suitable choice of $\eta$ as discussed in Remark \ref{rem:eta_choice}.

In a sequential regime, where researchers can choose treatments in a second stage, these policy gradients are useful to improve decisions on the experimental participants.

\subsection{Formal model}

We now state the formal version of the comparable-clusters model. 

\begin{ass}[Model] \label{ass:ass_0} 
Suppose that for any $(i,t,k)$, the following holds.
\begin{itemize} 
\item[(i)] No carryovers and common covariate distribution:
$ 
Y_{i,t}^{(k)}(\mathbf d_1^{(k)},\ldots,\mathbf d_t^{(k)})
$ 
is constant in $\mathbf d_1^{(k)},\ldots,\mathbf d_{t-1}^{(k)}$, and $X_i^{(k)}\sim F_X$ for an unknown distribution $F_X$.

\item[(ii)] Under an assignment in Assumption \ref{defn:bernoulli} with parameter $\hat\beta_{k,t}$,
\begin{equation} \label{eqn:main_Y0} 
\begin{aligned} 
\mathbb E_{\hat\beta_{k,t}}
\Big[
Y_{i,t}^{(k)}
\mid
X_i^{(k)}=x
\Big]
=
\tau_k+\alpha_t+y(x,\hat\beta_{k,t}),
\end{aligned} 
\end{equation}
where $Y_{i,t}^{(k)}=Y_{i,t}^{(k)}(\mathbf D_1^{(k)},\ldots,\mathbf D_t^{(k)})$, the expectation is over the treatment assignments and the distribution of potential outcomes, and $y(\cdot)$ is common across clusters.

\item[(iii)] Within-cluster dependence is local in the following sense: for some $\gamma_N\ge 1$
\[
Y_{i,t}^{(k)}
\perp
\{Y_{j,t}^{(k)}\}_{j\notin \mathcal I_i^{(k)}}
\mid
\hat\beta_{k,t},
\qquad
|\mathcal I_i^{(k)}|\le 2\gamma_N.
\]
\end{itemize} 
\end{ass}   

Assumption \ref{ass:ass_0}(i) rules out carryover effects and imposes a common covariate distribution across clusters. Assumption \ref{ass:ass_0}(ii) is the central comparable-clusters restriction. Conditional on an individual's own  covariates, expected outcomes may depend on the policy parameter $\beta$, which captures spillovers induced by the assignment rule, but this dependence must be the same across clusters up to the additive fixed effects $\tau_k$ and $\alpha_t$.

Assumption \ref{ass:ass_0}(iii) controls within-cluster dependence. The parameter $\gamma_N$ measures the size of the local dependence neighborhood. This condition is sufficient, but not necessary as Theorem \ref{thm:const1} below shows that such condition can be relaxed to other forms of weak dependence.  
Section \ref{sec:why_assumptions} discusses the applicability of these assumptions in our application. 

\begin{defn}[Reward, net-of fixed effects] \label{defn:welf1} 
For treatments assigned as in Assumption \ref{defn:bernoulli} with policy parameter $\beta$, define
$
W(\beta)
=
\int y(x,\beta)dF_X(x).
$
\end{defn} 

The reward $W(\beta)$ is the expected outcome, net of additive cluster and time effects, under assignment rule $\pi(\cdot;\beta)$. The expectation is over covariates, treatment assignments, and potential outcomes. Because $\tau_k$ and $\alpha_t$ are additive and do not depend on $\beta$, we directly define reward net of such fixed effects as these do not affect relevant comparisons.  Note that researchers may also subtract treatment costs $c\beta$ in the definition of $W(\beta)$ when relevant.

Under differentiability in Assumption \ref{ass:regularity_basic} below, define the optimal policy and the MPE
\begin{equation}  \label{eqn:estimand} 
\beta^* \in \mathrm{arg}\sup_{\beta\in\mathcal B} W(\beta),
\qquad
M(\beta)=\frac{\partial W(\beta)}{\partial \beta}.
\end{equation}


\begin{rem}[Direct and spillover effects]
In addition to marginal and welfare effects, we may also be interested separately in direct and spillover effects. Under the additional assumption that $\mathbb E_{\beta}
\Big[
Y_{i,t}^{(k)}
\mid D_{i,t}^{(k)}=d, 
X_i^{(k)}=x
\Big]
=
\tau_k+\alpha_t+m(d, x,\beta)$ for some common function $m(d,x,\beta)$ define the conditional direct effect and the marginal spillover effect as
\[
\Delta(x,\beta)
=
m(1,x,\beta)-m(0,x,\beta),
\qquad
S(d,x,\beta)
=
\frac{\partial m(d,x,\beta)}{\partial \beta},
\quad
d\in\{0,1\}.
\]
The direct effect compares treated and untreated units holding fixed the assignment environment. The marginal spillover effect captures how expected outcomes change when the assignment rule changes, holding fixed the individual's own treatment status.

Similar in spirit to the decomposition in \cite{hudgens2008toward}\footnote{We also note that in more recent work, \cite{hu2021average} motivate targeting as causal estimand the average indirect effect, different from $S(\cdot)$ with heterogeneous assignments.  \cite{graham2010measuring} present peer effects' decompositions in the different contexts of peer groups' formation. 
} 
 \begin{equation} \label{eqn:marginal}  
 \small 
 M(\beta) = \int  \Big[\underbrace{\pi(x;\beta) S(1, x, \beta) + (1 - \pi(x; \beta))  S(0, x, \beta)}_{(S)} +  \underbrace{\frac{\partial \pi(x;\beta)}{\partial \beta} \Delta(x, \beta)}_{(D)}\Big] dF_X(x). 
 \end{equation} 
The term $(D)$ captures the direct effect of changing the individual's treatment probability. The term $(S)$ captures the marginal spillover effect. \qed 
\end{rem} 

\subsection{Illustrative examples}

Two examples illustrate the applicability of Assumption \ref{ass:ass_0}. 
 
\begin{exmp}[Neighbors' effects] \label{exmp:main} 
Suppose, for simplicity, that there are no individual covariates and that treatments are assigned independently with common probability \(\beta\), so that
$ 
D_{i,t}^{(k)} \sim_{i.i.d.} \mathrm{Bern}(\beta).
$ 
Let \(\mathcal N_i^{(k)}\) denote the set of neighbors of individual \(i\) in cluster \(k\). Consider the outcome model
\begin{equation} \label{eqn:ex}
\small 
\begin{aligned} 
Y_{i,t}^{(k)}
&=
\tau_k+\alpha_t
+
D_{i,t}^{(k)}\phi_1
+
\frac{\sum_{j\in\mathcal N_i^{(k)}}D_{j,t}^{(k)}}{|\mathcal N_i^{(k)}|}\phi_2
-
\left(
\frac{\sum_{j\in\mathcal N_i^{(k)}}D_{j,t}^{(k)}}{|\mathcal N_i^{(k)}|}
\right)^2\phi_3  -
D_{i,t}^{(k)}
\frac{\sum_{j\in\mathcal N_i^{(k)}}D_{j,t}^{(k)}}{|\mathcal N_i^{(k)}|}
\phi_4
+
\nu_{i,t}^{(k)},
\end{aligned} 
\end{equation}
with $\mathbb E[\nu_{i,t}^{(k)}]=0$.  The coefficient \(\phi_1\) captures the direct effect of the treatment, while \(\phi_2\) and \(\phi_3\) capture spillovers from neighbors' treatment exposure, allowing for decreasing marginal spillovers. The interaction term \(\phi_4\) allows the effect of an individual's own treatment to depend on neighbors' treatment exposure. Let \(c\) denote the cost of treating an individual, and let
$ 
\iota=\mathbb E\!\left[\frac{1}{|\mathcal N_i|}\right]$.
Taking expectations over the independent treatment assignments and subtracting treatment costs, the reward is
\[
W(\beta)
=
\beta(\phi_1-c)
+
\beta\phi_2
-
\beta\phi_3\iota
-
\beta^2\phi_3(1-\iota)
-
\beta^2\phi_4 .
\]
 
The direct-effect component is \((\phi_1-c)-\beta\phi_4\), while the remaining terms capture how changing \(\beta\) changes neighbors' treatment exposure and hence spillovers.

In this example, realized networks may differ across clusters, and realized treatment assignments and outcomes may differ as well. What must be common across clusters is the average response to the assignment rule. In the model above, this requires the coefficients \(\phi_1,\ldots,\phi_4\) and the relevant distribution of network exposure, summarized here by \(\iota\), to be common across clusters, up to the additive fixed effects \(\tau_k\) and \(\alpha_t\). The fixed effects allow clusters to differ in outcome levels, but not in the slope of \(W(\beta)\). 

Appendix \ref{sec:1a} generalizes to non-linear models with local interference, under a particular network microfoundation, homogeneous across clusters, where we control the dependence structure through $\gamma_N^{1/2}$ characterizing an a-priori bound on the largest degree. 
\qed
\end{exmp}

\begin{exmp}[Linear-in-means peer effects] \label{exmp:linear_means} Consider again the assignment in Example \ref{exmp:main} but now denote  \(A^{(k)}\) the row-normalized adjacency matrix with 
\(\sum_j A_{ij}^{(k)}=1, A_{ii}^{(k)} = 0\). Let
$ 
\bar D_{i,t}^{(k)}=\sum_j A_{ij}^{(k)}D_{j,t}^{(k)}, 
\bar Y_{i,t}^{(k)}=\sum_j A_{ij}^{(k)}Y_{j,t}^{(k)}, 
$ 
and 
\begin{equation} \label{eqn:linear_means}
Y_{i,t}^{(k)}
=
\tau_k+\alpha_t
+
\phi_1D_{i,t}^{(k)}
+
\phi_2 \bar D_{i,t}^{(k)}
-
\phi_4D_{i,t}^{(k)}\bar D_{i,t}^{(k)}
+
\gamma \bar Y_{i,t}^{(k)}
+
\nu_{i,t}^{(k)},
\qquad
\mathbb E[\nu_{i,t}^{(k)}\mid A^{(k)}]=0 .
\end{equation}
Assume \(|\gamma|<1\), so that \(I-\gamma A^{(k)}\) is invertible. In vector form, we can write 
\begin{equation} \label{eqn:linear_means_reduced}
Y^{(k)}
=
(I-\gamma A^{(k)})^{-1}
\Big[
(\tau_k+\alpha_t)\mathbf 1
+
\phi_1D^{(k)}
+
\phi_2A^{(k)}D^{(k)}
-
\phi_4\operatorname{diag}(D^{(k)})A^{(k)}D^{(k)}
+
\nu^{(k)}
\Big].
\end{equation}
Since \(A\mathbf 1=\mathbf 1\), 
$ 
\mathbb E_\beta[D\mid A]=\beta\mathbf 1,
\mathbb E_\beta[AD\mid A]=\beta\mathbf 1,
\mathbb E_\beta[\operatorname{diag}(D)AD\mid A]=\beta^2\mathbf 1, 
$ 
which implies 
\[
\mathbb E_\beta[Y_i^{(k)}\mid A^{(k)}]
=
\frac{\tau_k+\alpha_t}{1-\gamma}
+
\frac{\beta\phi_1+\beta\phi_2-\beta^2\phi_4}{1-\gamma} \qquad \Rightarrow W(\beta)
=
\frac{\beta\phi_1+\beta\phi_2-\beta^2\phi_4}{1-\gamma}
-
c\beta, 
\] 
once we also consider a cost of treatment $c$, which satisfies Assumption \ref{ass:ass_0}(ii).

Thus, 
the average response to the assignment rule is common across clusters, provided that
\(\gamma,\phi_1,\phi_2,\phi_4\) are common across clusters. 
In this example, \(A^{(k)}\) may differ across clusters and, under this linear model, no restrictions on the network formation are imposed since we condition on $A^{(k)}$. In the presence of additional heterogeneity, governed either by $A^{(k)}$ or by covariates, we would also require a common distribution of the two across clusters. 

Finally, although the dependence restrictions need not satisfy Assumption \ref{ass:ass_0}(iii), the model satisfies the weak-dependence conditions required in Theorem \ref{thm:const1}, for independent idiosyncratic shocks $\nu_i$ with bounded variance, provided that the network multiplier is sufficiently diffuse, i.e., 
$
\sup_k \frac{1}{n}
\left\|
\left[(I-\gamma A^{(k)})^{-1}\right]^\top \mathbf 1
\right\|_2^2
< \infty .
$ 
\end{exmp}

\section{Single-wave experiment} \label{sec:designs}

This section introduces the single-wave local-perturbation design. The goal is to estimate the marginal policy effect \(M(\beta)\) at a pre-specified policy parameter \(\beta\). 

\subsection{Experimental design}

Let
\begin{equation} \label{eqn:ej} 
\underline e_j =
\begin{bmatrix}
0,\ldots,0,1,0,\ldots,0
\end{bmatrix},
\qquad
\underline e_j\in\{0,1\}^p,
\qquad
\underline e_j^{(j)}=1,
\end{equation}
denote the \(j\)-th coordinate vector. For expositional simplicity, Algorithm \ref{alg:my_pilot} focuses on the first coordinate, \(j=1\). The extension to multiple coordinates follows verbatim.

The design pairs clusters. Let \(G=K/2\), and index each pair by \(g=(k,k+1)\), with \(k\) odd. Within each pair, one cluster is assigned the perturbed policy parameter \(\beta+\eta_n\underline e_1\), and the other is assigned \(\beta-\eta_n\underline e_1\), where \(\eta_n>0\) is a small deterministic perturbation. Conditional on these policy parameters, treatments are assigned independently within each cluster as in Assumption \ref{defn:bernoulli}. The purpose of using opposite perturbations within a pair is to estimate a local slope while differencing out pair-level fixed differences.

The estimator for pair \((k,k+1)\) is
\begin{equation} \label{eqn:gradient_main} 
\widehat{M}_{(k,k+1)}(\beta)
=
\frac{
\Big[\bar Y_1^{(k)}-\bar Y_0^{(k)}\Big]
-
\Big[\bar Y_1^{(k+1)}-\bar Y_0^{(k+1)}\Big]
}{2\eta_n},
\end{equation}
where \(\bar Y_t^{(h)}\) is the sample average outcome in cluster \(h\) at time \(t\). The estimator is a difference-in-differences across the two clusters in a pair. The baseline outcomes remove additive cluster fixed effects; in the absence of fixed effects, one can equivalently set \(\bar Y_0^{(h)}=0\). 

The design uses two sources of variation. Within each cluster, individual-level randomization identifies direct effects. Across the two clusters, the between-cluster perturbation provides the variation needed to identify the spillover part in Equation \eqref{eqn:marginal}. Thus the same experiment estimates the overall marginal policy effect \(M(\beta)\), while also allowing separate estimation of direct and marginal spillover effects, discussed below.

The same estimator can be used for hypothesis testing. If \(\beta\) is an interior local optimum and \(W(\cdot)\) is differentiable, then \(M(\beta)=0\). More generally, for any subset of coordinates \(p_1\le p\), one may test
\begin{equation} \label{eqn:h0} 
H_0:
M^{(j)}(\beta)=0,
\qquad
\forall j\in\{1,\ldots,p_1\}.
\end{equation}
Rejection provides evidence that \(\beta\) is not an interior local optimum along the tested coordinates.\footnote{Failure to reject should not be interpreted as evidence of optimality unless the test is sufficiently powered.} We aggregate the pair-level estimators and construct a test statistic 
\begin{equation} \label{eqn:test_stat}
\small 
\begin{aligned} 
\mathcal T_n
=
\frac{\sqrt G\,\bar M_n(\beta)}
{
\sqrt{(G-1)^{-1}\sum_g\left(\widehat M_g(\beta)-\bar M_n(\beta)\right)^2}
}, \qquad \bar M_n(\beta)
=
\frac{1}{G}\sum_{g=1}^G \widehat M_g(\beta)
\end{aligned} 
\end{equation}
where \(\widehat M_g(\beta)\) denotes the estimator in Equation \eqref{eqn:gradient_main} for pair \(g\). 

Critical values are obtained by permuting the signs of the pair-level estimates \(\widehat M_g(\beta)\). Since inference uses variation across pairs, power depends on the number of pairs \(G\), while the within-cluster sample size \(n\) controls the precision of each pair-level estimate.

 Finally, the notation \(g=(k,k+1)\) is only for convenience. In applications, pairing can also use pre-treatment cluster characteristics.

\begin{algorithm} [!ht]   \caption{One-wave experiment for inference with $p_1 = 1$}\label{alg:my_pilot}
\footnotesize  
    \begin{algorithmic}[1]
    \Require Value $\beta \in \mathbb{R}^p$ (exogenous), $K$ clusters, constant $\bar{C}$, size $\alpha$; 
    \State Organize clusters into $G = K/2$ pairs with consecutive indexes $\{k, k+1\}$;  
   
    \State $t = 0$ (baseline): either nobody receives treatments or treatments are assigned with $\pi(\cdot;\beta)$.
    \begin{algsubstates}
        \State Experimenters collect baseline outcomes: for $n$ units in each cluster observe $Y_{i,0}^{(h)}, X_i^{(h)}, h  \in \{1, \cdots, K\}$. 
        \end{algsubstates} 
        
    \State $t = 1$: experiment starts
    \begin{algsubstates}
        \State  For each pair $g = \{k, k+1\}$, randomize
   \begin{equation} \label{eqn:perturbation}
     \small 
     \begin{aligned} 
    D_{i,1}^{(h)} | \beta, X_i^{(h)} = x \sim \begin{cases} 
    & \mathrm{Bern}(\pi(x, \beta + \eta_n\underline{e}_1)) \text{ if } h = k  \\ 
    & \mathrm{Bern}(\pi(x, \beta - \eta_n\underline{e}_1)) \text{ if } h = k+1 
    \end{cases} , \quad \bar{C} n^{-1/2} < \eta_n < \bar{C} n^{-1/4},
    \end{aligned}   
    \end{equation} 
    
        \State For $n$ units in each cluster $h$ observe $Y_{i,1}^{(h)}$; 
        \State For each pair $g = (k, k+1)$ construct 
        $$
         \widehat{M}_{(k, k+1)}(\beta)  = \frac{1}{2 \eta_n }  \Big[\bar{Y}_1^{(k)} - \bar{Y}_0^{(k)}\Big] - \frac{1}{2 \eta_n } \Big[\bar{Y}_1^{(k + 1)} - \bar{Y}_0^{(k + 1)}\Big],
         $$ 
      \end{algsubstates}
     
 \State Construct  the t-statistic to test $H_0$ in Equation \eqref{eqn:h0} (with $j = 1$) and return $\mathcal{T}_n$ in Equation \eqref{eqn:test_stat}. 
\State Construct tests $1\Big\{|\mathcal{T}_n| > \mathrm{cv}_{G}(\alpha) \Big\}$ with size $\alpha$, with critical values obtained by permuting the sign of the estimated marginal effect as described in Corollary \ref{thm:inference4b}. 

 
         \end{algorithmic}
\end{algorithm}

\subsection{Consistency, inference and choice of $\eta_n$}

Next, we study properties of the design under the following regularity condition. 

\begin{ass}[Regularity 1] \label{ass:regularity_basic} Suppose that for all $x \in\mathcal{X}, d \in \{0,1\}$, $\pi(x, \beta)$, and $y(x, \beta)$ are uniformly bounded and twice differentiable with bounded derivatives. 
\end{ass}

\begin{thm}[Marginal effects] \label{thm:const1} Suppose that $Y_{i,t}^{(k)}$ is sub-Gaussian with sub-Gaussian parameter bounded by $r^2 < \infty$. Let  Assumptions \ref{ass:ass_0}, \ref{ass:regularity_basic} hold. Let $\mathrm{Var}(\sqrt{n} \hat{M}_{(k, k+1)}(\beta)) \le \tilde{C}_{k , k+1}\frac{\rho_n}{\eta_n^2}$, for arbitrary $\rho_n$ and constant $\tilde{C}_{k, k+1}$. Then, with probability at least $1 - \delta$, for any $\delta \in (0,1)$, for a finite constant $c_0 < \infty$ independent of $(n, N, \gamma_N, K, \beta)$, 
$$
\small 
\begin{aligned} 
 \Big|\hat{M}_{(k, k+1)}(\beta) - M^{(1)}(\beta)\Big| \le c_0 \left(\eta_n +  \min\Big\{\sqrt{\frac{\gamma_N \log(\gamma_N/\delta)}{n \eta_n^2}},  \sqrt{\frac{\tilde{C}_{k, k+1} \rho_n}{n \eta_n^2 \delta}}\Big\}  \right), 
 \end{aligned}
$$  
where $\hat{M}_{(k, k+1)}$ is estimated as in Algorithm \ref{alg:my_pilot}. For $\gamma_N \log(\gamma_N)/n^{1/3} = o(1), \eta_n = n^{-1/3}, \hat{M}_{(k, k+1)}(\beta) \rightarrow_p M^{(1)}(\beta), \bar{M}_n \rightarrow_p M^{(1)}(\beta)$. 
\end{thm} 

The proof is in Appendix \ref{sec:proof1}.
Theorem \ref{thm:const1} shows we can consistently estimate the marginal effects with two large clusters. The convergence rate depends on the \textit{minimum} between the maximum degree of dependence, which is proportional to $\gamma_N^{1/2}$, and the covariances among unobservables, captured by $\rho_n$. 

\begin{rem}[Practical choice of \(\eta_n\)] \label{rem:eta_choice}
Theorem \ref{thm:const1} shows that \(\eta_n\) plays the role of a bandwidth. For coordinate \(j\), smaller perturbations reduce the approximation bias, while larger perturbations reduce the variance generated by dividing by \(2\eta_n\). In applications, we recommend choosing the perturbation in treatment-probability units, where we choose perturbations 
\[
\eta_{n,j}= a_j n^{-1/3} \Big/ \int \frac{\partial   \pi(x;\beta)}{\partial \beta_j}  dF_X(x),  
\]
denote the desired change in assignment probability along coordinate \(j\), where \(a_j\) is chosen based on feasibility and power (typically $a_j \approx \sqrt{2 \sigma^2}$ where $\sigma^2$ is the outcome variance). 
Thus, we recommend selecting \(\eta\) so that the actual probability perturbation, rather than the numerical change in the parameter, is of order \(n^{-1/3}\), to account for the parametrization through $\pi(x,\beta)$. Also, note that the reparametrization of $\pi$ does not affect the zero-gradient condition for local optimality when the chain rule is applicable. \qed 
\end{rem}

Given our consistency result, we next discuss inference.

\begin{ass}[Regularity 2] \label{ass:var} Assume that for treatments as assigned in Algorithm \ref{alg:my_pilot}, for all $k \in \{1, \cdots, K\}$, $Y_{i,t}^{(k)}$ has a bounded fourth moment, and for some $\bar{C}_k > 0$, $\rho_n \ge 1$,
\begin{equation} \label{eqn:variance}  
\small 
\begin{aligned} 
\mathrm{Var}\left(\sqrt{n}\Big[\bar{Y}_1^{(k)} - \bar{Y}_0^{(k)}\Big]\right) = \bar{C}_k \rho_n.
\end{aligned} 
\end{equation} 

\end{ass}

Assumption \ref{ass:var} imposes standard moment bounds and a \textit{lower bound} on the variance of the estimator. In particular, Assumption \ref{ass:var} is a mild regularity condition that states that the variance does not converge to zero at a rate faster than $1/n$.\footnote{Note that 
$
\bar{C}_k \rho_n = \frac{1}{n} \sum_{i=1}^n \mathrm{Var}\left(Y_{i, 1}^{(k)} - Y_{i,0}^{(k)}\right) +  \frac{1}{n} \sum_{i, j, j\neq i} \mathrm{Cov}\left(Y_{i,1}^{(k)} - Y_{i, 0}^{(k)}, Y_{j,1}^{(k)} - Y_{j, 0}^{(k)}\right). 
$
Assumption \ref{ass:var} states that $\rho_n \ge 1$, i.e., $\rho_n$ does not converge to zero. This requires that the negative covariance components (if any) do not outweigh the variances, holding for example with no or positive correlations (more generally it only requires non-degenerate variance).}


  \begin{thm} \label{thm:inference1}  Let  Assumptions \ref{ass:ass_0}, \ref{ass:regularity_basic}, \ref{ass:var} hold. Let $n^{1/4} \eta_n = o(1), \gamma_N/n^{1/4} = o(1)$, $K < \infty$. Then, for each pair $(k, k+1)$, for $\widehat{M}_{(k, k+1)}$ estimated as in Algorithm \ref{alg:my_pilot}, \\
$\mathrm{Var}\Big(\widehat{M}_{(k, k+1)}\Big)^{-1/2} \Big(\widehat{M}_{(k, k+1)} - M^{(1)}(\beta)\Big)  \rightarrow_d \mathcal{N}(0,1).  
$
  \end{thm} 
  
  The proof is in Appendix \ref{sec:inference_thm}. Given Theorem \ref{thm:inference1}, it is possible to conduct inference by using randomization tests \citep{canay2017randomization}.


\begin{cor}[Randomization tests] \label{thm:inference4b} Let the conditions in Theorem \ref{thm:inference1} hold. For any $\alpha \in (0,1)$,  $\lim_{n \rightarrow \infty} P\Big(|\mathcal{T}_n| \le \mathrm{cv}_{K/2}^P(\alpha)\Big| H_0\Big) = 1 - \alpha$, where $\mathrm{cv}_{K/2}^P(\alpha)$ is a $(1 - \alpha)^{th}$ quantile of t-statistics computed from all permutations over the pairs' sign as described in Appendix \ref{sec:perm_tests}, and $H_0$ is as in Equation \eqref{eqn:h0}.
\end{cor}

Corollary \ref{thm:inference4b} formalizes the inference results. The corollary shows control of Type I error for finite $K$, while the power of the test also depends on $K$.

\subsection{Additional estimands: direct effects and marginal spillovers}

We conclude this discussion by showing how 
Algorithm \ref{alg:my_pilot} also allows us to estimate the direct effect of the treatment and the (marginal) spillover effect separately, and the reward up-to a (time-specific) additive shift. To do this, we need to impose the stronger requirement that the conditional expectation given individual level treatment is common across clusters. 

\begin{ass}[Expectation conditional on individual level treatment] \label{ass:ass_0b} 
Suppose that for any $(i,t,k)$, under an assignment in Assumption \ref{defn:bernoulli} with parameter $\hat\beta_{k,t}$,
$ 
\mathbb E_{\hat\beta_{k,t}}
\Big[
Y_{i,t}^{(k)}
\mid D_{i,t}^{(k)} = d, 
X_i^{(k)}=x
\Big]
=
\tau_k+\alpha_t+ m(d,x,\hat\beta_{k,t})$ for some $m(d,x,\beta)$ common across clusters, uniformly bounded and twice-differentiable with bounded derivatives.
\end{ass}

We can now define these additional target parameters 
$$
\small 
\begin{aligned} 
\Delta(\beta) = \int \left[m(1, x, \beta) - m(0, x, \beta)\right] dF_X(x),  \qquad S_1(d, \beta)= \int \frac{\partial m(d, x, \beta)}{\partial \beta^{(1)}} dF_X(x), \qquad  W(\beta) + c_1,    
\end{aligned} 
$$  
where $c_1$ is a time specific additive shift (that also depends on the baseline policy intervention). 
Each estimand integrates over the distribution of covariates. Note also that the target reward effect includes an additive shift (that depends on time fixed effects and subtracts the effect at baseline); however any comparison between two global contrasts at ($\beta, \beta'$) both estimated using the same two consecutive periods would automatically remove the additive shift $c_1$. 

For a given pair of clusters $(k, k+1)$, we estimate direct effects as 
\begin{equation} \label{eqn:direct_estimated}
\footnotesize  
\begin{aligned} 
\widehat{\Delta}_{(k, k+1)}(\beta) & =  \frac{1}{2 n} \sum_{h \in \{k, k+1\}} \sum_{i=1}^n \left[\frac{D_{i,1}^{(h)} Y_{i,1}^{(h)}}{\pi(X_i^{(h)}, \beta + \eta_n v_h \underline{e}_1)} - \frac{(1 - D_{i,1}^{(h)}) Y_{i,1}^{(h)}}{1 - \pi(X_i^{(h)}, \beta + \eta_n v_h \underline{e}_1)}\right], \quad v_h = \begin{cases}  1 \text{ if } h = k \\
-1 \text{ if } h = k + 1. 
\end{cases}  
\end{aligned} 
\end{equation}  
 We average direct effects across cluster pairs to obtain a single estimate $\bar{\Delta}_n = \frac{1}{G} \sum_{g} \widehat{\Delta}_g(\beta)$. The marginal spillover effect is estimated as follows:
\begin{equation} \label{eqn:estimated_spill}
\small 
\begin{aligned} 
\widehat{S}_{(k, k+1)}(0, \beta) = \frac{1}{2 n} \sum_{h \in \{k, k+1\}} \frac{v_h}{\eta_n} \sum_{i=1}^n  \left[ \frac{Y_{i,1}^{(h)} (1 - D_{i,1}^{(h)})}{1 - \pi(X_i^{(h)}, \beta + v_h \eta_n \underline{e}_1)} - \bar{Y}_0^{(h)}\right]. 
\end{aligned} 
\end{equation}
Researchers can report the between-pairs average $\bar{S}_n(0, \beta) = \frac{1}{G} \sum_{g} \widehat{S}_g(0, \beta)$ (and similarly $\hat{S}(1, \beta)$ for treated units), which estimates spillovers on the control units.

For the reward effect, we construct an estimator
$
\bar{W}_n(\beta) = \frac{1}{K} \sum_{k=1}^K \Big[\bar{Y}_1^{(k)} - \bar{Y}_0^{(k)}\Big].
$

\begin{thm}[Asymptotically negligible bias of treatment effects] \label{thm:bias_direct} Let  Assumptions \ref{ass:ass_0}, \ref{ass:regularity_basic}, \ref{ass:ass_0b} hold, and $\eta_n = o(n^{-1/4})$. Suppose in addition that overlap holds, namely
$\pi(X_i, \beta + \eta_n v_h \underline{e}_1) \in (\delta, 1-\delta)$ for some $\delta \in (0,1)$ almost surely.
Then, $\mathbb{E}\left[\bar{\Delta}_n(\beta)\right] = \Delta(\beta) + o(n^{-1/2})$, where the second term does not depend on $K$. Similarly, 
$\mathbb{E}\left[\bar{W}_n(\beta)\right] = W(\beta) + c_1 + o(n^{-1/2})$, for an additive shift $c_1$ that only depends on $\beta$ at time $t = 0$ (but not $t = 1$), and where the second term does not depend on $K$.  In addition, for all pairs $(k, k+1)$,
$
\mathbb{E}\Big[\widehat{S}_{(k, k+1)}(0, \beta) \Big] = S_1(0, \beta) + \mathcal{O}(\eta_n). 
$
\end{thm} 

The proof is in Appendix \ref{sec:proof3}. 
The bias of the estimated direct effect is asymptotically negligible at a rate \textit{faster} than the parametric rate $n^{-1/2}$ when \textit{pooling} observations from different clusters.  Our main insight here is that, with pairing and perturbations of opposite signs, the first-order bias cancels out.
Here, $\eta_n = o(n^{-1/4})$ is consistent with requirements in previous theorems. Therefore, we can conduct inference on the direct effect using permutation tests, where we simply permute the sign of the estimated direct effect in each cluster $k$.\footnote{This is because, under the same conditions on $\gamma_N$ as those in Theorem \ref{thm:inference1} one can establish asymptotic normality of the estimated direct effect within each cluster from Gaussian approximations for local dependency graphs. As in Corollary \ref{thm:inference4b}, we can apply randomization inference for the direct effect.} Inference on the marginal spillover effects follows verbatim to inference on the MPE.

\begin{table}[!htbp] \centering 
  \scalebox{0.65}{
\begin{tabular}{@{\extracolsep{5pt}} ccccc} 
\\[-1.8ex]\hline 
\hline \\[-1.8ex] 
Estimand & Description & Estimator  & Estimator's Description & Randomization Inference \\ 
\hline \\[-1.8ex] 
$M(\beta)$  & Marginal effect & $\bar{M}_n(\beta) = \frac{1}{G} \sum_{g=1}^G \widehat{M}_g(\beta)$  & DiD (or DiM) estimators  & Permute signs of $\widehat{M}_g(\beta)$ \\ 
& & $\widehat{M}_g(\beta)$ as in Eq \eqref{eqn:gradient_main} & from each cluster pair $g$ & for each cluster pair $g$ \\ 
 \\
$\Delta(\beta)$ & Direct effect & $\bar{\Delta}_n = \frac{1}{G} \sum_g \hat{\Delta}_g(\beta)$ &  Pooled IPW estimators & Permute sign of estimated effect \\
& & $\hat{\Delta}_g$ as in Eq \eqref{eqn:direct_estimated}  & from each cluster & for each pair estimate \\  \\ 
$S(0,\beta)$ & Marginal spillovers & $\bar{S}_n(0, \beta) = \frac{1}{G} \sum_g \widehat{S}_g(0, \beta)$ & DiD (or DiM) + IPW estimators & Permute sign of $\widehat{S}_g(0, \beta)$ \\ 
& on controls & $\widehat{S}_g(0, \beta)$ as in Eq \eqref{eqn:estimated_spill} & from each cluster pairs $g$ & for each cluster pair $g$ \\ 
 \\ 
 $S(1,\beta)$ & Marginal spillovers & $\bar{S}_n(1, \beta) = \frac{1}{G} \sum_g \widehat{S}_g(1, \beta)$ & DiD (or DiM) + IPW estimator & Permute sign of  $\widehat{S}_g(1, \beta)$\\ 
& on treated & & from each cluster pairs $g$  & for each cluster pair $g$ \\ 
\hline \\[-1.8ex] 
\end{tabular} 
}
 \caption[Caption for LOF]{Estimands and estimators from \textit{single wave experiment} with treatment probability $\beta$. Inference procedure is formally described in Corollary \ref{thm:inference4b}, with DiD standing for Difference-in-Differences and DiM stands for difference in means (applicable in the absence of cluster fixed effects) and IPW standing for Inverse Probability Weights. } 
  \label{tab:methods1} 
\end{table}

\section{A large-scale implementation} \label{sec:field}

Next, we illustrate our method through a large-scale experiment where we implemented our single wave experiment over two consecutive experimentation waves. We use each wave to illustrate properties of the single wave experiment, and in particular how we can learn marginal, direct and spillover effects with our design and a \textit{few} large clusters. See Table \ref{tab:main_summaries}. 

\vspace{-4mm} 

\paragraph{Treatment $D$} The experiment was implemented through Precision Development (PxD), an NGO that provides farmers with phone-based agricultural advisory services. Farmers often lack access to geo-localized weather forecasts, and digital delivery offers solutions to address this challenge \citep{fabregas2019realizing}. Prior to the experiment, only 45\% of cotton growers reported consistent access to weather information, usually via radio or television. About 86\% of cotton growers indicated that weather information helps plan agricultural activities (\url{https://precisiondev.org/weather-forecasting-product-for-punjab-pakistan/}). In addition, those farmers with access to weather forecasts only access forecasts at the district level, a higher administrative unit with 3-4 tehsils  (tehsils are equivalent to US counties).   

In partnership with a private forecast provider, Precision Development developed calibrated (geo-localized) weather forecast information localized at the tehsil level. The treatment consists of calling farmers to provide weather forecasts via robocalls, meant to improve farmers' ability to take measures on their plots. The experiment was randomized at a large scale across approximately 400,000 farmers, of which about 287,000 are in our main experiment as we discuss below. In a survey, $80\%$ of the respondents said they actively shared weather information with other farmers, providing suggestive evidence of spillovers.

\vspace{-4mm}

 \paragraph{Policy $\beta$} Our policy is choosing how many farmers to treat. Each treatment costs $0.29\$$ per farmer/year, which can lead to significant costs once implemented at scale in Pakistan. 
 
 \vspace{-4mm}

\paragraph{Target outcome $Y$} We study the effect of the treatment on farmers' ability to predict short-run weather. This is relevant in these applications: correctly predicting weather improves efficiency in the use of resources by, for example, using irrigation or pesticides more efficiently and better invest, see e.g., \cite{burlig2024long}.   

As our main data source, we use repeated high-frequency (daily) cross-sectional survey data collected from June to October 2022. To measure farmers' weather forecasts, we ask farmers: ``What do you expect will be the maximum (minimum) temperature in your area
tomorrow?". We merge this information with PxD forecast weather for the day after the survey interview with the specific farmer. We measure the absolute difference between the farmer's predicted maximum (and minimum) temperature and those predicted by PxD forecasts. 
  To combine beliefs about maximum and minimum temperature, we construct a statistical index as described in \cite{viviano2021should} which serves as our main outcome.

 Temperature variables measure \textit{incorrect} beliefs, i.e.,  negative treatment effects indicate that the farmer's prediction is closer to the PxD forecast.  
Predicted temperature is a convenient proxy for farmers' one-day ahead weather perceptions since it is less volatile than precipitation \citep{grenci2001world}, and it is relatively stable within a tehsil. We focus on farmer's prediction relative to the PxD forecast, since PxD forecast is a good proxy for real temperature (Appendix Table \ref{tab:forecast}), and accuracy relative to PxD forecasts instead of real weather reduces variance, as it does not depend on idiosyncratic noise.
 
The survey was  run over approximately $6,000$ farmers, stratified across tehsils and treatment, of which we have approximately $1,000$ respondents for our main outcome. We show balance on take-up rates for survey responses on many dimensions in Appendix \ref{app:more_exp1}, and balance in baseline covariates across clusters assigned to different treatment probabilities.

\vspace{-4mm}

\paragraph{Clusters} In total, 25 tehsils are exposed to our main experiment/design (with in total 287,000 farmers). 
Figure \ref{fig:sample_size} illustrates the region in Pakistan exposed to experimental variation and the sample size within each district (not all tehsils in these districts were part of the experiment; only tehsils where PxD conducts its main operations were included). 
We consider a tehsil a cluster. The assumption is that spillovers between tehsils are negligible, justified by the fact that tehsils denote large geographic areas, and forecasts are geo-localized at the tehsil level. In contrast to some prior work  \citep[e.g.,][]{alatas2012targeting}, our design allows for spillovers across villages in the same tehsil.

\begin{figure}[!htp]
\centering 
\includegraphics[scale = 0.6, trim={0 1cm 0 1cm},clip]{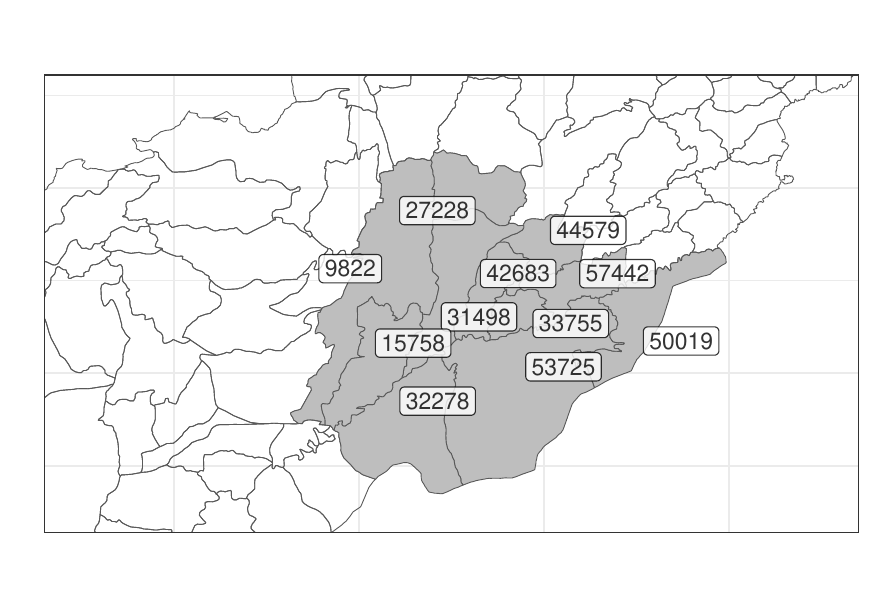}
\includegraphics[scale = 0.5, trim={0 0.3cm 0 0.8cm},clip]{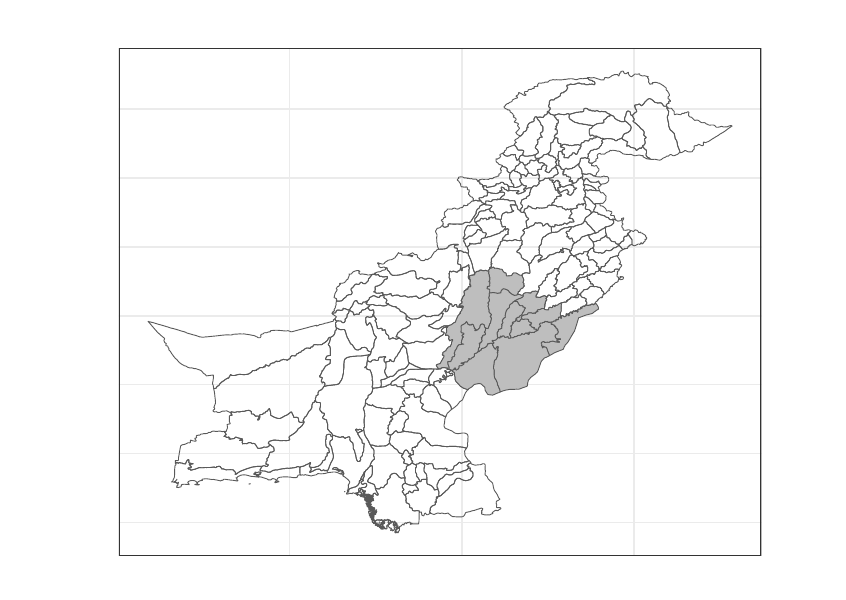}
\caption{Pakistan's map, organized in districts (each district contains multiple tehsils). Gray regions indicate areas selected for the experiment. Next to each district, we report the total sample size obtained from the tehsils in the experiment in the given district.}  \label{fig:sample_size}
\end{figure}

\vspace{-4mm}

 \paragraph{Two wave design} We deployed the \textit{local perturbation design} presented in Section \ref{sec:designs} over two consecutive waves: We induced perturbations around $\beta = 50\%$ in the first wave, where $\beta$ denotes the share of treated individuals, and in the second wave, we induced perturbations around $\beta = 70\%$. The first wave started in April 2022, and the second wave started in August 2022 and ended in October 2022.\footnote{This increase (and treatment probabilities) was planned ex-ante by the NGO's, since the NGO wanted to reach a larger number of treated units by the end of the intervention. This allows us to learn the marginal effects in each wave and guarantee that each treatment rule was chosen ex-ante (i.e., it is not data dependent).} We followed our design in each wave with, in our main specifications, $\eta_n = 10\%$. We also considered variations that allowed us to study multiple choices $\eta_n$ based on power considerations, see Appendix \ref{app:choice_eta}. 

Two additional remarks: of 42  tehsils where PxD conducts its main operations, 40 were originally enrolled into the geo-localized weather program, and five were set aside due to a different intervention conducted by PxD at the same time on these tehsils. In addition, 10 other tehsils (the remaining $\sim 110,000$ farmers) were exposed to tehsil's specific treatment probabilities that do not follow the local perturbation in our original design ($\beta = 11\%$ on average); this group was ex-ante randomly selected and not considered as part of our main analysis (we report effects on this group in Appendix \ref{app:more_exp1}).

\begin{table}[!ht]
\centering 
\scalebox{0.6}{\begin{tabular}{@{\extracolsep{5pt}} l|cc} 
\\[-1.8ex]\hline 
\hline \\[-1.8ex] 
 &  Wave 1  & Wave 2 \\ 
\hline \\[-1.8ex] 
 Treatment $D$ &  One-day ahead geo-localized weather forecast  &  One-day ahead geo-localized weather forecast \\ 
 & & \\ 
 Outcome $Y$ &  Farmer's one-day ahead correct forecast (temp) & Farmer's one-day ahead correct forecast (temp) \\ 
 & & \\ 
  Policy $\beta$ &  Share of treated farmers  & Share of treated farmers \\
  & & \\ 
  Choice of $\beta$ & $50\%$ & $70\%$ \\ 
  & &  \\ 
  Perturbation main exp $\eta_n$ & $10\%$ & $10\%$ \\ 
  & & \\ 
  Total \# of clusters & 25 & 25 \\ 
  & & \\ 
  Total \# of farmers in main experiment & 287,487 & 287,487 \\ 
  & & \\ 
Observed surveyed individuals in main exp/ & 259 & 583 \\ 
  & & \\ 
  \hline \\ 
  Estimated marginal effect [one sided p-value] & -4.00$^{**}$ [0.02] & -0.31 [0.41] \\ 
  & & \\ 
 Mechanism  & Large marginal spillover effects & Close to zero marginal spillover effects \\ 
    & & \\ 
  Cost treatment farmer/year  & $0.29\$$ farmer/year & $0.29\$$ farmer/year \\ 
  & & \\
  Policy implication & Increase share of treated individuals & Treat $\sim 70\%$ of farmers \\ & &  (save 1,000,000\$/year once implemented at scale)  
\end{tabular}
} 
  \caption{Illustration of how our theoretical framework maps to this experiment. P-value is for one-sided test computed via randomization inference. Marginal spillovers denote the marginal effect of increasing friends' treatment probabilities of the control units. } 
  \label{tab:main_summaries} 
\end{table}

\begin{table}[!htp] \centering 

\scalebox{0.65}{\begin{tabular}{lcccc}
\hline
Group of Tehsils & Number of Farmers & Number of Tehsils & Average $\beta$ (Wave 1) & Average $\beta$ (Wave 2) \\
\hline
Negative Perturbation (Medium Saturation) & 137 729 & 12 & $50\% - \eta_n = 40\%$ & $70\% - \eta_n = 60\%$ \\
Positive Perturbation (High Saturation) & 149 758 & 13 & $50\% + \eta_n = 60\%$ & $70\% + \eta_n = 80\%$ \\
Low Saturation, not following main design & 111 300 & 10 & 11\% & 25\% \\
\hline
\end{tabular}
} 
  \caption{Statistics of the experiment. $\beta$ indicates the average treatment probability across each group of tehsils. 
Each group of tehsils was stratified across districts. We use the Negative and Positive Perturbation groups to compute the marginal effects since these groups closely follow Section \ref{sec:designs}. The lower saturation group assigned tehsil's specific treatment probabilities that do not follow the local perturbation in our original design, and therefore not considered as part of our main analysis (we also report effects on this group in Appendix Table \ref{tab:beliefs_forecast2}, see Appendix \ref{app:more_exp1}). } 
  \label{tab:saturations} 
\end{table}


\subsection{Main results: marginal and spillover effects}

Next, we study marginal effects
on beliefs about PxD forecasts (i.e., whether the farmer's prediction agrees with PxD forecast), illustrating properties of our design on our main outcome (forecast temperature). We assume no cluster fixed effects because of lack of baseline outcomes. Although this is a strong assumption, it is motivated by balance across clusters on pre-treatment observables (Table \ref{tab:summaries}). In practice, we recommend to collect baseline outcomes when possible and, as in this case, when infeasible, to check for balance on observable covariates between clusters exposed to different treatment probabilities. 


\vspace{-4mm} 
\paragraph{Inference with a few large clusters} Figure \ref{fig:effects} plots the estimated marginal effects in the main design, i.e., for $\beta \in \{50, 70\}\%$ (with $\eta_n = 10\%$ on average). 
This estimator corresponds to the average difference between all negative and positive perturbation clusters.

We observe decreasing marginal effects when moving from $\beta = 0.5$ to $\beta = 0.7$.  To learn whether PxD should increase expenditure as $\beta = 50\%$ or $\beta = 70\%$, we test the  hypothesis of whether PxD should increase treatment probabilities through a one-sided test. 
The relevant hypothesis is whether PxD should treat more farmers (and expand its budget) to improve information diffusion. 

Table \ref{tab:marginal_effects_my_main} shows that the marginal effect is large and statistically significant at $\beta = 50\%$, preserves its sign but is smaller and non-significant at $\beta = 70\%$. P-values are computed via randomization inference for one-sided test, formally described in Appendix \ref{sec:perm_tests}.  This result is suggestive that treating $50\%$ of the population is sub-optimal, whereas treating $70\%$ of the individuals is close to a local optimum. 

Table \ref{tab:marginal_effects_my_main} reports direct and marginal spillover effects. We observe marginal effects are mostly driven by large and significant marginal spillover effects at $\beta = 50\%$ on control units, whereas \textit{marginal} spillover effects are close to zero at $\beta = 70\%$. 

\vspace{-4mm} 

\paragraph{Summary and cost benefit analysis} In summary, our design allows us to conduct inference on the marginal effect, marginal spillover effects and direct effect with a finite number of large clusters. We find large marginal effects at $\beta = 50\%$, driven by large marginal spillover effects on the control units. We find close to zero marginal effects near $\beta =70\%$. 

Using the two experimental waves, we can estimate the benefit of learning the marginal effects.
Figure \ref{fig:final_benefit} reports the relative improvement from the one wave experiment (i.e., treating $50\%$ of individuals) to the second wave experiment where $70\%$ of individuals are treated. Increasing the number of treated units from $50\%$ to $70\%$ of individuals leads to statistically significant increase in the objective function $W(\beta)$ (decrease in inaccuracy of farmers' predictions). This increase suggests the relevance of learning positive marginal effects at $50\%$. On the other hand, at $70\%$ we find very small marginal improvements of increasing treatment probabilities, suggesting a local optimum.

We contrast our design with a typical grid search method in Figure \ref{fig:final_benefit} that only learns welfare effects at points $\{0, 50\%, 100\%\}$ instead of inducing appropriate perturbations, therefore failing to identify decreasing marginal effects near $70\%$. Such a simpler design would recommend \textit{all} individuals to be assigned to treatment. This would lead to small improvements: We can use as a \textit{conservative} estimate (upper bound) of $W(\beta)$ at $\beta = 100\%$, its Taylor approximation at $\beta = 70\%$, $W(0.7) + 0.3 M(0.7)$.\footnote{This is a conservative estimate even without linearity because we might most likely expect decreasing marginal effects, as supported by Table \ref{tab:marginal_effects_my_main}.} The predicted improvements when treating all units in the population are small relative to only treating $70\%$ (equal to $3\%$) and non-significant. 
Treating only $70\%$ of the individuals instead of $100\%$ would save approximately 0.29\$ per farmer/year over $30\%$ of farmers in the population. This is economically significant if we consider a policy implemented on all farmers in Pakistan (approximately ten million), saving one million US dollars/year.

\begin{minipage}{\textwidth}
  \begin{minipage}[b]{0.49\textwidth}
    \centering
   \includegraphics[scale=0.4]{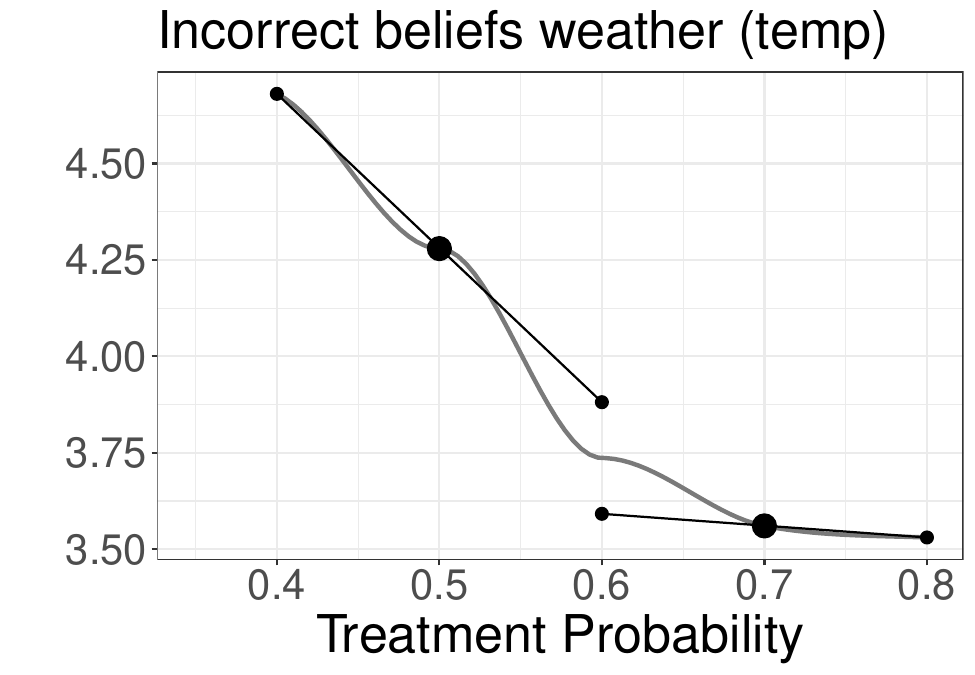}
    \captionof{figure}{Difference between the farmer's predicted temperature and PxD's temperature forecast for the day after the interview. The larger dots report the estimated effects at $\beta = 50\%, \beta = 70\%$, from the first and second wave. The lines report the estimated marginal effects, and the smaller dots the effect estimated at $\beta \in \{40, 60, 80\}\%$ over the first wave (first line) and second wave (second line).} \label{fig:effects}
  \end{minipage}
  \hfill
  \begin{minipage}[b]{0.49\textwidth}
    \centering
   \scalebox{0.6}{\begin{tabular}{@{\extracolsep{5pt}}lcc}
\\[-1.8ex]\hline
\hline \\[-1.8ex]
Incorrect Beliefs about  & \multicolumn{2}{c}{PxD Forecast Temperature}  \\
\\[-1.8ex] & $\beta = 50\%$ (Wave 1) & $\beta = 70\%$ (Wave 2)\\
\hline \\[-1.8ex]
Marginal Effect & -4.00$^{**}$ & -0.31  \\
One-sided p-value & [0.02] & [0.41] \\
& & \\
Direct Effect & -0.90$^{*}$ & -0.66$^{*}$ \\
One-sided p-value & [0.08] & [0.09]  \\
& & \\
Marginal Spillovers on Treated & 0.43 & -0.86  \\
One-sided p-value & [0.45] & [0.37]  \\
& &  \\
Marginal Spillovers on Controls & -6.39$^{*}$ & 0.65  \\One-sided p-value & [0.06] & [0.43]  \\
  & & \\
\hline \\[-1.8ex]
Observations & 259 & 583   \\
\hline \\[-1.8ex]
\textit{Note:}  & \multicolumn{2}{r}{$^{*}$p$<$0.1; $^{**}$p$<$0.05; $^{***}$p$<$0.01}
\end{tabular}
}
      \captionof{table}{Estimated effects over first and second wave from the main design. P-values are computed via randomization inference for one-sided tests.} \label{tab:marginal_effects_my_main} 
    \end{minipage}
  \end{minipage}


\begin{figure}[!ht]
    \centering
   \includegraphics[scale = 0.5]{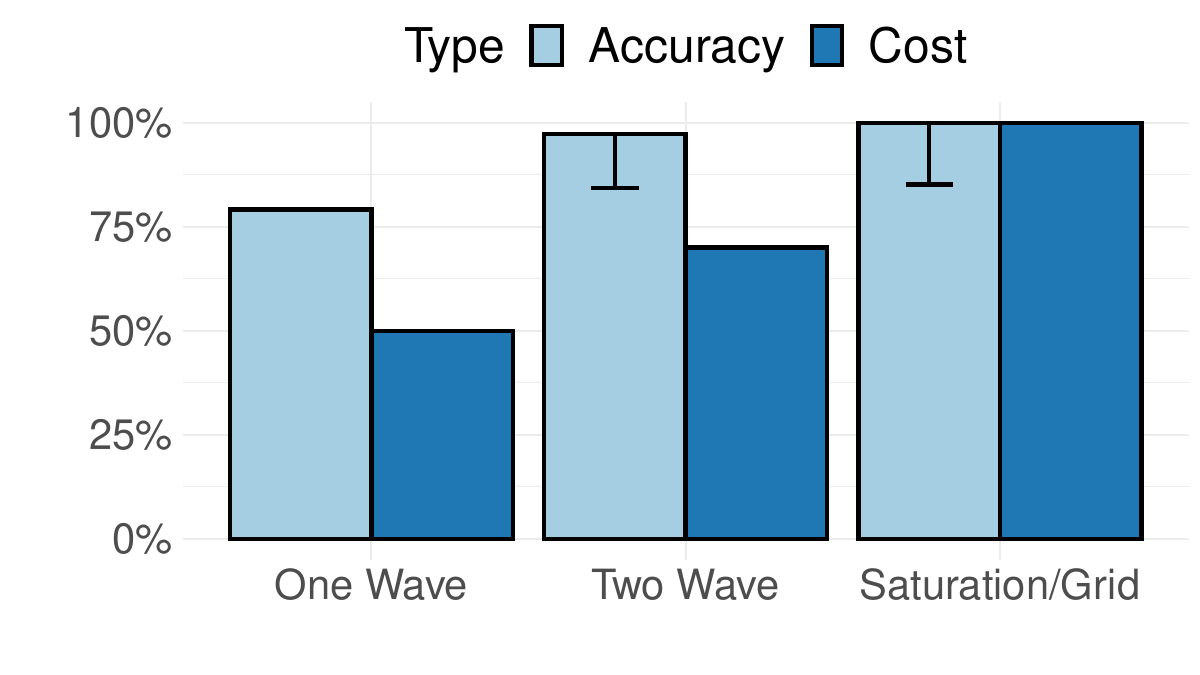}
    \caption{Benefits of a sequential experiment using predicted temperature as the target objective $W(\beta)$. The light blue column reports the percentage change in average forecast accuracy  generated by a policy recommendation using information from the marginal effect, either with one-wave experiment (first column), or two waves (third column), or using a standard saturation experiment with probabilities $\{0, 0.5, 1\}$ (fifth column). The second, fourth, and sixth columns report the cost of the intervention that would be recommended by each of these experiments relative to the cost of treating everybody in the population. 
    The error bars report $10\%$ lower confidence intervals over the improvement from the first to the second wave and from the two wave to treating everybody in the population (what Saturation/Grid would suggest). These are obtained via randomization inference on the gradient at $\beta = 0.5$ and $\beta = 0.7$, respectively, and using a first-order Taylor approximation to the reward around $0.7$ to obtain a conservative estimate of the effect at $\beta = 100\%$.     }
    \label{fig:final_benefit}
\end{figure}

 
\subsection{Robustness, assumptions, and applicability of the method} \label{sec:why_assumptions}

We conclude this discussion relating our main assumptions to our empirical setting, and a discussion around our main robustness checks, summarized in Table \ref{tab:robustness_checks}.

Inference rests on two major requirements. First, treatments 
must not generate cluster-level heterogeneity in expected outcomes.  Our outcome model is consistent with numerous 
applications in economics and industry, such as agronomy advice programs as in the context of our applications (e.g.,
\cite{duflo2023chat}) or network-based interventions \citep[e.g.][]{miguel2004worms}. 
We check for balance by comparing baseline characteristics across clusters (Appendix Table~\ref{tab:summaries}) where we observe substantial balance. If significant heterogeneity is present, we recommend balancing 
clusters accordingly (see Appendix~\ref{app:het} for details).

Second, within-cluster dependence must 
be limited (``weak dependence'')---that is, each individual can directly interact with a 
small subset of others in the same cluster relative to the population size. This is motivated by clusters being large regions in our application, where we may expect individuals to interact only with a subset of individuals in the region.

By contrast, if 
outcomes exhibit unobserved heterogeneity across clusters or strong dependence within them, it would be impossible to learn optimal policies without network information, unless we have many independent clusters, which in turn would imply more stringent assumptions on how interference propagates.

When we run more than a single wave in our experiments, the cost-benefit analysis in Figure \ref{fig:final_benefit} also rests on the assumption that there are no carry-over (dynamic) effects for the MPE to be comparable across different periods. 
No carry-over often holds when the intervention has only 
short-term effects \citep[e.g.][]{athey2018design, kasy2019adaptive}. In our application, the treatment (providing weather forecast for the upcoming few days) affects our main target outcome, a proxy for one-day ahead temperature predictions, but, as we show in Appendix \ref{app:more_experiment}, it does not affect the outcome in the upcoming weeks. 

 Finally, it is interesting to study how better weather predictions affect behavioral responses. In particular efficacy on belief is a relevant welfare proxy when this may also affect behavioral responses such as choices of irrigation, use of pesticides and ultimately profits. Two papers make a sharp connection between better forecast accuracy and outcomes: \cite{burlig2024long} study this in the context of monsoon season for long-run forecasts, and complementary follow-up work by \cite{rudder2024learning} illustrate effects on behavioral responses for short-run forecasts.


\begin{table}[!ht]
\centering
\footnotesize
\begingroup
\setlength{\tabcolsep}{3pt}
\renewcommand{\arraystretch}{1.22}
\begin{tabular}{
@{}
>{\raggedright\arraybackslash}p{0.16\textwidth}
>{\raggedright\arraybackslash}p{0.14\textwidth}
>{\raggedright\arraybackslash}p{0.34\textwidth}
>{\raggedright\arraybackslash}p{0.28\textwidth}
@{}
}
\hline
\hline
\textbf{Robustness} & \textbf{Reported in} & \textbf{Description} & \textbf{Summary of findings} \\
\hline
\\[-0.6em]

Regression estimates
&
Table \ref{tab:beliefs_forecast2}
&
Regressions of incorrect beliefs on individual treatment status, cluster treatment probability, and their interaction. 
&
Direct and spillover coefficients generally indicate improved forecasts. 
\\[0.8em]
\hline
\\[-0.6em]

Balance table for cluster heterogeneity
&
Table \ref{tab:summaries}
&
Tests whether clusters assigned to different treatment-probability groups are comparable in baseline administrative characteristics.
&
Baseline characteristics are similar across treatment-probability groups. 
\\[0.8em]
\hline
\\[-0.6em]

survey-response balance
&
Tables \ref{tab:summaries_response_rates}, \ref{tab:summaries_response_rates_v2}, \ref{tab:table_non_respondents}
&
Checks whether the surveyed sample is balanced across treatment-probability groups and whether respondents to the weather-belief question differ from non-respondents in baseline characteristics. 
&
Surveyed individuals are broadly balanced across treatment-probability groups. Respondents and non-respondents are also similar on most baseline covariates.
\\[0.8em]
\hline
\\[-0.6em]

dynamic effects
&
Table \ref{tab:dynamics}
&
Tests whether treatment effects vary across waves by including a second-wave indicator and its interaction with individual treatment. 
&
For our main outcome (temperature beliefs), treatment effects preserve sign and magnitude, and the treatment-by-second-wave interaction is close to zero and not significant. 
\\[0.8em]
\hline
\\[-0.6em]

treatment efficacy
&
Tables \ref{tab:preliminary_treatment2}, \ref{tab:forecast}
&
Checks that the intervention affected engagement and that the forecast-based outcome is meaningful.
&
Treated farmers receive more calls and have significantly higher response per call. PxD forecasts strongly track realized precipitation and temperature.
\\[0.8em]
\hline
\\[-0.6em]

More refined marginal effects
&
Table \ref{tab:marginal_effects_ma}
&
Uses smaller five-percentage-point perturbations to estimate additional marginal effects around \(\beta=40\%\) and \(\beta=60\%\). 
&
The smaller-perturbation estimates are substantially noisier. This supports using the better-powered \(10\%\) perturbation as the main specification, and illustrates how in practice we may choose \(\eta\).
\\[0.8em]
\hline
\\[-0.6em]

Pairing sensitivity
&
Figure \ref{fig:p_values_distribution}
&
Reports \(p\)-values over 5000 randomly selected pairings ex-ante feasible under the chosen design of positive- and negative-perturbation clusters.
&
average p-value for the marginal effect at $\beta= 0.5$ is below $0.05$ and always below $0.1$ across, and always above $0.1$ for $\beta = 0.7$. Studentized $t$-statistic is on average around $-3.41$ for $\beta = 0.5$ and $-0.83$ for $\beta = 0.7$. 
\\[0.4em]

\hline
\hline
\end{tabular}
\endgroup
\renewcommand{\arraystretch}{1}
\caption{Summary of robustness checks and additional empirical analyses.}
\label{tab:robustness_checks}
\end{table}

\section{Multi-wave experiment: regret guarantees} \label{sec:main_design}

How should we run the experiment if we are given the possibility of running more waves? This question is relevant for many applications, including those involving agronomy advice \citep[e.g.][]{kasy2019adaptive}.  We therefore discuss the design of sequential experiments. 

For illustrative purposes, we provide the algorithm for the one-dimensional case $p = 1$, in Algorithm \ref{alg:adaptive}, that is, when $\beta \in \mathcal{B} = [\mathcal{B}_1, \mathcal{B}_2]$ is a scalar. In Appendix \ref{app:algorithms}, we provide the complete algorithm for the $p$-dimensional case. Let $\hat{M}_{k, t}$ be as in Equation \eqref{eqn:defn1} for $k$ odd. 

 Theoretical results are for the general $p$-dimensional case ($p$ is finite). Let $\check{T} = T/p$.
 
  \begin{algorithm} [!h]   \caption{Multiple-wave experiment with $\beta$ scalar}\label{alg:adaptive}
  \footnotesize  
    \begin{algorithmic}[1]
    \Require Starting value $\beta_0$ in the interior of $[\mathcal{B}_1 + \eta_n, \mathcal{B}_2 - \eta_n]$, $K$ clusters, $T + 1$ periods, constant $\bar{C}$. 
    \State   Create pairs of clusters $\{k, k+1\}, k \in \{1, 3, \cdots, K-1\}$; 
    \State $t = 0$ (initialization): 
    \begin{algsubstates}
        \State  Assign treatments as 
    $
    D_{i,0}^{(h)} | X_i^{(h)} = x \sim \mathrm{Bern}(\pi(x, \beta_0)) \text{ for all } h \in \{1, \cdots, K\}. 
    $       
     \State For $n$ units in each cluster observe $Y_{i,0}^{(h)}, h \in \{1, \cdots, K\}$; initialize $\widehat{M}_{k,t} = 0$, $\check{\beta}_{k}^0 = \beta_0$. 
        \end{algsubstates} 
    \While{$1 \le t \leq T$}
    \begin{algsubstates}
    \State Define 
    $$
    \small 
    \begin{aligned} 
    \check{\beta}_{h}^t = \begin{cases} & P_{\mathcal{B}_1 + \eta_n, \mathcal{B}_2 - \eta_n}\Big[\check{\beta}_{h}^{t-1} + \alpha_{h + 2, t}\widehat{M}_{h + 2,t - 1}\Big], \quad h \in \{1, \cdots, K - 2\},  \\
    & P_{\mathcal{B}_1 + \eta_n, \mathcal{B}_2 - \eta_n}\Big[\check{\beta}_{h}^{t-1} + \alpha_{1, t} \widehat{M}_{1,t - 1}\Big], \quad \quad \quad h \in \{K - 1, K\};
    \end{cases} 
    \end{aligned}
    $$ 
    where $\alpha_{k,t}$ is the learning rate $P_{a, b}(x) = \mathrm{arg} \min_{x' \in [a, b]^p} ||x - x'||^2$.
        \State  Assign treatments as (for  $\bar{C} n^{-1/2}< \eta_n < \bar{C} n^{-1/4}$)
   \begin{equation} \label{eqn:rand_adaptive}
   \small \begin{aligned} 
    D_{i,t}^{(h)} | X_i^{(h)} = x \sim  
    \mathrm{Bern}(\pi(x, \hat{\beta}_{h,t})), \quad  \hat{\beta}_{h,t} = \begin{cases} & \check{\beta}_{h}^t  + \eta_n \text{ if } h \text{ is odd}  \\
    & \check{\beta}_{h}^t  - \eta_n \text{ if } h \text{ is even}
    \end{cases} 
    \end{aligned} 
    \end{equation} 
        \State For $n$ units in each cluster $h \in \{1, \cdots, K\}$ observe $Y_{i,t}^{(h)}$; 
        \State For each pair $\{k, k+1\}$, estimate
      \begin{equation} \label{eqn:defn1}
      \small 
      \begin{aligned} 
      \hat{M}_{k,t} = 
      \hat{M}_{k + 1,t} = \frac{1}{2 \eta_n} \Big[\bar{Y}_t^{(k)} - \bar{Y}_0^{(k)}\Big] -  \frac{1}{2 \eta_n} \Big[\bar{Y}_t^{(k + 1)} - \bar{Y}_0^{(k +1)}\Big]. 
      \end{aligned} 
      \end{equation}        
        \EndWhile
      \end{algsubstates}
 \State Return 
 $
 \hat{\beta}^* = \frac{1}{K} \sum_{k = 1}^K \check{\beta}_{k}^T
 $  
 
         \end{algorithmic}
\end{algorithm}

 The algorithm pairs clusters (here two consecutive clusters form a pair) and initializes clusters at the same starting value $\beta_0$, $\check{\beta}_1^1 = \cdots = \check{\beta}_K^1 = \beta_0$. At $t = 0$, it randomizes treatments independently using the same starting value $\beta_0$ for all clusters. 
 Here, $\beta_0$ is chosen exogenously, e.g., it is the current policy in place. Over each iteration $t$, we assign treatments based on $\hat{\beta}_{k,t}$  for cluster $k$ at time $t$, which equals the parameter $\check{\beta}_k^t$ obtained from a previous iteration plus a positive (negative) perturbation $\eta_n$ in the first (second) cluster in a pair. The local perturbation follows similarly to what is discussed in the single-wave experiment. Also, by construction, $\check{\beta}_k^t$ is the same for a given pair $(k, k+1)$, where $k$ is odd. We choose $\check{\beta}_k^{t + 1}$ via what we call \textit{sequential cross-fitting}: we wrap clusters in a \textit{circle} and update the parameter in a pair of clusters $(k, k+1)$ using information from the subsequent pair (see Appendix Figure \ref{fig:clusters}). The algorithm runs over $T$ periods and returns 
$
\hat{\beta}^* = \frac{1}{K} \sum_{k = 1}^K \check{\beta}_{k}^{T } .
$
Choosing the average is motivated by the theoretical properties of gradient descent, although other statistics are also possible.

\begin{figure}[!ht]
\vspace{-3mm}
\centering
\begin{tikzpicture}[
node distance = 6mm and 6mm, 
  start chain = going right,
    mw/.style = {minimum width=#1},
  list/.style = {rectangle split, rectangle split parts=1,
                 rectangle split horizontal, draw,
                 align=center,
                 text width=7mm, 
                 minimum height=9mm, 
                 inner sep=1mm, on chain}, 
          > = stealth, 
          ]

  \node[list] (A) {\nodepart[mw=7mm]{two}  };
  \node[list] (B) {\nodepart[mw=7mm]{two}  };
  \node[list] (C) {\nodepart[mw=7mm]{two} 3};

  \node[list, below=of A] (D) {\nodepart[mw=7mm]{two}  };
  \node[list] (E) {\nodepart[mw=7mm]{two}  };
  \node[list] (F) {\nodepart[mw=7mm]{two} 3};
 \node[circle] (g) at (4.5,-0.8) {$\longrightarrow$};

    \node[list, fill = cyan] (A) at (5, 0) {\nodepart[mw=7mm]{two}  };
  \node[list, fill = cyan] (B) {\nodepart[mw=7mm]{two}  };
  \node[list, fill = cyan] (C) {\nodepart[mw=7mm]{two} 3};

  \node[list, below=of A, fill = gray] (D) {\nodepart[mw=7mm]{two}  };
  \node[list, fill = gray] (E) {\nodepart[mw=7mm]{two}  };
  \node[list, fill = gray] (F) {\nodepart[mw=7mm]{two} 3};
  
   \node[circle] (g) at (10.5,-0.8) {$\longrightarrow$};
  \draw[-] (A) edge (D) (B) edge (E) (C) edge (F);
  \node[list, fill = cyan] (A) at (11, 0) {\nodepart[mw=7mm]{two}  };
  \node[list, fill = cyan] (B) {\nodepart[mw=7mm]{two}  };
  \node[list, fill = cyan] (C) {\nodepart[mw=7mm]{two} 3};

  \node[list, below=of A, fill = gray] (D) {\nodepart[mw=7mm]{two}  };
  \node[list, fill = gray] (E) {\nodepart[mw=7mm]{two}  };
  \node[list, fill = gray] (F) {\nodepart[mw=7mm]{two} 3};
  
  \draw[-] (A) edge (D) (B) edge (E) (C) edge (F);
  
  \path[->,  -triangle 90] (A) edge [bend right = 300]  (B);
   \path[->,  -triangle 90] (B) edge [bend right = 300]  (C);
    \path[->,  -triangle 90] (F) edge [bend right = -90]  (D);
\end{tikzpicture}
\vspace{-4mm}
\caption{Sequential cross-fitting method. Clusters (rectangles) are paired. Within each pair, researchers assign different treatment probabilities to clusters with different colors. Finally, the policy in each pair is updated using information from the consecutive pair, with the first pair using information from the last pair. Note that because $K > 2 T$, the outcomes in each cluster are independent of the gradient estimated in the subsequent cluster over each iteration $t$. 
 } \label{fig:clusters}
\end{figure}

\begin{rem}[Why sequential cross fitting] \label{rem:aaa} 
Next, we illustrate the source of bias if the sequential cross-fitting was not employed in the adaptive experiment. Every period, the researcher can only identify the expected outcome of $Y_{i,t}^{(k)}$ conditional on the parameter $\hat{\beta}_{k,t}$, namely (omitting fixed effects for simplicity)
$
\widetilde{W}(\hat{\beta}_{k,t}) = \mathbb{E}_{\hat{\beta}_{k,t}}[Y_{i,t}^{(k)} | \hat{\beta}_{k,t}]. 
$
If $\hat{\beta}_{k,t}$ were chosen exogenously, based on information from a different cluster, $\mathbb{E}_{\hat{\beta}_{k,t}}[Y_{i,t}^{(k)} | \hat{\beta}_{k,t}] = \mathbb{E}_{\hat{\beta}_{k,t}}[Y_{i,t}^{(k)}] = W(\hat{\beta}_{k,t})$, where $W(\hat{\beta}_{k,t})$ defines the expected outcome once we deploy the policy $\hat{\beta}_{k,t}$ on a new population. However, the equality conditional and unconditional on $\hat{\beta}_{k,t}$ does not occur when $\hat{\beta}_{k,t}$ is estimated using information on $Y_{i,t-1}^{(k)}$. Consider the example where the outcome depends on some auto-correlated unobservables $\nu_{i,t}$ and treatment assignments in Figure \ref{fig:dag}. 
The \textit{dependence} structure of Figure \ref{fig:dag} implies:
$
W(\hat{\beta}_{k,t}) = \mathbb{E}_{\hat{\beta}_{k,t}}[Y_{i,t}^{(k)}] \neq \mathbb{E}_{\hat{\beta}_{k,t}}[Y_{i,t}^{(k)} | \hat{\beta}_{k,t}] = \widetilde{W}(\hat{\beta}_{k,t}),  
$ 
if $\hat{\beta}_{k,t}$ depends on covariates and unobservables from previous outcomes (and so on unobservables $\nu_{i,t}^{(k)}$) in cluster $k$. 
Here, $W(\hat{\beta}_{k,t}) $ captures the estimand of interest. Instead, $\widetilde{W}(\hat{\beta}_{k,t})$ denotes what we can identify.
Algorithm \ref{alg:adaptive}   breaks such dependence and guarantees unconfounded experimentation.  We formalize this intuition in Appendix Lemma \ref{lem:1a}.  \qed 
\end{rem}

\begin{rem}[Learning rate] \label{rem:learning_rate}
We are left to discuss how ``large" the step size should be: if the marginal effect is positive, by how much should we increase the treatment probability? Assuming strong concavity of the objective function, the learning rate $\alpha_{k,t}$ should be of order $1/t$ (e.g., $J/t$) to control both the in-sample and out-of-sample regret for small $J$ (e.g., $10\%$). When $\beta$ denotes a treatment probability a natural choice is $J \in [10\%, 20\%]$.  A more robust choice with moderate or large $T$ as we formally discuss in Appendix Theorem \ref{thm:rate} is 
\begin{equation} \label{eqn:gradient}
\small 
\begin{aligned} 
\alpha_{k,t} = \begin{cases} 
&  \frac{J}{\check{T}^{1/2 - v/2} ||\hat{M}_{k,t}||} \text{ if } ||\hat{M}_{k,t}||_2 >  \frac{c}{\check{T}^{1/2 - v/2}} +  \epsilon_n, \\
&0\text{ otherwise} 
\end{cases} ,   
\end{aligned} 
\end{equation} 
for a positive $\epsilon_n$, $\epsilon_n \rightarrow 0$, and small constants $1 \ge v, J, c > 0$.\footnote{Formally, we let   $\epsilon_n$ be proportional to $\sqrt{\frac{\gamma_N}{\eta_n^2 n}} + \eta_n$. See Theorem \ref{thm:rate} for more details. }  
Here, the learning rate divides the estimated marginal effect by its norm (known as gradient norm rescaling, \citealt{hazan2015beyond}) and controls the out-of-sample regret under weaker conditions on $W(\beta)$ (strict quasi-concavity instead of strong concavity).    \qed 
\end{rem}

\begin{figure}[!h]
\centering 
    \begin{tikzpicture}[scale = 1.16]
    \node[draw, black,ultra thick, inner sep=0pt,
  text width=12mm,
  align=center,   circle] (h) at (-3,-2) {$Y_{i,t-1}^{(k)}$};
    \node[draw, black,ultra thick,  inner sep=0pt,
  text width=12mm,
  align=center, circle] (d) at (-3,-4) {$\hat{\beta}_{k,t}$};
    \node[draw, black,ultra thick,inner sep=0pt,
  text width=12mm,
  align=center, circle] (b) at (1,0) {$\nu_{i,t}^{(k)}$};

    \node[] (e) at (-8,-5) {Policy on a new population};
   \node[] (e) at (-1,-5) {Experiment with repeated sampling};
    \node[draw, black,ultra thick, inner sep=0pt,
  text width=12mm,
  align=center, circle] (e) at (1,-2) {$Y_{i,t}^{(k)}$};
 \node[draw, black,ultra thick,  inner sep=0pt,
  text width=12mm,
  align=center, circle] (f) at (1,-4) {$D_{i,t}^{(k)}$};
 \node[draw, black,ultra thick,  inner sep=0pt,
  text width=12mm,
  align=center, circle] (g) at (-3,0) {$\nu_{i,t - 1}^{(k)}$};

    \node[draw, black,ultra thick, inner sep=0pt,
  text width=12mm,
  align=center,   circle] (hh) at (-10,-2) {$Y_{i,t-1}^{(k)}$};
    \node[draw, black,ultra thick,  inner sep=0pt,
  text width=12mm,
  align=center, circle] (dd) at (-10,-4) {$\beta^*$};
    \node[draw, black,ultra thick,inner sep=0pt,
  text width=12mm,
  align=center, circle] (bb) at (-6,0) {$\nu_{i,t}^{(k)}$};

    \node[draw, black,ultra thick, inner sep=0pt,
  text width=12mm,
  align=center, circle] (ee) at (-6,-2) {$Y_{i,t}^{(k)}$};
 \node[draw, black,ultra thick,  inner sep=0pt,
  text width=12mm,
  align=center, circle] (ff) at (-6,-4) {$D_{i,t}^{(k)}$};
 \node[draw, black,ultra thick,  inner sep=0pt,
  text width=12mm,
  align=center, circle] (gg) at (-10,0) {$\nu_{i,t - 1}^{(k)}$};

    \draw[->, -triangle 90]    (g) edge (h)  (h) edge (d) 
    (d) edge (f) (f) edge (e) (b) edge (e) (g) edge (b);

   \draw[->, -triangle 90] (dd) edge (hh)   (gg) edge (hh) 
    (dd) edge (ff) (ff) edge (ee) (bb) edge (ee) (gg) edge (bb);

    \end{tikzpicture}
    \caption{The left panel shows the dependence structure when a static policy is implemented on a new population (we omit $D_{i,t-1}^{(k)}$ for expositional convenience), where $\nu_{i,t}$ denote unobservable characteristics. The right panel shows the dependence structure of a sequential experiment that uses the same units for policy updates over subsequent periods with \textit{repeated} sampling. 
} \label{fig:dag}
    \end{figure}

The main distinction from most of the previous literature on adaptive experiments \citep[e.g.][]{ wager2019experimenting, hadad2019confidence, zhang2020inference} is that in all such references repeated sampling does not occur, and batches are independent each period. Here, instead, clusters are dependent over each period, motivating our sequential estimation procedure.

\subsection{Regret guarantees} 

Next, we provide the main regret guarantees. 

\begin{ass}\label{ass:bounded}  Let (A) $Y_{i,t}^{(k)}$ be sub-Gaussian with sub-gassian parameter bounded by $r^2 < \infty$; and  (B) $K \ge 2(T/p + 1)$.  
 \end{ass}

Condition (A) states that unobservables have sub-Gaussian tails (attained by bounded random variables); (B) assumes that the number of clusters is at least twice the number of waves, which guarantees that Assumption \ref{defn:bernoulli} (as shown in Appendix Lemma \ref{lem:1a}) holds.

\begin{ass}[Strong concavity] \label{ass:strong_concavity} Assume $W(\beta)$ is $\sigma$-strongly concave, for some $\sigma > 0$ (i.e., $W(\beta)$'s Hessian is strictly negative definite), with interior $\beta^*$.  
\end{ass} 

An example is Example \ref{exmp:main}, where neighbors' effects induce decreasing marginal effects, and the treatment may entail some costs, see  real-world data example in Appendix Figure \ref{fig:objectives_analysis}. Strong concavity also arises in linear models with negative externalities. A comprehensive discussion is below in Remark \ref{rem:concavity}.

\begin{thm}[Main regret guarantees] \label{thm:rate2b}   Let  Assumptions \ref{ass:ass_0}, \ref{ass:regularity_basic}, \ref{ass:bounded}, \ref{ass:strong_concavity} hold. Take a small $1/4 > \xi > 0$, $\alpha_{k,w} = J/w$ for a finite $J \ge 1/\sigma$.   Let $n^{1/4 - \xi} \ge C \sqrt{p \log(n) \gamma_N T^{B p} \log(KT)}$, $\eta_n = 1/n^{1/4 + \xi}$, for finite constants $B, C > 0$. Then, with probability at least $1 - 1/n$, for constants $c, \bar{C} < \infty$, both independent of $(p, n, N, K, \check{T})$,
\begin{equation} \label{eqn:rate_all} 
\small 
\begin{aligned} 
\max_{k \in \{1,\cdots, K\}} \frac{1}{\check{T}} \sum_{w=1}^{\check{T}} \Big[W(\beta^*) - W(\check{\beta}_k^{w})\Big] \le  c \frac{p^2 \log(\check{T} + 1)}{\check{T}}, \quad W(\beta^*) - W(\hat{\beta}^*) \le \frac{p^2 \bar{C}}{\check{T}}. 
\end{aligned} 
\end{equation} 
For $K = 2 \check{T} + 2$, Equation \eqref{eqn:rate_all} holds with $K/2 - 1$ in lieu of $\check{T}$. 
\end{thm}  

The proof is in Appendix \ref{sec:proof5}. The theorem formalizes the in-sample and out-of-sample regret bound. Assuming that the sample size grows polynomially in the number of clusters (and $T$), and we choose $K = 2 (T/p + 1)$ (the smallest number of clusters under the stated assumptions), the theorem implies a rate of order $\log(K)/K$ for the in-sample regret.\footnote{Also note that a similar bound can be obtained as we control for the perturbed policies for sufficiently small perturbations $\eta_n$.} 

The rate in $T$ (and therefore $K$) does not depend on $p$,  as $n \rightarrow \infty$. This is different from grid-search procedures, where the rate in $K$ would be exponentially slower in $p$.

\begin{thm}[Out-of-sample regret with larger sample size] \label{thm:rate2bb}   Let  Assumptions \ref{ass:ass_0}, \ref{ass:regularity_basic}, \ref{ass:bounded}, \ref{ass:strong_concavity} hold, with $W(\beta)$ being $\tau$-smooth, and $K = 2 \check{T} + 2$. Take a small $1/4 > \xi > 0$, $\alpha_{k,w} = 1/\tau$.   Let $n^{1/4 - \xi} \ge C \sqrt{p \log(n) \gamma_N e^{T B p} \log(KT)}$, $\eta_n = 1/n^{1/4 + \xi}$, for finite constants $B, C > 0$. Then, with probability at least $1 - 1/n$, for constants $0 < c_{p}, c_{p}'< \infty$, independent of $(n, N, K, \gamma_N, \check{T})$,
$$
\small 
\begin{aligned} 
W(\beta^*) - W(\hat{\beta}^*) \le c_{p} \exp(- c_{p}' K).
\end{aligned} 
$$
\end{thm}  

The proof is in Appendix \ref{sec:exponential2}. The main restriction is that the sample size grows \textit{exponentially} in the number of iterations (instead of polynomially).  The theorem leverages properties of the gradient descent under strong concavity and smoothness \citep{bubeck2012regret}. Fast rates for the out-of-sample regret are achieved under an appropriate choice of the learning rate that leverages the smoothness of the objective function.  The choice of a learning rate invariant in the iteration $t$ requires a sample size exponential in $T$. This differs from the choice of a learning rate as $1/t$ in Theorem \ref{thm:rate2b}, where the adaptive learning rate enables controlling the cumulative error with $n$ growing polynomially in $T$. To our knowledge, these regret guarantees are the first under unknown (and partial) interference.

\begin{rem}[Rates and relevant assumptions] \label{rem:concavity}
The regret guarantees combine two distinct sets of restrictions. The first set is about identification: Assumption \ref{ass:ass_0}, together with the sequential cross-fitting scheme, ensures that marginal effects estimated in one pair of clusters are informative about policy updates in another pair by sharing a common slope. 

The second set of restrictions is about the shape of \(W(\beta)\), and determines the optimization rate. Assumption \ref{ass:strong_concavity} imposes globally decreasing marginal returns and a unique optimum. This condition is natural in models with decreasing marginal spillovers, as in Example \ref{exmp:main}, or in linear models with negative externalities and treatment costs. Under strong concavity and an adaptive learning rate, Theorem \ref{thm:rate2b} obtains the benchmark \(1/K\)-type rate.

The exponential out-of-sample rate in Theorem \ref{thm:rate2bb} requires stronger assumptions, in particular additional smoothness restrictions and a much larger per-cluster sample size in the number of iterations, which therefore we see as a favorable case. 

Importantly, 
strong concavity can fail in settings where information only spreads after enough individuals have been treated, or more generally in settings with phase-transitions. In this case we may expect that the objective function either has multiple optima or is quasi-concave. In the first case, we should interpret our regret guarantees relative to a \textit{local} instead of global optimum absent global strong concavity. The second case instead is possible to handle (assuming strong concavity only locally at the optimum and strict quasi-concavity globally), where rates are nearly $1/K$ (but non-exponential) for an appropriate choice of the learning rate, see Appendix \ref{sec:quasi_concavity}.  \qed
\end{rem}

We now contrast the above results with past literature. In the online optimization literature, the rate $1/T$ is common for convex optimization, assuming independent units \citep[see][for out-of-sample regret rates]{duchi2018minimax}. Here, because of interference, we leverage between-clusters perturbations. Also, we do not have direct access to the gradient, and related optimization procedures are those in the literature on zero-th order optimization \citep{kiefer1952stochastic}. \cite{flaxman2004online, agarwal2010optimal} in particular are related to our approach, where regret can converge at rate $O(1/T)$ in expectation only, whereas high-probability bounds are $1/\sqrt{T}$ \citep[see Theorem 6 in][and the discussion below]{agarwal2010optimal}.  Here, we exploit within-cluster concentration and between clusters' variation to control for large deviations of the estimated gradients and obtain faster rates for high-probability bounds. This approach also allows us to extend out-of-sample guarantees beyond global strong concavity (assumed in the above references) in Appendix \ref{sec:quasi_concavity}. In our derivations, the perturbation parameter depends on the sample size, differently from the references above, and the idea of sequential estimation is novel due to repeated sampling. 
 \cite{wager2019experimenting} derive $1/T$ regret guarantees in the different settings of market pricing, as $n \rightarrow \infty$, with independent units and samples each wave. 
Our results study a complementary framework with locally dependent units and partial interference.  \cite{viviano2019policy} considers a single network, with \textit{observed} neighbors of experiment participants, instead of a sequential experiment. He imposes geometric (VC) restrictions on the policy and solves a mixed-integer linear program. Here, we introduce an adaptive experiment and we do not require network information, using network concentration not studied in previous works.

These differences require a different set of techniques for derivations.  The proof of the theorem (i) uses concentration arguments for locally dependent graphs \citep{janson2004large}; (ii) uses the within-cluster and between-clusters variation for consistent estimation of the marginal effect, together with the cluster pairing; (iii) it uses a recursive argument to bound the cumulative error obtained through the estimation and sequential cross-fitting.

\subsection{Calibrated numerical studies} 

We complement our theoretical findings by collecting numerical studies \textit{calibrated} to real-world applications in \cite{cai2015social} and \cite{alatas2012targeting}. In these studies, we simulate a quadratic model of spillover effects estimated from the data and report welfare effects of our proposed procedure. Details on the implementation and estimation are contained in Appendix \ref{app:experiment}. We compare our method to the best competitor between the one that maximizes the estimated reward obtained from a correctly specified quadratic model and the one that chooses the treatment with the largest value within a grid of values. Our method leads to significant out-of-sample welfare improvements for both applications. Improvements are generally larger for larger $T$. The panel at the bottom of Table \ref{tab:cov02} reports positive and large improvements for the in-sample reward across all the designs, worst-case across clusters. These are often increasing in $T$ with a few exceptions since uniform concentration may deteriorate for large $T$ and small $n$. In Appendix \ref{sec:designs2}, we study the performance of the one-wave experiment.
In the online Appendices  \ref{sec:aa1}, \ref{sec:aa2}, we report results across many other specifications of the network, policy functions, and choice of different parameters and different starting values.

\begin{table}[!htp]\centering
\caption{Multiple-wave experiment. $200$ replications. The relative improvement in reward with respect to the best competitor for $\rho = 2$. The panel at the top reports the out-of-sample regret, and the one at the bottom the worst-case in-sample regret. A description of the data-generating process is in Appendix \ref{app:experiment}. 
}    \label{tab:cov02} 
\ra{1.3}
\scalebox{0.8}{\begin{tabular}{lcccc|cccc}
\hline\hline
& \multicolumn{4}{c}{Information} & \multicolumn{4}{c}{Cash Transfer} \\
\cline{2-5}\cline{6-9}
$T =$ & 5 & 10 & 15 & 20 & 5 & 10 & 15 & 20 \\
\hline
$n = 200$ & 0.058 & 0.147 & 0.297 & 0.212 & 0.621 & 0.520 & 0.737 & 0.752 \\
$n = 400$ & 0.227 & 0.210 & 0.355 & 0.346 & 0.653 & 0.746 & 0.874 & 0.899 \\
$n = 600$ & 0.299 & 0.281 & 0.345 & 0.493 & 0.647 & 0.801 & 0.942 & 1.125 \\
\hline
$n = 200$ & 0.233 & 0.243 & 0.265 & 0.287 & 0.247 & 0.279 & 0.300 & 0.321 \\
$n = 400$ & 0.243 & 0.274 & 0.321 & 0.335 & 0.267 & 0.307 & 0.344 & 0.353 \\
$n = 600$ & 0.262 & 0.314 & 0.343 & 0.360 & 0.294 & 0.360 & 0.388 & 0.387 \\
\hline
\end{tabular}
}
\end{table}

\section{Conclusions} \label{sec:conclusions}

This paper makes two main contributions. First, it introduces a single-wave experimental design to estimate the marginal effect of the policy and test for policy optimality. The experiment also enables identifying and estimating treatment effects, which can be of independent interest.  
Our design is motivated by our novel empirical application, that illustrates its advantages in a large scale implementation to study information campaigns for climate adaptation.  Our empirical application shows that using the marginal effect can be informative for decision-making even with few (two) waves. 
Second, it introduces an adaptive experiment to estimate optimal policies. We derive asymptotic properties for inference and provide a set of guarantees on the in-sample and out-of-sample regret. 


This work opens new questions in the context of experimental design with spillovers. We leave to future research the study of the properties of the estimators in settings where (i) clusters are not fully disconnected, in the spirit of   \cite{leung2023network};  (ii) clusters need to be estimated from the data in the spirit of \cite{viviano2023causal}; (iii)  matching clusters based on observable covariates, as we discuss in Appendix \ref{app:het}. Similarly, studying the properties of the proposed method, as the degree of interference is proportional to the sample size, is an interesting direction. This is theoretically possible, in the spirit of what discussed in Theorem \ref{thm:const1}, and we leave its comprehensive analysis to future research. Finally, an open question is how to estimate policies when the network is only partially observed \citep[e.g.,][]{breza2017using, manresa2013estimating}, and how to measure costs and benefits of collecting network data, for which Appendix \ref{sec:comp} provides novel directions for future research.

\bibliography{my_bib2}
\vspace{-2mm}
\bibliographystyle{chicago}

\newpage

\begin{center}
    \LARGE Appendix to ``Policy design in experiments with unknown interference''
\end{center}

\begin{appendices}

\section{Main extensions and discussions} \label{sec:extensions1}

\vspace{-1.5mm}

\subsection{General equilibrium effects} \label{sec:global} 

In this section, the treatment affects each unit in a cluster $k$ through a global interference mechanism mediated by a variable $p_t^{(k)}$. For simplicity, we let $X_i^{(k)} = 1$.

\begin{ass}[Global interference] \label{ass:global} Let treatments be assigned as in Assumption \ref{defn:bernoulli}. Let 
$$
\small 
\begin{aligned} 
Y_{i,t}^{(k)} = \alpha_t + \tau_k + g\Big(p_t^{(k)}, \hat{\beta}_{k,t}\Big) + \varepsilon_{i,t}^{(k)}, \quad \mathbb{E}_{\hat{\beta}_{k, 1:t}}\Big[\varepsilon_{i,t}^{(k)} | p_{t}^{(k)}\Big] = 0, 
\end{aligned}
$$  
for some function $g(\cdot)$ unknown to the researcher, bounded and twice continuously differentiable with bounded derivatives, and unobservable $p_t^{(k)}$. Assume in addition that $\varepsilon_{i,t}^{(k)} \perp \varepsilon_{j \not \in \mathcal{I}_i^{(k)},t}^{(k)} | \hat{\beta}_{k,1:t}, p_{t}^{(k)}$ for some set $|\mathcal{I}_i^{(k)}| = \mathcal{O}(\gamma_N)$.  
\end{ass} 

Assumption \ref{ass:global} states that the outcome within each cluster is a function of a common factor, and treatment assignment rule $\hat{\beta}_{k,t}$.

\begin{ass}[Global interference component] \label{ass:factor} Let treatments be assigned as in Assumption \ref{defn:bernoulli}. Assume that 
$
p_t^{(k)} = q(\hat{\beta}_{k,t}) +  o_p(\eta_n),  
$  
with $q(\beta)$ being unknown, bounded and twice continuously differentiable in $\beta$ with uniformly bounded derivatives. 
\end{ass}

Assumption \ref{ass:factor} states that the factor can be expressed as the sum of two components. The first component $q(\cdot)$ depends on the policy parameter $\hat{\beta}_{k,t}$ assigned at time $t$. The second component is a stochastic component that depends on the realized treatment effects. We illustrate an example below.

\begin{exmp}[Within cluster average] \label{exmp:general} Suppose that 
$
Y_{i,t}^{(k)} = t(\bar{D}_t^{(k)}, \nu_{i,t}),  \nu_{i,t}^{(k)} \sim_{i.i.d.} \mathcal{P}_{\nu},  D_{i,t}^{(k)} \sim_{i.i.d.} \mathrm{Bern}(\beta)
$ 
where $t(\cdot)$ is some arbitrary (smooth) function. Then
$p_t^{(k)} = \bar{D}_t^{(k)}$ i.e., individuals depend on the average exposure in a cluster. We can write 
$
Y_{i,t}^{(k)} = t(p_t^{(k)},\nu_{i,t}^{(k)})$  where $p_t^{(k)} = \beta + (\bar{D}_t^{(k)} - \beta),$
which satisfies Assumption \ref{ass:factor} for $\eta_n = n^{-1/3}$ or larger. \qed 
\end{exmp} 




We are interested in 
$
M_g(\beta) = \frac{\partial W_g(\beta)}{\partial \beta}, W_g(\beta) = g( q(\beta), \beta). 
$  
Estimation of the marginal effect follows similarly to Equation \eqref{eqn:gradient_main}. The following theorem guarantees consistency. 

\begin{thm} \label{thm:const_steady} Let Assumption \ref{ass:global}, \ref{ass:factor} hold with subgaussian $\varepsilon_{i,t}^{(k)}$, $X_i = 1$. For $\widehat{M}_{(k, k+1)}$ as in Algorithm \ref{alg:my_pilot}, for $k$ being odd: 
$
\Big|\widehat{M}_{(k, k+1)} - M_g(\beta)\Big| = \mathcal{O}_p\left(\sqrt{\frac{\gamma_N \log(n\gamma_N)}{\eta_n^2 n }} +  \eta_n \right) + o_p(1). 
$ 
\end{thm}

The proof is in Appendix \ref{sec:proof8}. Given the estimated marginal effect we can follows verbatim  Algorithm \ref{alg:adaptive} to estimate optimal policies.

 \subsection{Non-adaptive experiment with local perturbations} \label{app:non_adaptive}
 
 This subsection serves two purposes. First, it sheds light on comparisons of the adaptive procedure with grid-search-type methods, showing drawbacks of the grid-search approach in terms of convergence of the regret. Second, it shows how, when an adaptive procedure is not available, we can still use information from the marginal effect estimated as we propose in Algorithm \ref{alg:my_pilot}, to improve convergence rates in $K$.

 The algorithm that we propose is formally discussed in Algorithm \ref{alg:one_wave_optimal} and works as follows. First, we construct a fine grid $\mathcal{G}$ of the parameter space $\mathcal{B}$ (with $p$ dimensions), with equally spaced parameters. Second, we pair clusters, and assign a \textit{different} parameter $\beta^k$ for each pair $(k, k+1)$ from the grid $\mathcal{G}$. Third, in each pair, we estimate the gradient $\widehat{M}_{(k, k+1)} \in \mathbb{R}^p$, by perturbing, sequentially for $T = p$ periods, one coordinate at a time of the parameter $\beta^k$.\footnote{Sequentiality here is for notational convenience only, and can be  replaced by $T = 1$, but with $2 p$ clusters  allocated to each coordinate.}  We estimate reward using a first-order Taylor expansion 
\begin{equation} \label{eqn:expansion} 
\small 
\begin{aligned} 
 \widehat{W}(\beta) = \bar{W}^{k^*(\beta)} + \widehat{M}_{(k^*(\beta), k^*(\beta)+1)}^\top (\beta - \beta^{k^*(\beta)}), \quad \hat{\beta}^{ow} = \mathrm{arg} \max_{\beta \in \mathcal{B}} \widehat{W}(\beta), 
 \end{aligned} 
 \end{equation}  
 $$
 \small 
 \begin{aligned} 
\text{where }  k^*(\beta) = \mathrm{arg} \min_{k \in \{1, 3, \cdots, K - 1\}, \beta^k \in \mathcal{G}} ||\beta^k - \beta||^2, \quad \bar{W}^k = \frac{1}{2} \Big[\frac{1}{T} \sum_{t=1}^T \bar{Y}_t^k - \bar{Y}_0^k + \frac{1}{T} \sum_{t=1}^T \bar{Y}_t^{k+1} - \bar{Y}_0^{k+1}\Big].
\end{aligned}
 $$
Here, $\bar{Y}_t^k$ is the average outcome in cluster $k$ at time $t$, 
 and $\widehat{M}_{(k^*, k^*+1)}$ is estimated as in Algorithm \ref{alg:one_wave_optimal}.  
 We can now characterize guarantees as $n \rightarrow \infty$, and $K, p < \infty$.  
 
 \begin{thm} \label{thm:grid_search1} Suppose that $Y_{i,t}^{(k)}$ is sub-Gaussian. Let  Assumptions \ref{ass:ass_0}, \ref{ass:regularity_basic}, and $\eta_n = o(n^{-1/4})$, $\gamma_N \log(n \gamma_N K)/(\eta_n^2 n) = o(1)$. Consider $\hat{\beta}^{ow}$ as in Algorithm \ref{alg:one_wave_optimal}, with $\mathcal{B} \subseteq [0,1]^p$. For a constant $\bar{C} < \infty$ independent of $(n, T, K)$,  
 $
\lim_{n \rightarrow \infty} P\left(W(\beta^*) - W(\hat{\beta}^{ow}) \le \frac{\bar{C}}{K^{2/p}} \right) = 1.   
 $
 \end{thm}
 
 The proof is in Appendix \ref{sec:aah}. 
 Theorem \ref{thm:grid_search1} showcases two properties of the method. First, for $p = 1$, the rate of convergence is of order $1/K^2$, which is possible \textit{because} we also estimate and leverage the gradient $\widehat{M}$. The insight is to augment the estimator of the welfare with $\widehat{M}$, since, otherwise, the rate would be slower in $K$.\footnote{By a second-order Taylor expansion, using information from the gradient guarantees that $\widehat{W}(\beta)$ converges to $W(\beta)$ up-to a second-order term of order $O(||\beta - \beta^k||^2)$, instead of a first-order term $O(||\beta - \beta^k||)$.} One drawback of a grid search approach is that, as $p > 1$, the method suffers a curse of dimensionality, and the rate in $K$ decreases as $p$ increases. This is different from the adaptive procedure (e.g., Theorem \ref{thm:rate2b}), where the rate in $K$ does not depend on $p$. A second \textit{disadvantage} of the grid search is that the method does not control the in-sample regret, formalized below.

 \begin{prop}[Non-vanishing in-sample regret] \label{prop:in_sample_grid} Take $\sum_k \beta_k/K = 0.5$ (any other number would work). There exists a strongly concave $W(\cdot)$, such that, for $p = 1$, 
 $
 W(\beta^*) - \frac{1}{K} \sum_{k=1}^K W(\beta^k) \ge c,    
 $
 for $c > 0$ independent of $(n, K, T)$.  
  \end{prop} 
  
  \begin{proof}[Proof of Proposition \ref{prop:in_sample_grid}] By concavity,
$
W(\beta^*) - \frac{1}{K} \sum_{k=1}^K W(\beta^k) \ge 
W(\beta^*) - W( \frac{1}{K} \sum_{k=1}^K \beta^k) = W(\beta^*) - W(0.5),   
$ 
which completes the proof, for a suitable choice of $W(\cdot)$. 
  \end{proof} 
  

\vspace{-4mm}

\subsection{Permutation tests} \label{sec:perm_tests}

\paragraph{Permutation tests} For permutation tests, consider  the vectors 
\begin{equation} \label{eqn:Vs}
\small 
\begin{aligned} 
V_1 = \frac{1}{2 \eta_n} \begin{bmatrix}
& \bar{Y}_1^{(1)} - \bar{Y}_0^{(1)} \\ 
& \bar{Y}_1^{(3)} - \bar{Y}_0^{(3)} \\ 
& \vdots \\ 
& \bar{Y}_1^{(K - 1)} - \bar{Y}_0^{(K - 1)}
 \end{bmatrix}, \quad V_2 = \frac{1}{2 \eta_n} \begin{bmatrix}
& \bar{Y}_1^{(2)} - \bar{Y}_0^{(2)} \\ 
& \bar{Y}_1^{(4)} - \bar{Y}_0^{(4)} \\ 
& \vdots \\ 
& \bar{Y}_1^{(K)} - \bar{Y}_0^{(K )}. 
 \end{bmatrix}
 \end{aligned}
\end{equation}  
We consider permutation tests over the sign of $\tilde{V}_s = s(V_1 - V_2), s \in \{-1, 1\}^{K/2}$. We define $T(\tilde{V}_{s})$ the t-static obtained from the vector $\tilde{V}_{s}$, and $C_K^P(\alpha)$ the $(1 -\alpha)^{th}$ quantile of $|T(\tilde{V}_{s})|,s \in \{-1, 1\}^{K/2}$ (up-to rounding), for two sided tests (one-sided test follows similarly by studying the distribution of $T(\tilde{V}_{s})$). From Theorem \ref{thm:inference1}, the distribution of $T(\tilde{V}_{s})$ is asymptotically invariant under the null hypothesis. Therefore, the justification of the test is asymptotic in the sense of \cite{canay2017randomization}. The complete description of the permutation test is provided in Algorithm \ref{alg:permutation_test}.

\begin{algorithm} [!ht]   \caption{Permutation test with $p_1 = 1$}\label{alg:permutation_test}
\footnotesize 
    \begin{algorithmic}[1]
    \Require Value $\beta \in \mathbb{R}^p$ (exogenous), $K$ clusters, constant $\bar{C}$, size $\alpha$; 
    \State Organize clusters into $G = K/2$ pairs with consecutive indexes $\{k, k+1\}$;  
   
    \State $t = 0$ (baseline): either nobody receives treatments or treatments are assigned with $\pi(\cdot;\beta)$ (either case is allowed).
    \begin{algsubstates}
        \State Experimenters collect baseline outcomes: for $n$ units in each cluster observe $Y_{i,0}^{(h)}, X_i^{(h)}, h  \in \{1, \cdots, K\}$. 
        \end{algsubstates} 
        
    \State $t = 1$: experiment starts
    \begin{algsubstates}
        \State  For each pair $g = \{k, k+1\}$, randomize
   \begin{equation} \label{eqn:perturbation}
     \small 
     \begin{aligned} 
    D_{i,1}^{(h)} | \beta, X_i^{(h)} = x \sim \begin{cases} 
    & \mathrm{Bern}(\pi(x, \beta + \eta_n\underline{e}_1)) \text{ if } h = k  \\ 
    & \mathrm{Bern}(\pi(x, \beta - \eta_n\underline{e}_1)) \text{ if } h = k+1 
    \end{cases} , \quad \bar{C} n^{-1/2} < \eta_n < \bar{C} n^{-1/4},
    \end{aligned}   
    \end{equation} 
    
        \State For $n$ units in each cluster $h$ observe $Y_{i,1}^{(h)}$; 
        \State Construct $(V_1, V_2)$ as in Equation \eqref{eqn:Vs}. 
      \end{algsubstates}
     \State For all $s \in \{-1, 1\}^{K/2}$ 
     \begin{algsubstates} 
     \State Compute $\tilde{V}_s = s(V_1 - V_2)$ 
     \State Compute 
     $$
   T(\tilde{V}_s) = \frac{\bar{V}_s}{\sqrt{\frac{1}{K/2 (K/2 - 1)} \sum_{i=1}^{K/2} (\tilde{V}_{s,i} - \bar{V}_s)^2}} \quad \text{where} \quad \bar{V}_s = \frac{1}{K/2} \sum_{i=1}^{K/2} \tilde{V}_{s,i}
     $$ 
     \end{algsubstates} 
     \item Compute $q_\alpha$ defined as the $(1 - \alpha)^{th}$ quantile of $|T(\tilde{V}_s)|$ for two-sided tests and of $T(\tilde{V}_s)$ for one-sided test 
 \State Reject the null hypothesis if $|T(\tilde{V})| \ge q_\alpha$ for two sided test and if $T(\tilde{V}) \ge q_\alpha$ for one-sided test. 
 
         \end{algorithmic}
\end{algorithm}

\paragraph{Choice of the pairing in our application}

In our application, any randomly selected pairing would work. Therefore, for simplicity we simply use the alphabetical ordering of the clusters to sort clusters into pairs. However, any other randomly selected pairing would work. Since we may have more clusters in one group than the other (e.g., in our application we have 13 clusters with a positive perturbation and 12 with a negative perturbation). In this case, we aggregate the two clusters with the smallest number of units into a single cluster, so to make the number of clusters even. 

To establish sensitivity of the results to the choice of the pairing (which would not affect the point estimate, since this average over all positive and negative perturbations regardless, but may affect the studentized t-statistic and the p-value, we recommend report the distribution of p-values and t-statistics over several randomly shuffled pairing which are consistent with the design (if the design admits a single pairing this sensitivity analysis is clearly not needed). We do so in the context of our application, where we report the distribution of p-values and t-statistics in Figure \ref{fig:p_values_distribution} over 5000 randomly generated pairings.  Conclusions remain robust under pairing reshuffling, with the p-value always below 0.1 and on average below 0.05 for the marginal effect at $\beta = 0.5$ and above $0.1$ uniformly for $\beta = 0.7$.

\begin{figure}[!ht]
    \centering
    \includegraphics[scale = 0.5]{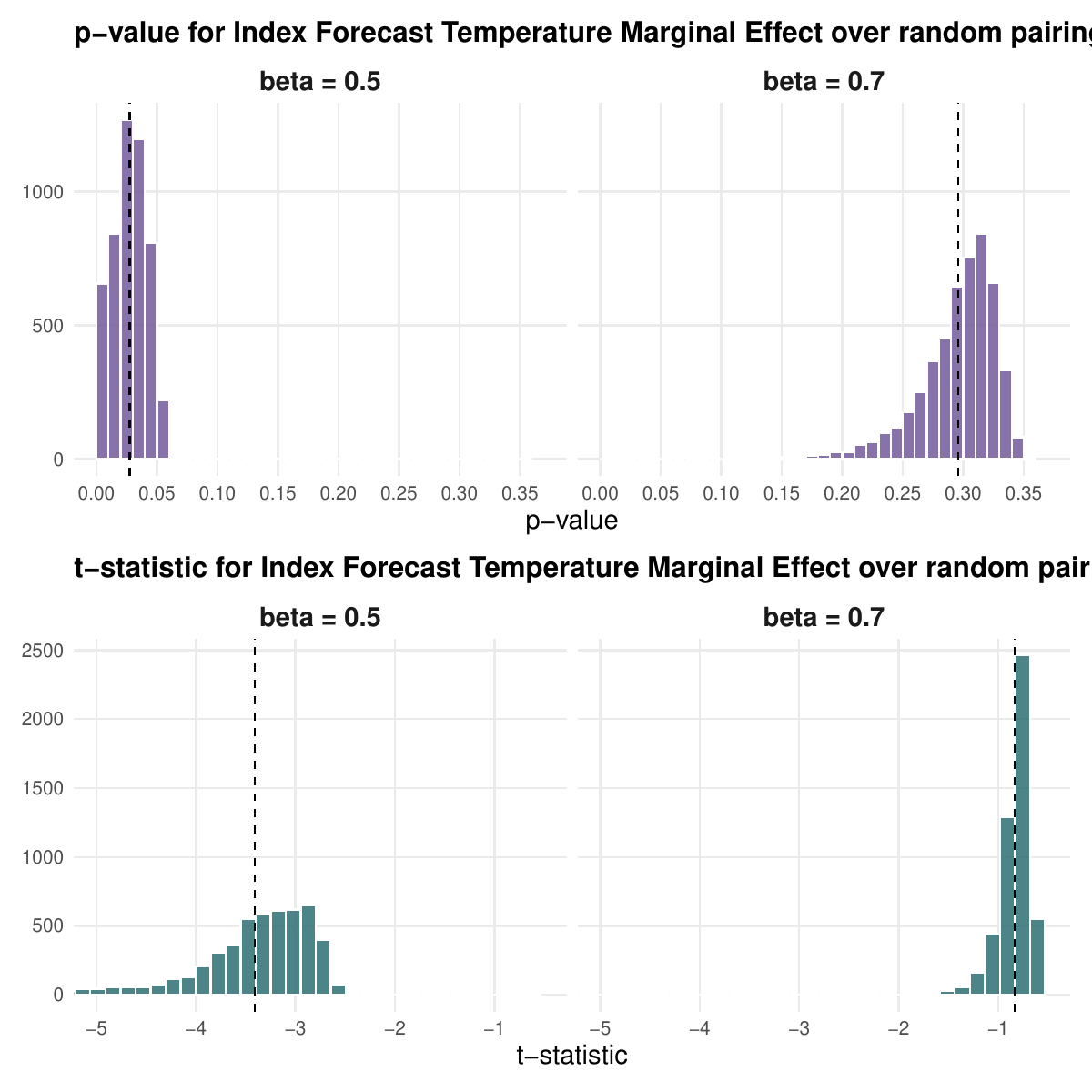}
    \caption{Distribution of  inference results across 5000 random pairings of positive and negative perturbation clusters for the marginal effect of the treatment on the index of forecast temperature accuracy. The panels at the top report the distribution of one-sided  $p$-values for $\beta = 0.5$ and $\beta = 0.7$, while the panels at the bottom report the corresponding distribution of studentized  $t$-statistics for $\beta = 0.5$ and $\beta = 0.7$. Dotted line corresponds to the average value}
\label{fig:p_values_distribution}
\end{figure}

\paragraph{Balance tables} In the context of permutation tests for balance tables (Tables \ref{tab:summaries}, \ref{tab:summaries_response_rates}), permutation tests are similar as described above, where, however, we replace $\bar{Y}_1^{(k)} - \bar{Y}_0^{(k)}$ in each entry of the vectors $V_1, V_2$ with the average baseline $j^{th}$ covariate $\bar{X}_j^{(k)}$ in cluster $k$. The null hypothesis of interest is therefore whether the average baseline covariate $\bar{X}_j^{(k)}$ has the same expectation across all clusters.

\subsection{A micro-foundation with network spillovers}\label{sec:1a}

We provide here a micro-foundation of Assumption \ref{ass:ass_0} in contexts with network spillovers, our leading application. Suppose individuals are connected with other individuals through an unobserved and cluster-specific adjacency matrix $A^{(k)}$. Individuals can form a link with an (unknown) subset of individuals in each cluster. Nodes in each cluster are spaced under some latent space \citep{hoff2002latent, lubold2020identifying} and can interact with at most the $\gamma_N^{1/2}$ closest nodes under the latent space, where we assume that $\gamma_N^{1/2}$ (the maximum number of possible connections) is the same across all clusters. We say $1\{i_k \leftrightarrow j_k\} = 1$ if individual $i$ can interact with $j$ in cluster $k$ (we assume symmetry for simplicity). 
Conditional on $1\{i_k \leftrightarrow j_k\}$, 
\begin{equation} \label{eqn:network} 
\small 
\begin{aligned} 
&  (X_i^{(k)}, U_i^{(k)}) \sim_{i.i.d.}  F_X F_{U|X}, \quad A_{i,j}^{(k)} = l\Big(X_i^{(k)}, X_j^{(k)}, U_i^{(k)}, U_j^{(k)}\Big)1\{i_k \leftrightarrow j_k \}, \quad l:\mathcal{X}^2 \times \mathcal{U}^2 \mapsto [0,1], 
\end{aligned} 
\end{equation}
for an arbitrary and unknown function $l(\cdot)$ and unobservables $U_i^{(k)}$. Whether two individuals interact depends on (i) whether they $i_k$ and $j_k$ are close enough within a certain latent space (captured by $1\{i_k \leftrightarrow j_k \}$); (ii) their covariates and unobserved individual heterogeneity (i.e., $X_i, U_i$), which capture homophily. Equation \eqref{eqn:network} also states that covariates are $i.i.d.$ unconditionally on $A^{(k)}$, but not necessarily conditionally.
Figure \ref{fig:network} provides an illustration.
Here, we condition on the indicators $1\{i_k \leftrightarrow j_k\}$ (which can differ across clusters) to control the network's maximum degree, but we do not condition on the network $A^{(k)}$. We can interpret such indicators as exogenously drawn from some arbitrary distribution, and then conditioning on such indicators.\footnote{Formally, $\mathcal{I}_k \sim \mathcal{P}_k, \quad (X_i^{(k)}, U_i^{(k)}) | \mathcal{I}_k \sim_{i.i.d.} F_{U|X} F_X, \quad A_{i,j}^{(k)} = l\left(X_i^{(k)}, X_j^{(k)}, U_i^{(k)}, U_j^{(k)}\right)1\{i_k \leftrightarrow j_k \}$, where $\mathcal{I}_k$ is the matrix of such indicators in cluster $k$ and $\mathcal{P}_k$ is a cluster-specific distribution left unspecified.}

Note that here the latent space can differ across clusters. The reason why this is possible is because, although the latent space (i.e., indicators $1\{i_k \leftrightarrow j_k\}$) can differ across clusters, the distribution of the outcomes remain the same whenever outcomes only depends on the unobservables and observables $(X_i, U_i)$ of their friends (since these are $i.i.d.$ conditional on the latent space), but not necessarily on the specific \textit{identity} of their friends. That is, a sufficient condition for Assumption \ref{ass:ass_0} to hold 
is that the outcome depends on neighbors not through their identity directly, but only through their individual specific observable and unobservables.  
We formalize this intuition below.





\begin{exmp}[Microfoundation with network model] \label{exmp:microfoundation} Consider the following restrictions: 
\begin{itemize}
\item[(A)]  For $k \in \{1, \cdots, K\}, i \in \{1, \cdots, \tilde{N}^{(k)}\}$, let Equation \eqref{eqn:network} hold given the indicators $1\{i_k \leftrightarrow j_k\}$, for some unknown $l(\cdot)$; in addition, $\sum_{j = 1}^{\tilde{N}^{(k)}} 1\{i_k \leftrightarrow j_k\} = \gamma_N^{1/2}$.  
\item[(B)] Suppose that for any $i,t,k, \mathbf{d}_s^{(k)} \in \{0,1\}^N, s \le t$
$$
\small 
\begin{aligned} 
Y_{i,t}^{(k)}(\mathbf{d}_1^{(k)}, \cdots, \mathbf{d}_{t}^{(k)}) = r\Big( \mathbf{d}_{i,t}^{(k)}, \mathbf{d}_{\mathcal{N}_i^{(k)},t}^{(k)}, X_i^{(k)},  X_{\mathcal{N}_i^{(k)}}^{(k)}, U_i, U_{\mathcal{N}_i^{(k)}}, |\mathcal{N}_i^{(k)}|, \nu_{i,t}^{(k)}\Big) + \tau_k + \alpha_{t}
\end{aligned} 
$$
where $\mathcal{N}_i^{(k)} = \{j: A_{i,j}^{(k)} > 0\}$, for some unknown $r(\cdot)$ symmetric in the neighbors' $\mathbf{d}_{\mathcal{N}_i^{(k)},t}^{(k)}, U_{\mathcal{N}_i^{(k)}}$, stationary (but possibly serially dependent) unobservables $\nu_{i,\cdot}^{(k)} | X^{(k)}, U^{(k)} \sim_{i.i.d.} P_{\nu}$, fixed effects $\tau_k, \alpha_t$.
\end{itemize}  
\end{exmp}

Condition (A) states the following: before being born, each individual may interact with $\gamma_N^{1/2}$ many other individuals (i.e., maximum degree), where the maximum number of potential connections $\gamma_N^{1/2}$ is homogeneous across different clusters. After birth, the individual's gender, income, and parental status determine her type and the distribution of her and her potential connections' edges.\footnote{See \cite{jackson1996strategic}, \cite{li2020random} for pairwise interactions. Extensions where the networks also depend on non-separable shocks $\omega_{i,j}$ are possible, as discussed in previous versions of this draft.} Condition (B) states that potential outcomes depend on neighbors' assignments, observables, and unobservables. \textit{Heterogeneity} in spillovers occurs arbitrarily through neighbors' observables and unobservables $(D_j, U_j, X_j)$. Such variables can interact with each other, allowing for observed and unobserved heterogeneity in direct and spillover effects. Whereas treatments may exhibit individual-level heterogeneity, treatments do not interact with clusters' fixed effects. 

Note that, however, as a sufficient condition, we require that the marginal distribution \(F_X\), the conditional distribution \(F_{U|X}\), the link function \(l(\cdot)\), and the potential outcome map \(r(\cdot)\) to be common across clusters.

\begin{prop}[Microfoundation with network spillovers] \label{lem:lem0} Consider treatments assigned as in Assumption \ref{defn:bernoulli}. Let (A) and (B) in Example \ref{exmp:microfoundation} hold. Then Assumption \ref{ass:ass_0} holds.  
\end{prop}   

The proof is in Appendix \ref{sec:lem_main}. Proposition \ref{lem:lem0} motivates Assumption \ref{ass:ass_0} in our leading example with network spillovers.

 \begin{figure}[!ht]
 \centering
 \vspace{-7mm}
    \begin{tikzpicture}

\coordinate (1) at (-10,1);
\coordinate (2) at (-2,1);
\coordinate (3) at (-2,5);
\coordinate (4) at (-10,5);

  \node[draw, circle] (a) at (-6, 3.2) {};
  \node[draw, circle] (b) at (-6, 4.1) {};
  \node[draw, circle] (c) at (-7, 3.5) {};
 \node[draw, circle] (d) at (-6.8, 2.5) {};
  \node[draw, circle] (e) at (-5.3, 2.5) {};
  \node[draw, circle] (f) at (-5, 3.5) {};
 \node[circle] (h) at (-5.1, 4.4) {};
 \node[circle] (i) at (-4.3, 2.9) {};
  \node[circle] (l) at (-4.3, 4) {};
   \node[circle] (m) at (-6.9, 4.4) {};
   \node[circle] (n) at (-7.6, 3) {};
   \node[circle] (o) at (-7.4, 2) {};
   \node[circle] (p) at (-7.7, 3.9) {};
     \node[circle] (q) at (-6, 4.8) {};
       \node[ circle] (r) at (-6, 1.9) {};
       \node[circle] (s) at (-4.8, 1.9) {};
     
 \node[circle] (g) at (-3,3) {$\longrightarrow$};
\node[circle] (g) at (3,3) {$\longrightarrow$};

 \node[circle] (g) at (-3,3) {$\longrightarrow$};
 
 \node[circle] (g) at (-6,1.5) {$\text{Possible connections}$};
  \node[circle] (g) at (0,1.5) {$\text{Types' assignment}$};
  \node[circle] (g) at (6,1.5) {$\text{Network formation}$};
  
  \node[draw, fill = green, circle] (aaa) at (0, 3.2) {};
  \node[draw, fill = green, circle] (bbb) at (0, 4.1) {};
  \node[draw, fill = blue, circle] (ccc) at (-1, 3.5) {};
 \node[draw, fill = blue, circle] (ddd) at (-0.8, 2.5) {};
  \node[draw, fill = blue, circle] (eee) at (0.7, 2.5) {};
  \node[draw, fill = blue, circle] (fff) at (1, 3.5) {};
 \node[circle] (hhh) at (0.9, 4.4) {};
 \node[circle] (iii) at (1.7, 2.9) {};
  \node[circle] (lll) at (1.7, 4) {};
   \node[circle] (mmm) at (-0.9, 4.4) {};
   \node[circle] (nnn) at (-1.6, 3) {};
   \node[circle] (ooo) at (-1.4, 2) {};
   \node[circle] (ppp) at (-1.7, 3.9) {};
     \node[circle] (qqq) at (0, 4.8) {};
       \node[ circle] (rrr) at (0, 1.9) {};
       \node[circle] (sss) at (1.2, 1.9) {};

  \node[draw, fill = green, circle] (aa) at (6, 3.2) {};
  \node[draw,fill = green,  circle] (bb) at (6, 4.1) {};
  \node[draw, fill = blue, circle] (cc) at (5, 3.5) {};
 \node[draw,  fill = blue,circle] (dd) at (5.2, 2.5) {};
  \node[draw, fill = blue, circle] (ee) at (6.7, 2.5) {};
  \node[draw,  fill = blue, circle] (ff) at (7, 3.5) {};
 \node[circle] (hh) at (6.9, 4.4) {};
 \node[circle] (ii) at (7.7, 2.9) {};
  \node[circle] (ll) at (7.7, 4) {};
   \node[circle] (mm) at (5.1, 4.4) {};
   \node[circle] (nn) at (4.4, 3) {};
   \node[circle] (oo) at (4.6, 2) {};
   \node[circle] (pp) at (4.3, 3.9) {};
     \node[circle] (qq) at (6, 4.8) {};
       \node[ circle] (rr) at (6, 1.9) {};
       \node[circle] (ss) at (7.2, 1.9) {};

    \draw[-, dashed] (a) edge (b) (a) edge (c) (a) edge (d) (a) edge (e) (a) edge (f);
    \draw[-] (aa) edge (bb) (aa) edge (ee); 
 
 \draw[-, dashed] (b) edge (c) (c) edge (d) (d) edge (e) (e) edge (f) (f) edge (b); 
 \draw[-, dashed] (f) edge (h) (f) edge (i) (e) edge (i) (b) edge (h) (f) edge (l); 
\draw[-, dashed] (m) edge (c) (n) edge (c) (p) edge (c) (n) edge (d) (o) edge (d) (r) edge (d) (r) edge (e) (s) edge (e) (b) edge (q) (b) edge (m); 
 
    \draw[-, dashed] (aaa) edge (bbb) (aaa) edge (ccc) (aaa) edge (ddd) (aaa) edge (eee) (aaa) edge (fff);
 
 \draw[-, dashed] (bbb) edge (ccc) (ccc) edge (ddd) (ddd) edge (eee) (eee) edge (fff) (fff) edge (bbb); 
 \draw[-, dashed] (fff) edge (hhh) (fff) edge (iii) (eee) edge (iii) (bbb) edge (hhh) (fff) edge (lll); 
\draw[-, dashed] (mmm) edge (ccc) (nnn) edge (ccc) (ppp) edge (ccc) (nnn) edge (ddd) (ooo) edge (ddd) (rrr) edge (ddd) (rrr) edge (eee) (sss) edge (eee) (bbb) edge (qqq) (bbb) edge (mmm); 
   
 \draw[-] (bb) edge (cc) (dd) edge (ee) (ee) edge (ff) (ff) edge (bb); 
 \draw[-] (ff) edge (ii)  (bb) edge (hh) (ff) edge (ll); 
\draw[-] (mm) edge (cc) (nn) edge (cc) (pp) edge (cc);


    \end{tikzpicture}
    \vspace{-18mm}
  \caption{Example of the network formation model, with $\gamma_N^{1/2} = 5$ satisfying Assumption \ref{ass:ass_0} under the microfoundation in Appendix \ref{sec:1a}.  Individuals are assigned different types, which may or may not be observed by the researcher (corresponding to different colors). Individuals interact based on their types and form links among the possible connections. The possible connections and the realized adjacency matrix remain unobserved to the researcher. }  \label{fig:network}
\end{figure}

\subsection{Computing the value of collecting network data} \label{sec:comp}

Here, we ask
how $\beta^*$ compares with the policy that assigns treatments without restrictions on the policy function, and provide useful bounds on the value of collecting network data. We focus on a setting with network spillovers, where $A$ denotes the unobserved adjacency matrix as in Example \ref{exmp:microfoundation}, and omit the super-script $k$ because the argument applies to any cluster. We study
\begin{equation} \label{eqn:difference_global} 
\small 
\begin{aligned} 
W_N^* - W(\beta^*),  \quad W_N^* = \sup_{\mathcal{P}_N(\cdot) \in \mathcal{F}} \frac{1}{N} \sum_{i=1}^N \mathbb{E}\Big[\mathbb{E}_{D \sim \mathcal{P}_N(A, X)}[Y_{i,t} | A, X]\Big]
\end{aligned} 
\end{equation} 
with $\mathcal{F}$ as the set of \textit{all} conditional distribution of the vector $D \in \{0,1\}^N$, given network $A$ and the covariates of all observations $X$ as defined in Section \ref{sec:1a}. 
Equation \eqref{eqn:difference_global} denotes the difference between the expected outcomes, evaluated at the global optimum over all possible assignments (with $A, X$ observed), and the welfare evaluated at $\beta^*$ (without observing $A$).  
\begin{ass}[Discrete parameter space, assignment, and minimum degree] \label{ass:discrete_space} Consider a network model as in Example \ref{exmp:microfoundation}. Assume that $X_i \in \mathcal{X}, \mathcal{X} = \{1, \cdots, |\mathcal{X}|\}, |\mathcal{X}| < \infty, P(X = x) > \bar{\kappa} > 0$ $\forall x \in \mathcal{X}$. Let $\pi(x, \beta) = \beta_x$, and $\mathcal{B} = [0,1]^{|\mathcal{X}|}$. Let $\inf_{x, x', u'} \int l(x,x', u, u') dF_{U| X = x}(u) \ge \underline{\kappa}$, for some $\underline{\kappa}, \bar{\kappa} \in (0, 1]$, with $l(\cdot)$ defined in Equation \eqref{eqn:network}. 
\end{ass} 
Assumption \ref{ass:discrete_space} states that researchers assign treatments based on finitely many observable types as in \cite{manski2004}, \cite{graham2010measuring}. Each type $x \in \mathcal{X}$ is assigned a different probability $\beta_x$, which can take any value between zero and one. Assumption \ref{ass:discrete_space} also states that conditional on individual's type $(X_i, U_i)$, any other unobserved type $U_j$ can form a connection with individual $i$ with some positive probability, provided that $i$ and $j$ are connected under the latent space representation (recall Equation \ref{eqn:network}). This condition is consistent with the model in Example \ref{exmp:microfoundation} (and restrictions on $\gamma_N$), because the assumption states that the expected minimum degree is bounded from below by $\underline{\kappa} \gamma_N^{1/2}$, which is smaller than the maximum degree $\gamma_N^{1/2}$.  The second restriction is on the potential outcomes. Let
\begin{equation} \label{eqn:outcome_optimality} 
\small 
\begin{aligned} 
Y_{i,t}(\mathbf{d}_t) & = \Big[\Delta(X_i) - v(X_i)\Big] \mathbf{d}_{i,t} + \mathcal{S}_{i,t}(\mathbf{d}_t) + \nu_{i,t}, \quad \mathbb{E}[\nu_{i,t} | X, A] = 0 \\ 
\mathcal{S}_{i,t}(\mathbf{d}_t) &=  s\Big(\frac{\sum_{j = 1}^N A_{i,j} \mathbf{d}_{j,t} 1\{X_j = 1\}}{\sum_{j = 1}^N A_{i,j} 1\{X_j = 1\}}, \cdots, \frac{\sum_{j = 1}^N A_{i,j} \mathbf{d}_{j,t} 1\{X_j = |\mathcal{X}|\}}{\sum_{j = 1}^N A_{i,j} 1\{X_j = |\mathcal{X}|\}}\Big), 
\end{aligned} 
\end{equation}  
where whenever $\sum_j A_{i,j} 1\{X_j = x\} = 0$ we indicate the corresponding entry as equal to zero. 
Here, $\Delta(\cdot)$ is the direct treatment effect, and $v(\cdot)$ is the cost of the treatment; $s(\cdot)$ captures the spillover effects. Spillovers depend on the fraction of treated neighbors and are heterogeneous in the neighbors' types, with no interactions with direct effects.

\begin{thm} \label{thm:optimality_global} Consider a model in Example \ref{exmp:microfoundation}. Let Equation \eqref{eqn:outcome_optimality} hold, with $s(\cdot)$ twice differentiable with bounded derivatives.  Suppose that Assumption  \ref{ass:discrete_space} hold. Then, with $W_N^*$ as in Equation \eqref{eqn:difference_global}, 
$
\lim_{N, \gamma_N \rightarrow \infty} \Big\{W_N^* - W(\beta^*)\Big\} \le  \mathbb{E}\Big[|\Delta(X) - v(X)|\Big]. 
$  

\end{thm}   

The proof is in Appendix \ref{sec:proof7}. 
Theorem \ref{thm:optimality_global} bounds the welfare difference by the expected direct effects minus costs. If direct effects are small compared with the treatment costs, such a difference is negligible (for any spillover effects). The bound is identified \textit{without} network data under separability of direct and spillover effects. The theorem assumes that the maximum degree converges to infinity, but it may converge at a slower rate than $N$, consistent with our conditions in previous theorems. This result is novel in the context of the literature on targeting networked individuals and provides a formal characterization of the \textit{value} of collecting network information.\footnote{We note \cite{akbarpour2018just} study network value from the different angle of network diffusion: for a class of network formation models and diffusion mechanisms, the authors show that random seeding is approximately optimal as researchers treat a few more individuals. The main differences are that here (i) we do not study the problem from the perspective of network diffusion but instead focus on an exogenous interference mechanism with heterogeneity; (ii) we provide an upper bound in terms of the direct treatment effect, leveraging a different model and theory. Different from \cite{akbarpour2018just}, the upper bound does not state that we should treat $\epsilon$-more individuals (since we consider a different model of spillovers). } 
Theorem \ref{thm:optimality_global} does \textit{not} state that spillovers are not relevant ($\beta^*$ depends on the spillovers). Instead, it states that one can compute best policies, without knowledge of the network in settings where direct effects are small. 

One can estimate the bound by taking an absolute difference between the treated and control units for different individual types, and average across types. 
In Example \ref{exmp:main}, the bound equals $\phi_1$ (the direct treatment effect) minus the cost of implementing the treatment. 

\begin{cor} \label{cor:c_e} Let the conditions in Theorem \ref{thm:optimality_global} hold. Let $C_e$ be the cost of collecting network information per individual (with total cost for observing the network $A$ equal to $N C_e$). Then, $\lim_{N, \gamma_N \rightarrow \infty} W_N^* - W(\beta^*) - C_e \le 0$, if $C_e \ge \mathbb{E}\Big[|\Delta(X) - v(X)|\Big]$. 
\end{cor}

\subsection{Pairing clusters with heterogeneity} \label{app:het}

\vspace{-1mm} 

\subsubsection{Inference and estimation with observed cluster heterogeneity}

In this subsection, we discuss an extension to allow for cluster heterogeneity. Consider $\theta_k \in \Theta$ to denote the cluster's type for cluster $k$, where $\Theta$ is a finite space (i.e., there are finitely many cluster types). Let $\theta_k$ be observable by the researcher and be non-random. For example, $\theta_k$ may denote a binary classification of clusters as large or small based on their observable size $\tilde{N}^{(k)}$.  For expositional, we focus on our leading example Example \ref{exmp:microfoundation} of network spillovers, but our discussion directly extend beyond this model. 

\begin{ass} \label{ass:het_cluster} Consider a data generating process as in Example \ref{exmp:microfoundation}. For each cluster $k$, Equation \eqref{eqn:network} holds, with $F_X, F_{U|X}$ replaced by $F_X(\theta_k), F_{U|X}(\theta_k)$ functions of $\theta_k$; the model in Example \ref{exmp:microfoundation} holds with $r(\cdot)$ that also depends on $\theta_k$.  
\end{ass} 

Assumption \ref{ass:het_cluster} allows for both the distribution of covariates and unobservables and potential outcomes to also depend on the cluster's type $\theta_k$. 

\begin{lem} \label{lem:lem02} Under Assumption \ref{ass:het_cluster}, under an assignment in Assumption \ref{defn:bernoulli}, for some function $y(\cdot)$ unknown to the researcher,  
\begin{equation} \label{eqn:main_Y} 
\small 
\begin{aligned} 
Y_{i,t}^{(k)} = y\Big(X_i^{(k)},\hat{\beta}_{k,t}, \theta_k\Big) + \varepsilon_{i,t}^{(k)} + \alpha_t + \tau_k, \quad \mathbb{E}_{\hat{\beta}_{k,t}}\Big[\varepsilon_{i,t}^{(k)} |  X_i^{(k)}\Big] = 0. 
\end{aligned} 
\end{equation}
\end{lem}  

Different from Proposition \ref{lem:lem0}, here the the functions also depend on the cluster's type $\theta_k$. 
The proof of Lemma \ref{lem:lem02} follows verbatim from the one of Proposition \ref{lem:lem0}, taking here into account also the (deterministic) cluster's type. 

\vspace{-3mm}
\paragraph{Single wave experiment} In the context of a single-wave experiment, 
we are interested in testing the null hypothesis of whether a \textit{class} of decisions $\beta(\theta), \theta \in \Theta$, which depends on the cluster's type, is optimal. Namely, let 
$
W\Big(\beta(\theta), \theta\Big) = \int_x y(x, \beta(\theta), \theta) dF_X(x; \theta), \quad \beta: \Theta \mapsto \mathcal{B}, \quad \theta \in \Theta 
$ 
be the reward corresponding to cluster's type $\theta$, for a decision rule $\beta(\theta)$. Also, let 
$ 
M\Big(\beta(\theta), \theta\Big) = \frac{\partial W(b, \theta)}{\partial b} \Big|_{b = \beta(\theta)}
$
be the marginal effect with respect to changing $\beta(\theta)$ (for fixed $\theta$). 
The null hypothesis is 
$
H_0: M\Big(\beta(\theta), \theta\Big) = 0, \forall \theta \in \Theta,  
$
i.e., the (baseline) policy $\beta(\theta)$ is optimal for all clusters under consideration. The algorithm follows similarly to Algorithm \ref{alg:my_pilot} with the following modification: instead of matching arbitrary clusters, we construct pairs such that elements in the same pair $(k, k+1)$ are such that $\theta_k = \theta_{k+1}$. 
 
\vspace{-3mm}

\paragraph{Multi-wave experiment} For the multi-wave experiment, our goal is to find $\beta^*(\theta)$ such that 
$
\beta^*(\theta) \in \mathrm{arg} \max_{b \in \mathcal{B}} W(b, \theta), \forall \theta \in \Theta. 
$ 
Similarly to the single-wave experiment, clusters $(k, k')$ of the same type $\theta_k = \theta_{k'}$ are first matched together. We can then consider two extensions. The first extension consists of grouping clusters of the same type together, estimating separately $\beta^*(\theta)$ for each $\theta \in \Theta$. In this case the regret bound holds up-to a factor of order $1/\min_t P(\theta = t)$, with $P(\theta = t)$ denoting the (exact) share of clusters of type $t$. The second approach instead consists of updating the same policy from a given pair using information from that \textit{same} pair. The validity of this latter extension relies on time independence.

\vspace{-3mm}

\subsubsection{Matching clusters with distributional embeddings} \label{sec:matching}

Next, we turn to settings where covariates have different distributions in different clusters. Let
$
X_i^{(k)} \sim_{i.i.d.} F_{X}^{(k)}, 
$
with $F_{X}^{(k)}$ being cluster-specific.
Treatments are assigned as follows 
\begin{equation} \label{eqn:tt} 
\small 
\begin{aligned} 
t = 0: \quad & D_{i,0}^{(h)} \sim  \pi(X_i^{(h)};\beta_0) , \quad h \in \{k, k'\} \\  
t = 1: \quad & D_{i,1}^{(k)} \sim   \pi(X_i^{(k)};\beta) , \quad  D_{i,1}^{(k')} \sim   \pi(X_i^{(k')};\beta') . 
\end{aligned} 
\end{equation}   
The estimand and estimator are respectively 
$$
\small 
\begin{aligned} 
\omega_k & =  \int y(x;\beta)  dF_X^{(k)}(x)  -  \int y(x;\beta') dF_X^{(k)}(x), \quad \widehat{\omega}_{k}(k') = \Big[\bar{Y}_1^{(k)} - \bar{Y}_1^{(k')}\Big] - \Big[\bar{Y}_0^{(k)} - \bar{Y}_0^{(k')}\Big]. 
\end{aligned} 
$$  
our focus is to control the bias of the estimator via matching.  

\begin{lem} \label{lem:discrepancy} Let Assumption \ref{ass:ass_0} hold, with $X_i^{(k)} \sim F_X^{(k)}$ having potentially different distributions across clusters, and treatments assigned as in Equation \eqref{eqn:tt}. Then
$$
\small 
\begin{aligned} 
\mathbb{E}[\widehat{\omega}_{k}(k')] - \omega_k = \int \Big(y(x;\beta') - y(x;\beta_0)\Big) d\Big(F_X^{(k)}(x) - F_X^{(k')}(x)\Big).
\end{aligned} 
$$ 
\end{lem}  
Lemma \ref{lem:discrepancy} shows the bias depends on the difference between the expectations averaged over two different distributions. The bias is unknown since it depends on the function $y(\cdot)$, which is not identifiable with finitely many clusters. We therefore bound the worst-case error over a class of functions $x \mapsto [y(x;\beta') - y(x;\beta_0)] \in \mathcal{M}$, with $\mathcal{M}$ defined below.

We start by defining $\mathcal{M}$ be a reproducing kernel Hilbert space (RKHS) equipped with a norm $||\cdot||_{\mathcal{M}}$.\footnote{A RKHS is a Hilbert space of functions where all the evaluations functionals are bounded, namely, where for each $f \in \mathcal{M}$, and $x \in \mathcal{X}$, $f(x) \le C ||f||_{\mathcal{M}}$  for a finite constant $C$.  Intuitively, assuming that $[y(\cdot; \beta') - y(\cdot; \beta_0)]  \in \mathcal{M}$ imposes smoothness conditions on the average effect as a function of $x$. } Without loss of generality, we study the worst-case functionals over the unit-ball. Formally, we focus on bounding the worst-case error of the form\footnote{Here Equation \eqref{eqn:ajk} follows directly from Lemma \ref{lem:discrepancy} and the fact that if $f \in \mathcal{M}, - f \in \mathcal{M}$. 
}  
\begin{equation} \label{eqn:ajk} 
\small 
\begin{aligned} 
& \sup_{[y(\cdot;\beta') - y(\cdot;\beta_0)] \in \mathcal{M}: ||y(\cdot;\beta') - y(\cdot;\beta_0)||_{\mathcal{M}} \le 1}  \Big|\omega_k - \mathbb{E}[\widehat{\omega}_{k}(k')] \Big| = \sup_{f \in \mathcal{M}: ||f||_{\mathcal{M}} \le 1}  \Big\{\int f(x) d(F_X^{(k)} - F_X^{(k')})\Big\}. 
\end{aligned} 
\end{equation}   

The right-hand side is know as the maximum mean discrepancy (MMD), a measure of distances in RKHS  \citep[see][and references therein]{muandet2016kernel}.  It is known that the MMD can be consistently estimated using kernels. In particular, given a particular choice of a kernel $k(\cdot)$, which corresponds to a certain RKHS, we can estimate 
 \begin{equation} \label{eqn:mmd_est} 
 \begin{aligned} 
 \widehat{\mathrm{MMD}}^2(k, k') & = \frac{1}{n(n-1)} \sum_{i=1}^n \sum_{j \neq i} h\Big(X_i^{(k)}, X_i^{(k')}, X_j^{(k)}, X_j^{(k')}\Big), \\ h(x_i, y_i, x_j, y_j) & = k(x_i, x_j) + k(y_i, y_j) - k(x_i, y_j) - k(x_j, y_i). 
 \end{aligned} 
 \end{equation}  
Consistency follows from results discussed in \cite{muandet2016kernel}.

 We now turn to the problem of matching clusters. The following matching algorithm is considered: (i) construct
$
k^* \in \mathrm{arg} \min_{k' \neq k}  \widehat{\mathrm{MMD}}^2(k, k'). 
$  based on the minimum estimated MMD in Equation \eqref{eqn:mmd_est}; (ii) randomize treatments as in Equation \eqref{eqn:tt}; (iii) estimate $\widehat{\omega}_k(k^*)$. With many clusters, we suggest minimizing the average MMD over cluster pairs.

\section{Derivations of results in main text}

First, we introduce conventions and notation. 
we say that $x \lesssim y$ if $x \le c y$ for a positive constant $c < 
\infty$. For $K$ many clusters, we say that $\nint{k} = k 1\{k \le K\} + (k - K)1\{k > K\}$, so that whenever $k$ is below the overall number of clusters $K$, $\nint{k}$ equals $k$ otherwise it ``loops back'' whenever $k > K$ and equal $k - K$. This  notation will be useful for our proof in the sequential cross-fitting argument. 
we will refer to $\widehat{M}_{(k, k+1)}$ as $\widehat{M}_k$ for $k$ is odd for short of notation. Also, we define $\check{M}_{k,s} = \widehat{M}_{\nint{k + 2}, s}$.  The following definition introduces the notion of a dependency graph \citep{janson2004large}.

 \begin{defn}[Dependency graph] For given random variables $R_1, \cdots, R_n$, $W_n \in \{0,1\}^{n \times n}$ is a non-random matrix defined as a dependency graph of $(R_1, \cdots, R_n)$ if, for any $i$, $R_i 
\perp R_{j: W_n^{(i,j)} = 0}$. we denote the dependency neighbors $N_i = \{j: W_n^{(i,j)} = 1\}$.   \qed 
  \end{defn} 
  

\begin{defn}[Cover] \label{defn:cover}
Given an adjacency matrix $A_n$, with $n$ rows and columns, a family $\mathcal{C}_n = \{\mathcal{C}_n(j)\}_j$ of  disjoint subsets of $\{1, \cdots, n\}$ is a proper cover of $A_n$ if $\cup_j \mathcal{C}_n(j) = \{1, \cdots, n\}$ and $\mathcal{C}_n(j)$ contains units such that for any pair of elements $\{i,k \in \mathcal{C}_n(j)\}$, $A_n^{(i,k)} = 0$.   \qed 
\end{defn} 
Namely, a proper cover of $A_n$ defines a set of disjoint sets, where each disjoint set contains some indexes of units that are not neighbors in $A_n$. Note that a proper cover always exists, since, if $A_n$ is fully connected, then the number of disjoint sets is just $n$, one for each element.  

The size of the smallest proper cover is the chromatic number, defined as $\chi(A_n)$. 
\begin{defn}[Chromatic number] \label{defn:chromatic} 
The chromatic number $\chi(A_n)$, denotes the size of the smallest proper cover of $A_n$. \qed 
\end{defn} 

We define the oracle descent procedure absent of sampling error. Let $\beta \in \mathcal{B} = [\mathcal{B}_1, \mathcal{B}_2]^p$, where $\mathcal{B}_1, \mathcal{B}_2$ are finite and $\beta$ assumed to be in the interior throughout. Also, let $P_{\mathcal{B}_1, \mathcal{B}_2}$ be the projection operator onto $\mathcal{B}$.

\begin{defn}[Oracle gradient descent under strong concavity]  We define, for $
\alpha_w = \frac{J}{w + 1}
$, 
\begin{equation}  \label{eqn:beta_start2}
\small 
\begin{aligned} 
\beta_{w}^{**} = 
 P_{\mathcal{B}_1, \mathcal{B}_2} \Big[\beta_{w-1}^{**} + \alpha_{w- 1} M(\beta_{w-1}^{**})\Big] , \quad \beta_{1}^{**} = \beta_0. 
 \end{aligned} 
\end{equation} 
\end{defn} 

Note that in the proofs, we will refer to the general $p$-dimensional case for the multi-wave experiment, which uses $\check{T} = T/p$ waves. See Algorithm \ref{alg:adaptive_complete}. 


\subsection{Lemmas and propositions}

\subsubsection{Preliminary lemmas} \label{sec:app_lemma}

\begin{lem}\label{lem:locally_dep} \citep[Theorem 3.5][]{ross2011fundamentals} Let $X_1, ..., X_n$ be random variables such that $\mathbb{E}[X_i^4] < \infty$, $\mathbb{E}[X_i] = 0$, $\tilde{\sigma}^2 = \mathrm{Var}(\sum_{i = 1}^n X_i)$ and define $W = \sum_{i = 1}^n X_i/\tilde{\sigma}$. Let the collection $(X_1, ..., X_n)$ have dependency neighborhoods $N_i$, $i = 1, ..., n$ and also define $D = \mathrm{max}_{1 \le i \le n} |N_i|$. Then for $Z$ a standard normal random variable, 
$
d_W(W, Z) \le \frac{D^2}{\tilde{\sigma}^3} \sum_{i = 1}^n \mathbb{E}| X_i|^3 + \frac{\sqrt{28} D^{3/2}}{\sqrt{\pi} \tilde{\sigma}^2} \sqrt{\sum_{i = 1}^n \mathbb{E}[X_i^4]},
$
where $d_W$ denotes the Wasserstein metric.
\end{lem}

\begin{lem}[From \cite{brooks1941colouring}] \label{lem:brooktheorem} 
For any connected undirected graph $G$ with maximum degree $\Delta$, the chromatic number of $G$ is at most $\Delta + 1$.  
\end{lem}


\begin{lem}[Concentration for dependency graphs] \label{lem:concentration1} Define $\{R_i\}_{i=1}^n$ sub-Gaussian random variables with parameter $r^2 < \infty$, forming a dependency graph with adjacency matrix $A_n$ with maximum degree bounded by $\gamma_N$. Then, with probability at least $1 - \delta$, for any $\delta \in (0,1)$, 
$
\Big|\frac{1}{n} \sum_{i=1}^n (R_i - \mathbb{E}[R_i])\Big| \le  \sqrt{ \frac{2 r^2 \gamma_N \log(2 \gamma_N/\delta)}{n}}. 
$
\end{lem} 

\begin{proof}[Proof of Lemma \ref{lem:concentration1}] For the smallest proper cover $\mathcal{C}_n$ as in Definition \ref{defn:cover}, 
 $$
 \small 
 \begin{aligned} 
 \Big|\frac{1}{n} \sum_{i=1}^n (R_i - \mathbb{E}[R_i])\Big| \le \sum_{j = 1}^{\chi(A_n)}  \underbrace{\Big| \frac{1}{n} \sum_{i \in \mathcal{C}_n(j)} (R_i - \mathbb{E}[R_i])\Big|}_{(A)}. 
 \end{aligned}
$$ 
Here, we sum over each subset of index $\mathcal{C}_n(j) \in \mathcal{C}_n$ in the proper cover, and then we sum over each element in the subset $\mathcal{C}_n(j)$. 
 Observe now that by definition of the dependency graph, components in $(A)$ are mutually independent. Using the Chernoff's bound \citep{wainwright2019high}, we have that with probability at least $1 - \delta$, for any $\delta \in (0,1)$ 
 $
 \Big| \sum_{i \in \mathcal{C}_n(j)} (R_i - \mathbb{E}[R_i])\Big| \le  \sqrt{2 r^2 |\mathcal{C}_n(j)| \log(2/\delta)}, 
 $ 
where $|\mathcal{C}_n(j)|$ denotes the number of elements in $\mathcal{C}_n(j)$. Using the union bound, we obtain that with probability at least $1 - \delta$, for any $\delta \in (0,1)$ 
 $$ 
 \small 
 \begin{aligned} 
  \Big|\frac{1}{n} \sum_{i=1}^n (R_i - \mathbb{E}[R_i])\Big| \le \underbrace{\frac{1}{n} \sum_{j = 1}^{\chi(A_n)} \sqrt{2 r^2 |\mathcal{C}_n(j)| \log(2 \chi(A_n)/\delta)}}_{(B)}. 
  \end{aligned}
$$ 
Using concavity of the square-root function, after multiplying and dividing (B) by $\chi(A_n)$, 
$$
\small 
\begin{aligned} 
(B) &\le \frac{1}{n} \chi(A_n)  \sqrt{2 r^2 \frac{1}{\chi(A_n)} \sum_{j = 1}^{\chi(A_n)} |\mathcal{C}_n(j)| \log(2 \chi(A_n)/\delta)} = \frac{1}{n}   \sqrt{2 r^2 \chi(A_n)n \log(2\chi(A_n)/\delta)}. 
\end{aligned} 
$$ 
The last equality follows since $ \sum_{j = 1}^{\chi(A_n)} |\mathcal{C}_n(j)| =n$. By Lemma \ref{lem:brooktheorem} the proof completes. 
\end{proof}

\subsubsection{Concentration of the average outcomes} 


\begin{lem} \label{lem:jjhu} Suppose that treatments are assigned as in Assumption \ref{defn:bernoulli} with 
$$
\small 
\begin{aligned} 
& D_{i,0}^{(k)} \sim \pi(X_i^{(k)}, \beta_0), \quad 
D_{i,0}^{(k + 1)} \sim \pi(X_i^{(k + 1)}, \beta_0), \quad 
  D_{i,t}^{(k)} \sim \pi(X_i^{(k)}, \beta), \quad 
D_{i,t}^{(k + 1)} \sim \pi(X_i^{(k + 1)}, \beta') 
\end{aligned} 
$$ 
with exogenous parameters $\beta_0, \beta, \beta'$ (i.e., independent of $\bar{Y}_t^{(k+1)}, \bar{Y}_0^{(k +1)}, \bar{Y}_t^{(k)}, \bar{Y}_0^{(k)}$). Let Assumption \ref{ass:ass_0} hold, with sub-Gaussian $Y_{i,t}^{(k)}$. Then with probability at least  $1 - \delta$, for any $\delta \in (0,1)$
$$
\Big| \bar{Y}_t^{(k)} - \bar{Y}_t^{(k + 1)} - \bar{Y}_0^{(k)} + \bar{Y}_0^{(k + 1)} - \int (y(x, \beta) -  y(x, \beta')) dF_X(x)\Big|  \le c_0 \sqrt{\frac{\gamma_N \log(\gamma_N/\delta)}{n}},  
$$  
for a finite constant $c_0 < \infty$ independent of $(n, N, \gamma_N, \delta, t, T, k, K)$. 
\end{lem} 

\begin{proof}[Proof of Lemma \ref{lem:jjhu}] First, note that by Assumption \ref{ass:ass_0}, we can write 
\begin{equation}
\small 
\begin{aligned}
\mathbb{E}\Big[\bar{Y}_t^{(k)} - \bar{Y}_t^{(k + 1)}\Big] & = \int (y(x, \beta) -  y(x, \beta')) dF_X(x)  + \tau_k - \tau_{k+1},\quad \mathbb{E}\Big[\bar{Y}_0^{(k)} - \bar{Y}_0^{(k + 1)}\Big] & = \tau_k - \tau_{k+1}. 
\end{aligned}  
\end{equation} 
In addition, by Assumption \ref{ass:ass_0}, $Y_{i,t}^{(k)}$ form a dependency graph with maximum degree bounded by $2 \gamma_N$. The proof completes by invoking Lemma \ref{lem:concentration1}. 
\end{proof} 

\begin{lem} \label{lem:expect} Let $y(x, \beta)$ be twice differentiable with uniformly bounded derivatives for all $x \in \mathcal{X}, \beta \in \mathcal{B}$. Then for all interior $\beta \in \mathcal{B}$, where $\mathcal{B}$ is a compact space 
$$
\small 
\begin{aligned} 
\Big|\int \Big[y(x, \beta + \eta_n \underline{e}_j) - y(x, \beta - \eta_n \underline{e}_j)\Big] dF_X(x) -  2 \eta_n M^{(j)}(\beta)\Big| \le c_0 \eta_n^2. 
\end{aligned} 
$$  
for a finite constant $c_0 < \infty$, 
\end{lem} 
\begin{proof}[Proof of Lemma \ref{lem:expect}] 
The lemma follows from the mean-value theorem, and the dominated convergence theorem (used to interchange integration and differentiation). 
\end{proof} 

\begin{lem} \label{lem:2} Let the conditions in Lemma \ref{lem:jjhu} hold.  Let $y(x, \beta)$ be twice differentiable in $\beta$ with uniformly bounded derivatives for all $x \in \mathcal{X}, \beta \in \mathcal{B}$. Suppose that $\beta = \check{\beta} + \eta_n \underline{e}_j$ and $\beta' = \check{\beta} - \eta_n \underline{e}_j$, with an $\check{\beta}$ exogenous parameter (i.e., independent of $\bar{Y}_t^{(k+1)}, \bar{Y}_0^{(k +1)}, \bar{Y}_t^{(k)}, \bar{Y}_0^{(k)}$), both in the interior of $\mathcal{B}$. Then with probability at least $1 - \delta$, for any $\delta \in (0,1)$
$$
\small 
\begin{aligned} 
\Big| \frac{\bar{Y}_t^{(k)} - \bar{Y}_t^{(k + 1)} - \bar{Y}_0^{(k)} + \bar{Y}_0^{(k + 1)}}{2 \eta_n} - M^{(j)}(\check{\beta})\Big|  \le c_0 \sqrt{\frac{\gamma_N \log(\gamma_N/\delta)}{\eta_n^2 n}} + c_0 \eta_n,
\end{aligned} 
$$ 
for a finite constant $c_0 < \infty$ independent of $(n, N, \gamma_N, \delta, t, T, k, K)$. 
\end{lem} 

\begin{proof}[Proof of Lemma \ref{lem:2}] The proof is immediate from Lemma \ref{lem:jjhu}, and Lemma \ref{lem:expect}.
\end{proof}


\subsubsection{Lemmas for the adaptive experiment} \label{lem:main_adaptive}

The following 
lemma shows that the parameters used in the experiment are independent of potential outcomes and covariates in the same cluster. The sequential cross-fitting breaks the dependence due to repeated sampling, which would otherwise confound the experiment.

\begin{lem}[Unconfoundedness] \label{lem:1a} Let $T/p +  1\le K/2$. Consider the experimental design in Algorithm \ref{alg:adaptive_complete} for generic $p$-dimensions (and Algorithm \ref{alg:adaptive} for $p = 1$). Then, for any $k$, 
$$
\small 
\begin{aligned} 
\Big(\hat{\beta}_{k,1}, \cdots, \hat{\beta}_{k, T} \Big) \perp \Big\{Y_{i,t}^{(k)}(\mathbf{d}), X_i^{(k)}, \mathbf{d} \in \{0,1\}^N\Big\}_{i \in \{1, \cdots, N\}, t \le T}. 
\end{aligned}
$$
\end{lem}

\begin{proof}[Proof of Lemma \ref{lem:1a}]
Let
$ 
G:=K/2, P_g:=\{2g-1,2g\}, g\in\{1,\ldots,G\},
$ 
denote the cluster pairs. Define the successor map on pairs by
\[
s(g)=
\begin{cases}
g+1, & g<G,\\
1, & g=G,
\end{cases}
\qquad
s^r(g):=\underbrace{s\circ\cdots\circ s}_{r\text{ times}}(g).
\]
For a set of pair indices \(A\subseteq\{1,\ldots,G\}\), let \(\mathcal F_A\) denote the sigma-field generated by the covariates, potential outcomes, and experimental randomization draws in all pairs \(P_a\), \(a\in A\). Let \(\mathcal Z_g\) denote the sigma-field generated by the covariates and potential outcomes in pair \(P_g\). By independence across clusters, and by independent experimental randomization, \(\mathcal F_A\perp \mathcal Z_g\) whenever \(g\notin A\).

It is enough to prove that, for every pair \(P_g\), the center parameter used in that pair is independent of \(\mathcal Z_g\). The perturbation sign and the coordinate perturbation are deterministic functions of the pair position and the coordinate, so this implies the desired statement for each cluster \(k\in P_g\).

Write \(b_g^w\) for the common center parameter \(\check\beta_h^w\) used for the two clusters \(h\in P_g\). This is common within a pair by construction of Algorithm \ref{alg:adaptive_complete}. Write \(\widehat M_{g,w}\) for the vector of estimated marginal effects from pair \(P_g\) during wave \(w\). In pair notation, the update in Algorithm \ref{alg:adaptive_complete} can be written as
$ 
b_g^w
=
P_{\mathcal B_1,\mathcal B_2-\eta_n}
\left[
b_g^{w-1}
+
\alpha_{s(g),w-1}\widehat M_{s(g),w-1}
\right],
 b_g^0=\beta_0 .
$ 
We prove by induction that, for every \(w\in\{0,\ldots,\check T\}\),
$ 
b_g^w
\quad\text{is } \mathcal F_{A_g(w)}\text{-measurable},
A_g(w):=\{s(g),s^2(g),\ldots,s^w(g)\},
$ 
with \(A_g(0)=\varnothing\).

For \(w=0\), \(b_g^0=\beta_0\) is deterministic, so the claim holds. Suppose the claim holds for \(w-1\), for all pairs. The estimator \(\widehat M_{s(g),w-1}\) is computed from the data in pair \(P_{s(g)}\) during wave \(w-1\), using the policy center \(b_{s(g)}^{w-1}\). By the induction hypothesis applied to pair \(s(g)\),
$ 
b_{s(g)}^{w-1}
\text{ is }
\mathcal F_{\{s^2(g),\ldots,s^w(g)\}}\text{-measurable}.
$ 
Therefore \(\widehat M_{s(g),w-1}\) is measurable with respect to
$ 
\mathcal F_{\{s(g),s^2(g),\ldots,s^w(g)\}}.
$ 
Also, by the induction hypothesis,
$ 
b_g^{w-1}
\text{ is }
\mathcal F_{\{s(g),\ldots,s^{w-1}(g)\}}\text{-measurable}.
$ 
Since \(b_g^w\) is a deterministic function of \(b_g^{w-1}\) and \(\widehat M_{s(g),w-1}\), it follows that \(b_g^w\) is \(\mathcal F_{A_g(w)}\)-measurable. This proves the induction claim.

Now use the cluster-count condition. Since
$ 
\check T+1=T/p+1\le K/2=G,
$ 
we have \(\check T\le G-1\). Hence, for every \(w\le \check T\),
$ 
g\notin A_g(w)=\{s(g),\ldots,s^w(g)\};
$ 
that is, the circular chain of successor pairs used to construct \(b_g^w\) has not yet returned to the original pair \(P_g\). Therefore \(b_g^w\) is measurable with respect to a sigma-field generated only by pairs different from \(P_g\). By independence across clusters,
$ 
b_g^w\perp \mathcal Z_g .
$ 
For any cluster \(k\in P_g\), the parameter \(\hat\beta_{k,t}\) used at any period \(t\le T\) is a deterministic function of \(b_g^w\), the coordinate being perturbed, and the deterministic sign assigned to cluster \(k\) within the pair. Therefore
$ 
\Big(\hat{\beta}_{k,1},\ldots,\hat{\beta}_{k,T}\Big)
\perp
\Big\{
Y_{i,t}^{(k)}(\mathbf d),X_i^{(k)}
:
\mathbf d\in\{0,1\}^N,\ i\in\{1,\ldots,N\},\ t\le T
\Big\}.
$ 
This proves the lemma.
\end{proof}

The following lemma follows by standard properties of the gradient descent algorithm. Recall the definition of $\beta^*$ in Equation \eqref{eqn:estimand} (main text) and $\beta_w^{**}$ in Equation \eqref{eqn:beta_start2}. 

  \begin{lem} \label{lem:gradient_descent} For the learning rate $\alpha_w = J/(w+1)$, and $\beta_w^{**}$ as in Equation \eqref{eqn:beta_start2}, under Assumption \ref{ass:regularity_basic}, \ref{ass:bounded}, \ref{ass:strong_concavity}, with $\sigma$-strong concavity, for $J \ge 1/\sigma$,  then  
  $
  ||\beta_{w}^{**} - \beta^*||^2 \le \frac{L p}{w}, 
  $ 
  where $L = \max\{2(\mathcal{B}_2 - \mathcal{B}_1)^2,  G^2 J^2, 1\}$, $G = \sup_{ \beta} ||\frac{\partial W(\beta)}{\partial \beta}||_{\infty}$. 
  \end{lem}

\begin{proof}[Proof of Lemma \ref{lem:gradient_descent}]    The proof follows standard arguments of the gradient descent method \citep{bottou2018optimization}, where, here, we leverage strong concavity and the assumption that the gradient is uniformly bounded.  Denote $\beta^*$ the estimand of interest and recall the definition of $\beta^{**}_w$ in Equation \eqref{eqn:beta_start2}, and since it is an interior the marginal effect evaluated at $\beta^*$ equals zero. We define $\nabla_{w-1}$ the gradient evaluated at $\beta_{w-1}^{**}$. By strong concavity,  
 
  \begin{equation} \label{eqn:helper222} 
  \small 
  \begin{aligned} 
\left(\frac{\partial W\left(\beta_{w-1}^{* *}\right)}{\partial \beta}-\frac{\partial W\left(\beta^{*}\right)}{\partial \beta}\right)\left(\beta^{*}-\beta_{w-1}^{* *}\right) \geq \sigma\left\|\beta_{w-1}^{* *}-\beta^{*}\right\|_{2}^{2}. 
 \end{aligned} 
  \end{equation} 
 In addition, we can write: (because $\beta^* \in [\mathcal{B}_1, \mathcal{B}_2]^p$)
 $$
 \small 
 \begin{aligned} 
 ||\beta_w^{**} - \beta^*||_2^2 = ||\beta^* - P_{\mathcal{B}_1, \mathcal{B}_2}(\beta_{w-1}^{**} + \alpha_{w-1} \nabla_{w-1})||_2^2 \le  ||\beta^* - \beta_{w-1}^{**} - \alpha_{w-1} \nabla_{w-1}||_2^2.  
 \end{aligned} 
 $$
 Observe that we have 
$
||\beta^* - \beta_w^{**}||_2^2 \le 
||\beta^* - \beta_{w-1}^{**}||_2^2 - 2\alpha_{w-1} \nabla_{w-1} (\beta^* - \beta_{w-1}^{**}) + \alpha_{w-1}^2 ||\nabla_{w-1}||_2^2.  
 $
 Using Equation \eqref{eqn:helper222}, we can write 
 $
 ||\beta_{w+1}^{**} - \beta^*||_2^2 \le (1 - 2 \sigma \alpha_w) ||\beta_{w}^{**} - \beta^*||_2^2 +  \alpha_w^2 G^2 p . 
 $ 
 we prove the statement by induction. At time $w = 1$, the statement trivially holds. For general $w$,   
 $$
 \small 
 \begin{aligned} 
 ||\beta_{w+1}^{**} - \beta^*||_2^2\le (1 - 2\frac{1}{w + 1}) \frac{L p}{w} +  \frac{L p}{(w+1)^2} & \le  (1 - 2\frac{1}{w + 1}) \frac{L p}{w} +  \frac{L p}{w(w+1)} = (1 - \frac{1}{w + 1}) \frac{L p}{w}.
 \end{aligned}  
 $$ 
 The right-hand side above equals $\frac{L p}{w+1}$, completing the proof. \end{proof} 
 
 \begin{lem}  \label{lem:fg2}  Let Assumptions \ref{ass:ass_0}, \ref{ass:regularity_basic}, \ref{ass:bounded} hold. Let $\alpha_w$ be as defined in Lemma \ref{lem:gradient_descent}.  Then with probability at least $1 - \delta$, for any $\delta \in (0,1)$, for all $w \ge 1$, assuming $\beta_0 \in [\mathcal{B}_1 + \eta_n, \mathcal{B}_2 - \eta_n]^p$,
 
$$
 \small 
 \begin{aligned} 
 \Big|\Big| \check{\beta}_k^w - \beta_w^{**} \Big|\Big|_{\infty} \le c_0 P_{w}(\delta)
 \end{aligned} 
$$
 where $P_{1}(\delta) = \alpha_1 \times \mathrm{err}(\delta)$ and $P_{w}(\delta) =  B p \alpha_{w} P_{w-1}(\delta) + P_{w-1}(\delta) +  \alpha_w \mathrm{err}(\delta)$, and $\mathrm{err}(\delta)\le c_0 \Big(\sqrt{ \gamma_N\frac{\log(p \check{T} K/\delta)}{\eta_n^2 n}} + p\eta_n\Big)$,  for finite constants $B < \infty, c_0 < \infty$ independent of $(n, N, \gamma_N, \delta, t, T, k, K, p)$.  
 \end{lem}

\begin{proof}[Proof of Lemma \ref{lem:fg2}]
Under Assumption \ref{ass:regularity_basic}, let
\[
B
:=
\sup_{\beta\in\mathcal B}
\max_{j,\ell\in\{1,\ldots,p\}}
\left|
\frac{\partial^2 W(\beta)}
{\partial \beta_j\partial \beta_\ell}
\right| <\infty .
\]
By Lemmas \ref{lem:2} and \ref{lem:1a}, and by a union bound over
\(j\in\{1,\ldots,p\}\), \(k\in\{1,\ldots,K\}\), and
\(w\in\{1,\ldots,\check T\}\), with probability at least \(1-\delta\),
\[
\left|
\check M_{k,w}^{(j)}
-
M^{(j)}\left(\check\beta_{k+2}^w\right)
\right|
\le
c_0
\left(
\sqrt{
\gamma_N
\frac{\log(p\check T K/\delta)}
{\eta_n^2 n}
}
+\eta_n
\right)
\]
for all \(j,k,w\). Therefore, on the same event,
\[
\left\|
\check M_{k,w}
-
M\left(\check\beta_{k+2}^w\right)
\right\|_\infty
\le
\mathrm{err}(\delta),
\]
where, increasing \(c_0\) if necessary,
\[
\mathrm{err}(\delta)
\le
c_0
\left(
\sqrt{
\gamma_N
\frac{\log(p\check T K/\delta)}
{\eta_n^2 n}
}
+
p\eta_n
\right).
\]
We prove the result on this event.

We use induction on \(w\). For \(w=1\), the algorithm initializes all pairs at the common value \(\beta_0\). Since the wave-zero gradient is initialized at zero,
\[
\check\beta_k^1=\beta_0,
\qquad
\beta_1^{**}=\beta_0.
\]
Therefore
\[
\left\|
\check\beta_k^1-\beta_1^{**}
\right\|_\infty
=
0
\le c_0P_1(\delta),
\]
which proves the base case.

Suppose now that, for some \(w\ge2\),
\[
\left\|
\check\beta_k^{w-1}-\beta_{w-1}^{**}
\right\|_\infty
\le
c_0P_{w-1}(\delta)
\]
for all \(k\). By the recursive update in Algorithm \ref{alg:adaptive_complete},
\[
\check\beta_k^w
=
P_{\mathcal B_1+\eta_n,\mathcal B_2-\eta_n}
\left[
\check\beta_k^{w-1}
+
\alpha_{w-1}\check M_{k,w-1}
\right],
\]
where \(\check M_{k,w-1}\) is the gradient estimated in the subsequent pair at wave \(w-1\). The oracle recursion satisfies
\[
\beta_w^{**}
=
P_{\mathcal B_1,\mathcal B_2}
\left[
\beta_{w-1}^{**}
+
\alpha_{w-1}M(\beta_{w-1}^{**})
\right].
\]
Since projection on a rectangle is non-expansive in the sup norm, and since the two projection rectangles differ by at most \(\eta_n\) in each coordinate,
\[
\begin{aligned}
\left\|
\check\beta_k^w-\beta_w^{**}
\right\|_\infty
&\le
\left\|
\check\beta_k^{w-1}
-
\beta_{w-1}^{**}
\right\|_\infty
+
\alpha_{w-1}
\left\|
\check M_{k,w-1}
-
M(\beta_{w-1}^{**})
\right\|_\infty
+
c_0\eta_n .
\end{aligned}
\]
Because \(P_{w-1}(\delta)\ge P_1(\delta)=\alpha_1\mathrm{err}(\delta)\) and \(\mathrm{err}(\delta)\) contains the term \(p\eta_n\), the last term is absorbed into \(c_0P_{w-1}(\delta)\). Thus,
\[
\begin{aligned}
\left\|
\check\beta_k^w-\beta_w^{**}
\right\|_\infty
&\le
c_0P_{w-1}(\delta)
+
\alpha_{w-1}
\left\|
\check M_{k,w-1}
-
M(\beta_{w-1}^{**})
\right\|_\infty .
\end{aligned}
\]

Next,
\[
\begin{aligned}
\left\|
\check M_{k,w-1}
-
M(\beta_{w-1}^{**})
\right\|_\infty
&\le
\left\|
\check M_{k,w-1}
-
M(\check\beta_{k+2}^{w-1})
\right\|_\infty
\\
&\quad+
\left\|
M(\check\beta_{k+2}^{w-1})
-
M(\beta_{w-1}^{**})
\right\|_\infty .
\end{aligned}
\]
The first term is bounded by \(\mathrm{err}(\delta)\). For the second term, by the mean-value theorem and the definition of \(B\),
\[
\left\|
M(\check\beta_{k+2}^{w-1})
-
M(\beta_{w-1}^{**})
\right\|_\infty
\le
Bp
\left\|
\check\beta_{k+2}^{w-1}
-
\beta_{w-1}^{**}
\right\|_\infty .
\]
By the induction hypothesis applied to the pair \(k+2\),
\[
\left\|
\check\beta_{k+2}^{w-1}
-
\beta_{w-1}^{**}
\right\|_\infty
\le
c_0P_{w-1}(\delta).
\]
Therefore,
\[
\left\|
\check M_{k,w-1}
-
M(\beta_{w-1}^{**})
\right\|_\infty
\le
\mathrm{err}(\delta)
+
c_0BpP_{w-1}(\delta).
\]
Combining the previous displays gives
\[
\begin{aligned}
\left\|
\check\beta_k^w-\beta_w^{**}
\right\|_\infty
&\le
c_0P_{w-1}(\delta)
+
c_0Bp\alpha_{w-1}P_{w-1}(\delta)
+
c_0\alpha_{w-1}\mathrm{err}(\delta).
\end{aligned}
\]
Since \(\alpha_w=J/(w+1)\), for \(w\ge2\) we have \(\alpha_{w-1}\le 2\alpha_w\). Hence, increasing \(c_0\) if necessary,
\[
\left\|
\check\beta_k^w-\beta_w^{**}
\right\|_\infty
\le
c_0
\left[
P_{w-1}(\delta)
+
Bp\alpha_wP_{w-1}(\delta)
+
\alpha_w\mathrm{err}(\delta)
\right].
\]
By the definition of \(P_w(\delta)\),
\[
P_w(\delta)
=
Bp\alpha_wP_{w-1}(\delta)
+
P_{w-1}(\delta)
+
\alpha_w\mathrm{err}(\delta),
\]
and therefore
\[
\left\|
\check\beta_k^w-\beta_w^{**}
\right\|_\infty
\le
c_0P_w(\delta).
\]
This completes the induction and proves the lemma.
\end{proof}

 \begin{lem} \label{thm:regret2} Let the conditions in  Lemmas \ref{lem:gradient_descent},  \ref{lem:fg2} hold. Then  with probability at least $1 - \delta$, for any $\delta \in (0,1)$, for all $w \ge 1, k \in \{1, \cdots, K\}$,  for finite constants $B, L< \infty$
  \begin{equation} \label{eqn:lemmab9} 
  \small 
  \begin{aligned} 
 ||\beta^{*} - \check{\beta}_{k}^{w}||_2^2 \le \frac{L p}{w} +p  w^{p B} e^{Bp} \times c_0 \Big(\gamma_N \frac{\log(p \check{T}K/\delta)}{\eta_n^2 n} + p^2 \eta_n^2 \Big),
 \end{aligned}
 \end{equation} 
  for a finite constant $c_0 < \infty$ independent of $(n, N, \gamma_N, \delta, t, T, k, K, p)$, 
 \end{lem} 
 \begin{proof}[Proof of Lemma \ref{thm:regret2}] \renewcommand{\qedsymbol}{} We can write 
 $
 ||\beta^{*} - \check{\beta}_{k}^{w}||_2^2 \le 
 2 ||\beta^{*} - \beta_{w}^{**}||_2^2 + 2 ||\check{\beta}_{k}^{w} - \beta_{w}^{**}||_2^2. 
 $
 The first component on the right-hand side is bounded by Lemma \ref{lem:gradient_descent}. Using Lemma \ref{lem:fg2}, we bound the second component with probability at least $1 - \delta$, as follows 
 $
 ||\check{\beta}_{k}^{w} - \beta_{w}^{**}||_2^2 \le 
 p ||\check{\beta}_{k}^{w} - \beta_{w}^{**}||_{\infty}^2   \le p c_0 (P_{w}^2(\delta)), 
$ 
for a finite constant $c_0$.
We conclude the proof by characterizing $P_{w}(\delta)$ as defined in Lemma \ref{lem:fg2}. Following Lemma \ref{lem:fg2}, we can define recursively $P_{w}(\delta)$ for any $1 \le w \le \check{T}$ (recall that $\alpha_{w} \propto 1/w$) as
$$
\small 
\begin{aligned} 
P_{w}(\delta) \le (1 + \frac{B p}{w}) P_{w - 1}(\delta) + \frac{1}{w} \mathrm{err}_n(\delta), \quad P_1(\delta) = \mathrm{err}_n(\delta). 
\end{aligned} 
$$  
where $\mathrm{err}_n(\delta) \le c_0 \left( \sqrt{\gamma_N \frac{\log(p \check{T} K/\delta)}{\eta_n^2 n}} + p \eta_n \right)$. Take, without loss of generality, $B \ge 1$ (if $B < 1$, we can find an upper bound with a different $B =1$). 
Substituting recursively each term, we can write
$
P_{w}(\delta) \le \mathrm{err}_n(\delta) \sum_{s=1}^w \frac{1}{s} \prod_{j=s}^w (\frac{B p}{j} + 1). 
$  
We now write 
$$
\small 
\begin{aligned} 
\sum_{s=1}^w \frac{1}{s} \prod_{j=s}^w (\frac{B p}{j} + 1)  & \le  \sum_{s=1}^w \frac{1}{s} \exp(\sum_{j=s}^w \frac{B p}{j}) \le  \sum_{s=1}^w \frac{1}{s} e^{\Big(Bp + B p \log(w) - B p \log(s) \Big)}   & \lesssim \sum_{s=1}^w \frac{1}{s^2} e^{B p \log(w) + Bp} \lesssim w^{B p} e^{Bp},   
\end{aligned} 
$$ 
completing the proof.
 \end{proof}

 \subsection{Proofs of the theorems}  \label{sec:const1}

For the following proofs, define a finite constant $c_0 < \infty$ independent of $(n, N, \gamma_N, \delta, t, T, k, K, p)$.

\subsubsection{Proof of Theorem \ref{thm:const1}} \label{sec:proof1}

First observe that for any $\delta \in (0,1)$, 
$
\Big|\mathbb{E}\Big[\widehat{M}_{k}(\beta) \Big] - M^{(1)}(\beta)\Big| \le c_0 \eta_n, \quad P\Big(\Big|\widehat{M}_{k}(\beta) - M^{(1)}(\beta)\Big| > c_0 \Big(\eta_n + \sqrt{\frac{\gamma_N \log(\gamma_N/\delta)}{n \eta_n^2}}\Big)\Big) \le \delta, 
$ 
with the proof of the first claim follows similarly as in the proof of Lemma \ref{lem:expect} and the second claim being a direct corollary of Lemma \ref{lem:2}. Finally observe that with probability at least $1 - \delta$, for any $\delta \in (0,1)$, we also have 
$
\Big|\widehat{M}_{k}(\beta) - M^{(1)}(\beta)\Big| \le c_0 \eta_n + c_0 \Big( \sqrt{\frac{\rho_n}{ \delta n \eta_n^2}}\Big),  
$
by Chebyshev inequality and the triangular inequality. 

\vspace{-2mm}

\subsubsection{Proof of Theorem \ref{thm:inference1}} \label{sec:inference_thm}

Consider Algorithm \ref{alg:my_pilot} for a generic coordinate $j$. Let $\beta$ be the target parameter as in Algorithm \ref{alg:my_pilot}. 
By Lemma \ref{lem:expect}, we have 
$
\Big|\mathbb{E}[\widehat{M}_{k}^{(j)}] -  M^{(j)}(\beta) \Big| \le  c_0 \eta_n.  
$
we have  
\begin{equation} \label{eqn:rhs_inference}
\small 
\begin{aligned} 
\Big|\frac{\widehat{M}_{k}^{(j)} -\mathbb{E}[\widehat{M}_{k}^{(j)} ]}{\sqrt{\mathrm{Var}(\widehat{M}_{k}^{(j)})}} - \frac{\widehat{M}_{k}^{(j)} - M^{(j)}(\beta)}{\sqrt{\mathrm{Var}(\widehat{M}_{k}^{(j)})}} \Big| \le c_0 \left(\frac{\eta_n  }{\sqrt{\mathrm{Var}(\widehat{M}_{k}^{(j)})}}\right). 
\end{aligned}
\end{equation} 
Observe that under Assumption \ref{ass:var}, 
$
\frac{\eta_n}{\sqrt{\mathrm{Var}(\widehat{M}_{k}^{(j)})}} \le (C_k + C_{k + 1}) \eta_n^2 \times \sqrt{n} , 
$
because $\mathrm{Var}(\sqrt{n} \widehat{M}_{k}^{(j)}) \ge (C_k + C_{k + 1}) \rho_n/\eta_n^2$, where $\rho_n \ge 1$ by Assumption \ref{ass:var} (i.e., the variance is not degenerate), and $(C_k + C_{k + 1}) > 0$ are positive constants in Assumption \ref{ass:var}. For $\eta_n = o(n^{-1/4})$, the right-hand side in Equation \eqref{eqn:rhs_inference} is $o(1)$. 
 
Observe now that by Assumption \ref{ass:ass_0}, $Y_{i,t}^{(k)} - Y_{i,0}^{(k)}$ form a locally dependent graph of maximum degree of order $\mathcal{O}(\gamma_N)$. By Lemma \ref{lem:locally_dep}, 
$$
\small 
\begin{aligned} 
& d_W\Big(\frac{1}{2 \eta_n \sqrt{\mathrm{Var}(\widehat{M}_{k}^{(j)})}} \Big[\bar{Y}_t^{(k)} - \bar{Y}_0^{(k)}\Big] - \frac{1}{2 \eta_n \sqrt{\mathrm{Var}(\widehat{M}_{k}^{(j)})}} \Big[\bar{Y}_t^{(k + 1)} - \bar{Y}_0^{(k+1)}\Big] - \frac{\mathbb{E}[\widehat{M}_{k}^{(j)}]}{\sqrt{\mathbb{V}(\widehat{M}_{k}^{(j)})}}, \mathcal{Z}\Big) \\ &\le \underbrace{\frac{\gamma_N^2}{\tilde{\sigma}^3} \sum_{h \in \{k, k+1\}} \sum_{i = 1 }^n  \Big[\mathbb{E}\Big| \frac{Y_{i,t}^{(k)} - Y_{i,0}^{(k)}}{\eta_n n}\Big|^3\Big]}_{(A)} + \underbrace{\frac{\sqrt{28} \gamma_N^{3/2}}{\sqrt{\pi} \tilde{\sigma}^2}  \sqrt{\sum_{i  = 1}^n \Big[\mathbb{E}\Big| \frac{Y_{i,t}^{(k)} - Y_{i,0}^{(k)}}{\eta_n n}\Big|^4 \Big]}}_{(B)}, 
\end{aligned} 
$$
where $\mathcal{Z} \sim \mathcal{N}(0,1), \quad \tilde{\sigma}^2 = \mathrm{Var}\Big(\frac{1}{2 \eta_n}  \Big[\bar{Y}_t^{(k)} - \bar{Y}_0^{(k)}\Big] - \frac{1}{2 \eta_n} \Big[\bar{Y}_t^{(k + 1)} - \bar{Y}_0^{(k+1)}\Big]\Big)$, 
and $d_W$ denotes the Wasserstein metric. Under Assumption \ref{ass:var}, $\tilde{\sigma}^2 \ge (C_k + C_{k'})\frac{1}{n \eta_n^2}$ for a constant $C_k + C_{k'} > 0$, and the third and fourth moment are bounded. Hence, we have for a constant $C' < \infty$, 
$ 
(A) \le C' \frac{\gamma_N^2}{n^3 \eta_n^3} \times n^{5/2} \eta_n^3 \lesssim \frac{\gamma_N^2}{n^{1/2}} \rightarrow 0.   
$
Similarly, for (B), we have 
$
(B) \le c' \frac{\gamma_N^{3/2} n \eta_n^2 }{\eta_n^2 n^{3/2}} \lesssim \frac{\gamma_N^{3/2} }{ n^{1/2}} \rightarrow 0.  
$

\subsubsection{Proof of Theorem \ref{thm:bias_direct}} \label{sec:proof3}

\paragraph{Direct effect and reward} By Assumptions \ref{ass:ass_0}, \ref{ass:ass_0b}, we can write (we omit the superscript $k$ from $X^{(k)}$ for the sake of brevity)
$$
\small 
\begin{aligned} 
& \mathbb{E}\Big\{\frac{1}{2 n} \sum_{i=1}^n \Big[\frac{D_{i,1}^{(k + 1)} Y_{i,1}^{(k + 1)}}{\pi(X_i, \beta + \eta_n \underline{e}_1)} - \frac{(1 - D_{i,1}^{(k + 1)}) Y_{i,1}^{(k + 1)}}{1 - \pi(X_i, \beta + \eta_n \underline{e}_1)}\Big] + \frac{1}{2 n} \sum_{i=1}^n\Big[\frac{D_{i,1} Y_{i,1}^{(k)}}{\pi(X_i, \beta - \eta_n \underline{e}_1)} - \frac{(1 - D_{i,1}^{(k)}) Y_{i,1}^{(k)}}{1 - \pi(X_i, \beta - \eta_n \underline{e}_1)} \Big]\Big\} \\ 
&= \underbrace{\frac{1}{2} \int \Big[m(1, x, \beta + \eta_n \underline{e}_1) - m(0, x, \beta + \eta_n \underline{e}_1) + m(1, x, \beta - \eta_n \underline{e}_1) - m(0, x, \beta - \eta_n \underline{e}_1)\Big] dF_X(x)}_{(i)}.  
\end{aligned} 
$$ 
The last equality follows from Assumption \ref{ass:ass_0b} and exogeneity of $\beta$. 
By the mean-value theorem
  $$
  \small 
  \begin{aligned} 
  (i) & = \int \Big[ m(1,x, \beta) - m(0, x, \beta) + \frac{\partial m(1, x, \beta)}{2 \partial \beta^1} \eta_n - \frac{\partial m(0, x, \beta)}{2\partial \beta^1} \eta_n - \frac{\partial m(1, x, \beta)}{2 \partial \beta^1} \eta_n  + \frac{\partial m(0, x, \beta)}{2 \partial \beta^1} \eta_n\Big] dF_X(x) \\ & + \mathcal{O}(\eta_n^2) = \int \Big[m(1,x, \beta) -  m(0, x, \beta)\Big] dF_X(x) + o(n^{-1/2}) \quad (\because \eta_n = o(n^{-1/4})).  
  \end{aligned} 
$$ 
The case for $\bar{W}_n(\beta)$ follows verbatim and omitted for brevity, where we also need to account for the difference in time fixed effect and the effect at the baseline intervention $W(\beta_0)$, all included as part of the additive constant $c_1$.

\paragraph{Marginal spillover effect} Finally, consider studying
$$
\small 
\begin{aligned} 
(I) = \mathbb{E}\Big\{\frac{1}{2 n} \sum_{h \in \{k, k+1\}} \frac{v_h}{\eta_n} \sum_{i=1}^n  \Big[ \frac{Y_{i,1}^{(h)} (1 - D_{i,1}^{(h)})}{1 - \pi(X_i^{(h)}, \beta + v_h \eta_n \underline{e}_1)} - \bar{Y}_0^{(h)}\Big] \Big\}, 
\end{aligned}
$$
where $v_h = 1 \{h = k\} - 1\{h = k +1\}$. 
Using Assumption \ref{ass:ass_0b}, similarly to the derivation of Lemma \ref{lem:expect}, we can write $(I)$ equal to 
$
\frac{1}{2 \eta_n} \int [m(0, x, \beta + \eta_n \underline{e}_1) - m(0, x, \beta - \eta_n \underline{e}_1)] dF_X(x). 
$ 
Note that from the mean value theorem, and Assumption \ref{ass:ass_0b}
$
m(0, x, \beta + \eta_n \underline{e}_1) - m(0, x, \beta - \eta_n \underline{e}_1) = 
m(0, x, \beta ) - m(0, x, \beta)  + 2 \frac{\partial m(0, x, \beta)}{\partial \beta^1} \eta_n + \mathcal{O}(\eta_n^2)
$ 
which completes the proof.

\vspace{-3mm} 

 \subsubsection{Proof of Theorem \ref{thm:rate2b}} \label{sec:proof5}

\paragraph{In-sample regret} By the mean value theorem and Assumption \ref{ass:regularity_basic}, we have 
 $
 \sum_{w=1}^{\check{T}} W(\beta^*) - W(\check{\beta}_{k}^w) \le \bar{C} p  \sum_{w=1}^{\check{T}} ||\beta^* - \check{\beta}_{k}^w||_2^2   ,
 $ 
 for a finite constant $\bar{C} < \infty$. This follows from two facts: first, in the mean value expansion $\frac{\partial W(\beta^*)}{\partial \beta} = 0$; second, for the quadratic component, the Hessian is uniformly bounded (Assumption \ref{ass:regularity_basic}), and therefore the maximum eigenvalue of the Hessian can be bounded by $p$ times a constant equal to the largest entry of the Hessian. This latter fact implies that the second order term in the quadratic expansion for $ W(\beta^*) - W(\check{\beta}_{k}^w)$ obtained through the mean value theorem is bounded by $\bar{C} p  ||\beta^* - \check{\beta}_{k}^w||_2^2$,  for a constant $\bar{C}$ proportional to the largest entry of the Hessian. By Lemma \ref{thm:regret2}, choosing $\delta = 1/n$, and for $n$ satisfying the conditions in the statement of the Theorem, it follows that for all $k$, 
 $
 \sum_{w  = 1}^{\check{T}} W(\beta^*) - W(\check{\beta}_{k}^w) \le \sum_{w  = 1 }^{\check{T}} \frac{p^2 \kappa'}{w} \lesssim  p^2 \log(\check{T} + 1) 
 $
 for $\kappa' < \infty$ being a finite constant. The proof completes.

\paragraph{Out-of-sample regret} Consider Lemma \ref{thm:regret2} where we choose $\delta = 1/n$. We can write for each $k$ 
 $
 ||\beta^* - \check{\beta}_k^{\check{T}}||_2^2 \le \frac{p L}{\check{T}} + c_0(1/\check{T}), 
 $ 
 for a finite constant $L < \infty$, 
 under the conditions for $n$ stated in Theorem \ref{thm:rate2b} by Lemma \ref{thm:regret2} (since the second component in the right-hand side of Equation \eqref{eqn:lemmab9} in Lemma \ref{thm:regret2} is bounded by $c_0/\check{T}$).  Note that 
 $
 ||\beta^* - \frac{1}{K} \sum_{k=1}^K \check{\beta}_k^{\check{T}}||_2^2 \le 
 \frac{1}{K} \sum_{k=1}^K ||\beta^* -  \check{\beta}_k^{\check{T}}||_2^2  
 $ 
 by Jensen's inequality. We can write $W(\beta^*) - W(\hat{\beta}) \le c_0 p  ||\beta^* - \frac{1}{K} \sum_{k=1}^K \check{\beta}_k^{\check{T}}||_2^2$ from Assumption \ref{ass:regularity_basic} and the mean-value theorem, for a finite constant $c_0 < \infty$. The proof is complete.

\subsubsection{Proof of Theorem \ref{thm:rate2bb}} \label{sec:exponential2} First, note that for a finite constant $c_0$, under Assumption \ref{ass:regularity_basic} and Assumption \ref{ass:strong_concavity}
$
W(\beta^*) - W(\hat{\beta}^*) \le c_0 ||\beta^* - \hat{\beta}||^2 \le c_0 p \frac{1}{K} \sum_{k=1}^K ||\beta^* - \check{\beta}_k^{\check{T}}||^2 
$ 
where in the first inequality we used strong concavity (gradient equals zero) together with Assumption \ref{ass:regularity_basic} (bounded Hessian) and the mean-value theorem, and in the second equality  we used Jensen's inequality. 
 Define $\beta_w^{**}$ as in Equation \eqref{eqn:beta_start2}, where, however, the learning rate is chosen so that $\alpha_w = 1/\tau$. 
we can write  
 $
 ||\beta^{*} - \check{\beta}_{k}^{\check{T}}||_2^2 \le 
 2 ||\beta^{*} - \beta_{\check{T}}^{**}||_2^2 + 2 ||\check{\beta}_{k}^{\check{T}} - \beta_{\check{T}}^{**}||_2^2. 
 $
 The first component is bounded by Theorem 3.10 in \cite{bubeck2014convex} (using the fact that $\mathcal{B}$ is compact) as follows:
 $
 ||\beta^{*} - \beta_{\check{T}}^{**}||_2^2 \le c_{0,p} \exp(-c_{0,p}' 2(\check{T} + 1)) = c_{0,p} \exp(-K c_{0,p}') 
 $ 
 for finite constants $0< c_{0,p}, c_{0,p}' < \infty$ that do not depend on $(N,\gamma_N, K, \check{T})$, where we used the fact that $2(\check{T} + 1) = K$. Using Lemma \ref{lem:fg2}, we bound the second component with probability at least $1 - \delta$, as follows (for any $w \le \check{T}$) 
 $
 ||\check{\beta}_{k}^{w} - \beta_{w}^{**}||_2^2 \le 
 p ||\check{\beta}_{k}^{w} - \beta_{w}^{**}||_{\infty}^2   = p \times c_0(P_{w}^2(\delta)),  
$ 
for a finite constant $c_0 < \infty$. 
We conclude the proof by characterizing $P_{w}(\delta)$ as defined in Lemma \ref{lem:fg2}. Following Lemma \ref{lem:fg2}, we can define recursively $P_{w}(\delta)$ for any $1 \le w \le \check{T}$ as
$$
\small 
\begin{aligned} 
P_{w}(\delta) \le (1 + B p) P_{w - 1}(\delta) + \mathrm{err}_n(\delta), \quad P_1(\delta) = \mathrm{err}_n(\delta). 
\end{aligned}
$$  
where $\mathrm{err}_n \le c_0( \sqrt{\gamma_N \frac{\log(p \check{T} K/\delta)}{\eta_n^2 n}} + p \eta_n )$, and $B > 0$ is a finite constant as in Lemma \ref{lem:fg2}. Using a recursive argument, we can write 
$
P_w(\delta) \lesssim w (1 + p B)^w \mathrm{err}_n(\delta). 
$  
The proof completes as we choose $n$ sufficiently large as stated in the theorem.

 \subsection{Corollaries} \label{sec:corollary}

\begin{proof}[Corollary \ref{thm:inference4b}]
Theorem \ref{thm:inference1} implies asymptotic normality. Because cluster are mutually independent this implies that the conditions of Theorem 3.1 
The result follows from Theorem 3.1 in \cite{canay2017randomization} (and the corresponding application in Section 4 in \cite{canay2017randomization}) are attained proving the result. 
\end{proof}

 \section{Additional results from the experiment} \label{app:more_experiment}

\subsection{Balance and robustness checks: Overview} \label{app:more_exp1}

Next, we provide an overview of additional analyses and balance checks in the Appendix.

\vspace{-4mm}

\paragraph{Balance} We use auxiliary data about farmers' baseline characteristics for \textit{all} farmers enrolled with PxD in the main experiment (more than 287,000 farmers) to test for homogeneity in covariates between different clusters, a relevant assumption in our framework. 
Appendix Table \ref{tab:summaries} reports the sample means across observable baseline covariates from program administrative data (each covariate is described below Table \ref{tab:summaries}).  
We test for differences in covariates between clusters exposed to different treatment probabilities. (When estimating marginal effects, it is easy to show that our framework only requires homogeneity restrictions between groups of clusters used to estimate the marginal effects (e.g., the group of clusters in different treatment exposures), but not necessarily between individual clusters having the same exposures.)    The relevant null hypothesis is that the expected value of each covariate in Table \ref{tab:summaries} in each cluster is the same across all clusters. We construct these tests via randomization inference formally described in Appendix \ref{sec:perm_tests}.
These tests are informative of whether such groups are comparable and are conducted with a large sample size ($n \approx 10,000$ on average in each tehsil). We observe similar estimates across all covariates. The smallest p-value is $0.21$, the median is above $0.5$, suggesting lack of  tehsil-level heterogeneity.

In Appendix Tables \ref{tab:summaries_response_rates}, \ref{tab:summaries_response_rates_v2}, \ref{tab:table_non_respondents} we also report balance table on response rates (both among all surveyed individuals and between respondents and non-respondents individuals), where results show substantial balance in relevant baseline characteristics. 

\vspace{-4mm} 
\paragraph{Treatment take-up and accuracy} The treatment group received approximately three times more frequent calls than the control group by design -- where the control group's calls were about other activities of the NGO. Appendix Table \ref{tab:preliminary_treatment2} shows that the larger number of calls does \textit{not} negatively affect response rates. Treated individuals present higher (and statistically significant at the $1\%$ level) response rates per call, engaging more with calls. 

Table \ref{tab:forecast} shows that forecast and real precipitation and temperature are strongly positively correlated, motivating our main focus on farmers' beliefs about PxD forecasts: PxD predicted and real weather follow very similar patterns, but beliefs about PxD forecasts are less noisy.  

\begin{table}[!htbp] \centering 
 
\scalebox{0.6}{\begin{tabular}{lccc}
\hline
Dependent variable: & Real Precipitation & Real Temperature Max & Correct Rain Forecast \\
\hline
Forecast Precipitation & 0.675*** &  &  \\
 & (0.020) &  &  \\
Forecast Temperature Max &  & 0.914*** &  \\
 &  & (0.029) &  \\
Constant & 1.585*** & 0.274 & 0.786*** \\
 & (0.069) & (1.112) & (0.005) \\
\hline
\multicolumn{4}{l}{Sample size: 22230} \\
\multicolumn{4}{l}{\textit{Note:} $^{*}$p$<$0.1; $^{**}$p$<$0.05; $^{***}$p$<$0.01} \\
\end{tabular}
}
  \caption{Forecast vs real weather in 2022. Sample size equal to 22230. The first column uses precipitation as a continuous variable and the last column regresses the indicator of whether the forecast of whether it will rain correctly predicts whether it rains.   In parenthesis standard errors clustered at the tehsil level.} 
  \label{tab:forecast}
\end{table}

\vspace{-4mm} 
\paragraph{Parametric regression estimates} Our design allows for standard regression methods. We illustrate this in Tables \ref{tab:beliefs_forecast2}. Table \ref{tab:beliefs_forecast2} reports regression estimates of farmers' incorrect beliefs about temperature and rain with respect to forecast rain from PxD, for which we find mostly significant spillover effects. For parametric regression estimates we can use information from all clusters, including the lower saturation group, after appropriately controlling for the treatment probability, since also in this group treatment are randomized. 
 
\paragraph{Dynamics and additional outcomes} Appendix Table \ref{tab:dynamics} illustrates lack of dynamics on our primary outcome (temperature forecasts). In Table \ref{tab:dynamics}, we also illustrates effects on other outcomes. We collect information about predicted rain, asking ``Do you think it will rain in your
area tomorrow?" We use a binary indicator indicating whether the farmers incorrectly predict no rain and, instead, it rains or vice versa (or replies ``I do not know"). As shown in Table \ref{tab:dynamics}, we do not consider rain as the main target outcome because, different from temperature, this may exhibit treatment effect heterogeneity over time, since the experiment spans seasons of different rain intensity (dry and monsoon seasons). Finally, we use survey information about farming activities to show effects on these in Appendix \ref{app:more_experiment}.

\subsection{Balance tables}

\begin{table*}[!ht]\centering

\ra{1.3}
\scalebox{0.6}{\begin{tabular}{lcccccc}
\hline
Saturation & \multicolumn{2}{c}{Medium} & \multicolumn{2}{c}{High} & \multicolumn{2}{c}{Medium/High} \\
First wave $\beta$ & 0.35 & 0.45 & 0.55 & 0.65 & $0.4 \pm 0.05$ & $0.6 \pm 0.05$ \\
\hline
Average \# of Farmers per tehsil & 11817 & 11137 & 10031 & 12795 & 11477 & 11519 \\
(p-value) & \multicolumn{2}{c}{(0.875)} & \multicolumn{2}{c}{(0.718)} & \multicolumn{2}{c}{(0.982)} \\
Education & 0.539 & 0.515 & 0.564 & 0.595 & 0.527 & 0.583 \\
(p-value) & \multicolumn{2}{c}{(0.875)} & \multicolumn{2}{c}{(0.875)} & \multicolumn{2}{c}{(0.211)} \\
Female & 0.016 & 0.019 & 0.021 & 0.030 & 0.018 & 0.026 \\
(p-value) & \multicolumn{2}{c}{(0.500)} & \multicolumn{2}{c}{(0.250)} & \multicolumn{2}{c}{(0.223)} \\
Acres & 4.158 & 4.159 & 4.468 & 4.067 & 4.158 & 4.228 \\
(p-value) & \multicolumn{2}{c}{(0.875)} & \multicolumn{2}{c}{(0.562)} & \multicolumn{2}{c}{(0.901)} \\
Male Dependants & 2.491 & 2.795 & 2.606 & 2.669 & 2.639 & 2.644 \\
(p-value) & \multicolumn{2}{c}{(0.593)} & \multicolumn{2}{c}{(0.937)} & \multicolumn{2}{c}{(0.988)} \\
Female Dependants & 2.485 & 2.750 & 2.645 & 2.637 & 2.613 & 2.641 \\
(p-value) & \multicolumn{2}{c}{(0.718)} & \multicolumn{2}{c}{(1)} & \multicolumn{2}{c}{(0.942)} \\
Age & 50.9 & 51.5 & 50.9 & 50.9 & 51.2 & 50.9 \\
(p-value) & \multicolumn{2}{c}{(0.937)} & \multicolumn{2}{c}{(1)} & \multicolumn{2}{c}{(0.970)} \\
Wheat & 0.644 & 0.510 & 0.470 & 0.546 & 0.579 & 0.515 \\
(p-value) & \multicolumn{2}{c}{(0.562)} & \multicolumn{2}{c}{(0.343)} & \multicolumn{2}{c}{(0.617)} \\
Whatsapp & 0.257 & 0.295 & 0.263 & 0.273 & 0.276 & 0.269 \\
(p-value) & \multicolumn{2}{c}{(0.812)} & \multicolumn{2}{c}{(0.937)} & \multicolumn{2}{c}{(0.702)} \\
\hline
\end{tabular}
}
\caption{Clusters' balance table. Each entry reports the average value (average between the clusters in a given group) of a given baseline characteristic for clusters exposed to different treatment probabilities. Each column collects results for two groups of clusters. For example, the first row/first column reports the average number of farmers in clusters with $\beta = 0.35$ (note that, similarly, also the last two columns also report the \textit{average} value across the clusters in a given group and not their sum).   P-values test the two-sided null hypothesis that the point estimates for the two groups are the same and are computed via randomization inference. Covariates are the average number of individuals in the experiment in each cluster, whether individuals have only attended primary or no education, the percentage of female farmers, the size of landholding in acres, the number of male and female dependants, the farmer's age, whether farmers are also wheat farmers, and whether they have ``Whatsapp". (When estimating marginal effects, it is easy to show that our framework only requires homogeneity restrictions between groups of clusters used to estimate the marginal effects (e.g., the group of clusters in different treatment exposures), since we can redefine a cluster as a group of cluster exposed to the same treatment exposure).}    \label{tab:summaries} 
\end{table*}

\paragraph{Balance checks on all individuals enrolled in the program} In Table \ref{tab:summaries} we report the balance checks among \textit{all} individuals enrolled in the program, where we see no significant imbalance (note that in the first row we report the \textit{average} number of individuals per tehsils).

\vspace{-4mm} 
\paragraph{Effects on response rates} In Table \ref{tab:preliminary_treatment2} we report further evidence of the effectiveness of the intervention on response rates. 

\vspace{-4mm} 
\paragraph{Balance checks on surveyed individuals} In Table \ref{tab:summaries_response_rates}, we report balance checks among surveyed individuals. We see no imbalance in relevant covariates, with all tests being nonsignificant, with one single exception. This exception  is for the difference in the number of females for the Negative Perturbation group, where in one group, the number of females is $0.5\%$, and in the other group is $2\%$. Even in this scenario, the small proportion of females in either groups of clusters (at most $3\%$) makes this difference not economically relevant.\

\vspace{-4mm} 
\paragraph{Response rates}  In Table \ref{tab:summaries_response_rates_v2} we report balance table for individuals who answered the question about ``what do you expect the maximum (minimum) temperature will be tomorrow?". We see the same patterns as in the other tables. In Table \ref{tab:table_non_respondents} we report the difference in means in baseline characteristics between the respondents to the question about predicted temperature and the non-respondents. We do observe that respondents are similar in all characteristics  to non-respondents (and for which we cannot reject the null hypothesis that means in baseline characteristics are different) except that they tend to be those that are more likely to have installed the App ``Whatsapp" on their mobile phone ($40\%$ have Whatsapp between respondents and $30\%$ have Whatsapp between non-respondents). 

\begin{table*}[!ht]\centering

\ra{1.3}
\scalebox{0.6}{\begin{tabular}{lcccccc}
\hline
Saturation & \multicolumn{2}{c}{Medium} & \multicolumn{2}{c}{High} & \multicolumn{2}{c}{High/Medium} \\
First wave $\beta$ & 0.35 & 0.45 & 0.55 & 0.65 & $0.6 \pm 0.05$ & $0.4 \pm 0.05$ \\
\hline
Average \# of Sampled Farmers per tehsil & 155 & 178 & 154 & 168 & 161 & 166 \\
(p-value) & \multicolumn{2}{c}{(0.625)} & \multicolumn{2}{c}{(0.750)} & \multicolumn{2}{c}{(0.841)} \\
Education & 0.487 & 0.455 & 0.501 & 0.580 & 0.545 & 0.470 \\
(p-value) & \multicolumn{2}{c}{(0.625)} & \multicolumn{2}{c}{(0.750)} & \multicolumn{2}{c}{(0.192)} \\
Female & 0.005 & 0.021 & 0.013 & 0.030 & 0.023 & 0.013 \\
(p-value) & \multicolumn{2}{c}{(0.031)} & \multicolumn{2}{c}{(0.187)} & \multicolumn{2}{c}{(0.155)} \\
Acres & 4.48 & 4.29 & 4.573 & 3.849 & 4.169 & 4.384 \\
(p-value) & \multicolumn{2}{c}{(0.781)} & \multicolumn{2}{c}{(0.375)} & \multicolumn{2}{c}{(0.878)} \\
Male Dependants & 2.68 & 2.69 & 2.732 & 2.837 & 2.790 & 2.689 \\
(p-value) & \multicolumn{2}{c}{(0.781)} & \multicolumn{2}{c}{(0.906)} & \multicolumn{2}{c}{(0.657)} \\
Female Dependants & 2.72 & 2.69 & 2.696 & 2.556 & 2.618 & 2.712 \\
(p-value) & \multicolumn{2}{c}{(0.875)} & \multicolumn{2}{c}{(0.812)} & \multicolumn{2}{c}{(0.773)} \\
Age & 50.79 & 51.40 & 51.019 & 51.130 & 51.081 & 51.122 \\
(p-value) & \multicolumn{2}{c}{(0.937)} & \multicolumn{2}{c}{(0.937)} & \multicolumn{2}{c}{(0.998)} \\
Wheat & 0.633 & 0.506 & 0.451 & 0.572 & 0.518 & 0.565 \\
(p-value) & \multicolumn{2}{c}{(0.375)} & \multicolumn{2}{c}{(0.468)} & \multicolumn{2}{c}{(0.681)} \\
Whatsapp & 0.322 & 0.343 & 0.320 & 0.328 & 0.324 & 0.333 \\
(p-value) & \multicolumn{2}{c}{(0.968)} & \multicolumn{2}{c}{(1)} & \multicolumn{2}{c}{(0.737)} \\
\hline
\end{tabular}
}
\caption{Clusters' balance table on response rate. Each entry reports the average value (average between the clusters in a given group) of a given baseline characteristic for clusters exposed to different treatment probabilities, averaging over individuals who replied to the survey. Each column collects results for two groups of clusters. For example, the first row/first column reports the average number of farmers in clusters with $\beta = 0.35$ (note that, similarly, also the last two columns also report the \textit{average} value across the clusters in a given group and not their sum).   P-values test the two-sided null hypothesis that the point estimates for the two groups are  the same and are computed via randomization inference. }    \label{tab:summaries_response_rates} 
\end{table*}

\begin{table*}[!ht]\centering

\ra{1.3}
\scalebox{0.6}{\begin{tabular}{lcccccc}
\hline
Saturation & \multicolumn{2}{c}{Medium} & \multicolumn{2}{c}{High} & \multicolumn{2}{c}{High/Medium} \\
First wave $\beta$ & 0.35 & 0.45 & 0.55 & 0.65 & $0.6 \pm 0.05$ & $0.4 \pm 0.05$ \\
\hline
Average \# of Sampled Farmers per tehsil & 32 & 36.6 & 30.1 & 35.2 & 32.92 & 34.50 \\
(p-value) & \multicolumn{2}{c}{(0.468)} & \multicolumn{2}{c}{(0.718)} & \multicolumn{2}{c}{(0.759)} \\
Education & 0.378 & 0.409 & 0.397 & 0.556 & 0.489 & 0.395 \\
(p-value) & \multicolumn{2}{c}{(0.656)} & \multicolumn{2}{c}{(0.750)} & \multicolumn{2}{c}{(0.226)} \\
Female & 0.000 & 0.036 & 0.011 & 0.024 & 0.018 & 0.019 \\
(p-value) & \multicolumn{2}{c}{(0.125)} & \multicolumn{2}{c}{(0.500)} & \multicolumn{2}{c}{(0.937)} \\
Acres & 4.907 & 5.211 & 5.395 & 4.223 & 4.719 & 5.068 \\
(p-value) & \multicolumn{2}{c}{(0.781)} & \multicolumn{2}{c}{(0.625)} & \multicolumn{2}{c}{(0.838)} \\
Male Dependants & 2.748 & 2.941 & 2.296 & 2.731 & 2.547 & 2.851 \\
(p-value) & \multicolumn{2}{c}{(0.562)} & \multicolumn{2}{c}{(1)} & \multicolumn{2}{c}{(0.417)} \\
Female Dependants & 2.726 & 2.703 & 2.403 & 2.689 & 2.568 & 2.714 \\
(p-value) & \multicolumn{2}{c}{(0.906)} & \multicolumn{2}{c}{(1)} & \multicolumn{2}{c}{(0.712)} \\
Age & 50.914 & 54.614 & 48.481 & 53.127 & 51.162 & 52.880 \\
(p-value) & \multicolumn{2}{c}{(0.500)} & \multicolumn{2}{c}{(1)} & \multicolumn{2}{c}{(0.844)} \\
Wheat & 0.603 & 0.509 & 0.458 & 0.542 & 0.507 & 0.553 \\
(p-value) & \multicolumn{2}{c}{(0.312)} & \multicolumn{2}{c}{(0.625)} & \multicolumn{2}{c}{(0.660)} \\
Whatsapp & 0.407 & 0.440 & 0.403 & 0.404 & 0.404 & 0.425 \\
(p-value) & \multicolumn{2}{c}{(0.906)} & \multicolumn{2}{c}{(0.812)} & \multicolumn{2}{c}{(0.759)} \\
\hline
\end{tabular}
}
\caption{Clusters' balance table on response rate for \textit{respondents} to the question about temperature (i.e., for individuals for which we do not observe missing values). Each entry reports the average value (average between the clusters in a given group) of a given baseline characteristic for clusters exposed to different treatment probabilities, averaging over individuals who replied to the survey. Each column collects results for two groups of clusters. For example, the first row/first column reports the average number of farmers in clusters with $\beta = 0.35$ (note that, similarly, also the last two columns also report the \textit{average} value across the clusters in a given group and not their sum).   P-values test the two-sided null hypothesis that the point estimates for the two groups are the same and are computed via randomization inference. }    \label{tab:summaries_response_rates_v2} 
\end{table*}

\begin{table}[!ht] \centering 

\scalebox{0.65}{\begin{tabular}{lcccc}
\hline
 &  & nCalls/Person & Total Response/Person & Average Response \\
\hline
Treated & 158 697 & 121 & 29 & 0.235 \\
Controls & 240 354 & 49 & 11 & 0.221 \\
p-value Response & [0.000] & & & \\
\hline
\end{tabular}
}
  \caption{Summary statistics of treated and control units for May - July (Wave 1), pooled across all tehsils in the experiment. $p$-value is obtained via randomization inference at the cluster level.} 
  \label{tab:preliminary_treatment2} 
\end{table}

\begin{table}[!htbp] \centering 

\scalebox{0.6}{\begin{tabular}{lccccc}
\hline
 & Non Respondents & Respondents & P-value High Saturation & P-value Medium Saturation & P-value Low Saturation \\
\hline
Education & 0.537 & 0.446 & 0.378 & 0.081 & 0.287 \\
Female & 0.019 & 0.017 & 0.452 & 0.539 & 0.227 \\
Acres & 4.119 & 4.881 & 0.321 & 0.071 & 0.062 \\
Male Dependants & 2.822 & 2.782 & 0.462 & 0.268 & 0.836 \\
Female Dependants & 2.787 & 2.684 & 0.951 & 0.973 & 0.369 \\
Age & 50.768 & 52.043 & 0.893 & 0.398 & 0.684 \\
Wheat & 0.533 & 0.524 & 0.675 & 0.668 & 0.797 \\
Whatsapp & 0.301 & 0.409 & 0.009 & 0.000 & 0.023 \\
Average \# Farmers per Cluster & 130 & 32 &  &  &  \\
\hline
\end{tabular}
}
 \caption{Balance table between respondentents to the question about maximum (minimum) temperature and non respondents using baseline characteristics. The first two columns report the mean of each covariate for the non respondents and respondents and the last three columns the p-values obtained via permutation tests for each group of tehsil. Permutation tests are at the cluster level.} 
  \label{tab:table_non_respondents} 
\end{table} 

\subsection{Regression estimates and dynamics} \label{sec:activities}

\paragraph{Regression estimates for forecast and real weather} In Table \ref{tab:beliefs_forecast2} we report regression estimates for change in beliefs relative to forecast weather. We present a descriptive regression of the outcome on the treatment status and the share of treated individuals in the same clusters \citep[e.g., as in][]{cai2015social}. This is a standard regression that our design, as well as other designs, can allow for. We also consider other specifications, controlling for the interaction between the individual treatment status and the share of treated individuals in the clusters. We control for three variables: \textit{Treatment} measures the effect on treated farmers,  \textit{Cluster Treat Prob} measures the spillover effect, and  \textit{Cluster Treat Prob} $\times$ \textit{Treatment} measures the interaction between the share of treated farmers and individual treatment. Throughout each specification, we include time(wave)-fixed effects since treatment probability is increased over the second wave. 
We include information about low, medium, and high saturation level, after appropriately controlling for tehsil-specific treatment probabilities (and fixed effects).

 Table \ref{tab:beliefs_forecast2} reports regression estimates of farmers' incorrect beliefs about temperature and rain relative to PxD forecasts. Note that response rates for rain are higher than for temperature. Standard errors in parentheses are clustered at the tehsil level. Results suggest that both treatment and spillover effects improve forecasts. In the specifications without the interaction between direct and spillover effects, the treatment estimates are negative and statistically significant throughout. The spillover estimates are also negative throughout and statistically significant in three of the four specifications; the temperature spillover estimate is imprecise when tehsil fixed effects are included. Spillovers exhibit similar or larger coefficients than direct effects, suggesting a ``multiplier effect" when \textit{all} farmers are informed. The multiplier effect can be due to farmers being more attentive to what other farmers report or receiving (the same) information from multiple farmers, and can be found also in other information campaigns \citep[e.g.][]{cai2015social}. When the interaction term is included, point estimates are noisier, as often occurs in experiments \citep[see][]{muralidharan2023factorial}; spillover estimates generally preserve their negative signs, although they are not always statistically significant. 
 
\begin{table}[!ht] \centering 
\scalebox{0.65}{\begin{tabular}{lcccccccc}
\hline
Dependent variable: & \multicolumn{4}{c}{Incorrect beliefs about PxD forecast Temperature} & \multicolumn{4}{c}{PxD forecast Rain} \\
 & (1) & (2) & (3) & (4) & (5) & (6) & (7) & (8) \\
\hline
Treatment & -0.837*** & -0.845*** & -0.833** & -0.867** & -0.041*** & -0.035** & 0.007 & 0.003 \\
 & (0.164) & (0.171) & (0.416) & (0.415) & (0.013) & (0.014) & (0.024) & (0.027) \\
Cluster Treat Prob & -0.722* & -2.915 & -0.718 & -2.939 & -0.086** & -1.165*** & -0.040 & -1.126*** \\
 & (0.389) & (1.798) & (0.595) & (1.886) & (0.035) & (0.161) & (0.043) & (0.168) \\
Cluster Treat Prob $\times$ Treatment &  &  & -0.008 & 0.042 &  &  & -0.092* & -0.072 \\
 &  &  & (0.729) & (0.750) &  &  & (0.050) & (0.055) \\
\hline
Time (Wave) Fixed Effects & Yes & Yes & Yes & Yes & Yes & Yes & Yes & Yes \\
Tehsil Fixed Effects & No & Yes & No & Yes & No & Yes & No & Yes \\
Observations & 1,143 & 1,143 & 1,143 & 1,143 & 5,104 & 5,104 & 5,104 & 5,104 \\
\hline
\multicolumn{9}{l}{\textit{Note:} $^{*}$p$<$0.1; $^{**}$p$<$0.05; $^{***}$p$<$0.01} \\
\end{tabular}
}
 \caption{The left-hand-side panel reports regressions of the absolute difference between the farmer's temperature prediction and the PxD forecast. The right-hand-side panel reports regressions of whether the farmer incorrectly predicts whether it will rain. Standard errors in parentheses are clustered at the tehsil level. The regressions use observations from Low, Medium, and Positive Perturbation tehsils, where treatment is randomly assigned using time- and cluster-specific treatment probabilities.} 
  \label{tab:beliefs_forecast2} 
\end{table}

\paragraph{Tests on dynamic effects on beliefs} In Table \ref{tab:dynamics} (left-panel), we present results on dynamic effects by controlling for whether individuals are surveyed during the second wave and the interaction between the individual treatment and the second wave. We focus on dynamics on direct effects for simplicity, whereas results for spillovers are robust (preserve sign and magnitude but are noisier) and omitted for brevity. The first two columns report the effects and dynamics of beliefs about temperature, our main outcomes. Treatment effects preserve the sign and magnitude as in our main specification in the main text, after controlling for the interaction on dynamics. Importantly, the coefficient interacting the treatment with the second wave experiment is very close to zero and non-significant. This is suggestive that effects in improving predictions on weather do not exhibit dynamic treatment effects. This result formalizes the intuition that correctly predicting short-term temperature the next day may not affect correct short-term predictions in the upcoming weeks or months. 

\paragraph{Effect on farming activities and power of tests on dynamics} An interesting question is whether our specification for beliefs is sufficiently powered to detect dynamics or treatment effect heterogeneity over time. To do so, 
in Table \ref{tab:dynamics} (second panel), we explore how treatment affects predicting rain and short-term farming activities. Different from temperature, for rain we do find some effect heterogeneity over time as farmers in different periods may differently being impacted by the treatment. This is intuitive, since different periods correspond to different rain seasons. This further motivates using temperature as a welfare proxy for sequential experimentation, as rain may exhibit some time effect heterogeneity. 

We also measure effects on activities. We use survey information on the timing of farming tasks, such as ``Can you recall the exact day when
you applied pesticides?" and use the same questions for irrigation, use of fertilizers, and planting decisions. We then match the reported date of the farming task with the realized rainfall for the same day and create an indicator variable if it did not rain on the day of the farming task. The treatment-by-second-wave estimates are positive and statistically significant for irrigation and fertilizer, but are not statistically significant for pesticides or planting. This suggests that individuals may adjust some actions dynamically and shows that our specifications are sufficiently powered to detect dynamics. These results provide suggestive evidence of a lack of dynamics in temperature forecast accuracy, our main outcome, but not necessarily in other outcomes such as farming actions.

\subsection{Details about perturbations $\eta_n$ and variants of the experiment} \label{app:choice_eta} 

In this section, we provide further details on the choice of $\eta_n$ in the main experiment.

\paragraph{Details about first wave and choice of $\eta_n$}   The first experimentation wave allows us to learn the marginal effect around $\beta = 50\%$. 
Over the first experimental wave, we randomly draw a group of twelve tehsils (``Negative Perturbation/Medium Saturation") to have an average treatment probability across tehsils in this group of $\beta = 0.4$, hence inducing a negative perturbation $\eta_n = 10\%$. The choice of the perturbation should depend on power considerations, as, in principle, we may also be interested in more refined marginal effects, at the expense of lower power. To study here trade-offs in the choice of the perturbation parameter $\eta_n$, we select the Negative Perturbation group to have $\beta = 0.4$ \textit{on average}, with half of the (randomly selected) clusters in the Negative Perturbation group having exactly $\beta = 0.35$ and half of the clusters with $\beta = 0.45$. 
We repeat the same with a ``Positive Perturbation/High Saturation" group with approximately $\beta = 0.6$ on average (and, similarly as before, with six tehsils in this group having $\beta = 0.55$ and seven $\beta = 0.65$). This gives us two (nested) perturbation designs. First, we obtain a better-powered perturbation design (which we refer to as our \textit{main design}) with a total of 25 clusters and perturbations around $\beta = 0.5$, with perturbations equal to $\eta_n = 10\%$ on average. The second design induces \textit{within} group perturbation of smaller order $5\%$, which allows us to also learn marginal effects at two more values $\beta \in \{40, 60\}\%$, with half of the clusters. A key intuition is that, by pooling clusters around smaller perturbations, all of our theoretical results directly apply to the main design, up to a small bias.\footnote{That is, within our framework it is possible to choose different levels of perturbations in the same design by assigning a group of clusters to a treatment probability $\beta + \eta_n$, a different group to $\beta -\eta_n$; within each group, then repeat this same procedure, inducing more minor perturbation of order $\beta + \eta_n \pm \eta_n', \eta_n' = o(\eta_n)$, and similarly for the second group. These two nested designs do not affect our theoretical results for inference on $M(\beta)$, as long as $\eta_n' = o(\eta_n)$. Choosing different levels of perturbations has the advantage that we can learn a larger set of marginal effects (both at $\beta$, and at  $\beta \pm \eta_n$), while avoiding under-powered studies for the main effect $M(\beta)$. } We report results from the main (better powered) design with $\beta = 50\%, \eta_n = 10\%$ on average; we show that for smaller choice of $\eta_n$ estimates can be under-powered, see Appendix Table \ref{tab:marginal_effects_ma}. We recommend the choice of two nested designs to avoid under-powered studies.

\paragraph{Details about second wave} The second wave experiment allows us to learn marginal effects at $\beta = 0.7$. 
Over the second wave (August - October), the ``Negative Perturbation" group was exposed to a larger treatment probability $\beta = 0.6$ and the ``Positive Perturbation" group was exposed to a treatment probability $\beta = 0.8$. Therefore, over the second experimentation wave, we have two groups with treatment probabilities $\beta = 70\% \pm \eta_n, \eta_n = 10\%$. Over the second wave, we also perturbed by 0.05 the probability of treatment for different types of
farmers, those below and above the median response rate in the first round, keeping the overall treatment
probability constant. This latter perturbation enables estimating heterogeneous treatment effects, omitted
from the main analysis for brevity and discussed in this Appendix below (in the following subsection).

\subsection{Additional results} 

\paragraph{More refined marginal effects from first wave experiment} In Table \ref{tab:marginal_effects_ma} we report estimated effects over the first wave for a secondary design where we use half of the clusters to learn marginal effects at $\beta \in \{40, 60\}\%$. Results report noisy estimates, due to lower sample size and smaller perturbation. This is suggestive that using a too small choice of the perturbation may lead to under-powered studies. We therefore recommend in practice to consider at least two-nested design as we did here to increase power.  


\vspace{-4mm} 
\paragraph{Marginal effects on response rates in the first wave (using baseline outcomes to control for tehsils fixed effects)} For illustrative purposes, in Table \ref{tab:main_experiment3a} we collect marginal effects for response rates for our secondary design, for which we can control for baseline outcomes. We use the estimators proposed in Section \ref{sec:designs}, with baseline outcomes as the outcomes in the control group over the first week of the intervention (assuming no spillovers during the first week of the experiment). We find significant direct treatment effects whereas marginal spillover effects are noisier/closer to zero as we may expect (since spillovers may less likely occur on higher response to phone calls given that the control group does not receive phone calls about weather forecasts).

\begin{table}[!ht] \centering 
  \caption{Study on dynamics. The first four columns report regressions of the absolute difference between the farmer's predicted maximum temperature tomorrow and the forecast maximum temperature (first column) or true maximum temperature (second column), and the inaccuracy in predicting forecast and realized rain (third and fourth columns). The last four columns report effects on farming actions (irrigation, use of fertilizers, pesticides, and planting) as defined in the main text. The regressions control for individual treatment, an indicator for the second wave, and their interaction. Results that also control for spillover effects are robust and omitted. Standard errors in parentheses are clustered at the tehsil level.} 
  \label{tab:dynamics} 
\scalebox{0.5}{\begin{tabular}{lcccccccc}
\hline
Dependent variable: & Incorrect Beliefs Forecast Temp & Beliefs Real Temp & Beliefs Forecast Rain & Beliefs Real Rain & Irrigation & Fertilizer & Pesticides & Planting \\
 & (1) & (2) & (3) & (4) & (5) & (6) & (7) & (8) \\
\hline
Treatment & -0.585* & -0.757* & 0.004 & 0.009 & -0.020 & -0.039* & -0.007 & -0.035 \\
 & (0.354) & (0.392) & (0.023) & (0.0213) & (0.018) & (0.022) & (0.020) & (0.038) \\
Second Wave & -0.846*** & -1.048*** & -0.150*** & -0.101*** & 0.128*** & 0.040** & 0.201*** & 0.034 \\
 & (0.317) & (0.341) & (0.020) & (0.025) & (0.023) & (0.020) & (0.024) & (0.033) \\
Treatment $\times$ Second Wave & -0.135 & -0.180 & -0.047 & -0.057** & 0.088*** & 0.074** & 0.040 & -0.001 \\
 & (0.395) & (0.393) & (0.030) & (0.0272) & (0.025) & (0.031) & (0.025) & (0.051) \\
\hline
Tehsil Fixed Effects & Yes & Yes & Yes & Yes & Yes & Yes & Yes & Yes \\
\hline
\multicolumn{9}{l}{\textit{Note:} $^{*}$p$<$0.1; $^{**}$p$<$0.05; $^{***}$p$<$0.01} \\
\end{tabular}
}
\end{table}

\begin{table}[!ht] \centering 
  \caption{The panel at the top reports the effects in percentage points for unconditional case with $\beta = 0.4$, $\beta = 0.5$, and $\beta = 0.6$. The panel at the bottom reports the difference in engagement by increasing treatment probababilities from $35\%$ to $65\%$. $p$-value for one-sided tests computed via randomization inference at the cluster level are in parenthesis.} 
  \label{tab:main_experiment3a} 
\scalebox{0.65}{\begin{tabular}{lccc}
\hline
Y: Phone Response Rate & $\beta = 0.4$ & $\beta = 0.5$ & $\beta = 0.6$ \\
 & May - July & May - July & May - July \\
\hline
Marginal Effect & 4.862 & 2.548 & -2.277 \\
p-value & [0.135] & [0.241] & [0.295] \\
Direct Effect & 1.789 & 1.136 & 1.248 \\
p-value & [0.009] & [0.012] & [0.000] \\
Marginal Spillovers on the Treated & 7.609 & -0.115 & -5.599 \\
p-value & [0.139] & [0.491] & [0.183] \\
Marginal Spillovers on the Controls & 0.938 & 2.265 & 0.022 \\
p-value & [0.406] & [0.264] & [0.500] \\
\hline
$\beta = 0.35 \uparrow$ & Improvement & p-value & \\
$\beta = 0.45$ & 0.48 & [0.133] & \\
$\beta = 0.55$ & 0.74 & [0.058] & \\
$\beta = 0.65$ & 0.51 & [0.261] & \\
\hline
\end{tabular}
}
\end{table}

\begin{table}[!ht] \centering 
  \caption{Marginal effect for secondary design (perturbation is $\eta_n = 5\%$) over the first experimentation wave.} 
  \label{tab:marginal_effects_ma} 
\scalebox{0.7}{\begin{tabular}{@{\extracolsep{5pt}}lcccc}
\\[-1.8ex]\hline
\hline \\[-1.8ex]
Dependent variable: & \multicolumn{2}{c}{Incorrect beliefs about PxD forecast Temperature} & \multicolumn{2}{c}{PxD forecast Rain} \\
\\[-1.8ex] & $\beta = 40\%$ (Wave 1) & $\beta = 60\%$ (Wave 1) & $\beta = 40\%$ (Wave 1) & $\beta = 60\%$ (Wave 1) \\
\hline \\[-1.8ex]
Marginal Effect & 6.04 & -0.18 & -0.06 & 0.11 \\
p-value & [0.20] & [0.49] & [0.42] & [0.47] \\
Direct Effect & -1.43$^{*}$ & -0.07 & 0.03 & -0.00 \\
p-value & [0.07] & [0.46] & [0.35] & [0.49] \\
Spillover on Treated & -1.24 & -2.74 & 0.23 & 0.43 \\
p-value & [0.47] & [0.34] & [0.42] & [0.31] \\
Spillover on Controls & 8.37 & 3.63 & -0.24 & -0.26 \\
p-value & [0.19] & [0.38] & [0.33] & [0.40] \\
Observations & 129 & 130 & 370 & 375 \\
\hline \\[-1.8ex]
\end{tabular}
}
\end{table}

\begin{table}[!ht] \centering 
  \caption{Heterogeneous treatment effect on engagement (response rates) as the main outcome computed in Wave 2. Wave 2 indicates the average effect in August 2022.  In parenthesis $p$-value computed via randomization inference at the cluster level.} 
  \label{tab:main_experiment3} 
\scalebox{0.7}{\begin{tabular}{lccc}
\hline
Y: Phone Response Rate & Low & Medium & High \\
 & Wave 2 & Wave 2 & Wave 2 \\
\hline
Marginal Effect & -3.857 & 1.019 & 1.944 \\
p-value & (0.341) & (0.442) & (0.397) \\
Direct Effect & 0.383 & 3.700 & 4.780 \\
p-value & (0.422) & (0.014) & (0.004) \\
Spillovers on High Types & -0.544 & -21.457 & 14.629 \\
p-value & (0.492) & (0.116) & (0.242) \\
Spillovers on Low Types & 6.902 & -1.636 & 5.886 \\
p-value & (0.250) & (0.369) & (0.177) \\
\hline
\end{tabular}
}
\end{table}

\section{Additional Extensions}  \label{app:e}


\subsection{Tests with a $p$-dimensional vector of marginal effects} \label{sec:pilot_general}

In the following lines we extend Algorithm \ref{alg:my_pilot} to testing the following null
$
H_0: M^{(j)}(\beta) = 0, \text{ for some } p \ge p_1 \ge 1, 
$ 
where we consider a generic number of dimensions tested $p_1$.

\begin{algorithm} [h]   \caption{One wave experiment for inference}\label{alg:my_pilot_general}
\footnotesize 
    \begin{algorithmic}[1]
    \Require Value $\beta \in \mathbb{R}^p$, $K$ clusters, $2$ periods of experimentation, number of tests $t$. 
    \State Match clusters into pairs $K/2$ pairs with consecutive indexes $\{k, k+1\}$;  
    \State $t = 0$ (\textit{baseline}): 
    \begin{algsubstates}
    \State Treatments are assigned at some baseline $\beta_0$ 
    $
    D_{i,0}^{(h)} \sim \pi(X_i^{(h)}, \beta_0), h \in \{1, \cdots, K\}$.
        \State Collect baseline values: for $n$ units in each cluster observe $Y_{i,0}^{(h)}, h \in \{1, \cdots, K\}$.  
        \end{algsubstates} 
    \State $t = 1$ (\textit{experimentation-wave})
    \State Assign each pair of clusters $\{k, k+1\}$ to a coordinate $j \in \{1, \cdots, p\}$ (with the same number of pairs to each coordinate)
    \State For each pair $\{k, k+1\}$, $k$ is odd, assigned to coordinate $j$
    \begin{algsubstates}
        \State  Randomize 
    $$
    \small 
    \begin{aligned} 
    D_{i,1}^{(h)} \sim \begin{cases} 
    & \pi(X_i^{(h)}, \beta + \eta_n \underline{e}_j ) \text{ if } h = k  \\ 
    & \pi(X_i^{(h)}, \beta - \eta_n \underline{e}_j ) \text{ if } h = k+1 
    \end{cases} , \quad n^{-1/2} < \eta_n \le n^{-1/4}
    \end{aligned}
    $$
        \State For $n$ units in each cluster $h \in \{k,k+1\}$ observe $Y_{i,1}^{(h)}$.  
     
 \State Estimate the marginal effect for coordinate $j$ as 
$
\widehat{M}_{k} = \frac{1}{2 \eta_n}\Big[\bar{Y}_1^{(k)} - \bar{Y}_0^{(k)}\Big] - 
 \frac{1}{2 \eta_n}\Big[\bar{Y}_1^{(k + 1)} - \bar{Y}_0^{(k + 1)}\Big] 
$
 \Return 
$
\widetilde{M}_n = \Big[\widehat{M}_1, \widehat{M}_3, \cdots, \widehat{M}_{K-1}\Big]$. 
  \end{algsubstates}
 
         \end{algorithmic}
\end{algorithm}

We define $\mathcal{K}_j$ the set of pairs in Algorithm \ref{alg:my_pilot_general} used to estimate the $j^{th}$ entry of $M(\beta)$. 
  Define 
  $
\bar{M}_{n}^{(j)} =\frac{2p_1}{K} \sum_{k \in \mathcal{K}_j} \widehat{M}_{k}, 
  $ 
the average marginal effect for coordinate $j$ estimated from those clusters is used to estimate the effect of the $j^{th}$ coordinate. 
we construct 
\begin{equation} \label{eqn:test_statB}
\small 
\begin{aligned} 
Q_{j,n} = \frac{\sqrt{K/(2p_1)}  \bar{M}_{n}^{(j)}}{\sqrt{(K/(2p_1) - 1)^{-1} \sum_{k \in \mathcal{K}_j} (\widehat{M}_{k}^{(j)} - \bar{M}_{n}^{(j)})^2}}, \quad \mathcal{T}_n = \max_{j \in \{1, \cdots, p_1\}} |Q_{j,n}|, 
\end{aligned} 
\end{equation} 
where $\mathcal{T}_n$ denotes the test statistics. The proposed test-statistic is particularly suited when a large deviation occurs over one dimension of the vector. 


\begin{thm}[Nominal coverage] \label{thm:inference4}  Let Assumptions \ref{ass:ass_0}, \ref{ass:var} hold. Let $n^{1/4} \eta_n = o(1),  \gamma_N^2/N^{1/4} = o(1)$, $K < \infty$. Let $K \ge 4p_1$, $H_0$ be as defined in Equation \eqref{eqn:h0}. For any $\alpha \le 0.08$,  
$
 \lim_{n \rightarrow \infty} P\Big(\mathcal{T}_n \le q_\alpha\Big| H_0\Big) \ge 1 - \alpha, \text{ where } q_\alpha = \mathrm{cv}_{K/(2p_1) - 1}\Big(1 - (1 - \alpha)^{1/p_1}\Big), 
$
with $\mathrm{cv}_{K/(2p_1) - 1}(h)$ denotes the critical value of a two-sided t-test with level $h$ with test-statistic having $K/(2p_1) - 1$ degrees of freedom. 
\end{thm}

The proof is in Appendix \ref{sec:proof10}.

 \subsection{Out-of-sample regret with strict quasi-concavity} \label{sec:quasi_concavity}

 In the following lines, we provide guarantees on the regret bounds for the adaptive algorithm in Section \ref{sec:main_design} under quasi-concavity. We replace Assumption \ref{ass:strong_concavity} with the following condition. 
 
\begin{ass}[Local strong concavity and strict quasi-concavity] \label{ass:quasi_convexity} Assume that the following conditions hold: (A) For every $\beta, \beta' \in \mathcal{B}$, such that $W(\beta') - W(\beta) \ge 0$, then $
M(\beta)^\top (\beta' - \beta) \ge 0; 
$
(B) For every $\beta \in \mathcal{B}$, $||M(\beta)||_2 \ge \mu ||\beta - \beta^*||_2$, for a positive constant $\mu > 0$; (C) $W(\beta)$ is $\sigma$-strongly concave at $\beta^*$ (but not necessarily for $\beta \neq \beta^*$), with $\beta^* \in \tilde{\mathcal{B}} \subset \mathcal{B}$ being in the interior of $\mathcal{B}$. 
\end{ass}

Condition (A) imposes a quasi-concavity of the objective function.  Condition (B) assumes that the marginal effect only vanishes at the optimum, ruling out regions over which marginal effects remain constant at zero. A notion of strict quasi-concavity can be found in \cite{hazan2015beyond}. Condition (C) imposes strong concavity locally at $\beta^*$ but not necessarily globally. The choice of the learning rate consists of a gradient norm rescaling, as discussed in Remark \ref{rem:learning_rate}.

\begin{thm} \label{thm:rate}  Let  Assumptions \ref{ass:ass_0}, \ref{ass:bounded}, \ref{ass:quasi_convexity} hold. Consider a learning rate $\alpha_{k,w}$ as in Equation \eqref{eqn:gradient}, for arbitrary $v \in (0,1)$, and $\epsilon_n$ such that $
\epsilon_n \ge \sqrt{p} \Big[\bar{C} \sqrt{ \gamma_N\frac{\log(\gamma_N \check{T} K/\delta)}{\eta_n^2 n}} + \eta_n \Big], \quad \frac{\mu \sqrt{\kappa}}{ \check{T}^{1/2 - v/2}}  - 2\epsilon_n \ge 0  .  
$ 
Take a small $1/4 > \xi > 0$, and let $n^{1/4 - \xi} \ge \bar{C} \sqrt{\log(n) p \gamma_N T^2 e^{B p T} \log(KT)}$, $\eta_n = 1/n^{1/4 + \xi}$, for finite constants $\infty > B, \bar{C} > 0$. 
Then, for $T \ge \zeta^{1/v}$, for a finite constant $\zeta < \infty$, there exists a sufficiently small and finite  $\kappa > 0$ in Equation \eqref{eqn:gradient} such that with probability at least $1 - 1/n$, 
$
W(\beta^*) - W(\hat{\beta}^*) = \mathcal{O}(\check{T}^{-1 + v}). 
$
\end{thm}  

The proof of Theorem \ref{thm:rate} leverages properties of gradient descent with gradient norm rescaling in \cite{hazan2015beyond}, together with concentration bounds similar to those obtained to derive Theorem \ref{thm:rate2b}. The rate obtained differs from Theorem \ref{thm:rate2b} in two aspects: it is of order $T^{-1 + v}$ for arbitrary small $v$ instead of $T^{-1}$ and the sample size grows \textit{exponentially} instead of polynomially in $T$. The reason for the first is to control the inverse gradient when close to zero, and the reason for the second is due to the different learning rate which does not divide by $1/t$ (see the proof of Lemma \ref{lem:fg2} for details). See Appendix \ref{sec:proof11} for more details.

\subsection{Non separable fixed effects} \label{sec:non_separable}

 In the following lines, we show how we can leverage direct and marginal spillover effects to identify the marginal effects when fixed effects are non-separable in time and cluster identity. 

\begin{thm}[Marginal effects with non-separable fixed effects] \label{thm:non_separable} Let $X = 1$, and suppose that $m(d, 1, \beta)$ is bounded and twice differentiable with bounded derivatives for $d \in \{0,1\}$.  Suppose that fixed-effects are non-separable, with  
\begin{equation} \label{eqn:non_s}
\small 
\begin{aligned} 
Y_{i,t}^{(k)} = m(D_{i,t}^{(k)}, 1, \beta) + \alpha_{k,t} + \varepsilon_{i,t}^{(k)},  \quad \mathbb{E}[\varepsilon_{i,t}^{(k)} | D_{i,t}^{(k)}] = 0, \quad D_{i,t}^{(k)} \sim_{i.i.d.} \mathrm{Bern}(\beta), 
\end{aligned} 
\end{equation}
and $m(1, 1, \beta)$ being a constant function in $\beta$. Then 
$$
\small 
\begin{aligned} 
\mathbb{E}\Big[\hat{\Delta}_k(\beta) + \hat{S}(0, \beta)(1 - \beta) - (1 - \beta)\hat{S}(1, \beta)\Big] = M(\beta) + \mathcal{O}(\eta_n). 
\end{aligned} 
$$  
\end{thm} 

The proof is in Appendix \ref{sec:proof12}. 
Theorem \ref{thm:non_separable} shows that we can use the information on the spillover and direct treatment effects to identify the marginal effects in the presence of non-separable time and cluster fixed effects. The theorem leverages the assumption that spillovers only occur on the control individuals but not the treated. The assumption of lack of spillovers on the treated may hold in some \citep[e.g][]{duflo2023chat} but not all settings.



\subsection{Policy choice with dynamic treatments} \label{sec:dynamics_main}


This section studies an experimental design with carry-overs occur. Let $X_i= 1$ for simplicity.

\begin{ass}[Dynamic model] \label{ass:dynamic} 
For treatments assigned with exogenous parameters $(\beta_{k,1}, \cdots, \hat{\beta}_{k,t})$ as in Assumption \ref{defn:bernoulli}, let 
$
Y_{i,t}^{(k)} = \Gamma(\beta_t, \beta_{t-1}) + \varepsilon_{i,t}^{(k)}, \mathbb{E}_{\hat{\beta}_{k, 1:t}}\Big[\varepsilon_{i,t}^{(k)}\Big] = 0,  
$ 
for some unknown $\Gamma(\cdot)$, $\varepsilon_{i,t}^{(k)}$. 
\end{ass} 

 The components $\hat{\beta}_{k,t}, \beta_{k,t-1}$ capture present and carry-over effects that result from individual and neighbors' treatments in the past two periods. We estimate a \textit{path} of policies $(0, \beta_1, \cdots, \beta_T)$ from an experiment, where, in the first period, we assume for simplicity that none of the individuals is treated. This path is then implemented on a new population.

 \begin{exmp} \label{exmp:main_dynamic} Suppose that 
 $
 Y_{i,t}^{(k)} =  D_{i,t}^{(k)} \phi_1 +  \frac{\sum_{j = 1}^N A_{i,j}^{(k)} D_{j,t-1}}{\sum_{j = 1}^N A_{i,j}^{(k)}} \phi_2  + \nu_{i,t}^{(k)}, D_{i,t}^{(k)} \sim_{i.i.d.} \mathrm{Bern}(\beta_t). 
 $  
 Let $\nu_{i,t}$ be a zero-mean random variable. The expression simplifies to 
 $
 Y_{i,t}^{(k)} = \beta_t \phi_1 + \beta_{t-1} \phi_2 + \varepsilon_{i,t}^{(k)}
 $ 
 where $\varepsilon_{i,t}^{(k)}$ is zero mean, and depends on neighbors' and individual assignments. \qed 
 \end{exmp} 

Given a horizon $T^*$, define the long-run welfare as follows: 
 $
\mathcal{W}(\{\beta_s\}_{s=1}^{T^*}) = \sum_{t=1}^{T^*} q^t \Gamma(\beta_t, \beta_{t-1}), 
$
for a known discounting factor $q < 1$, where $\beta_0 =0$. The long-run welfare deifines the cumulative (discounted) welfare obtained from a certain sequence of decisions $(\beta_1, \beta_2, \cdots)$. The goal is to maximize the long-run welfare. 

The choice of future treatment probabilities must depend on the ones chosen in the past. 
We parametrize future treatment probabilities based on past treatment probabilities as follows 
$
\beta_{t+1} = h_\theta(\beta_{t}, \beta_{t-1}),  \theta \in \Theta.
$  
The parametrization is imposed for computational convenience. The choice of letting $\beta_{t+1}$ be a function of the past two $(\beta_t,\beta_{t-1})$ only follows from the first order conditions with respect to $\beta_{t+1}$.  
For some arbitrary large $T^*$, the objective function takes the following form 
\begin{equation} \label{eqn:jjc} 
\small 
\begin{aligned} 
\widetilde{W}(\theta) &=  \sum_{t=1}^{T^*} q^t \Gamma\Big(\beta_t, \beta_{t-1} \Big), \quad \beta_t = h_\theta(\beta_{t-1},\beta_{t-2}) \quad \text{ for all } t \ge 1, \quad \beta_0 = \beta_{-1} = 0. 
\end{aligned} 
\end{equation} 
Here $\widetilde{W}(\theta)$ denotes the long-run welfare indexed by a given policy's parameter $\theta$. The objective function defines the discounted cumulative welfare induced by the policy $h_\theta$.

Algorithm \ref{alg:sequential} estimates the function $\Gamma(\cdot)$ using a single wave experiment. It then uses the estimated function $\Gamma(\cdot)$ and \textit{its gradient} to estimate the welfare-maximizing parameter $\theta$. Specifically, we conduct the randomization using two periods of experimentation only. We partition the space $[0,1]^2$ into a grid $\mathcal{G}$ of equally spaced components $(\beta_1^r, \beta_2^r)$ for each triad of clusters $r$. Within each triad, we induce small deviations to the parameters $\beta$. 
For each triad $r$, the algorithm returns 
$
\widetilde{\Gamma}(\beta_2^r, \beta_1^r), \widehat{g}_1(\beta_2^r, \beta_1^r), \widehat{g}_2(\beta_2^r, \beta_1^r)
$ 
where the latter two components are the estimated partial derivatives of $\Gamma(\cdot)$, and $\widetilde{\Gamma}(\beta_2^r, \beta_1^r)$ is the within cluster average.

For each pair of parameters $(\beta_2, \beta_1)$, we estimate $\widehat{\Gamma}(\beta_2, \beta_1)$ as follows 
\begin{equation} 
\small 
\begin{aligned} 
& \widehat{\Gamma}(\beta_2, \beta_1) = \widetilde{\Gamma}(\beta_2^r, \beta_1^r) + \widehat{g}_2(\beta_2^r, \beta_1^r)(\beta_2 - \beta_2^r) + \widehat{g}_1(\beta_2^r, \beta_1^g)(\beta_1 - \beta_1^r), \\ & \text{ where } \quad  (\beta_1^r, \beta_2^r) = \mathrm{arg} \min_{(\tilde{\beta}_1, \tilde{\beta}_2) \in \mathcal{G}} \Big\{||\beta_1 - \tilde{\beta}_1||^2 + ||\beta_2 - \tilde{\beta}_2||^2  \Big\}. 
\end{aligned} 
\end{equation}  
we estimate $\Gamma(\beta_2, \beta_1)$ at $(\beta_2, \beta_1)$ using a a first-order Taylor approximation around the closest pairs of parameters in the grid $\mathcal{G}$. Given $\widehat{\Gamma}$, we estimate the welfare-maximizing parameter
$$
\small 
\begin{aligned} 
\widehat{\theta} \in \mathrm{arg} \max_{\theta \in \Theta} \sum_{t=1}^{T^*} q^t \widehat{\Gamma}(\beta_t, \beta_{t-1}), \quad \beta_t = h_\theta(\beta_{t-1}, \beta_{t-2}) \quad \forall t \ge 1, \quad \beta_0 = \beta_{-1} = 0. 
\end{aligned} 
$$

In the following theorem, we study the out-of-sample regret. 

\begin{thm}[Out-of-sample regret] \label{thm:opt_dynamics} Let Assumption \ref{ass:dynamic} hold.  Let $X = 1$, and suppose that $\Gamma(\beta_2, \beta_1)$ is twice differentiable with bounded derivatives. Let treatments be assigned as in Algorithm \ref{alg:sequential}. 
Suppose in addition that $\varepsilon_{i,t}^{(k)} \perp \varepsilon_{j \not \in \mathcal{I}_i^{(k)}}^{(k)}$ where $|\mathcal{I}_i^{(k)}| \le \gamma_N$, for some arbitrary $\gamma_N$ and $\varepsilon_{i,t}^{(k)}$ is sub-Gaussian. Let $\gamma_N\log(\gamma_N)/(\eta_n^2 n) = o(1)$. Then 
$
\lim_{n \rightarrow \infty} P\Big(\sup_{\theta \in \Theta} \widetilde{W}(\theta) - \widetilde{W}(\widehat{\theta}) \le  \frac{\bar{C}}{K} \Big) = 1
$ 
for a constant $\bar{C}$ independent of $K$. 
\end{thm} 

The proof is in Appendix \ref{sec:proof9}. To our knowledge, Algorithm \ref{alg:sequential} is novel to the literature of experimental design.\footnote{We note that optimal dynamic treatments have been studied in the literature on bio-statistics, see, e.g., \cite{laber2014dynamic}, while here we consider the different problem of the design of the experiment. 
\cite{adusumilli2019dynamically} discuss off-line policy estimation in the presence of dynamic budget constraints with $i.i.d.$ observations. The authors assume no carry-overs, and do not discuss the problem of experimental design.}
 
Theorem \ref{thm:opt_dynamics} shows that the regret scales at a rate $1/K$. The key insight is to use information of the estimated gradient. Different from previous sections, the rate $1/K$ is specific to the one-dimensional setting and carry-overs over two consecutive periods. In $p$ dimensions, the rate would be of order $1/K^{2/(p+1)}$ due to the \textit{curse of dimensionality}.


\subsection{Staggered adoption}  \label{sec:staggered}

  In this section, we sketch the experimental design in the presence of staggered adoption, i.e., when treatments are assigned only once to individuals and post-treatment outcomes are collected once. The algorithm works similarly to what was discussed in Section \ref{sec:main_design} with one small difference: every period, we only collect information from a given clusters' pair and update the policy for the subsequent pair and proceed in an iterative fashion. 
 

\begin{thm}[In-sample regret] \label{thm:rate3b} Let the conditions in Theorem \ref{thm:rate2b} hold and let $\beta \in \mathbb{R}$, with $\check{\beta}^{t}$ estimated as in Algorithm \ref{alg:adaptive2}. Then 
$
P\Big(\frac{1}{T} \sum_{t=1}^{T} \Big[W(\beta^*) - W(\check{\beta}^{t})\Big] \le  \bar{C} \frac{ p^2 \log(T)}{T} \Big) \ge 1 - 1/n
$
for a finite constant $\bar{C} < \infty$. 
\end{thm} 
 
See Appendix \ref{sec:proof_last} for the proof. The disadvantage of the staggered adoption case is that we cannot control the in-sample regret \textit{worst-case} over all clusters as in Section \ref{sec:main_design}, but only the average regret across clusters.

\begin{algorithm} [h]   \caption{Adaptive Experiment with staggered adoption}\label{alg:adaptive2}
\footnotesize 
    \begin{algorithmic}[1]
    \Require Starting value $\beta \in \mathbb{R}$, $K$ clusters, $T + 1$ periods of experimentation. 
    \State   Create pairs of clusters $\{k, k+1\}, k \in \{1, 3, \cdots, K-1\}$; 
    \State $t = 0$: 
    \begin{algsubstates}
        \State For $n$ units in each cluster observe  the baseline outcome $Y_{i,0}^{(h)}, h \in \{1, \cdots, K\}, \check{\beta}^0 = \beta$.  
        \State Initalize a gradient estimate $\widehat{M}_{t} = 0$
        \end{algsubstates} 
    \While{$1 \le t \leq T$}
    \begin{algsubstates}
    \State Sample without replacement one pair of clusters $\{k, k+1\}$ not observed in previous iterations; 
    \State Define 
    $
    \check{\beta}^t = \check{\beta}^{t-1} + \alpha_t \widehat{M}_{t}; 
    $
        \State  Assign treatments for this pair $h$ as 
    $$
    \small 
    \begin{aligned}
    D_{i,t}^{(h)} \sim  
    \pi(1, \beta_{t}), \quad  \beta_{t} = \begin{cases} & \check{\beta}^t  + \eta_n \text{ if } h \text{ is even}  \\
    & \check{\beta}^t  - \eta_n \text{ if } h \text{ is odd}
    \end{cases} ,  \quad n^{-1/2} < \eta_n \le n^{-1/4}
    \end{aligned}
    $$
        \State For $n$ units in each cluster $h \in \{1, \cdots, K\}$ observe $Y_{i,t}^{(h)}$ for this pair 
         \State Update the gradient $\hat{M}_t$ accordingly
        \EndWhile
      \end{algsubstates}
 \State Return 
 $
 \hat{\beta}^* =  \check{\beta}^T
 $  
 
         \end{algorithmic}
\end{algorithm}

\subsection{Learning the marginal effect with arbitrary cluster heterogeneity and many (possibly small) clusters} \label{app:arbitrary_heterogeneity}

In this subsection we sketch an extension where clusters are allowed to exhibit arbitrary unobserved cluster level heterogeneity. We show that in this case, it is still possible to run an adaptive experiment by randomly allocating clusters into groups of clusters, and inducing local perturbation between each \textit{group} of clusters. Because of the random allocation to groups, researchers are able to estimate the marginal effect of the reward, averaged across all clusters in the population.

However, rates of convergence of the in-sample (and out-of-sample) regret are slower in the number of clusters compared to our main results in Theorem \ref{thm:rate2b} and therefore rely on researchers having access to a very large (growing) number of (possibly small) clusters. This is intuitive: as we allow for arbitrary cluster level heterogeneity, it is not possible to consistently estimate reward and marginal effect with a finite number of clusters. 

For simplicity, consider a univariate parameter $\beta \in \mathbb{R}$. Suppose we can write 
\begin{equation} \label{eqn:outcomes_heterogeneous}
\small 
\begin{aligned} 
\mathbb{E}_\beta[Y_{i,t}^{(k)}] = W_k(\beta) + \alpha_t 
\end{aligned} 
\end{equation} 
where $\mathbb{E}_\beta[\cdot]$ integrates over the distribution of covariates, individual and neighbors' treatment assignments. Equation \eqref{eqn:outcomes_heterogeneous} is a generalization of the model in Equation \eqref{eqn:main_Y0} as it allows for arbitrary (unobserved) cluster-level heterogeneity. Here $W_k(\beta)$ denotes a reward function which is cluster specific (therefore incorporating arbitrary heterogeneity in addition to fixed effects). Define the average reward function, the marginal effect and the optimal policy as 
$$
W(\beta) = \frac{1}{K} \sum_{k=1}^K W_k(\beta), \quad M(\beta) = \frac{\partial W(\beta)}{\partial \beta}, \quad \beta^* \in \mathrm{arg} \max_\beta W(\beta) 
$$ 
That is, $W(\beta)$ denotes the average reward across all clusters and $M(\beta)$ its derivative. 

Ideally, we would like to learn the marginal effect to then control 
notions of in-sample and out-of-sample regret, respectively defined for a sequence of policies $\hat{\beta}_{k,t}$ deployed in the experiment and final estimated policy $\hat{\beta}^*$, $W(\beta^*)- \frac{1}{K T} \sum_{t=1}^T \sum_{k=1}^K W_k(\hat{\beta}_{k,t}), W(\beta^*) - W(\hat{\beta}^*)$. 

To do so, we can follows similar experiments as those in the main text, with appropriate modifications. 
In particular, we learn the marginal effect by randomizing clusters into two groups. Specifically, denote $R_k \in \{-1,1\}$ a sampling indicator for each cluster $k \in \{1, \cdots, K\}$. The sampling indicator is such that if $R_k = 1$, a cluster $k$ is assigned to the first group of clusters and if $R_k = -1$ cluster $k$ is assigned to the second group. We can then proceed similarly to the local perturbation method presented in Section \ref{sec:designs}, with the difference that here the perturbation $\eta$ should depend on the overall number of clusters $K$. Specifically, suppose at time $t = 1$ we randomize treatments as follows
$$
D_{i,t = 1}^{(k)} | R_1, \cdots, R_K \sim \pi(X_i^{(k)}, \beta + R_k \eta_{K}), \quad P(R_k = 1) = P(R_k = -1) = 1/2
$$ 
where treatments are independent across individuals and sampling indicators $R_k$ are independent across clusters. 
We estimate the marginal effect by taking the difference between the clusters exposed to each perturbation so that 
$$
\hat{M}(\beta) = \frac{1}{n \eta_{K} K} \sum_{k=1}^{K} \sum_{i=1}^n  Y_{i, t = 1}^{(k)} \cdot R_k. 
$$ 

The approach therefore mimics Section \ref{sec:designs} with the difference that here we first randomly allocate clusters into two groups (instead of using two clusters only) and then induce perturbations between these two groups. The estimated marginal effect does not require to differentiate out baseline outcomes because of the random allocation of clusters into groups. 

Using a second order Taylor expansion, assuming bounded second derivative of each $W_k(\cdot)$ we can write 
$$ 
\begin{aligned} 
\mathbb{E}[\hat{M}(\beta)] & = \frac{1}{\eta_K} \left\{\frac{1}{K} \sum_{k=1}^K \mathbb{E}[R_k] W_k(\beta) + \frac{1}{K} \sum_{k=1}^K \mathbb{E}[\mathrm{sign}(R_k) R_k] \frac{\partial W_k(\beta)}{\partial \beta} \eta_K + \mathcal{O}(\eta_K^2) \right\} \\ 
& = M(\beta)  + \mathcal{O}(\eta_K), 
\end{aligned} 
$$ 
where we used the fact that $\mathbb{E}[R_k] = 0$ and $\mathbb{E}[\mathrm{sign}(R_k) R_k] = \mathbb{E}[|R_k|] = 1$. Therefore, the estimated marginal effect has a bias proportional to $\eta_K$. In addition, under mild regularity conditions the variance is of order (even with small $n$ in each cluster)
$$
\mathbb{V}(\hat{M}(\beta)) = \mathcal{O}\left(\frac{1}{\eta_K^2 K}\right). 
$$ 

It follows that the fastest rate of convergence we can achieve for $\hat{M}(\beta)$ to be consistent for $M(\beta)$ takes $\eta_K \approx K^{-1/4}$ with convergence rate of order $\hat{M}(\beta) - M(\beta) = \mathcal{O}_p(K^{-1/4})$. 
Using the estimator for the marginal effect, we can then
leverage adaptive experiments similar in spirit to \cite{bubeck2012regret} and \cite{lattimore2024bandit} to control the in-sample and out-of-sample regret as a function of periods $T$ through a sequential experiment with $T$ iterations, while letting $K \rightarrow \infty$.
Finite sample regret bounds in $K$ and $T$ would be slower in $K$ than those in Theorem \ref{thm:rate2b} due to the dependence of the convergence rate of the marginal effect on the number of clusters $K$.\footnote{For finite sample bound in $K$, we can adopt the sequential cross-fitting procedure similar to Algorithm \ref{alg:adaptive} where clusters are randomly allocated into pairs, each pair containing two groups of randomly allocated clusters. The algorithm then proceeds similarly to Algorithm \ref{alg:adaptive} using group of clusters in lieu of single clusters, and choosing $\eta_K \approx K^{-1/3}$. One can provide finite sample bounds for the in-sample and out-of-sample regret both as a functions of $T$ and $K$ using finite sample concentration inequalities and following similar derivations as those we derived for the adaptive experiment in Theorem \ref{thm:rate2b} using groups of clusters in lieu of single clusters, and adjusting the rates as functions of the number of clusters within each group.  }

\section{Additional Algorithms} \label{app:algorithms}

\vspace{-5mm}

\begin{algorithm} [H]   \caption{Welfare maximization with a ``non-adaptive" experiment}\label{alg:one_wave_optimal}
\footnotesize 
    \begin{algorithmic}[1]
    \Require $K$ clusters, $T = p$ periods of experimentation, $n^{-1/2} < \eta_n \le n^{-1/4}$. 
    \State   Create pairs of clusters $\{k, k+1\}, k \in \{1, 3, \cdots, K-1\}$; 
    \State $t = 0$: For $n$ units in each cluster observe  the baseline outcome $Y_{i,0}^{(h)}, h \in \{1, \cdots, K\}$. Assign each pair $(k, k+1)$ to an element $\beta^k \in \mathcal{G}$, where $\mathcal{G}$ is an equally spaced grid.
   \While{$1 \le t \leq T$}
    \begin{algsubstates}
        \State  Assign treatments as 
    $
    D_{i,t}^{(h)} \sim  
    \pi(1, \beta^h), \quad  \beta^h = \check{\beta}^{h}  \pm \eta_n \underline{e}_t (h \text{ is even/odd}),  
    $
        \State For $n$ units in each cluster $h \in \{1, \cdots, K\}$ observe $Y_{i,t}^{(h)}$; estimate for pair $(k, k+1)$, entry $t$, 
        $ 
        \widehat{M}_{(k, k+1)}^{(t)}(\beta^k) = \frac{1}{2 \eta_n} \Big[ \bar{Y}_t^k - \bar{Y}_0^k\Big] -  \frac{1}{2 \eta_n} \Big[ \bar{Y}_t^{k+1} - \bar{Y}_0^{k+1}\Big]. 
        $
  \end{algsubstates} 
   \EndWhile, \Return $\hat{\beta}^{ow}$ as in Equation \eqref{eqn:expansion}. 
 
         \end{algorithmic}
\end{algorithm}

\vspace{-4mm}

\begin{algorithm} [H]   \caption{Adaptive Experiment with Many Coordinates}\label{alg:adaptive_complete}
\footnotesize 
\begin{algorithmic}[1]
    \Require Starting value $\beta_0$ in the interior of $[\mathcal{B}_1 + \eta_n, \mathcal{B}_2 - \eta_n]^p$, $K$ clusters, $T + 1$ periods of experimentation, constant $\bar{C}$. 
    \State   Create pairs of clusters $\{k, k+1\}, k \in \{1, 3, \cdots, K-1\}$; 
    \State $t = 0$(\textit{baseline}): 
 Assign treatments as 
    $
    D_{i,0}^{(h)} | X_i^{(h)} = x  \sim \pi(x;\beta_0)  \text{ for all } h \in \{1, \cdots, K\}; 
    $
    for $n$ units in each cluster observe $Y_{i,0}^{(h)}, h \in \{1, \cdots, K\}$; for cluster $k$ initalize a gradient estimate $\widehat{M}_{k,t} = 0$ and initial parameters $\check{\beta}_{k}^o = \beta_0$. 
    \While{$1 \le w \leq \check{T} = \frac{T}{p}$}
    \begin{algsubstates}
    \ForEach {$j \in \{1, \cdots, p\}$} ($P_{\mathcal{B}_1, \mathcal{B}_2 - \eta_n}$ is the projection operator onto $[\mathcal{B}_1, \mathcal{B}_2 - \eta_n]^p$)
    $$
    \small 
    \begin{aligned} 
    \check{\beta}_{h}^w = 
   &   P_{\mathcal{B}_1 + \eta_n, \mathcal{B}_2 - \eta_n}\Big[\check{\beta}_{h}^{w-1} + \alpha_{\nint{h+2}, w - 1} \widehat{M}_{\nint{h + 2},w - 1}\Big],  \quad \nint{h} = h 1\{h \le K\} + (h - K)1\{h > K\}.
    \end{aligned}
    $$ 
        \State  Assign treatments as (for a finite constant $\bar{C}$, $\underline{e}_j$ in Equation \eqref{eqn:ej})
    $$
    \small 
    \begin{aligned}
    D_{i,t}^{(h)} | X_{i,t}^{(h)} = x \sim  
    \pi(x, \hat{\beta}_{h,w}), \quad  \hat{\beta}_{h,w} = \check{\beta}_{h}^w  \pm \eta_n \underline{e}_j (h\text{ is even/odd}) ,  \quad \bar{C} n^{-1/2}< \eta_n < \bar{C} n^{-1/4}
    \end{aligned}
    $$ 
        \State For $n$ units in clusters $h \in \{1, \cdots, K\}$ observe $Y_{i,t}^{(h)}$, and for pair $\{k, k+1\}$, estimate
      $
      \widehat{M}_{k,w}^{(j)} = 
      \widehat{M}_{k + 1,w}^{(j)} = \frac{1}{2 \eta_n} \Big[\bar{Y}_t^{(k)} - \bar{Y}_0^{(k)}\Big] -  \frac{1}{2 \eta_n} \Big[\bar{Y}_t^{(k + 1)} - \bar{Y}_0^{(k +1)}\Big]. 
      $  
       \State $t  \leftarrow t + 1$.  
      \EndFor
      \State $w  \leftarrow w + 1$.          
        \EndWhile
      \end{algsubstates}
,  \Return 
 $
 \hat{\beta}^* = \frac{1}{K} \sum_{k = 1}^K \check{\beta}_{k}^{\check{T}}
 $  
 
         \end{algorithmic}
\end{algorithm}

\begin{algorithm} [H]   \caption{Dynamic Treatment Effects with $\beta \in \mathbb{R}$}\label{alg:sequential}
\footnotesize 
    \begin{algorithmic}[1]
    \Require Parameter space $\mathcal{B}$, clusters $\{1, \cdots, K\}$, two periods $\{t, t+1\}$, perturbation $\eta_n$.
    \State Group clusters into triads $r \in \{1, \cdots, K/3\}$ with consecutive indeces $\{k, k+1, k+2\}$; construct a grid of parameters $\mathcal{G} \subset [0,1]^2$ equally spaced on $[0,1]^2$; assign each parameter $(\beta_1^r, \beta_2^r) \in \mathcal{G}$ to a different triad $r$. 
   \State For each $r \in \{1, \cdots, K/3\}$, and triad $(k, k+1, k+2)$ randomize treatments 
\begin{equation} 
\small 
\begin{aligned} 
 D_{i,t}^{(k)} | X_i^{(k)}, \beta_1^r, \beta_2^r &\sim \pi(X_i^{(k)}, \beta_2^r), \quad 
 & D_{i,t + 1}^{(k)} | X_i^{(k)}, \beta_1^r, \beta_2^r \sim \pi(X_i^{(k)}, \beta_1^r), \\ 
 D_{i,t}^{(k + 1)} | X_i^{(k)}, \beta_1^r, \beta_2^r & \sim \pi(X_i^{(k)}, \beta_2^r + \eta_n) , \quad 
 &  D_{i,t + 1}^{(k + 1)} | X_i^{(k)}, \beta_1^r, \beta_2^r \sim \pi(X_i^{(k)}, \beta_1^r) \\ 
 D_{i,t}^{(k + 2)} | X_i^{(k)},\beta_1^r, \beta_2^r &\sim \pi(X_i^{(k)}, \beta_2^r), \quad 
 & D_{i,t + 1}^{(k + 2)} | X_i^{(k)}, \beta_1^r, \beta_2^r \sim\pi(X_i^{(k)}, \beta_1^r + \eta_n)
\end{aligned} . 
\end{equation} 
 \State For each $k \in \{1, 4, \cdots, K - 2\}$ estimate 
\begin{equation} \label{eqn:dynamic_est}  
\small 
\begin{aligned}  
 & \widehat{g}_{1, k} = \frac{\bar{Y}_{t+1}^{(k+2)} - \bar{Y}_{t+1}^{(k)}}{\eta_{n}}, \quad \widehat{g}_{2, k} = \frac{\bar{Y}_{t+1}^{(k+1)} - \bar{Y}_{t+1}^{(k)}}{\eta_{n}}, \quad \widetilde{\Gamma}_k = \frac{1}{3} \sum_{h \in \{k, k+1, k+2\}} \bar{Y}_{t+1}^{(h)}
 \end{aligned} 
 \end{equation}  

         \end{algorithmic}
\end{algorithm}

\section{Proofs for the extensions} \label{app:gradient}

 \subsection{Proof of Theorem \ref{thm:const_steady}} \label{sec:proof8}

We write
$
\mathbb{E}\Big[\bar{Y}_t^{(k)}\Big| p_t^{(k)}\Big]  = \alpha_t + \tau_k + g\Big(  q(\beta + \eta_n) + o_p(\eta_n), \beta + \eta_n\Big) . 
$ 
From a Taylor expansion in its first argument around $q(\beta + \eta_n)$, we obtain 
$
g\Big(  q(\beta + \eta_n) + o_p(\eta_n), \beta + \eta_n\Big) =  
g\Big(  q(\beta + \eta_n), \beta + \eta_n\Big) + o_p(\eta_n). 
$ 
Similarly, $\mathbb{E}\Big[\bar{Y}_t^{(k)}\Big| p_t^{(k+1)}\Big]  = \alpha_t + \tau_k + g\Big(   q(\beta - \eta_n) + o_p(\eta_n), \beta  - \eta_n\Big) = g\Big(  q(\beta - \eta_n), \beta - \eta_n\Big) + o_p(\eta_n)$.
Therefore, 
$$ 
\small 
\begin{aligned} 
\mathbb{E}\Big[\bar{Y}_t^{(k)}\Big| p_t^{(k)}\Big] - \mathbb{E}\Big[\bar{Y}^{(k+1)}\Big| p_t^{(k+1)}\Big] & = \tau_k - \tau_{k+1} + g\Big(  q(\beta + \eta_n), \beta + \eta_n\Big) + o_p(\eta_n) - g\Big(   q(\beta - \eta_n), \beta  - \eta_n\Big). 
\end{aligned}
$$ 
We can now proceed with a Taylor expansion around of the functions $g(\cdot)$ around $\beta$ to obtain (this follows similarly to Lemma \ref{lem:expect})
$
g\Big(  q(\beta + \eta_n), \beta + \eta_n\Big) - g\Big(   q(\beta - \eta_n), \beta  - \eta_n\Big) = 2 M_g(\beta) \eta_n + O(\eta_n^2). 
$ 
In addition observe that since at the baseline $\beta_0$ is the same for both clusters, 
$
\mathbb{E}[Y_0^{(k)} - Y_0^{(k + 1)} | p_t^{(k)}, p_t^{(k+1)}] = \tau_k - \tau_{k+1} + o_p(\eta_n).  
$ 
The proof concludes from Lemma \ref{lem:concentration1} with $\delta = 1/n$ and the local dependence assumption in Assumption \ref{ass:global}.

\subsection{Proof of Theorem \ref{thm:grid_search1}} \label{sec:aah}

Recall that $\mathcal{G}$ denotes a finite grid with $K/2$ elements. 
 First, we bound 
$
W(\beta^*) - W(\hat{\beta}^{ow}) \le 2 \sup_{\beta \in [0,1]^p} \Big|W(\beta) - \hat{W}(\beta) \Big|.
$ 
By the mean value theorem, we can write for any $\beta^k \in \mathcal{G}$
$
W(\beta) = W(\beta^k) + M(\beta^k)^\top(\beta - \beta^k) + O\Big(||\beta^k - \beta||^2\Big), 
$ 
Since we construct $\hat{W}(\beta)$ as in Equation \eqref{eqn:expansion}, we can choose $\beta^k$ closest to $\beta$, such that $ O\Big(||\beta^k - \beta||^2\Big) = O(1/K^{2/p})$ by construction of the grid. We can write
$$
\small 
\begin{aligned} 
\sup_{\beta \in [0,1]} \Big|W(\beta) - \hat{W}(\beta) \Big| & \le \sup_{\beta \in [0, 1]^p, k \in \{1, \cdots, K\}} \Big| W(\beta^k) + M(\beta^k)^\top(\beta - \beta^k) - \bar{W}^{k} - \widehat{M}_{(k, k+1)}^\top(\beta - \beta^{k})\Big| + O(1/K^{2/p}) \\ 
&\le 
\sup_{k \in \{1, \cdots, K\}} \Big| W(\beta^k)- \bar{W}^{k}\Big| + ||M(\beta^k) - \widehat{M}_{k, k+1}||_{\infty} O(1) + O(1/K^{2/p}) 
\end{aligned} 
$$ 
In addition, similarly to what discussed in Lemma \ref{lem:expect}, it follows 
$$
\small 
\begin{aligned} 
2 \mathbb{E}\Big[\bar{W}^k\Big] = \int y(x, \beta^k + \eta_n)dF_X(x) + \int y(x, \beta^k -  \eta_n)dF_X(x) = 2 \int y(x, \beta^k)dF_X(x) + O(\eta_n^2). 
\end{aligned}
$$ 
Using Lemma \ref{lem:concentration1}, we can write for all $k \le K$, with probability at least $1 - \delta$, 
$
\Big|\bar{W}^k-  W(\beta^k)\Big| \le c_0\Big(\sqrt{\gamma_N \log(p K \gamma_N/\delta)/n} + \eta_n^2 \Big),   
$ 
where we used the union bound over $K, p$ in the expression. Similarly, from Lemma \ref{lem:2}, also using the union bound over $K$ and $p$, with probability at least $1 - \delta$, 
$
||\widehat{M}_{(k, k+1)} - M(\beta^k)||_{\infty} \le c_0\Big(\sqrt{\gamma_N \log(K p \gamma_N/\delta)/(n \eta_n^2)} + \eta_n\Big),  
$ 
which concludes the proof as we choose $\delta = 1/n$, since $\eta_n = o(1)$, and $p$ is finite.

\subsection{Proof of Proposition \ref{lem:lem0}} \label{sec:lem_main}

\paragraph{Proof of $(i)$ in Assumption \ref{ass:ass_0}} The condition $(i)$ state no carryovers and covariates with same distribution. This condition holds directly under $(B)$ in Example \ref{exmp:microfoundation} and Equation \eqref{eqn:network} 

\paragraph{Proof of $(ii)$ in Assumption \ref{ass:ass_0}} 
Under Condition (B) in Example \ref{exmp:microfoundation}, and using the fact that $r(\cdot)$ is symmetric in $A_{i, \cdot}^{(k)}$, we can write for some function $g$,  
 $$
 \small 
 \begin{aligned} 
 r\Big(D_{i,t}^{(k)}, D_{j: A_{i,j}^{(k)} > 0,t}^{(k)}, X_i^{(k)}, X_{j: A_{i,j}^{(k)} > 0}^{(k)},  U_i^{(k)}, U_{j: A_{i,j}^{(k)} > 0}^{(k)}, |\mathcal{N}_i^{(k)}|, \nu_{i,t}^{(k)}\Big) = g(Z_{i,t}^{(k)}). 
 \end{aligned} 
 $$ 
 Here, $Z_{i,t}^{(k)}$ depends on $A_i^{(k)}$, i.e., the edges of individual $i$, and on unobservables and observables of all those individuals such that $A_{i,j}^{(k)} > 0$, namely, 
 $$
 \small 
 \begin{aligned} 
 Z_{i,t}^{(k)} = \Big[D_{i,t}^{(k)}, X_i^{(k)}, U_i^{(k)}, \nu_{i,t}^{(k)}, A_i^{(k)} \otimes \Big( X^{(k)}, U^{(k)}, D_{t}^{(k)}\Big), \Big\{\Big[X_j^{(k)}, U_j^{(k)}\Big], j :  1\{i_k \leftrightarrow j_k\} = 1\Big\}\Big].
 \end{aligned} 
 $$ 
 The last element in $Z_{i,t}$ captures the dependence of $r(\cdot)$ with $A_{i, \cdot}^{(k)}$. 
Such representation follows under Condition (A) in Example \ref{exmp:microfoundation}, $A_i^{(k)}$ is a function of 
 $
 \Big(X_i^{(k)}, U_i^{(k)}, \Big\{\Big[X_j^{(k)}, U_j^{(k)}\Big], j :  1\{i_k \leftrightarrow j_k\} = 1\Big\}\Big), 
 $ 
 only, and each entry depends on $(X_j, U_j, X_i, U_i)$ through the same function $f$ for each individual. What is important, is that $\sum_j 1\{i_k \leftrightarrow j_k\} = \gamma_N^{1/2}$ for each unit $i$. 
 $$
 \small 
 \begin{aligned} 
 Z_{i,t}^{(k)} = \tilde{g}(D_{i,t}^{(k)}, \nu_{i,t}^{(k)}, X_i^{(k)}, U_i^{(k)}, \tilde{Z}_{i,t}^{(k)}), \quad \tilde{Z}_{i,t}^{(k)} =   \Big\{\Big[X_j^{(k)}, U_j^{(k)}, D_{j,t}^{(k)}\Big], j \neq i :  1\{i_k \leftrightarrow j_k\} = 1 \Big\}, 
 \end{aligned} 
 $$ 
where $\tilde{Z}_{i,t}^{(k)}$ is the vector of $\Big[X_j^{(k)}, U_j^{(k)}, D_{j,t}^{(k)}\Big]$ of all individuals $j$ with $1\{i_k \leftrightarrow j_k\} = 1$.

To show that Assumption \ref{ass:ass_0} $(ii)$ holds it suffices to show that $Z_{i,t}^{(k)}$ are identically distributed (conditional on the indicators $1\{i_k \leftrightarrow j_k\}$). 

To show this, observe first that since $(U_i^{(k)}, X_i^{(k)}) \sim_{i.i.d.} F_{X|U} F_U$, and $\{\nu_{i,t}\}$ are $i.i.d.$ conditionally on $U^{(k)}, X^{(k)}$ (Condition (B) in Example \ref{exmp:microfoundation}) and treatments are randomized independently (Assumption \ref{defn:bernoulli}), we have  
$ 
\Big[X_j^{(k)}, U_j^{(k)}, \nu_{j,t}^{(k)}, D_{j,t}^{(k)}\Big] \Big| \hat{\beta}_{k,t} \sim_{i.i.d.} \mathcal{D}(\hat{\beta}_{k,t})
$
is $i.i.d$ with some distribution $\mathcal{D}(\hat{\beta}_{k,t})$ which only depends on the coefficient $\hat{\beta}_{k,t}$ governing the distribution of $D_{i,t}^{(k)}$ under Assumption \ref{defn:bernoulli}. 

As a result for $\hat{\beta}_{k,t} \perp (X^{(k)}, \nu_t^{(k)}, U^{(k)})$, $Z_{i,t}^{(k)}$ are identically distributed across $i,k$ and therefore Assumption \ref{ass:ass_0} $(ii)$ holds. The fact that $Z_{i,t}^{(k)}$ are identically distributed follows from the fact that $\sum_j  1\{i_k \leftrightarrow j_k\} = \gamma_N^{1/2}$ for all $i$, and $\tilde{Z}_{i,t}^{(k)} \perp [D_{i,t}^{(k)}, \nu_{i,t}^{(k)}, X_i^{(k)}, U_i^{(k)}]$ for all $(i,k,t)$ because $[D_{i,t}^{(k)}, \nu_{i,t}^{(k)}, X_i^{(k)}, U_i^{(k)}]$ are iid from Assumption \ref{defn:bernoulli} and Condition (B) in Example \ref{exmp:microfoundation}. 
 
 \paragraph{Proof of Assumption \ref{ass:ass_0} $(iii)$}
From the argument above, $Y_{i,t}^{(k)} | \hat{\beta}_{k,t}$ is a measurable function of a vector $\Big[X_j^{(k)}, U_j^{(k)}, \nu_{j,t}^{(k)}, D_{j,t}^{(k)}\Big]_{j :  1\{i_k \leftrightarrow j_k\} = 1}$.\footnote{Here for notational convenience convenience only, we are letting $1\{i_k \leftrightarrow i_k\} = 1$.} Therefore, given $\hat{\beta}_{k,t}$, $Y_{i,t}^{(k)}$ is mutally independent of $Y_{v, t}^{(k)}$ for all $v$ such that they do not share a common element $\Big[X_j^{(k)}, U_j^{(k)}, \nu_{j,t}^{(k)}, D_{j,t}^{(k)}\Big]$, that is, such that 
 $
\max_j 1\{i_k \leftrightarrow j_k\} 1\{v_k \leftrightarrow j_k\} = 0. 
 $  There are at most $\gamma_N^{1/2} + \gamma_N$ many of $Y_{v,t}^{(k)}$ which can share a common neighbor with $Y_{i,t}^{(k)}$ ($\gamma_N^{1/2}$ neighhbors and $\gamma_N$ neighbors of the neighbors).

 \subsection{Proof of Theorem \ref{thm:optimality_global}} \label{sec:proof7}
 
\vspace{-2mm}
We break the proof into several steps. Recall that  the theorem assumes the outcome model in Example \ref{exmp:microfoundation}. 
\paragraph{Upper bound on $W_N^*$}  Recall that from (B) in Example \ref{exmp:microfoundation}, the maximum degree is $\gamma_N^{1/2}$. Consider first the case where Assumption $\Delta(x) = v(x)$. We return to the case where $\Delta(x) \neq v(x)$ at the end of the proof. For $\Delta(x) = v(x)$ 
$$
\small 
\begin{aligned} 
W_N^* \le \frac{1}{N} \sum_{i=1}^N \sup_{\mathcal{P}} \mathbb{E}\Big[\mathbb{E}_{D \sim \mathcal{P}(A, X)}\Big[s\Big(\Big[\frac{\sum_{j = 1}^N A_{i,j} D_j 1\{X_j = x\}}{\sum_{j = 1}^N A_{i,j} 1\{X_j = x\}}\Big]_{x \in \{1, \cdots, |\mathcal{X}|\}}\Big) \Big| A, X\Big]\Big]. 
\end{aligned} 
$$  
Let 
$
\beta^G = \mathrm{arg} \max_{\beta_1, \cdots, \beta_{|\mathcal{X}|} \in [0,1]^{|\mathcal{X}|}} s\Big(\beta_1, \cdots, \beta_{|\mathcal{X}|}\Big).
$ 
Note that since $D_j \in \{0,1\}$, we can write 
$$
\small 
\begin{aligned} 
\sup_{\mathcal{P}} \mathbb{E}\Big[\mathbb{E}_{D \sim \mathcal{P}(A, X)}\Big[s\Big(\Big[\frac{\sum_{j = 1}^N A_{i,j} D_j 1\{X_j = x\}}{\sum_{j = 1}^N A_{i,j} 1\{X_j = x\}}\Big]_{x \in \{1, \cdots, |\mathcal{X}|\}}\Big) \Big| A, X\Big]\Big] \le s\Big(\beta_1^G, \cdots, \beta_{|\mathcal{X}|}^G\Big).
\end{aligned} 
$$ 
\vspace{-8mm} 
\paragraph{Lower bound on $W(\beta^*)$} Using the fact that $\mathcal{B} = [0,1]^{|\mathcal{X}|}$, we can write\footnote{$\mathbb{E}_\beta\Big[s\Big(\Big[\frac{\sum_{j = 1}^N A_{i,j} D_j 1\{X_j = x\}}{\sum_{j = 1}^N A_{i,j} 1\{X_j = x\}}\Big]_{x \in \{1, \cdots, |\mathcal{X}|\}}\Big)\Big]$ does not depend on  $i$ similarly to the proof of Proposition \ref{lem:lem0}.}
$$
\small 
\begin{aligned} 
W(\beta^*) & = \max_{\beta \in [0,1]^{|\mathcal{X}|}} \mathbb{E}_\beta\Big[s\Big(\Big[\frac{\sum_{j = 1}^N A_{i,j} D_j 1\{X_j = x\}}{\sum_{j = 1}^N A_{i,j} 1\{X_j = x\}}\Big]_{x \in \{1, \cdots, |\mathcal{X}|\}}\Big)\Big] \ge \mathbb{E}_{\beta^G}\Big[s\Big(\Big[\frac{\sum_{j = 1}^N A_{i,j} D_j 1\{X_j = x\}}{\sum_{j = 1}^N A_{i,j} 1\{X_j = x\}}\Big]_{x \in \{1, \cdots, |\mathcal{X}|\}}\Big)\Big], 
\end{aligned} 
$$ 
where we use the fact that $\beta^G = (\beta_1^G, \cdots, \beta_{|\mathcal{X}|}^G) \in [0,1]^{|\mathcal{X}|}$, and $\Delta(\cdot) = v(\cdot)$. It follows
$$
\small
\begin{aligned} 
&  \mathbb{E}_{\beta^G}\Big[s\Big(\Big[\frac{\sum_{j = 1}^N A_{i,j} D_j 1\{X_j = x\}}{\sum_{j = 1}^N A_{i,j} 1\{X_j = x\}}\Big]_{x \in \{1, \cdots, |\mathcal{X}|\}}\Big)\Big] 
\\ & =  s(\beta^G) + \mathbb{E}_{\beta^G}\Big\{\frac{\partial s(\beta)}{\partial \beta}\Big|_{\tilde{\beta}} \times  \Big(\Big[\frac{\sum_{j = 1}^N A_{i,j} D_j 1\{X_j = x\}}{\sum_{j = 1}^N A_{i,j} 1\{X_j = x\}}\Big]_{x \in \{1, \cdots, |\mathcal{X}|\}} - \beta^G \Big)\Big\},  
\end{aligned} 
$$ 
with $\frac{\partial s(\cdot)}{\partial \beta}$ evaluated at a (random) $\tilde{\beta} \in \Big[\Big[\frac{\sum_{j = 1}^N A_{i,j} D_j 1\{X_j = x\}}{\sum_{j = 1}^N A_{i,j} 1\{X_j = x\}}\Big]_{x \in \{1, \cdots, |\mathcal{X}|\}}, \beta^G \Big]$. It follows 
$$
\small 
\begin{aligned} 
& W_N^* - W(\beta^*) \le \underbrace{\Big|\mathbb{E}_{\beta^G}\Big\{\frac{\partial s(\beta)}{\partial \beta}\Big|_{\tilde{\beta}} \times  \Big(\Big[\frac{\sum_{j = 1}^N A_{i,j} D_j 1\{X_j = x\}}{\sum_{j = 1}^N A_{i,j} 1\{X_j = x\}}\Big]_{x \in \{1, \cdots, |\mathcal{X}|\}} - \beta^G \Big) \Big\}\Big|}_{(I)}. 
\end{aligned} 
$$ 
\vspace{-8mm}
\paragraph{Bound with Cauchy-Schwarz} We can now bound $(I)$ as follows. 
$$
\small
\begin{aligned} 
&(I) \le \sup_{\beta} \Big|\Big|\frac{\partial s(\beta)}{\partial \beta}\Big|\Big|_2  \times |\mathcal{X}| \max_{x \in \mathcal{X}} \underbrace{\sqrt{\mathbb{E}_{\beta^G}\Big[\Big(\frac{\sum_{j = 1}^N A_{i,j} D_j 1\{X_j = x\}}{\sum_{j = 1}^N A_{i,j} 1\{X_j = x\}} - \beta_x^G\Big)^2\Big]}}_{(II)}, 
\end{aligned} 
$$ 
where we used Cauchy-Schwarz and then bound the first component by the supremum over $\beta,x$ and the second component by the largest term over $x \in \mathcal{X}$ times $|\mathcal{X}|$. 
\vspace{-3mm}
\paragraph{Bound for $(II)$}
Recall that here $\mathbb{E}_{\beta^G}$ indicates that  $D_{i,t} | X_i^{(k)} = x \sim_{i.i.d.} \mathrm{Bern}(\beta_x^G)$. It follows 
$
\mathbb{E}_{\beta^G}\Big[\frac{\sum_{j = 1}^N A_{i,j} D_j 1\{X_j = x\}}{\sum_{j = 1}^N A_{i,j} 1\{X_j = x\}} \Big| X^{(k)},  A^{(k)}\Big] = \beta_x^G 1\{\sum_{j = 1}^N A_{i,j} 1\{X_j = x\} > 0\}
$ (since we defined $0/0 = 0$). 
Let $p_x = P(\sum_{j = 1}^N A_{i,j} 1\{X_j = x\} > 0)$ By the law of total variance,
$$
\small 
\begin{aligned} 
\mathbb{E}_{\beta^G}\Big[\Big(\frac{\sum_{j = 1}^N A_{i,j} D_j 1\{X_j = x\}}{\sum_{j = 1}^N A_{i,j} 1\{X_j = x\}} - \beta_x^G\Big)^2\Big] = \mathbb{E}_{\beta^G}\Big[\mathrm{Var}\Big(\frac{\sum_{j = 1}^N A_{i,j} D_j 1\{X_j = x\}}{\sum_{j = 1}^N A_{i,j} 1\{X_j = x\}} \Big| X^{(k)}, A^{(k)}\Big)\Big] + (\beta_x^G)^2 (1 - p_x). 
\end{aligned} 
$$ 
In addition, 
$
\mathrm{Var}_{\beta^G}\Big(\frac{\sum_{j = 1}^N A_{i,j} D_j 1\{X_j = x\}}{\sum_{j = 1}^N A_{i,j} 1\{X_j = x\}} \Big| X^{(k)}, A^{(k)}\Big) \le 1\{\sum_{j = 1}^N A_{i,j} 1\{X_j = x\} > 0\}\beta_x^G(1 - \beta_x^G)\Big/\sum_{j = 1}^N A_{i,j} 1\{X_j = x\}.
$ 
Let $\kappa ' = \underline{\kappa} P(X = x)$, $P(X = x) > 0$ and $\kappa'$ is bounded away from zero by Assumption \ref{ass:discrete_space}. Let $1_x = 1\{\sum_{j = 1}^N A_{i,j} 1\{X_j = x\} > 0\}$ (recall that whenever $\sum_{j = 1}^N A_{i,j} 1\{X_j = x\} = 0$ we indicate the entry in $s(\cdot)$ as zero so that effectively in our notation we adopt the convention that $0/0 = 0$ by definition). 
$$
\small 
\begin{aligned} 
& \mathbb{E}\Big[1_x \beta_x^G(1 - \beta_x^G)\Big/\sum_{j = 1}^N A_{i,j} 1\{X_j = x\}\Big]  = \beta_x^G(1 - \beta_x^G) 
\mathbb{E}\Big[1_x\Big/\sum_{j = 1}^N A_{i,j} 1\{X_j = x\}\Big] \\ 
& \le \beta_x^G(1 - \beta_x^G)  P\Big(0 < \sum_{j = 1}^N A_{i,j} 1\{X_j = x\} < \kappa' \gamma_N^{1/4}\Big) +  \beta_x^G(1 - \beta_x^G) P\Big(\sum_{j = 1}^N A_{i,j} 1\{X_j = x\} \ge \kappa' \gamma_N^{1/4}\Big) \frac{1}{\kappa' \gamma_N^{1/4}} \\
& \le \beta_x^G(1 - \beta_x^G)  P\Big(\sum_{j = 1}^N A_{i,j} 1\{X_j = x\} < \kappa' \gamma_N^{1/4}\Big) + \frac{1}{\kappa' \gamma_N^{1/4}} .  
\end{aligned} 
$$ 
\vspace{-5mm} 
\paragraph{Final bound} Next, we derive a bound for $P\Big(\sum_{j = 1}^N A_{i,j} 1\{X_j = x\} < \kappa' \gamma_N^{1/4}\Big)$, since $\frac{1}{\kappa' \gamma_N^{1/4}} = o(1)$ as $\gamma_N \rightarrow \infty$. 
Define $h_x(X_i, U_i) = P(X = x)\int l(X_i, U_i, x, u)dF_{U | X = x}(u)$. Note that (for $i \neq j$)
$
\mathbb{E}[A_{i,j} 1\{X_j = x\}| X_i, U_i] = h_x(X_i, U_i) 1\{i \leftrightarrow j\}, 
$
since, conditional on $(X_i, U_i)$,the indicator $1\{i \leftrightarrow j\}$ is fixed (exogenous), and $(X_i, U_i) \sim_{i.i.d.} F_XF_{U|X}$. Also, recall that $\sum_j  1\{i \leftrightarrow j\} = \gamma_N^{1/2}$. Hence, only $\gamma_N^{1/2}$ many edges of $i$ can at most be non-zero, while the remaining ones are zero almost surely. Therefore, using Hoeffding's inequality \citep{wainwright2019high}, and using independence conditional on $X_i, U_i$,
\begin{equation} \label{eqn:helper1a}
\small 
\begin{aligned} 
P\Big(\Big|\frac{1}{\gamma_N^{1/2}} \sum_{j = 1}^N A_{i,j} 1\{X_j = x\} - h_x(X_i, U_i)\Big| \le \bar{C} \sqrt{\frac{\log(2 \gamma_N)}{\gamma_N^{1/2}}}\Big| X_i, U_i\Big) \ge 1 - 1/\gamma_N,  
\end{aligned}
\end{equation}  
for a finite constant $\bar{C} < \infty$. 
Observe that $h_x(X_i, U_i) \ge \kappa' > 0, \kappa' =P(X = x) \kappa$ almost surely by assumption. 
Define the event
$
\mathcal{E} =  \Big\{|\sum_{j = 1}^N A_{i,j} 1\{X_j = x\} - \gamma_N^{1/2} h_x(X_i, U_i)| \le \bar{C} \sqrt{\log(2 \gamma_N) \gamma_N^{1/2}} \Big\},  
$ and $\mathcal{E}^c$ its complement. 
we can write 
\begin{equation} \label{eqn:helper2}
\small 
\begin{aligned} 
 & P\Big(\sum_{j = 1}^N A_{i,j} 1\{X_j = x\} < \kappa' \gamma_N^{1/4}\Big)   =  P\Big(\sum_{j = 1}^N A_{i,j} 1\{X_j = x\} - \gamma_N^{1/2} h_x(X_i, U_i) + \gamma_N^{1/2} h_x(X_i, U_i) < \kappa' \gamma_N^{1/4}\Big) \\ &\le 
    P\Big( \gamma_N^{1/2} h_x(X_i, U_i) < \kappa' \gamma_N^{1/4} + |\sum_{j = 1}^N A_{i,j} 1\{X_j = x\} - \gamma_N^{1/2} h_x(X_i, U_i)| \Big)  \\
    &\le 
 P\Big( \gamma_N^{1/2} h_x(X_i, U_i) < \kappa \gamma_N^{1/4} + |\sum_{j = 1}^N A_{i,j} 1\{X_j = x\} - \gamma_N^{1/2} h_x(X_i, U_i)| \Big| \mathcal{E} \Big) \\ &\quad \quad + P\Big( \gamma_N^{1/2} h_x(X_i, U_i) < \kappa' \gamma_N^{1/4} + |\sum_{j = 1}^N A_{i,j} 1\{X_j = x\} - \gamma_N^{1/2} h(X_i, U_i)| \Big | \mathcal{E}^c \Big)  \times
 P\Big( \mathcal{E}^c \Big). 
 \end{aligned} 
\end{equation} 
Note that by Equation \eqref{eqn:helper1a} (which holds conditionally and so also unconditionally)
$$
\small 
\begin{aligned}
& P\Big( \gamma_N^{1/2} h_x(X_i, U_i) < \kappa' \gamma_N^{1/4} + |\sum_{j = 1}^N A_{i,j} 1\{X_j = x\} - \gamma_N^{1/2} h_x(X_i, U_i)| \Big | \mathcal{E}^c \Big)  \times
 P\Big( \mathcal{E}^c \Big) \le \frac{1}{\gamma_N} = o(1).   
\end{aligned} 
$$ 
Finally, we can write for a finite constant $\bar{C} < \infty$, 
 $$
 \small 
 \begin{aligned}
 & P\Big( \gamma_N^{1/2} h_x(X_i, U_i) < \kappa' \gamma_N^{1/4} + |\sum_{j = 1}^N A_{i,j} 1\{X_j = x\} - \gamma_N^{1/2} h_x(X_i, U_i)| \Big| \mathcal{E}\Big)  \\
 & \le P\Big( \gamma_N^{1/2} h_x(X_i, U_i) < \kappa' \gamma_N^{1/4} + \bar{C} \sqrt{\log(2 \gamma_N) \gamma_N^{1/2}} \Big| \mathcal{E}\Big) \le  1\Big\{\inf_{x,x', u'} h_x(x', u') < \kappa' \gamma_N^{-1/4} + \bar{C} \sqrt{\log(2 \gamma_N)}\gamma_N^{-1/4}\Big\}
 \end{aligned} 
 $$ 
 which equals to zero for $N, \gamma_N$ large enough, since $\inf_{x,x', u'} h_x(x', u') > 0$. Using a similar argument (which we omit for space constraints), it is easy to show that $p_x \rightarrow 1$ as $\gamma_N, N \rightarrow \infty$. 
\vspace{-5mm} 
\paragraph{Case where $\Delta(x) \neq v(x)$} Consider the case where $\Delta(x) \neq v(x)$. We have
$$
\small 
\begin{aligned} 
& W_N^* \le \sum_{x \in \mathcal{X}} \Big[\Delta(x) - v(x)\Big]_{+} P(X = x) + \frac{1}{N} \sum_{i=1}^N \sup_{\mathcal{P}} \mathbb{E}\Big[\mathbb{E}_{D \sim \mathcal{P}(A, X)}\Big[s\Big(\Big[\frac{\sum_{j = 1}^N A_{i,j} D_j 1\{X_j = x\}}{\sum_{j = 1}^N A_{i,j} 1\{X_j = x\}}\Big]_{x \in \{1, \cdots, |\mathcal{X}|\}}\Big) \Big| A, X\Big]\Big] \\ 
&W(\beta^*) \ge  \sum_{x \in \mathcal{X}} \Big[\Delta(x) - v(x)\Big]_{-} P(X = x) + \max_{\beta \in [0,1]^{|\mathcal{X}|}} \mathbb{E}_\beta\Big[s\Big(\frac{\sum_{j = 1}^N A_{i,j} D_j 1\{X_j = x\}}{\sum_{j = 1}^N A_{i,j} 1\{X_j = x\}}\Big]_{x \in \{1, \cdots, |\mathcal{X}|\}}\Big) \Big], 
\end{aligned} 
$$ 
where $[x]_{+} = \max\{0,x\}, [x]_{-} = \min\{0,x\}$. 
Note that $\sum_{x \in \mathcal{X}} \Big[\Delta(x) - v(x)\Big]_{+} P(X = x) - \sum_{x \in \mathcal{X}} \Big[\Delta(x) - v(x)\Big]_{-} P(X = x) = \mathbb{E}\Big[|\Delta(X) - v(X)|\Big]$. The rest of the proof follows as above, taking into account the additional term $\mathbb{E}\Big[|\Delta(X) - v(X)|\Big]$.

\subsection{Proof of Theorem \ref{thm:inference4}} \label{sec:proof10} Let $\tilde{K} = K/2p_1$. Take 
 $
 t_z^j = \frac{\frac{1}{\sqrt{z}} \sum_{i=1}^z X_i^j}{\sqrt{(z - 1)^{-1} \sum_{i=1}^z (X_i^j - \bar{X}^j)^2}}, X_i^j \sim \mathcal{N}(0, \sigma_i^j).  
 $
By Theorem 1 in \cite{ibragimov2010t}, we have that for $\alpha \le 0.08$
 $
 \sup_{\sigma_1, \cdots, \sigma_q} P(|t_z| \ge \mathrm{cv}_\alpha) = P(|T_{q-1}| \ge \mathrm{cv}_\alpha), 
 $ 
 where $\mathrm{cv}_\alpha$ is the critical value of a t-test with level $\alpha$, and $T_{z-1}$ is a t-student random variable with $z-1$ degrees of freedom. The equality is attained under homoskedastic variances \citep{ibragimov2010t}.   We now write 
 $$
 \small 
 \begin{aligned} 
 P\Big(\mathcal{T}_n \ge q | H_0\Big) &= 
 P\Big(\max_{j \in \{1, \cdots, l\}}|Q_{j,n}| \ge q | H_0\Big) = 1 -   P\Big(|Q_{j,n}| \le q \forall j| H_0\Big) = 1 -   \prod_{j=1}^{p_1} P\Big(|Q_{j,n}| \le q| H_0\Big),  
 \end{aligned} 
 $$
 where the last equality follows by between cluster independence.
 Observe now that by Theorem \ref{thm:inference1} and the fact that the rate of convergence is the same for all clusters (Assumption \ref{ass:var}), for all $j$, for some $(\sigma_1, \cdots, \sigma_{z})$, $z = \tilde{K}$, 
 $
\sup_q \Big| P\Big(|Q_{j,n}| \le q| H_0\Big) - P\Big(|t_{\tilde{K}}^j| \le q\Big)\Big| = o(1).  
 $
Using the result in \cite{ibragimov2010t}, we have 
$
\inf_{\sigma_1^j, \cdots, \sigma_{\tilde{K}}^j}  P\Big(|t_{\tilde{K}}^j| \le q\Big) = P\Big(|T_{\tilde{K} - 1}| \le q| H_0\Big). 
$
For size equal to $\alpha$, we obtain 
$
1 - P^{p_1}\Big(|T_{\tilde{K} - 1}| \le q\Big) = \alpha \Rightarrow  P\Big(|T_{\tilde{K} - 1}| \ge q\Big) = 1 - (1 - \alpha)^{1/p_1}. 
$
The proof completes after solving for $q$.

\subsection{Proof of Theorem \ref{thm:rate}} \label{sec:proof11} 

In this subsection, we derive the theorem for the gradient descent method under Assumption \ref{ass:quasi_convexity}. The derivation is split into the following lemmas.

\begin{defn}[Oracle gradient descent] \label{defn:oracle_gradient}
Fix \(v\in(0,1)\), \(J\in(0,1]\), and a positive constant \(\kappa\) defined in Lemma \ref{lem:gradient_descent3}. Set \(\beta_0^*=\beta_0\). For \(w=0,\ldots,\check T-1\), define
\begin{equation} \label{eqn:beta_start}
\small
\begin{aligned}
\beta_{w+1}^*
=
\begin{cases}
P_{\mathcal B_1,\mathcal B_2}
\left[
\beta_w^*
+
\displaystyle
\frac{J}{\check T^{1/2-v/2}}
\frac{M(\beta_w^*)}{\|M(\beta_w^*)\|_2}
\right],
&
\text{if }
\|M(\beta_w^*)\|_2^2>\mu^2\kappa\check T^{-1+v},
\\[1.2em]
\beta_w^*,
&
\text{otherwise.}
\end{cases}
\end{aligned}
\end{equation}
\end{defn}

\begin{lem}[Adaptive gradient descent for quasi-concave functions and locally strongly concave functions]
\label{lem:gradient_descent3}
Let \(\mathcal B\) be compact and define
\[
G
=
\max\left\{
\sup_{\beta,\beta'\in\mathcal B}\|\beta-\beta'\|_2^2,
1
\right\}.
\]
Let Assumptions \ref{ass:regularity_basic}, \ref{ass:bounded}, and \ref{ass:quasi_convexity} hold. Then there exists a positive finite constant \(\kappa\), defined in Equation \eqref{eqn:kappa}, such that, for any \(v\in(0,1)\) and for \(\check T\) sufficiently large with
$ 
\check T\ge \left(\frac{G+1}{J}\right)^{1/v},
$ 
the oracle sequence in Definition \ref{defn:oracle_gradient} satisfies
$ 
\|\beta_{\check T}^*-\beta^*\|_2^2
\le
\kappa\check T^{-1+v}.
$ 
\end{lem}

\begin{proof}[Proof of Lemma \ref{lem:gradient_descent3}]
The proof follows the normalized-gradient argument of \cite{hazan2015beyond}, adapted to obtain the explicit rate \(\check T^{-1+v}\). Let
$ 
\epsilon:=\check T^{-1+v}.
$ 
Since \(\beta^*\) is an interior optimum and \(W\) is locally strongly concave at \(\beta^*\), compactness of \(\mathcal B\), continuity of \(W\), and Assumption \ref{ass:regularity_basic} imply that there exist finite constants
\[
0<\lambda_{\min}\le \lambda_{\max}<\infty
\]
such that, for all \(\beta\in\mathcal B\),
\[
\lambda_{\min}\|\beta-\beta^*\|_2^2
\le
W(\beta^*)-W(\beta)
\le
\lambda_{\max}\|\beta-\beta^*\|_2^2 .
\]
Define
\begin{equation} \label{eqn:kappa}
\small
\begin{aligned}
\bar\kappa:=\frac{\lambda_{\max}}{\lambda_{\min}},
\qquad
\kappa:=4\left(\bar\kappa+J^2+1\right).
\end{aligned}
\end{equation}
Clearly, if the oracle algorithm stops at some \(w\), then
\[
\|M(\beta_w^*)\|_2^2
\le
\mu^2\kappa\epsilon.
\]
By Assumption \ref{ass:quasi_convexity}(B),
\[
\|\beta_w^*-\beta^*\|_2^2
\le
\kappa\epsilon.
\]
Because the same point is then repeated at all subsequent iterations, the conclusion follows. It remains to study the case in which the oracle algorithm updates.

Fix \(w\) such that the oracle algorithm updates, and write
$ 
\nabla_w:=M(\beta_w^*).
$ 
Suppose first that
\[
\|\beta_w^*-\beta^*\|_2^2>\bar\kappa\epsilon.
\]
For \(\check T\) large enough, the point
\[
\tilde\beta
:=
\beta^*
-
\sqrt{\epsilon}
\frac{\nabla_w}{\|\nabla_w\|_2}
\]
belongs to \(\mathcal B\), since \(\beta^*\) is an interior point of \(\mathcal B\). By the quadratic bounds above,
$ 
W(\beta^*)-W(\tilde\beta)
\le
\lambda_{\max}\epsilon,
$ 
whereas
$ 
W(\beta^*)-W(\beta_w^*)
>
\lambda_{\min}\bar\kappa\epsilon
=
\lambda_{\max}\epsilon.
$ 
Therefore \(W(\tilde\beta)>W(\beta_w^*)\). By Assumption \ref{ass:quasi_convexity}(A),
\[
\nabla_w^\top(\tilde\beta-\beta_w^*)\ge 0.
\]
Substituting the definition of \(\tilde\beta\) gives
\begin{equation} \label{eqn:jhg}
\nabla_w^\top(\beta^*-\beta_w^*)
\ge
\sqrt{\epsilon}\,\|\nabla_w\|_2 .
\end{equation}

Using the update rule and the fact that projection onto the convex set \(\mathcal B\) is non-expansive,
\[
\begin{aligned}
\|\beta^*-\beta_{w+1}^*\|_2^2
&\le
\left\|
\beta^*
-
\beta_w^*
-
\frac{J}{\check T^{1/2-v/2}}
\frac{\nabla_w}{\|\nabla_w\|_2}
\right\|_2^2
\\
&=
\|\beta^*-\beta_w^*\|_2^2
-
\frac{2J}{\check T^{1/2-v/2}}
\frac{\nabla_w^\top(\beta^*-\beta_w^*)}{\|\nabla_w\|_2}
+
\frac{J^2}{\check T^{1-v}} .
\end{aligned}
\]
Since \(\sqrt{\epsilon}=\check T^{-1/2+v/2}\), Equation \eqref{eqn:jhg} implies
\[
\|\beta^*-\beta_{w+1}^*\|_2^2
\le
\|\beta^*-\beta_w^*\|_2^2
-
(2J-J^2)\epsilon
\le
\|\beta^*-\beta_w^*\|_2^2
-
J\epsilon,
\]
where the last inequality uses \(J\le 1\).

Now suppose instead that
\[
\|\beta_w^*-\beta^*\|_2^2\le \bar\kappa\epsilon.
\]
If the algorithm stops, the conclusion follows from the preceding stopping argument. If it updates, then the update has length \(J\sqrt{\epsilon}\), and therefore
\[
\|\beta_{w+1}^*-\beta^*\|_2
\le
\|\beta_w^*-\beta^*\|_2
+
J\sqrt{\epsilon}
\le
(\sqrt{\bar\kappa}+J)\sqrt{\epsilon}.
\]
By the definition of \(\kappa\), this implies
\[
\|\beta_{w+1}^*-\beta^*\|_2^2
\le
\kappa\epsilon.
\]
Moreover, once the oracle sequence enters the ball
\[
\left\{\beta:\|\beta-\beta^*\|_2^2\le \kappa\epsilon\right\},
\]
it remains in this ball. To see this, suppose first that the algorithm stops at \(w\). Then
\[
\beta_{w+1}^*=\beta_w^*,
\]
and the claim is immediate. Suppose instead that the algorithm updates. If
\[
\|\beta_w^*-\beta^*\|_2^2\le \bar\kappa\epsilon,
\]
then the one-step bound above gives
\[
\|\beta_{w+1}^*-\beta^*\|_2^2\le \kappa\epsilon.
\]
If instead
\[
\bar\kappa\epsilon
<
\|\beta_w^*-\beta^*\|_2^2
\le
\kappa\epsilon,
\]
then the descent inequality applies and gives
\[
\|\beta_{w+1}^*-\beta^*\|_2^2
\le
\|\beta_w^*-\beta^*\|_2^2-J\epsilon
\le
\kappa\epsilon.
\]
Thus the ball of radius squared \(\kappa\epsilon\) is invariant for the oracle sequence.

It remains to show that the sequence enters that ball by time \(\check T\). Suppose not. Then, for every \(w=0,\ldots,\check T-1\),
\[
\|\beta_w^*-\beta^*\|_2^2>\kappa\epsilon\ge \bar\kappa\epsilon,
\]
so the descent inequality applies at every step. Therefore,
\[
\|\beta_{\check T}^*-\beta^*\|_2^2
\le
\|\beta_0-\beta^*\|_2^2
-
J\check T\epsilon
\le
G-J\check T^v.
\]
If
\[
\check T\ge \left(\frac{G+1}{J}\right)^{1/v},
\]
the right-hand side is negative, a contradiction. Therefore the sequence must enter the \(\kappa\epsilon\)-ball by time \(\check T\), and by the invariance argument above,
\[
\|\beta_{\check T}^*-\beta^*\|_2^2
\le
\kappa\epsilon
=
\kappa\check T^{-1+v}.
\]
This completes the proof.
\end{proof}

\begin{lem} \label{lem:fg}
Let Assumptions \ref{ass:ass_0}, \ref{ass:bounded}, and \ref{ass:quasi_convexity} hold. Assume that
\[
\small
\begin{aligned}
\epsilon_n
&\ge
\sqrt{p}
\left[
\bar C
\sqrt{
\gamma_N
\frac{\log(\gamma_N p \check T K/\delta)}
{\eta_n^2 n}
}
+\eta_n
\right],
\qquad
\frac{\mu\sqrt{\kappa}}{\check T^{1/2-v/2}}-2\epsilon_n>0
\end{aligned}
\]
for a finite constant \(\bar C>0\). Then, for any \(\delta\in(0,1)\), with probability at least \(1-\delta\), for any \(w\le \check T\),
\[
\small
\begin{aligned}
\text{either } \quad
(i)\quad
\left\|
\check\beta_k^w-\beta_w^*
\right\|_\infty
&=
\mathcal O\big(P_w(\delta)+p\eta_n\big),
\\
\text{or } \quad
(ii)\quad
\left\|
\check\beta_k^w-\beta^*
\right\|_2^2
&\le
\frac{p}{\check T^{1-v}} .
\end{aligned}
\]
Here \(P_1(\delta)=\mathrm{err}(\delta)\) and
\[
\small
\begin{aligned}
P_w(\delta)
=
P_{w-1}(\delta)
+
\frac{2\sqrt p}{\nu_n}Bp\frac{1}{\check T^{1/2-v/2}}P_{w-1}(\delta)
+
\frac{2\sqrt p}{\nu_n}\frac{1}{\check T^{1/2-v/2}}\mathrm{err}(\delta),
\end{aligned}
\]
for a finite constant \(B<\infty\), where
\[
\small
\begin{aligned}
\mathrm{err}(\delta)
\le
c_0
\left(
\sqrt{
\gamma_N
\frac{\log(\gamma_N p \check T K/\delta)}
{\eta_n^2 n}
}
+
p\eta_n
\right),
\end{aligned}
\]
with
\[
\small
\begin{aligned}
\nu_n
=
\frac{\mu\sqrt{\kappa}}{\check T^{1/2-v/2}}
-
2\epsilon_n,
\end{aligned}
\]
and a finite constant \(c_0<\infty\).
\end{lem}

\begin{proof}[Proof of Lemma \ref{lem:fg}]
Let
\[
\tau_{\check T}:=\frac{1}{\check T^{1/2-v/2}}.
\]
We prove the result on an event that has probability at least \(1-\delta\). By Lemmas \ref{lem:1a} and \ref{lem:2}, and by a union bound over \(j\in\{1,\ldots,p\}\), \(k\in\{1,\ldots,K\}\), and \(w\in\{1,\ldots,\check T\}\), with probability at least \(1-\delta\),
\[
\left|
\check M_{k,w}^{(j)}
-
M^{(j)}(\check\beta_{k+2}^w)
\right|
\le
c_0
\left(
\sqrt{
\gamma_N
\frac{\log(\gamma_N p\check T K/\delta)}
{\eta_n^2 n}
}
+\eta_n
\right)
\]
for all \(j,k,w\). Therefore, increasing \(c_0\) if necessary,
\[
\left\|
\check M_{k,w}
-
M(\check\beta_{k+2}^w)
\right\|_\infty
\le
\mathrm{err}(\delta),
\]
where
\[
\mathrm{err}(\delta)
\le
c_0
\left(
\sqrt{
\gamma_N
\frac{\log(\gamma_N p\check T K/\delta)}
{\eta_n^2 n}
}
+
p\eta_n
\right).
\]
We work on this event throughout the proof.

Let
\[
B:=
\sup_{\beta\in\mathcal B}
\max_{j,\ell\in\{1,\ldots,p\}}
\left|
\frac{\partial^2 W(\beta)}
{\partial\beta_j\partial\beta_\ell}
\right|<\infty .
\]
We use that the sample-size condition in Theorem \ref{thm:rate} implies, after increasing the constant in the lower bound for \(\epsilon_n\), that
\[
Bp\big\{P_{\check T}(\delta)+p\eta_n\big\}
\le
\epsilon_n,
\qquad
P_{\check T}(\delta)+p\eta_n
\le
c_0\tau_{\check T}.
\]
These inequalities are used only to keep the normalized-gradient denominator bounded away from zero.

We prove the claim by induction. The statement is immediate for \(w=1\), since
\[
\check\beta_k^1=\beta_1^*=\beta_0
\]
for all \(k\). Suppose the claim holds up to wave \(w\), with \(w<\check T\). We show it for wave \(w+1\).

First suppose that, for all \(k\),
\[
\left\|
\check\beta_k^w-\beta^*
\right\|_2^2
\le
c_0\frac{p}{\check T^{1-v}}.
\]
Then the conclusion also holds at wave \(w+1\). Indeed, if the algorithm stops, then
\[
\check\beta_k^{w+1}=\check\beta_k^w .
\]
If it does not stop, the update has length at most \(\tau_{\check T}\). Since \(\beta^*\) is in the interior of \(\mathcal B\), for \(n\) large enough it belongs to
\[
[\mathcal B_1+\eta_n,\mathcal B_2-\eta_n]^p,
\]
and projection onto this rectangle is non-expansive. Hence
\[
\left\|
\check\beta_k^{w+1}-\beta^*
\right\|_2
\le
\left\|
\check\beta_k^w-\beta^*
\right\|_2
+
\tau_{\check T}
\le
c_0\sqrt p\,\tau_{\check T},
\]
which implies
\[
\left\|
\check\beta_k^{w+1}-\beta^*
\right\|_2^2
\le
c_0\frac{p}{\check T^{1-v}}.
\]

It remains to consider the case where, at wave \(w\),
\[
\left\|
\check\beta_k^w-\beta_w^*
\right\|_\infty
\le
c_0\big\{P_w(\delta)+p\eta_n\big\}
\qquad\text{for all }k .
\]
For any \(k\), by the mean-value theorem and the preceding display,
\[
\begin{aligned}
\left\|
\check M_{k,w}
-
M(\beta_w^*)
\right\|_\infty
&\le
\left\|
\check M_{k,w}
-
M(\check\beta_{k+2}^w)
\right\|_\infty
+
\left\|
M(\check\beta_{k+2}^w)
-
M(\beta_w^*)
\right\|_\infty
\\
&\le
\mathrm{err}(\delta)
+
c_0Bp\big\{P_w(\delta)+p\eta_n\big\}.
\end{aligned}
\]
By the small-error condition above, this implies
\[
\left\|
\check M_{k,w}
-
M(\beta_w^*)
\right\|_\infty
\le
\mathrm{err}(\delta)+c_0BpP_w(\delta)+c_0\epsilon_n .
\]

If the stochastic update stops for some \(k\), then
\[
\left\|
\check M_{k,w}
\right\|_2
\le
\mu\sqrt{\kappa}\tau_{\check T}
+
\epsilon_n .
\]
Using the previous display and the fact that \(\epsilon_n\ge \sqrt p\,\mathrm{err}(\delta)\), we obtain
\[
\left\|
M(\beta_w^*)
\right\|_2
\le
\mu\sqrt{\kappa}\tau_{\check T}
+
c_0\epsilon_n .
\]
Increasing constants, Assumption \ref{ass:quasi_convexity}(B) gives
\[
\left\|
\beta_w^*-\beta^*
\right\|_2
\le
c_0\tau_{\check T}.
\]
Therefore, for every \(h\),
\[
\begin{aligned}
\left\|
\check\beta_h^w-\beta^*
\right\|_2
&\le
\left\|
\check\beta_h^w-\beta_w^*
\right\|_2
+
\left\|
\beta_w^*-\beta^*
\right\|_2
\\
&\le
c_0\sqrt p\,\big\{P_w(\delta)+p\eta_n\big\}
+
c_0\tau_{\check T}
\le
c_0\sqrt p\,\tau_{\check T}.
\end{aligned}
\]
As above, one more update has length at most \(\tau_{\check T}\), and hence
\[
\left\|
\check\beta_h^{w+1}-\beta^*
\right\|_2^2
\le
c_0\frac{p}{\check T^{1-v}}
\]
for all \(h\). This gives alternative (ii).

Now suppose the stochastic update does not stop. Then
\[
\left\|
\check M_{k,w}
\right\|_2
\ge
\mu\sqrt{\kappa}\tau_{\check T}
+
\epsilon_n .
\]
Using again the previous bounds,
\[
\begin{aligned}
\left\|
M(\beta_w^*)
\right\|_2
&\ge
\left\|
\check M_{k,w}
\right\|_2
-
\sqrt p\,
\left\|
\check M_{k,w}
-
M(\beta_w^*)
\right\|_\infty
\\
&\ge
\mu\sqrt{\kappa}\tau_{\check T}
-
2\epsilon_n
=:
\nu_n .
\end{aligned}
\]
Also, by the same inequality and the condition
\[
\mu\sqrt{\kappa}\tau_{\check T}-2\epsilon_n>0,
\]
we have
\[
\left\|
\check M_{k,w}
\right\|_2
\ge
\nu_n .
\]
Therefore,
\[
\begin{aligned}
\left\|
\frac{\check M_{k,w}}{\|\check M_{k,w}\|_2}
-
\frac{M(\beta_w^*)}{\|M(\beta_w^*)\|_2}
\right\|_\infty
&\le
\frac{2\sqrt p}{\nu_n}
\left\|
\check M_{k,w}
-
M(\beta_w^*)
\right\|_\infty
\\
&\le
\frac{2\sqrt p}{\nu_n}
\left\{
\mathrm{err}(\delta)
+
c_0BpP_w(\delta)
+
c_0p\eta_n
\right\}.
\end{aligned}
\]

The stochastic and oracle updates are
\[
\check\beta_k^{w+1}
=
P_{\mathcal B_1+\eta_n,\mathcal B_2-\eta_n}
\left[
\check\beta_k^w
+
\tau_{\check T}
\frac{\check M_{k,w}}{\|\check M_{k,w}\|_2}
\right],
\]
and
\[
\beta_{w+1}^*
=
P_{\mathcal B_1,\mathcal B_2}
\left[
\beta_w^*
+
\tau_{\check T}
\frac{M(\beta_w^*)}{\|M(\beta_w^*)\|_2}
\right],
\]
up to constants such as \(J\), which are absorbed into \(c_0\). Since projection on a rectangle is non-expansive in the sup norm, and since the two projection rectangles differ by at most \(\eta_n\) in each coordinate,
\[
\begin{aligned}
\left\|
\check\beta_k^{w+1}-\beta_{w+1}^*
\right\|_\infty
&\le
\left\|
\check\beta_k^w-\beta_w^*
\right\|_\infty
\\
&\quad
+
\tau_{\check T}
\left\|
\frac{\check M_{k,w}}{\|\check M_{k,w}\|_2}
-
\frac{M(\beta_w^*)}{\|M(\beta_w^*)\|_2}
\right\|_\infty
+
c_0p\eta_n
\\
&\le
c_0\big\{P_w(\delta)+p\eta_n\big\}
+
\frac{2\sqrt p}{\nu_n}
\tau_{\check T}
\left\{
\mathrm{err}(\delta)
+
c_0BpP_w(\delta)
\right\}
+
c_0p\eta_n .
\end{aligned}
\]
By the recursive definition of \(P_{w+1}(\delta)\),
\[
P_{w+1}(\delta)
=
P_w(\delta)
+
\frac{2\sqrt p}{\nu_n}Bp\tau_{\check T}P_w(\delta)
+
\frac{2\sqrt p}{\nu_n}\tau_{\check T}\mathrm{err}(\delta),
\]
and therefore
\[
\left\|
\check\beta_k^{w+1}-\beta_{w+1}^*
\right\|_\infty
\le
c_0\big\{P_{w+1}(\delta)+p\eta_n\big\}.
\]
This gives alternative (i) at wave \(w+1\).

The induction is complete.
\end{proof}

\begin{lem} \label{thm:regret1}
Let the conditions in Lemma \ref{lem:fg} hold. Then, with probability at least \(1-\delta\), for any \(k\in\{1,\ldots,K\}\), \(v\in(0,1)\), \(\delta\in(0,1)\), and \(\check T\ge \zeta^{1/v}\),
\[
\small
\begin{aligned}
\left\|
\beta^{*} - \check{\beta}_{k}^{\check T}
\right\|_2^2
\le
\frac{\kappa}{\check T^{1-v}}
+
\check T
\exp\left(
\frac{B p \sqrt p}{\mu\sqrt{\kappa}}\check T
\right)
c_0
\left(
\gamma_N
\frac{\log(p\gamma_N \check T K/\delta)}{\eta_n^2 n}
+
p^2\eta_n^2
\right),
\end{aligned}
\]
with \(0<\zeta,\kappa,B<\infty\) being constants independent of \((n,\check T)\), \(\epsilon_n\) as defined in Lemma \ref{lem:fg}, and a finite constant \(c_0<\infty\).
\end{lem}

\begin{proof}[Proof of Lemma \ref{thm:regret1}]
We invoke Lemma \ref{lem:fg}. Observe that we only have to check that the result holds for case (i) in Lemma \ref{lem:fg}, since otherwise the claim holds directly up to a change in constants. Using the triangle inequality, we can write
\[
\small
\begin{aligned}
\left\|
\beta^{*} - \check{\beta}_{k}^{\check T}
\right\|_2^2
&\le
2
\left\|
\beta^{*} - \beta_{\check T}^*
\right\|_2^2
+
2
\left\|
\check{\beta}_{k}^{\check T} - \beta_{\check T}^*
\right\|_2^2 .
\end{aligned}
\]
The first component on the right-hand side is bounded by Lemma \ref{lem:gradient_descent3}, with \(\check T\ge \zeta^{1/v}\), where \(\zeta\) is a constant defined in Lemma \ref{lem:gradient_descent3}. Renaming constants, this gives the first term in the statement.

Using Lemma \ref{lem:fg}, we bound with probability at least \(1-\delta\) the second component as follows:
\[
\small
\begin{aligned}
\left\|
\check{\beta}_{k}^{\check T} - \beta_{\check T}^*
\right\|_2^2
&\le
p
\left\|
\check{\beta}_{k}^{\check T} - \beta_{\check T}^*
\right\|_{\infty}^2
\\
&\le
p \times
\mathcal O\left(
\big(P_{\check T}(\delta)+p\eta_n\big)^2
\right).
\end{aligned}
\]
We conclude the proof by explicitly bounding \(P_{\check T}(\delta)\). From Lemma \ref{lem:fg}, for all \(1<w\le \check T\),
\[
\small
\begin{aligned}
P_w(\delta)
=
\left(
1+
\frac{2Bp\sqrt p}{\nu_n\check T^{1/2-v/2}}
\right)
P_{w-1}(\delta)
+
\frac{2\sqrt p}{\nu_n\check T^{1/2-v/2}}
\mathrm{err}(\delta),
\qquad
P_1(\delta)=\mathrm{err}(\delta),
\end{aligned}
\]
where
\[
\small
\begin{aligned}
\mathrm{err}(\delta)
\le
c_0
\left(
\sqrt{
\gamma_N
\frac{\log(p\gamma_N \check T K/\delta)}
{\eta_n^2 n}
}
+
p\eta_n
\right),
\end{aligned}
\]
and \(B<\infty\) denotes a finite constant. Using a recursive argument, we obtain
\[
\small
\begin{aligned}
P_w(\delta)
&\le
\mathrm{err}(\delta)
\sum_{s=1}^w
\left(
\frac{2\sqrt p}{\nu_n\check T^{1/2-v/2}}
\right)
\prod_{j=s}^w
\left(
\frac{2Bp\sqrt p}{\nu_n\check T^{1/2-v/2}}
+
1
\right)
\\
&\quad
+
\mathrm{err}(\delta)
\prod_{j=1}^w
\left(
\frac{2Bp\sqrt p}{\nu_n\check T^{1/2-v/2}}
+
1
\right).
\end{aligned}
\]
By Lemma \ref{lem:fg},
\[
\nu_n
=
\frac{\mu\sqrt{\kappa}}{\check T^{1/2-v/2}}
-
2\epsilon_n .
\]
Under the condition on \(\epsilon_n\) in Lemma \ref{lem:fg}, after decreasing constants if necessary,
\[
\nu_n
\ge
\frac{\mu\sqrt{\kappa}}{2\check T^{1/2-v/2}}.
\]
As a result,
\[
\small
\begin{aligned}
\frac{2Bp\sqrt p}{\nu_n\check T^{1/2-v/2}}
&\le
\frac{4Bp\sqrt p}{\mu\sqrt{\kappa}},
\\
\frac{2\sqrt p}{\nu_n\check T^{1/2-v/2}}
&\le
\frac{4\sqrt p}{\mu\sqrt{\kappa}}.
\end{aligned}
\]
Therefore,
\[
\small
\begin{aligned}
\prod_{j=s}^w
\left(
\frac{2Bp\sqrt p}{\nu_n\check T^{1/2-v/2}}
+
1
\right)
&\le
\exp\left(
\sum_{j=s}^w
\frac{4Bp\sqrt p}{\mu\sqrt{\kappa}}
\right)
\\
&\le
\exp\left(
\frac{4Bp\sqrt p}{\mu\sqrt{\kappa}}\check T
\right),
\end{aligned}
\]
since \(w\le \check T\). Hence, increasing constants if necessary,
\[
\small
\begin{aligned}
P_w(\delta)
&\le
c_0\check T
\exp\left(
\frac{4Bp\sqrt p}{\mu\sqrt{\kappa}}\check T
\right)
\mathrm{err}(\delta).
\end{aligned}
\]
Thus,
\[
\small
\begin{aligned}
p\big(P_{\check T}(\delta)+p\eta_n\big)^2
&\le
c_0
\check T
\exp\left(
\frac{Bp\sqrt p}{\mu\sqrt{\kappa}}\check T
\right)
\left(
\gamma_N
\frac{\log(p\gamma_N\check T K/\delta)}
{\eta_n^2 n}
+
p^2\eta_n^2
\right),
\end{aligned}
\]
where constants have been renamed and polynomial factors in \(\check T\) are absorbed into the exponential term. Combining the preceding displays completes the proof.
\end{proof}


 \begin{cor} Theorem \ref{thm:rate} holds. 
 \end{cor}

\begin{proof}
Consider Lemma \ref{lem:fg} with \(\delta=1/n\). By the condition in Lemma \ref{lem:fg}, we need
\[
\frac{\mu\sqrt{\kappa}}{\check T^{1/2-v/2}}-2\epsilon_n>0.
\]
This is satisfied, for \(n\) large enough, if
\[
\small
\begin{aligned}
\sqrt p
\left[
\bar C
\sqrt{
\gamma_N
\frac{\log(p\gamma_N \check T K n)}
{\eta_n^2 n}
}
+
\eta_n
\right]
\le
\frac{\mu\sqrt{\kappa}}{4\check T^{1/2-v/2}} .
\end{aligned}
\]
This condition is implied by the rate assumptions of the theorem. As a result,
\[
\nu_n
=
\frac{\mu\sqrt{\kappa}}{\check T^{1/2-v/2}}
-
2\epsilon_n
\ge
\frac{\mu\sqrt{\kappa}}{2\check T^{1/2-v/2}} .
\]

By Lemma \ref{thm:regret1}, under the stated rate assumptions, with probability at least \(1-1/n\), for all \(k\in\{1,\ldots,K\}\),
\[
\left\|
\check{\beta}_k^{\check T}-\beta^*
\right\|_2^2
\lesssim
\frac{p}{\check T^{1-v}}.
\]
Let
\[
\hat{\beta}^*
=
\frac{1}{K}\sum_{k=1}^K \check{\beta}_k^{\check T}.
\]
By Jensen's inequality,
\[
\begin{aligned}
\left\|
\hat{\beta}^*-\beta^*
\right\|_2^2
&=
\left\|
\frac{1}{K}\sum_{k=1}^K
\left(\check{\beta}_k^{\check T}-\beta^*\right)
\right\|_2^2
\\
&\le
\frac{1}{K}
\sum_{k=1}^K
\left\|
\check{\beta}_k^{\check T}-\beta^*
\right\|_2^2
\lesssim
\frac{p}{\check T^{1-v}}.
\end{aligned}
\]
Finally, by Assumption \ref{ass:regularity_basic}, since \(\beta^*\) is an interior maximizer and \(W\) has uniformly bounded second derivatives,
\[
W(\beta^*)-W(\hat{\beta}^*)
\lesssim
\left\|
\hat{\beta}^*-\beta^*
\right\|_2^2.
\]
Combining the previous displays completes the proof.
\end{proof}

\subsection{Proof of Theorem \ref{thm:non_separable}} \label{sec:proof12}

By Equation \eqref{eqn:non_s}, we can write 
$
\mathbb{E}[\hat{\Delta}_k(\beta)] = m(1, 1, \beta) - m(0, 1, \beta) + O(\eta_n^2).
$
Following the same strategy as in the proof of Theorem \ref{thm:bias_direct},  it is easy to show that
$$
\small 
\begin{aligned} 
\mathbb{E}[\hat{S}(0, \beta)] = \frac{\partial m(0, 1, \beta)}{\partial \beta} + \frac{1}{2} \Big[\alpha_{t, k} - \alpha_{t - 1, k} - \alpha_{t, k+1}  + \alpha_{t-1, k+1}\Big] + \mathcal{O}(\eta_n). 
\end{aligned} 
$$ 
Similarly, 
$
\mathbb{E}[\hat{S}(1, \beta)] = \frac{\partial m(1, 1, \beta)}{\partial \beta} + \frac{1}{2} \Big[\alpha_{t, k} - \alpha_{t - 1, k} - \alpha_{t, k+1}  + \alpha_{t-1, k+1}\Big] + \mathcal{O}(\eta_n). 
$ 
The proof completes because $\frac{\partial m(1, 1, \beta)}{\partial \beta} = 0$.

\subsection{Proof of Theorem \ref{thm:opt_dynamics}} \label{sec:proof9}
We bound
$$
\small 
\begin{aligned} 
\sup_{\theta \in \Theta} \widetilde{W}(\theta) - W(\widehat{\theta}) \le 2 \sum_t q^t \times \underbrace{\sup_{(\beta_1, \beta_2) \in [0,1]^2} \Big|\widehat{\Gamma}(\beta_2, \beta_1) - \Gamma(\beta_2, \beta_1)\Big|}_{(A)}. 
\end{aligned} 
$$  
We focus on bounding $(A)$ since $\sum_t q^t < \infty$. 
To bound $(A)$ observe first that each element in the grid $\mathcal{G}$ has a distance of order $1/\sqrt{K}$, since the grid has two dimensions and $K/3$ components. \
Let $||\beta - \beta^r||_2^2 = |\beta_1 - \beta_1^r|^2 + |\beta_2 - \beta_2^r|^2$, denoting the $l2$-norm and similarly $||\beta - \beta^r||_1$ denoting the $l1$-norm. 
For any element $(\beta_2, \beta_1)$, we can write 
$$
\small 
\begin{aligned} 
\Gamma(\beta_2, \beta_1) = \underbrace{\Gamma(\beta_2^r, \beta_1^r)}_{(B)} + \underbrace{\frac{\partial \Gamma(\beta_2^r, \beta_1^r) }{\partial \beta_1^r} (\beta_1 - \beta_1^r) +  \frac{\partial \Gamma(\beta_2^r, \beta_1^r) }{\partial \beta_2^r} (\beta_2 - \beta_2^r)}_{(C)} + \underbrace{\mathcal{O}\Big(||\beta - \beta^r||_2^2  \Big)}_{(D)}
\end{aligned} 
$$ 
where $\beta^r \in \mathcal{G}$ is some value in the grid such that the remainder term (D) is of order $1/K$. We can now write 
$$
\small 
\begin{aligned}
(A) & \le \underbrace{\sup_{(\beta_1^r, \beta_2^r) \in \mathcal{G}, ||\beta - \beta^r||^2 \lesssim 1/K} \Big|\widetilde{\Gamma}(\beta_2^r, \beta_1^r) - \Gamma(\beta_2^r, \beta_1^r)\Big|}_{(i)}  + \underbrace{ \sup_{(\beta_1^r, \beta_2^r) \in \mathcal{G}} \Big|\widehat{g}_2(\beta_2^r, \beta_1^r) - \frac{\partial \Gamma(\beta_2^r , \beta_1^r)}{\partial \beta_2^r}\Big| \Big(||\beta - \beta^r||_1\Big)}_{(ii)} \\ 
& + \underbrace{\sup_{(\beta_1^r, \beta_2^r) \in \mathcal{G}} \Big|\widehat{g}_1(\beta_2^r, \beta_1^r) - \frac{\partial \Gamma(\beta_2^r , \beta_1^r)}{\partial \beta_1^r}\Big| \Big(||\beta - \beta^r||_1\Big)}_{(iii)} + \mathcal{O}\Big(||\beta - \beta^r||_2^2\Big). 
\end{aligned} 
$$ 
We now study each component separately. We start from $(i)$. We observe that under Assumption \ref{ass:dynamic}, by doing a Taylor expansion around $(\beta_1^r, \beta_2^r)$, it follows 
$
\mathbb{E}[\bar{Y}_{t+1}^{(k)}] = \Gamma(\beta_2^r, \beta_1^r)+ \mathcal{O}(\eta_n). 
$ 
Therefore by Lemma \ref{lem:concentration1}, and the union bound over $K$ many elements in $\mathcal{G}$ as $\gamma_N \log(\gamma_N K)/n \rightarrow 0$, $(i) \rightarrow 0$. Consider now $(ii)$. We observe that since $\mathcal{B}$ is compact, we have $\Big(|\beta_2 - \beta_2^r| + |\beta_1 - \beta_1^r|\Big) = \mathcal{O}(1)$. In addition, similarly to what discussed in Lemma \ref{lem:2}, it follows that with probability at least $1 - \delta$, 
$
\Big|\widehat{g}_1(\beta_2^r, \beta_1^r) - \frac{\partial \Gamma(\beta_2^r , \beta_1^r)}{\partial \beta_1^r}\Big| \le c_0\Big(\sqrt{\frac{\gamma_N \log(\gamma_N/\delta)}{\eta_n^2 n}} + \eta_n \Big). 
$
Therefore, by the union bound as $\frac{\gamma_N \log(\gamma_N K)}{\eta_n^2 n} = o(1)$ $(ii) = o_p(1)$ and similarly $(iii)$. The proof concludes because  $|\beta_1^r - \beta_1|^2 + |\beta_2^r - \beta_2|^2 \lesssim 1/K$ by construction of the grid.

\subsection{Proof of Theorem \ref{thm:rate3b} } \label{sec:proof_last}

 The proof mimics the proof of Theorem \ref{thm:rate2b}. 
  
 Consider Lemma \ref{thm:regret2} where we choose $\delta = 1/n$. Note that we can directly apply Lemma \ref{thm:regret2} also to the gradient estimated with Algorithm \ref{alg:adaptive2}, since, by the circular-cross fitting argument, each parameter $\check{\beta}_k^w$ is estimated using sequentially pairs of different clusters as in   Algorithm \ref{alg:adaptive2}. The rest of the proof follows verbatim from the one of Theorem \ref{thm:rate2b}.

 \section{Numerical studies: main results} \label{app:more_sim}

\subsection{Design and overview of main simulation results} \label{app:experiment}

To evaluate the performance of our design with many waves, we calibrate simulations to data from \cite{cai2015social} and \cite{alatas2012targeting, alatas2016network}, while making simplifying assumptions whenever necessary. As in our application, we let $\beta$ denote the treatment probability and $\eta_n = 10\%$.\footnote{In the online supplement, we report results as we vary $\eta_n$ (Figure \ref{fig:comparison_eta}).} In the first calibration, the outcome is insurance adoption, and the treatment is whether an individual received an intensive information session. In the second calibration, the treatment is whether a household received a cash transfer, and the outcome is program satisfaction. The experiment of \cite{cai2015social} contains multiple arms. Here, we only focus on the treatment effects of intensive information sessions, pooling the remaining arms together for simplicity. The experiment of \cite{alatas2012targeting} contains different arms assigned at the village level, as well as information on cash transfers assigned at the household level. Here, we study the effect of cash transfers only and control for village-level treatments when estimating the parameters of interest.

In each cluster $k$, we generate
\begin{equation} \label{eqn:calibrated_model} 
\small 
\begin{aligned} 
Y_{i,t} = \phi_0 + \phi_1 D_{i,t} + \phi_2 S_{i,t} + \phi_3 S_{i,t}^{2} - c D_{i,t} + \eta_{i,t}, \quad S_{i,t} = \frac{\sum_{j\neq i} A_{i,j} D_{j,t} }{ \sum_{j \neq i} A_{i,j}}, \eta_{i,t} \sim_{i.i.d.} \mathcal{N}(0, \sigma^2), 
\end{aligned} 
\end{equation}  
where $c$ is the cost of the treatment. 
We consider two sets of parameters $\Big(\phi_0, \phi_1, \phi_2, \phi_3, \sigma^2\Big)$ calibrated to data from \cite{cai2015social} and \cite{alatas2012targeting, alatas2016network} respectively. We obtain information on neighbors' treatment directly from data from  \cite{cai2015social}. For the second application, we merge data from  \cite{alatas2012targeting}, and \cite{alatas2016network}, and use information from approximately 100 observations whose neighbors' treatments are all observable to estimate the parameters.\footnote{This approach introduces a sampling bias in the estimation procedure, which we ignore for simplicity, given that our goal is not the analysis of the original experiment but only calibrating numerical studies.} For either application, we estimate a linear model as in Equation \eqref{eqn:calibrated_model}, also controlling for additional covariates to guarantee the unconfoundedness of the treatment.\footnote{For \cite{cai2015social} the covariates are gender, age, rice area, literacy level, a coefficient that captures the risk aversion, the baseline disaster probability, education, and a dummy containing information on whether the individual has one to five friends. For \cite{alatas2012targeting}, we control for the education level, village-level treatments, i.e., how individuals have been targeted in a village (i.e., via a proxy variable for income, a community-based method, or a hybrid), the size of the village, the consumption level, the ranking of the individual poverty level, the gender, marital status, household size, the quality of the roof and top.}  For simplicity, we consider as cost of treatment $c = \phi_1$, i.e., the opportunity cost of allocating the treatment to a population of disconnected individuals.

 We generate $K$ clusters, each with $N = 600$ units, and sample $n \in \{200, 400, 600\}$. 
We generate a geometric network 
$
A_{i,j} = 1\Big\{||U_i - U_j||_1 \le  2 \rho/\sqrt{N}\Big\},  U_i \sim_{i.i.d.} \mathcal{N}(0, I_2), 
$ 
where the parameter $\rho$ governs the density of the network. The geometric formation process and the $1/\sqrt{N}$ follow similarly to simulations in \cite{leung2019treatment}. We report results for $\rho = 2$ in Table \ref{tab:cov02}, while results are robust as we increase $\rho$ (see Appendix \ref{app:more_sim}). Throughout the analysis, without loss, we report reward divided by its maximum $W(\beta^*)$ (i.e., $W(\beta^*) = 1$), and we subtract the intercept $\phi_0$.


In Appendix \ref{sec:designs2}, we study the performance of the one-wave experiment. We show that the proposed test controls size uniformly across specifications and present desirable properties for power. In Table \ref{tab:cov02} we present simulations for the multi-wave experiment. In the adaptive experiment, we choose an adaptive learning rate $10\%/\sqrt{t}$ with gradient norm rescaling as Remark \ref{rem:learning_rate}.\footnote{This choice guarantees that for each iteration, we only vary treatment probabilities by at most $10\%$, and the size of the variation is decreasing over each iteration, as for the learning rate under strong concavity without norm rescaling. This choice is preferable to $10\%/\sqrt{T}$ because it allows for larger steps in the initial iterations. A valid alternative is $10\%/t$. The latter case has a practical drawback: updates become very small after a few iterations. Comparisons for different learning rates are in the online supplement (Fig \ref{fig:learning_rate_comparisons}). } Since the model does not allow for time-varying fixed effects, we estimate marginal effects without baseline outcomes. For the multi-wave experiment, we initialize parameters at a small treatment probability $\beta = 0.2$ (here the optimum is around $60\%$).

 We let $T \in \{5, 10, 15, 20\}$. In Table \ref{tab:cov02}, we report the improvement of the proposed method with respect to a grid search method that samples observations from an equally spaced grid between $[0.1, 0.9]$ with a size equal to the number of clusters (i.e., $2T$). We consider the best competitor between the one that maximizes the estimated reward obtained from a correctly specified quadratic function and the one that chooses the treatment with the largest value within the grid. For both the competing methods, but not for the proposed procedure, we divide the outcomes' variance $\sigma^2$ by $T$, simulating settings where researchers may sample outcomes $T$ times (hence outcomes with a \textit{lower} variance) from each cluster before estimating treatment effects, and obtaining \textit{more precise} information.
The panel at the top of Table \ref{tab:cov02} reports the out-of-sample improvement in the target outcome. The improvement is positive, and up to three percentage points for targeting information and up to sixty percentage points for targeting cash transfers. Improvements are generally larger for larger $T$. The panel at the bottom of Table \ref{tab:cov02} reports positive and large improvements for the in-sample reward across all the designs, worst-case across clusters. For the worst-case regret, we fix the number of clusters to $K = 40$ for the proposed method and study the properties as a function of the number of iterations.  The improvements are twice as large for targeting information and thirty percentage points larger for targeting cash transfers. These are often increasing in $T$ with a few exceptions since uniform concentration may deteriorate for large $T$ and small $n$ as we consider the worst-case reward across clusters. 

In the online Appendices  \ref{sec:aa1}, \ref{sec:aa2}, we report results across many other specifications of the network, policy functions, and choice of different parameters and different starting values (e.g., also when $\beta$ is initialized near the optimum).

\subsection{Main additions} \label{sec:designs2}

Here, we study the properties of the one wave experiment as we vary the number of clusters $K$ and the sample size from each cluster $n$. We are interested in testing the one-sided null of whether we should increase the number of treated individuals to increase welfare, i.e., 
\begin{equation} \label{eqn:h0_experiment} 
\small 
\begin{aligned} 
H_0: \frac{\partial W(\beta)}{\partial \beta} \le 0, \quad H_1 = \frac{\partial W(\beta)}{\partial \beta} > 0 \quad \beta \in [0.1, \cdots, \beta^*].
\end{aligned} 
\end{equation} 
In Figure \ref{fig:power}, we report the power of the test as a function of the regret, where the test is computed using Theorem \ref{thm:inference4} through the pivotal test statistic with the critical value in Theorem \ref{thm:inference4} (both pivotal test statistic with $t$-student critical value and randomization tests control size; here we use the pivotal test statistic for computational advantages in simulations over randomization tests). Power is increasing in the regret, the number of clusters, and sample size. However, the marginal improvement in the power from twenty to thirty clusters is small. This result is suggestive of the benefit of the method even with few clusters and a small sample size.

In Table \ref{tab:cov01} we report the size of the test. 


\begin{table}[!htp]\centering
\caption{One wave experiment. $200$ replications. Coverage for testing $H_0$ (size is $5\%$). First panel corresponds to $\rho = 2$, and second panel to $\rho = 6$.
}    \label{tab:cov01} 
\ra{1.3}
\scalebox{0.7}{\begin{tabular}{@{}lrrrrrrrrrrrrrr@{}}\toprule
& \multicolumn{6}{c}{Information} & \ & \multicolumn{6}{c}{Cash Transfer} \\
\cmidrule{2-7}  \cmidrule{9-14}  
$K = $ &  10 & 20  & 30 & 40&  &   && & 10 & 20  & 30 & 40 & &  \\
\Xhline{.8pt} 
$n=200$ & $0.915$ & $0.945$ & $0.910$ & $0.900$ &  &   && & $0.915$ & $0.940$ & $0.920$ & $0.905$ \\ 
$n=400$ & $0.980$ & $0.960$ & $0.915$ & $0.930$ &  &   && & $0.980$ & $0.960$ & $0.905$ & $0.915$ \\ 
 $n=600$ & $0.980$ & $0.995$ & $0.975$ & $0.935$ &  &   && & $0.980$ & $0.995$ & $0.995$ & $0.930$ \\  
 \Xhline{.8pt}  \\ 
 \Xhline{.8pt} 
$n=200$ & $0.925$ & $0.945$ & $0.910$ & $0.900$ &  &   && & $0.910$ & $0.940$ & $0.915$ & $0.900$ \\ 
$n=400$ & $0.980$ & $0.960$ & $0.925$ & $0.930$ &  &   && & $0.980$ & $0.960$ & $0.900$ & $0.930$ \\ 
 $n=600$ & $0.970$ & $0.995$ & $0.970$ & $0.935$ &  &   && & $0.985$ & $0.995$ & $0.970$ & $0.930$ \\ 
 \bottomrule
\end{tabular}
}
\end{table}

 \begin{figure}[!ht]
 \centering 
\includegraphics[scale=0.5]{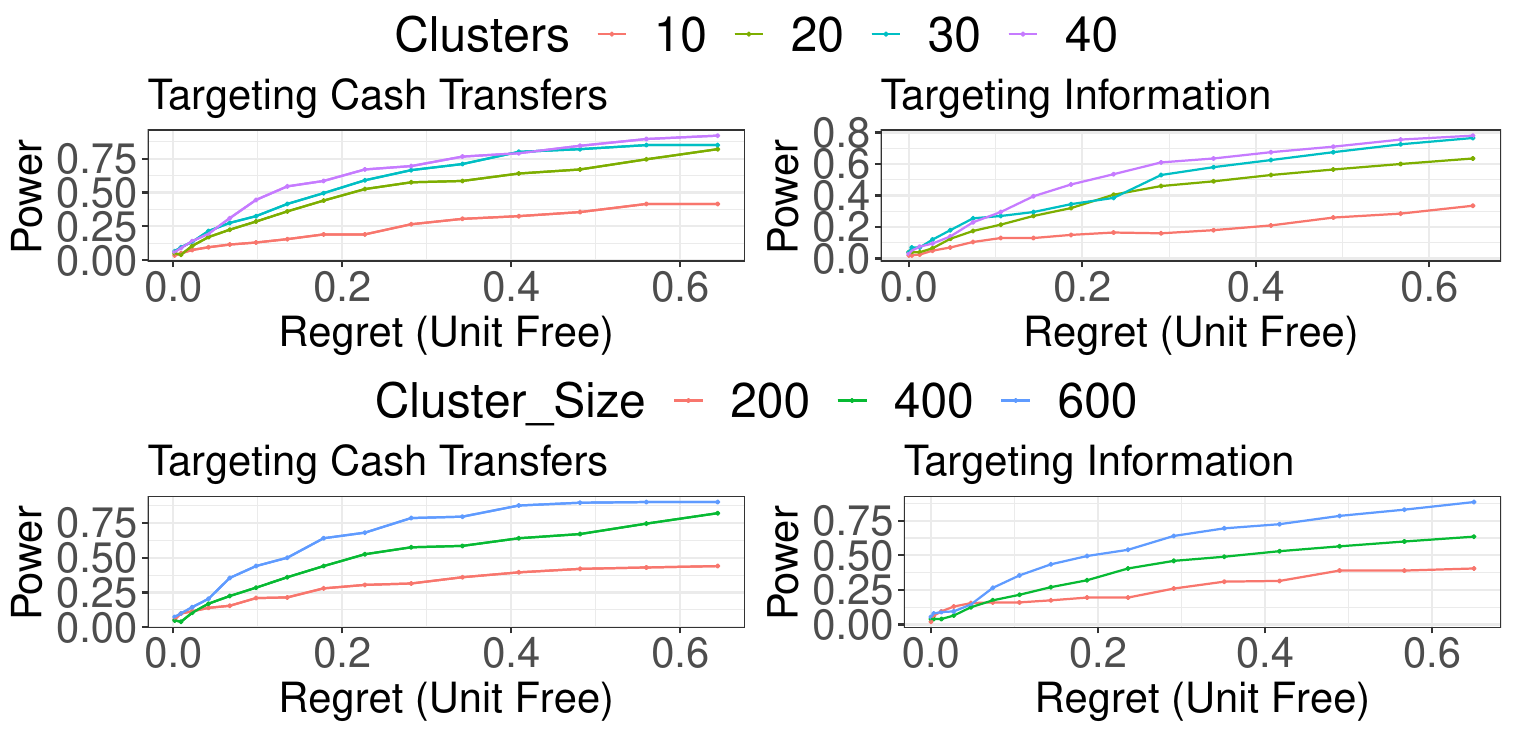}
\caption{One-wave experiment in Section \ref{sec:designs}. $200$ replications. Power plot for $\rho = 2$. The panels at the top fix $n = 400$ and varies $K$. The panels at the bottom fix $K = 20$ and vary $n$. } \label{fig:power}
\end{figure}

\begin{figure}[!ht]
\centering 
\includegraphics[scale=0.5]{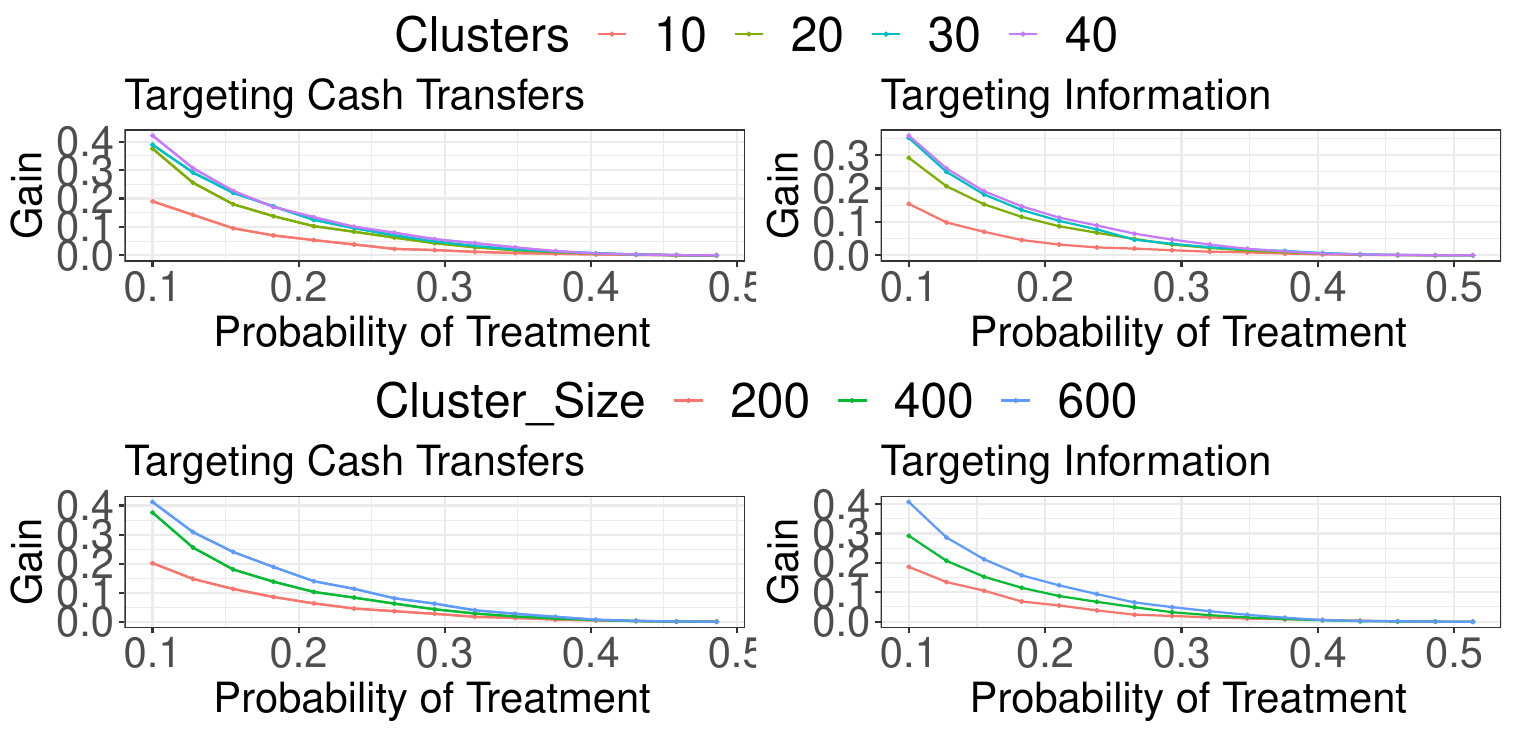}
\caption{One-wave experiment. $\rho = 2$. Expected percentage increase in welfare from increasing the probability of treatment $\beta$ by $5\%$ upon rejection of $H_0$. Here, the x-axis reports $\beta \in [0.1, \cdots, \beta^* - 0.05]$. The panels at the top fix $n = 400$ and vary the number of clusters. The panels at the bottom fix $K = 20$ and vary $n$.  } \label{fig:benefit} 
\end{figure}

\begin{figure}[!ht]
    \centering
    \includegraphics[scale = 0.3]{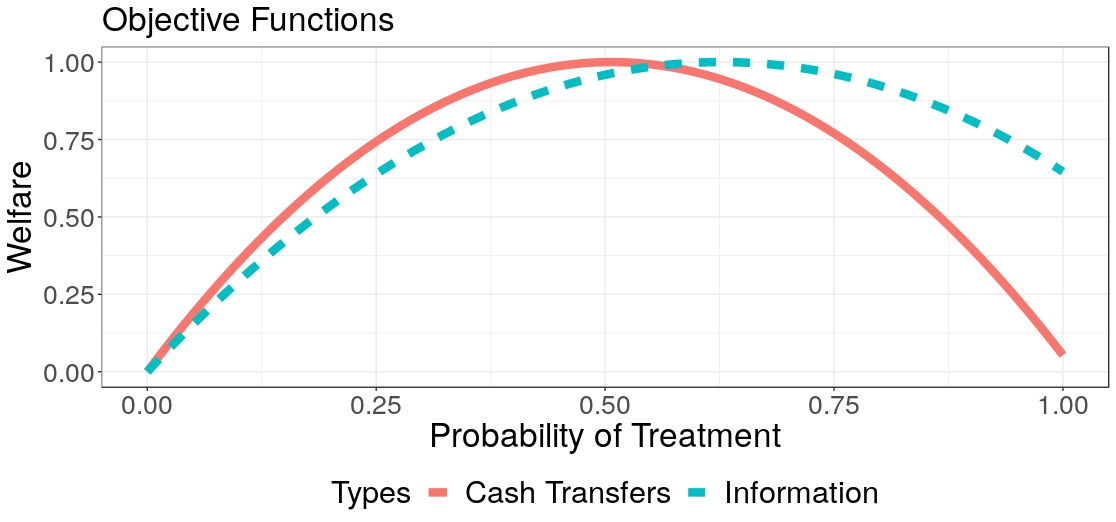}
    \caption{ Objective functions (rescaled by $W(\beta^*)$, and minus the intercept $\phi_0$ as functions of unconditional treatment probabilities, with cost of treatments $c = \phi_1$. The objectives are estimated using data  from \cite{alatas2012targeting} for the cash transfers and \cite{cai2015social} for the information campaign as described in Section \ref{sec:designs}. 
}
    \label{fig:objectives_analysis}
\end{figure}

 \subsection{One-wave experiment} \label{sec:aa1}
 
 In Figure \ref{fig:power2} we report the power plot for $\rho = 6$. In Figure \ref{fig:impr2} we report the welfare gain from increasing $\beta$ by $5\%$ upon rejection of $H_0$ for $\rho = 6$. In Figure \ref{fig:comparison_eta}, we report comparisons for different values of $\eta_n$.

\subsection{Multiple-wave experiment} \label{sec:aa2}

In Table \ref{tab:cov02b}, we provide comparison with competitors for $\rho = 6$. Results are robust as in the main text.

In Figure \ref{fig:learning_rate_comparisons} we report a comparison among different learning rates, which are the one which rescales by $1/t$, the one that rescales by $1/\sqrt{T}$ and the one that rescales by $1/\sqrt{t}$. 

In Figure \ref{fig:different_starting_values} we study the adaptive experiment as the starting value is the optimum minus $5\%$ and show that  the out-of-sample regret is small and close to zero. 


\subsection{Calibrated experiment with covariates} \label{sec:example_contd}

In this subsection, we turn to a calibrated experiment where we also control for covariates. We use data from \cite{alatas2012targeting, alatas2016network}. we estimate a function heterogenous in the distance of the household's village from the district's center. we use information from approximately four hundred observations, whose eighty percent or more neighbors are observed. We let $X_i \in \{0,1\}, X_i = 1$ if the household is far from the district's center than the median household, and estimate
\begin{equation} \label{eqn:dd} 
Y_i | X_i = x \quad =  \quad \phi_{0} + \tilde{X}_i \tau + D_i \phi_{1, x} + \frac{\sum_{j \neq i} A_{j,i} D_j}{\max\{\sum_{j \neq i} A_{j,i}, 1\}} \phi_{2, x} +\Big( \frac{\sum_{j \neq i} A_{j,i} D_j}{\max\{\sum_{j \neq i} A_{j,i}, 1\}}\Big)^2 \phi_{3, x} + \eta_i, 
\end{equation}   
where $\eta_i$ are unobservables centered on zero conditional on $X_i = x$, and $\tilde{X}_i$ denotes controls which also include $X_i$.\footnote{We also control for the education level, village-level treatments, i.e., how individuals have been targeted in a village (i.e., via a proxy variable for income, a community-based method, or a hybrid), the size of the village, the consumption level, the ranking of the individual poverty level, the gender, marital status, household size, the quality of the roof and top (which are indicators of poverty). }  Using the estimated parameter, we can then calibrate the simulations as follows.
 
We let $\eta_{i,t}, \sim \mathcal{N}(0, \sigma^2)$, where $\sigma^2$ is the residual variance from the regression. We then generate the network and the covariate as follows:
$$
A_{i,j} = 1\Big\{||U_i - U_j||_1 \le  2 \rho/\sqrt{N}\Big\}, \quad  U_i \sim_{i.i.d.} \mathcal{N}(0, I_2), \quad X_i = 1\{U_i^{(1)} > 0\}. 
$$
Here, $U_i^{(1)}$ is continuous and captures a measure of distances. Individuals are more likely to be friends if they have similar distances from the center, and $X_i$ is equal to one if an individual is far from the district's center from the median household. We fix $\rho = 1.5$ to guarantee that the objective's function optimum is approximately equal to the optimum observed from the data (in calibration, the optimum is $\beta \approx 0.26$, while $\beta^* \approx 0.29$ on the data). We then generate data 
\begin{equation} \label{eqn:heterogenous}
Y_{i,t} | X_i = x \quad =  \quad D_i \hat{\phi}_{1, x} + \frac{\sum_{j \neq i} A_{j,i} D_j}{\max\{\sum_{j \neq i} A_{j,i}, 1\}} \hat{\phi}_{2, x} +\Big( \frac{\sum_{j \neq i} A_{j,i} D_j}{\max\{\sum_{j \neq i} A_{j,i}, 1\}}\Big)^2 \hat{\phi}_{3, x} + \eta_{i,t}. 
\end{equation}   
where we removed covariates that did not interact with the treatment rule (i.e., do not affect welfare computations).
The policy function is 
$
\pi(x;\beta) = x \beta + (1 - x) (1 - \beta)
$
 where $\beta$ is the probability of treatment for individuals farer from the center. Here, we implicitely imposed a budget constraint $\beta P(X_i = 1) + (1 - \beta) P(X_i = 0) = 1/2$, where, by construction $P(X_i = 1) = 1/2$. 
 
 We collect results for the one-wave experiment in Figure \ref{fig:power_app}, \ref{fig:welfare_app} (left-panel), where we report power and the relative improvement from improving by $5\%$ the treatment probability for people in remote areas as discussed in the main text. Welfare improvements (and power) are increasing in the cluster size and the number of clusters. However, such improvements are negligible as we increase clusters from twenty to forty, suggesting that twenty clusters are sufficient to achieve the largest welfare effects. In the right-hand side panel of Figure \ref{fig:welfare_app} we report the out-of-sample regret. The regret is generally decreasing in the number of iterations, especially as the regret is further away from zero. As the regret gets almost zero (0.06$\%$), the regret oscillates around zero as the number of iterations increases due to sampling variation. This behavior is suggestive that for some applications, few iterations (in this case, ten) are sufficient to reach the optimum, up to a small error.  In Table \ref{tab:cov_app}, we observe perfect coverage for $n = 600$, and under-coverage by no more than five percentage points in the remaining cases.

\begin{figure}
\center
\includegraphics[scale=0.7]{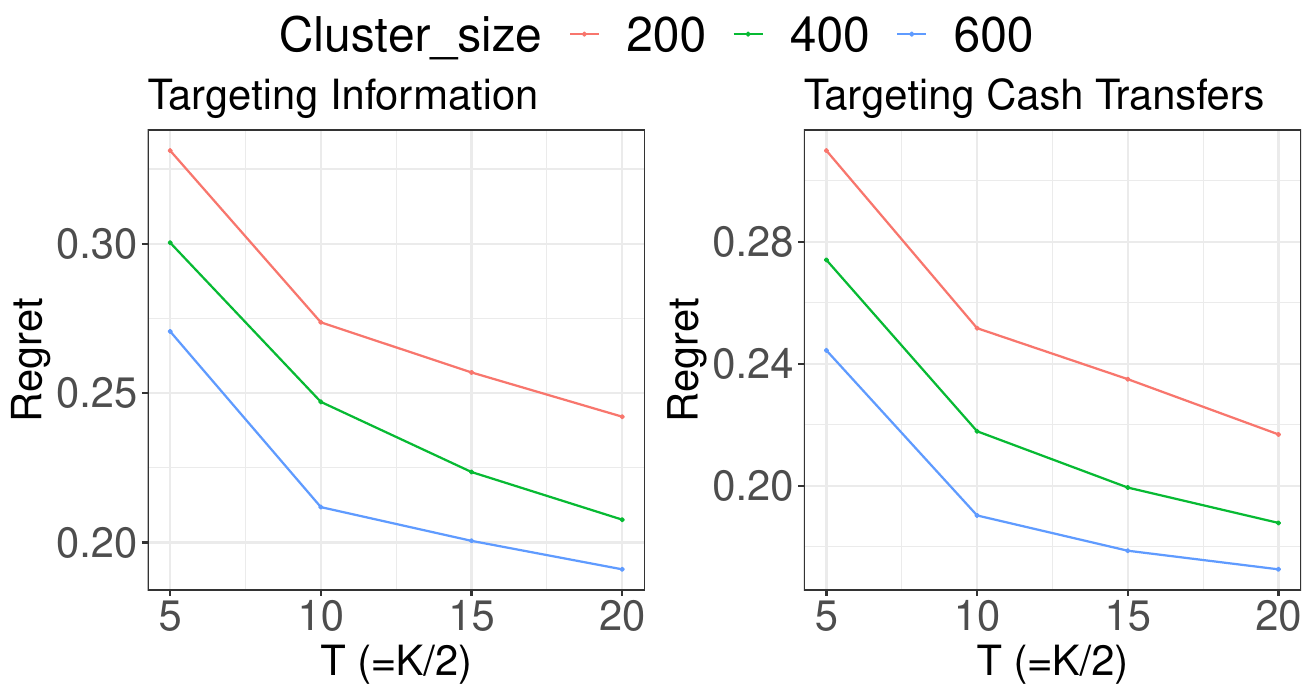}
\caption{Multi-wave experiment in Section \ref{sec:designs}. $200$ replications. In-sample regret, average across clusters, for $\rho = 2$.} \label{fig:in_sample}
\end{figure}

 \begin{figure}
 \centering 
\includegraphics[scale=0.7]{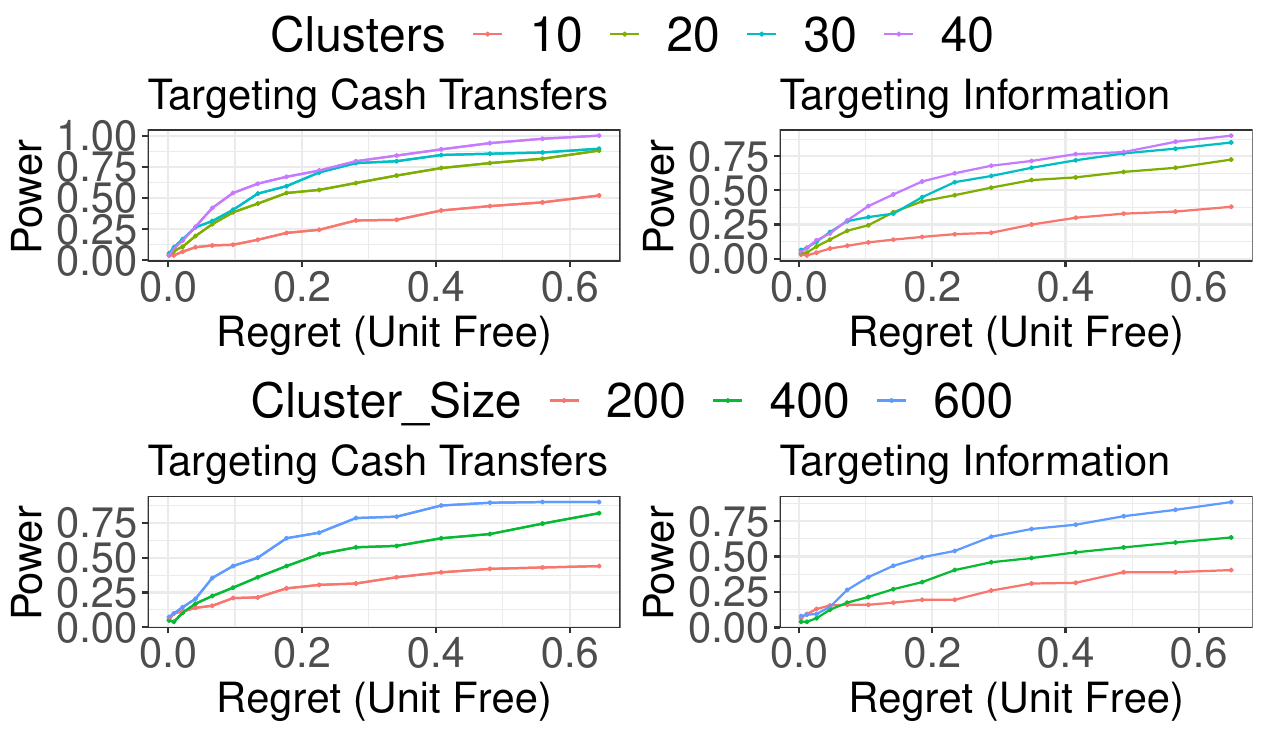}
\caption{One-wave experiment in Section \ref{sec:designs}. Power plot for $\rho = 6$. The panels at the top fix $n = 400$ and varies $K$. The panels at the bottom fix $K = 20$ and vary $n$.} \label{fig:power2}
\end{figure}

\begin{figure}
\centering 
\includegraphics[scale=0.7]{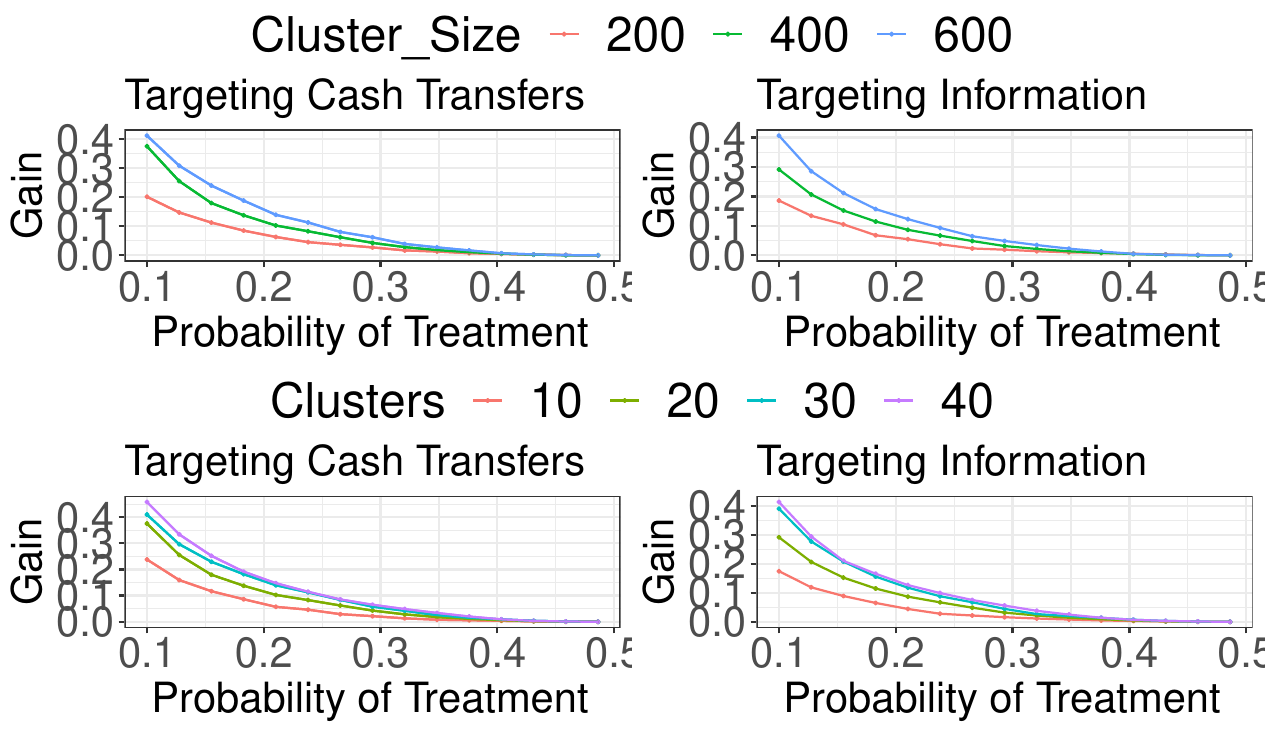}
\caption{One-wave experiment in Section \ref{sec:designs}. $\rho = 6$. Expected percentage increase in welfare from increasing the probability of treatment $\beta$ by $5\%$ upon rejection of $H_0$. Here, the x-axis reports $\beta \in [0.1, \cdots, \beta^* - 0.05]$. The panels at the top fix $n = 400$ and varies the number of clusters. The panels at the bottom fix $K = 20$ and vary $n$.} \label{fig:impr2}
\end{figure}

\begin{figure}
\center 
\includegraphics[scale=0.7]{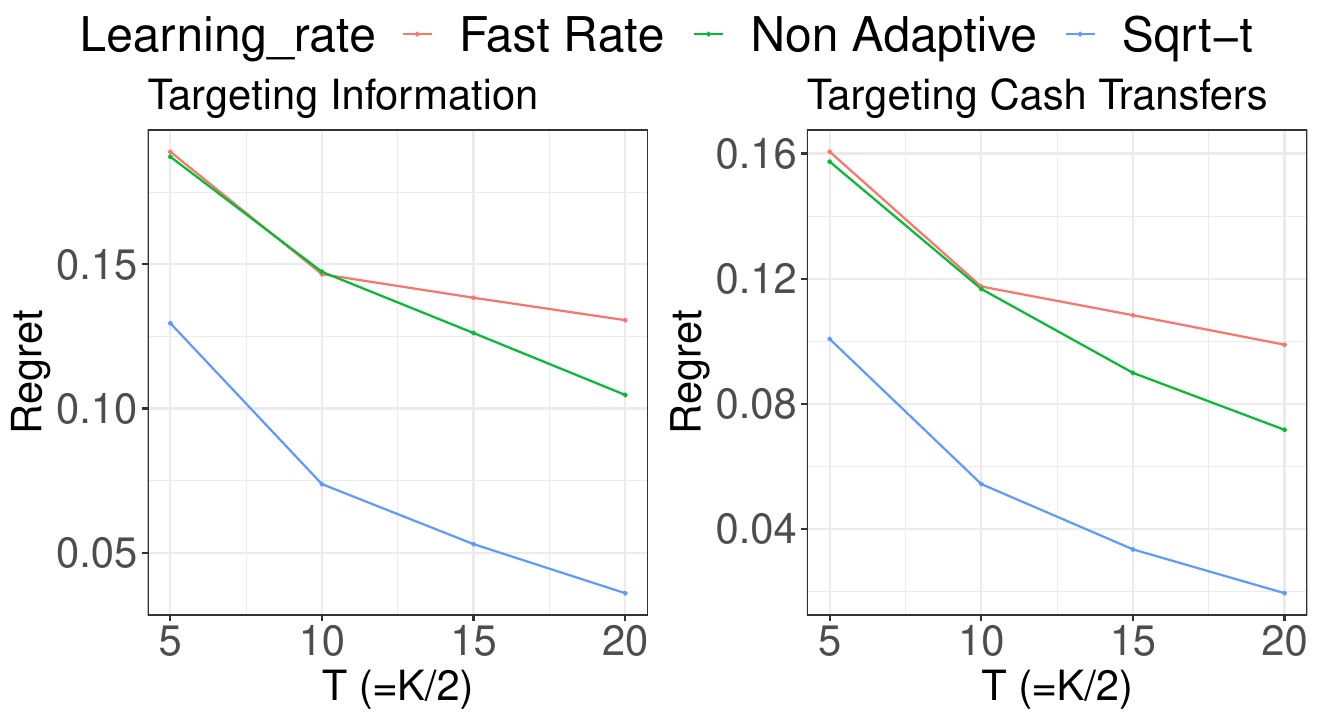}
\caption{Comparisons among different learning rates with experiment as in Section \ref{sec:main_design}. 200 replications, $\rho = 2, n = 600, K = 2T$. Fast rate denotes a rescaling of order $1/t$; non-adaptive depends on a rescaling of order $1/\sqrt{T}$; the last one (Sqrt-t) depends on a rescaling of order $1/\sqrt{t}$.} \label{fig:learning_rate_comparisons}
\end{figure}

\begin{table}[!ht]\centering
\caption{Multiple-wave experiment in Section \ref{sec:designs}. Relative improvement in welfare with respect to best competitor for $\rho = 6$. The panel at the top reports the out-of-sample regret and the one at the bottom the worst case in-sample regret across clusters. 
}    \label{tab:cov02b} 
\ra{1.3}
\begin{tabular}{@{}lrrrrrrrrrrrrrr@{}}\toprule
& \multicolumn{6}{c}{Information} & \ & \multicolumn{6}{c}{Cash Transfer} \\
\cmidrule{2-7}  \cmidrule{9-14}  
$T = $ &  5 & 10  & 15 & 20&  &   && & 5 & 10  & 15 & 20 & &  \\
\Xhline{.8pt} 
$n=200$ & $0.03$ & $0.105$ & $0.243$ & $0.156$ &  &   && & $0.233$ & $0.243$ & $0.264$ & $0.287$ \\ 
$n=400$ & $0.135$ & $0.130$ & $0.244$ & $0.258$ &  &   && & $0.243$ & $0.274$ & $0.321$ & $0.335$ \\ 
 $n=600$ & $0.217$ & $0.214$ & $0.281$ & $0.344$ &  &   && & $0.261$ & $0.313$ & $0.343$ & $0.360$\\  
 \Xhline{.8pt}  \\ 
 \Xhline{.8pt} 
$n=200$ & $0.587$ & $0.695$ & $0.670$ & $0.627$ &  &   && &  $0.247$ & $0.279$ & $0.300$ & $0.320$  \\ 
$n=400$ & $0.551$ & $0.667$ & $0.830$ & $0.869$ &  &   && &  $0.266$ & $0.306$ & $0.343$ & $0.352$ \\ 
 $n=600$ &  $0.589$ & $0.771$ & $0.897$ & $0.955$ &  &   && & $0.294$ & $0.360$ & $0.387$ & $0.387$ \\ 
 \bottomrule
\end{tabular}
\end{table}

 \begin{figure}
 \center
 \includegraphics[scale=0.6]{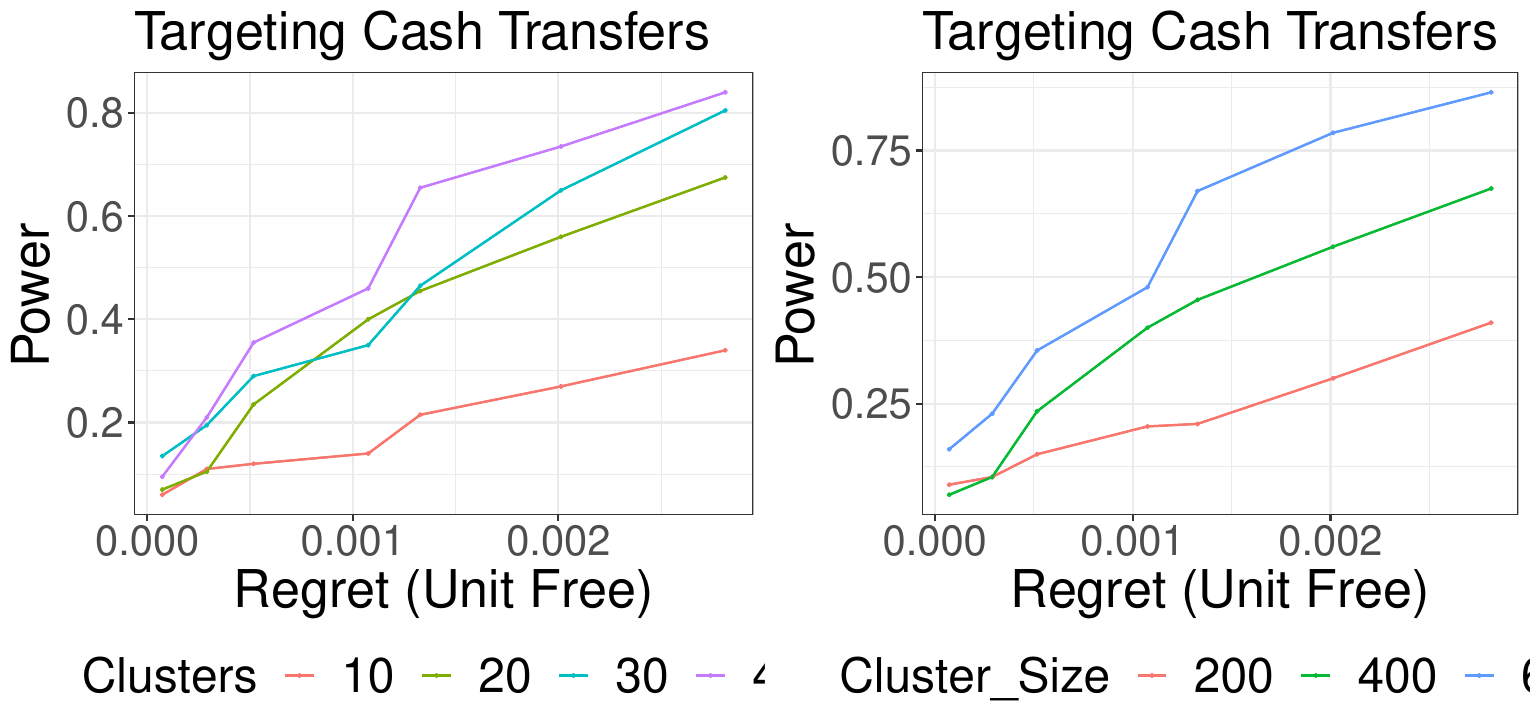}
 \caption{Single-wave experiment in Section \ref{sec:example_contd}. Power, 200 replications.}
 \label{fig:power_app}
 \end{figure}

 \begin{table}[!htbp] \centering 
  \caption{Single-wave experiment in Section \ref{sec:example_contd}, 200 replications. Coverage for tests with size $5\%$.} 
  \label{tab:cov_app} 
\begin{tabular}{@{\extracolsep{5pt}} ccccc} 
\\[-1.8ex]\hline 
$K = $ & $10$ & $20$ & $30$ & $40$ \\ 
\hline \\[-1.8ex] 
$n = 200$ & $0.955$ & $0.935$ & $0.900$ & $0.905$ \\ 
$n  = 400$ & $0.965$ & $0.945$ & $0.900$ & $0.950$ \\ 
$n  = 600$ & $0.935$ & $0.965$ & $0.920$ & $0.965$ \\ 
\hline \\[-1.8ex] 
\end{tabular} 
\end{table} 

 \begin{figure}
 \centering 
 \includegraphics[scale=0.7]{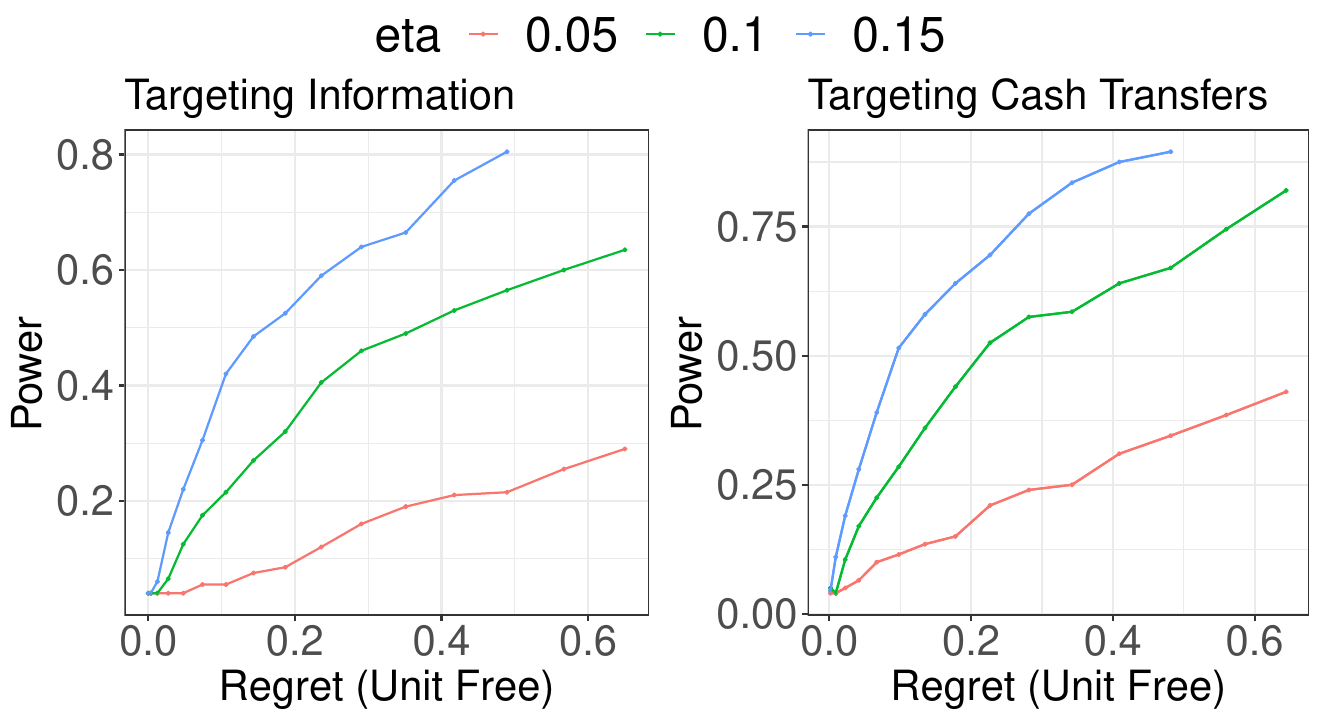}
 \caption{One wave experiment calibrated to \cite{alatas2012targeting} and \cite{cai2015social}. The plot reports power for different values of $\eta_n$ varies, with $K = 200, n = 400$, with 200 replications. } \label{fig:comparison_eta}
 \end{figure}
 
\begin{figure}
\center 
\includegraphics[scale=0.6]{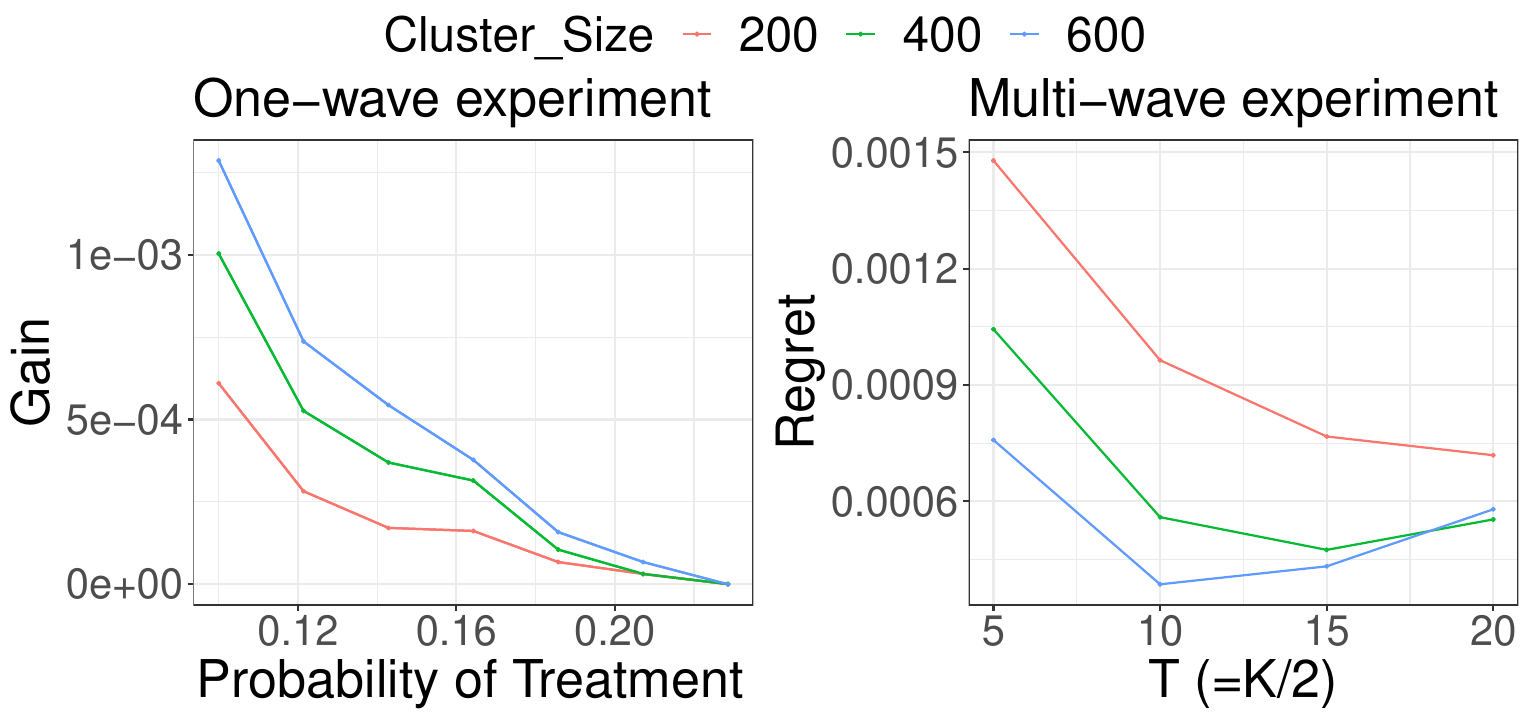}
\caption{Experiment in Section \ref{sec:example_contd}. Left-hand side panel reports the expected percentage increase in welfare from increasing the probability of treatment $\beta$ by $5\%$ to individuals in remote areas upon rejection of $H_0$. Here, the x-axis reports $\beta \in [0.1, \cdots, \beta^* - 0.05]$. The right-hand side panel reports the in-sample regret. 400 replications.  }
\label{fig:welfare_app} 
\end{figure}

 \begin{figure}
 \centering 
 \includegraphics[scale = 0.5]{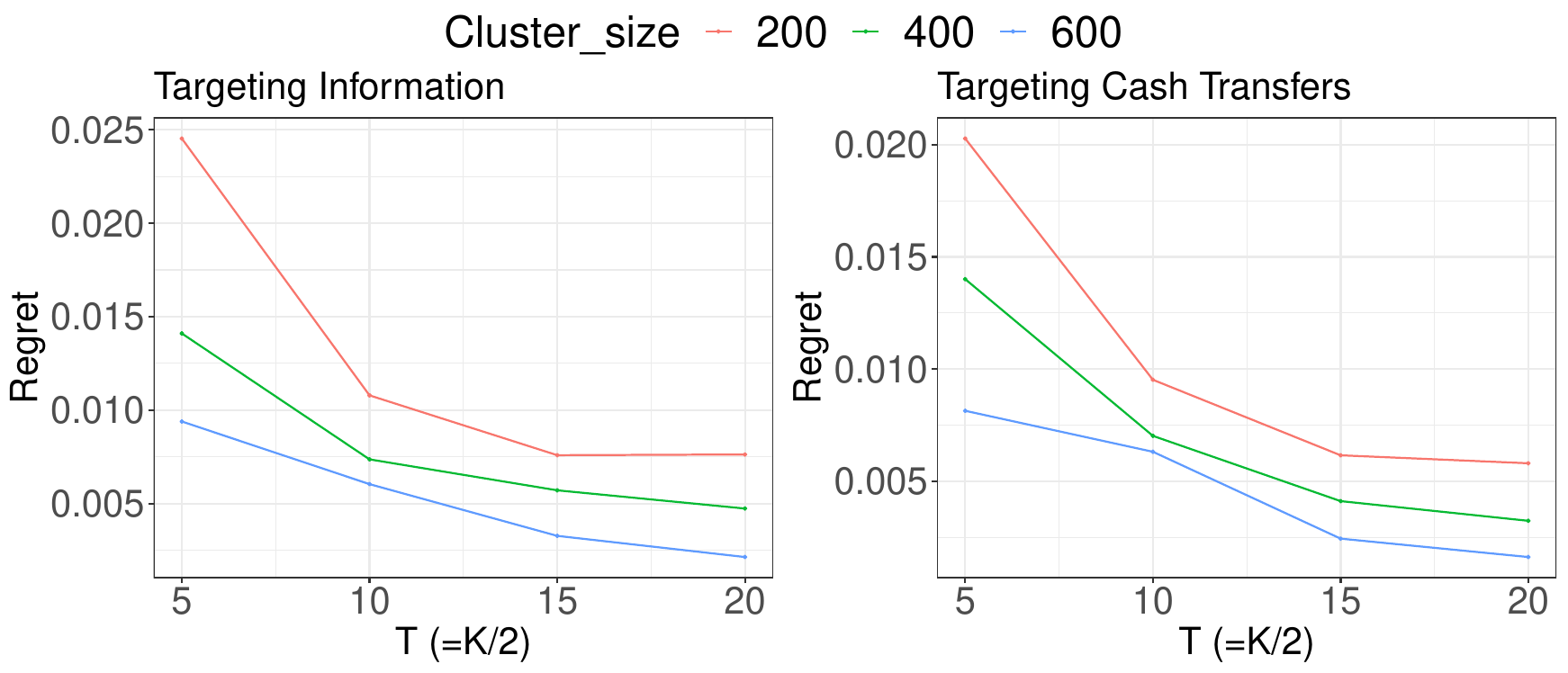}
 \caption{Multi-wave experiment wave experiment in Section \ref{sec:main_design}, as $\beta$ is initialized at the optimum value minus $5\%$. Reported in the figure is the out-of-sample welfare. 200 replications.}
 \label{fig:different_starting_values}
 \end{figure}

\begin{figure}[!htp]
\centering 
\includegraphics[scale=0.5]{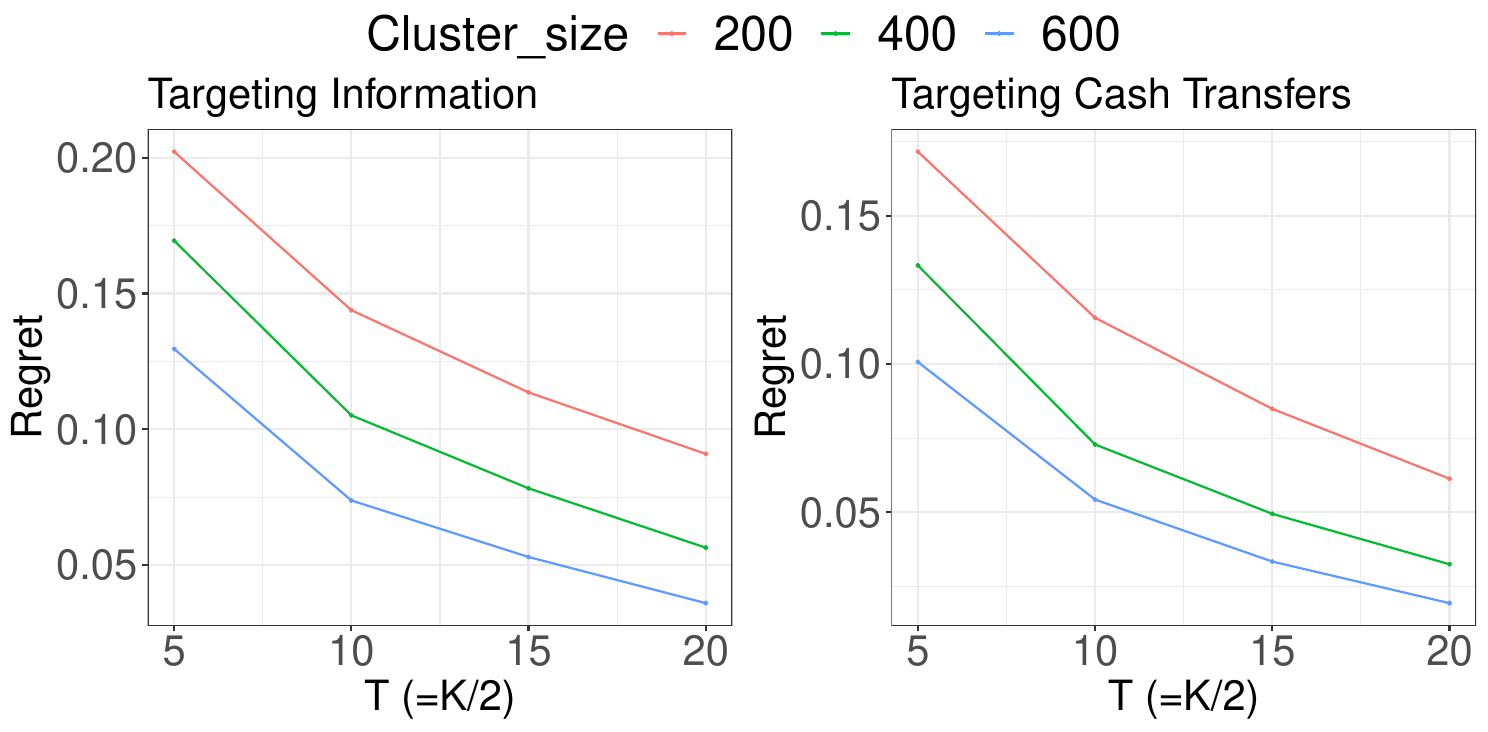}
\caption{Adaptive experiment $\rho = 2$. $200$ replications. The panel reports the out-of-sample regret of the method as a function of the number of iterations. } \label{fig:oos}
\end{figure}

\end{appendices}

 \end{document}